\documentclass[]{article}

\usepackage{microtype}
\usepackage{listings}
\usepackage{nicefrac}
\usepackage{xspace}
\usepackage{paralist}
\usepackage{amsmath, amsthm, amssymb}
\usepackage{color}
\usepackage[left=2.25cm,right=2.25cm,top=3cm,bottom=3cm]{geometry}
\usepackage{hyperref}
\usepackage{natbib}

\usepackage{SessionTypesForLinkFailures}

\title{Session Types for Link Failures (Technical Report)}

\author{Manuel Adameit
	\qquad\qquad Kirstin Peters
	\qquad\qquad Uwe Nestmann
	\vspace{0.5em}\\
	TU Berlin, Germany
}

\begin{document}

\maketitle

\begin{abstract}
	We strive to use session type technology to prove behavioural properties of fault-tolerant distributed algorithms.
	Session types are designed to abstractly capture the structure of (even multi-party) communication protocols. The goal of session types is the analysis and verification of the protocols' behavioural properties. One important such property is progress, \ie the absence of (unintended) deadlock.
	Distributed algorithms often resemble (compositions of) multi-party communication protocols. In contrast to protocols that are typically studied with session types, they are often designed to cope with system failures. An essential behavioural property is (successful) termination, despite failures, but it is often elaborate to prove for distributed algorithms.

	We extend multi-party session types (and multi-party session types with nested sessions) with optional blocks that cover a limited class of link failures. This allows us to automatically derive termination of distributed algorithms that come within these limits. To illustrate our approach, we prove termination for an implementation of the ``rotating coordinator'' Consensus algorithm.
	This paper is an extended version of \cite{adameitPetersNestmann17}.
\end{abstract}

\section{Introduction}

Session types are used to statically ensure correctly coordinated behaviour in systems without global control. One important such property is progress, \ie the absence of (unintended) deadlock. Like with every other static typing approach to guarantee behavioural properties, the main advantage is that the respective properties are then provable without unrolling the process, \ie without computing its executions. Thereby, the state explosion problem is avoided.
Hence, after the often elaborate task of establishing a type system, they allow to prove properties of processes in a quite efficient way.
 
Session types describe global behaviours of a system---or protocols---as \emph{sessions}, \ie units of conversations. The participants of such sessions are called \emph{roles}.
\emph{Global types} specify protocols from a global point of view, whereas \emph{local types} describe the behaviour of individual roles within a protocol.
\emph{Projection} ensures that a global type and its local types are consistent.
These types are used to reason about processes formulated in a corresponding \emph{session calculus}. Most of the existing session calculi are extensions of the well-known $ \pi $-calculus \cite{milnerParrowWalker92} with specific operators adapted to correlate with local types.
Session types are designed to abstractly capture the structure of (even multi-party) communication protocols \cite{BettiniAtall08,BocciAtall10}.
The literature on session types provides a rich variety of extensions.
Session types with \emph{nested} protocols were introduced by \cite{DemangeonHonda12} as an extension of multi-party session types as defined \eg in \cite{BettiniAtall08,BocciAtall10}.
They offer the possibility to define sub-protocols independently of their parent protocols.

It is essentially the notion of nested protocols that led us to believe that session types could be applied to capture properties of distributed algorithms, especially the so-called round-based distributed algorithms. The latter are typically structured by a repeated execution of communication patterns by $n$~distributed partners. Often, like it will also be in our running example, such a pattern involves an exposed coordinator role, whose incarnation may differ from round to round.
As such, distributed algorithms very much resemble compositions of nested multi-party communication protocols.
Moreover, an essential behavioural property of distributed algorithms is (successful) termination \cite{Tel94,Lynch96}, despite failures, but it is often elaborate to prove. It turns out that progress (as provided by session types) and termination (as required by distributed algorithms) are closely related.
For these reasons, our goal is to apply session type technology to prove behavioural properties of distributed algorithms.

Particularly interesting round-based distributed algorithms were designed in a fault-tolerant way, in order to work in a model where they have to cope with system failures---be it links dropping or manipulating messages, or processes crashing with or without recovery.
As the current session type systems are not able to cover fault-tolerance (except for exception handling as in \cite{CarboneHondaYoshida08,capecchi2016}), it is necessary to add an appropriate mechanism to cover system failures.

\paragraph{Optional Blocks.}
While the detection of conceptual design errors is a standard property of type systems, proving correctness of algorithms despite the occurrence of uncontrollable system failures is not.
In the context of distributed algorithms, various kinds of failures have been studied. Often, the correctness of an algorithm does not only depend on the kinds of failures but also of the phase of the algorithm in which they occur, the number of failures, or their likelihood.
Here, we only consider a very simple case, namely algorithms that terminate despite arbitrarily many link failures that may occur at any moment in the execution of the algorithm.

Therefore, we extend session types with \emph{optional blocks}. Such a block specifies chunks of communication that may at some point fail due to a link failure.
This partial communication protocol is protected by the optional block, to ensure that no other process can interfere before the optional block was resolved and to ensure, that in the case of failure, no parts of the failed communication attempt may influence the further behaviour.
In case a link fails, the ambition to guarantee progress requires from our mechanism that the continuation behaviour is not blocked.
Therefore, the continuation of an optional block $ C $ can be parametrised by a set of values that are either computed by a successful termination of an optional block or are provided beforehand as default values, \ie we require that for each value that $ C $ uses the optional block specifies a default value.
An optional block can cover parts of a protocol or even other optional blocks.
The type system ensures that communication with optional blocks requires an optional block as communication partner and that only a successful termination of an optional block releases the protection around its values derived within optional blocks.
The semantics of the session calculus then allows us to abort an unguarded optional block at any point. If an optional block models a single communication, its abortion represents a message loss.
In summary, optional blocks will allow us to automatically derive termination despite arbitrary link failures of distributed algorithms.

\paragraph{Running Example.}
Fault-tolerant Consensus algorithms are used to solve the problem of reaching agreement on a decision in the presence of faulty processes or otherwise unreliable systems.
A simple, but prominent example is the rotating coordinator algorithm \cite{Tel94} to be used in an asynchronous message-passing system model.
We use this algorithm as running example throughout this paper.

\begin{example}[The Rotating Coordinator Algorithm]
	\label{exa:RCAlgorithm}
	$ $
	\begin{lstlisting}
*$ x_i :=\! $ input;
*for $ r := 1 $ to $ n $ do
*\{
  *if $ r = i $ then broadcast($ x_i $);
  *else if alive($ p_r $) then $ x_i :=\! $ input\_from\_broadcast()
*\};
*output $ x_i $;
	\end{lstlisting}
\end{example}

\noindent
The above example describes the rotating coordinator algorithm for participant~$ i $. A network then consists of $ n $ such participants composed in parallel that try to reach Consensus on a value $ x $. Each participant receives an initial value for $ x $ from the environment, then performs the $ n $ rounds of the algorithm described in the Lines 2--6, and finally outputs its decision value. Within the $ n $ rounds each participant $i$ is exactly once (if $ r = i $) the coordinator and broadcasts its current value to all other participants. In the remaining rounds it receives the values broadcasted by other participants and replaces its own value with the received value. 

Due to the asynchronous nature of the underlying system, messages do not necessarily arrive in the order in which they were sent; also, different participants can be in different (local) rounds at the same (global) time. In case of link failure, the system may lose an arbitrary number of messages. In case of process crash, only its messages that are still in transition may possibly be received. In a system with only crash failures but no link failures, messages that are directed from a non-failing participant to some other non-failing participants are eventually received; this is often called reliable point-to-point communication. 

Following  \cite{Tel94, Lynch96}, a network of processes like above solves Consensus if
\begin{inparaenum}[(1)]
	\item all non-failing participants eventually output their value (\emph{Termination}),
	\item all emitted output values are the same (\emph{Agreement}), and
	\item each output value is an initial value of some participant (\emph{Validity}).
\end{inparaenum}
We show that an implementation of the above algorithm reaches termination despite an arbitrary number of link failures, \ie all participants (regardless of whether their messages are lost) eventually terminate. Note that we implement broadcast by a number of binary communications---one from the sender to each receiver---and thus a link failure not necessary implies the loss of all messages of a broadcast. This way we consider the more general and more realistic case of such failures.
Moreover we concentrate our attention to the main part of the algorithm, \ie the implementation of the rounds in the Lines 2--6.

\paragraph{Related Work.}
Type systems are usually designed for scenarios that are free of system failures.
An exception is \cite{KouzapasGutkovasGay14} that introduces unreliable broadcast. Within such an unreliable broadcast a transmission can be received by multiple receivers but not necessarily all available receivers. In the latter case, the receiver is deadlocked. In contrast, we consider failure-tolerant unicast, \ie communications between a single sender and a single receiver, where in the case of a failure the receiver is not deadlocked but continues using default values.

\cite{CarboneHondaYoshida08,capecchi2016} extends session types with exceptions thrown by processes within \textsc{try}-and-\textsc{catch}-blocks.
Both concepts---\textsc{try}-and-\textsc{catch}-blocks and optional blocks---introduce a way to structurally and semantically encapsulate an unreliable part of a protocol and provide some means to 'detect' a failure and 'react' to it. They are, however, conceptionally and technically different.
An obvious difference is the limitation of the inner part of optional blocks towards the computation of values; there is no such limitation in the \textsc{try}-and-\textsc{catch}-blocks of \cite{capecchi2016}.
More fundamentally these approaches differ in the way they allow to 'detect' failures and to 'react' to them.

Optional blocks are designed for the case of system errors that may occur non-deterministically and not necessarily reach the whole system or not even all participants of an optional block, whereas \textsc{try}-and-\textsc{catch}-blocks model controlled interruption requested by a participant. Hence these approaches differ in the source of an error; raised by the underlying system structure or by a participant. Technically this means that in the presented case failures are introduced by the semantics, whereas in \cite{capecchi2016} failures are modelled explicitly as \textsc{throw}-operations. In the latter case the model also describes, why a failure occurred. Here we deliberately do not model causes of failures, but let them occur non-deterministically.
In particular we do not specify, how a participant 'detects' a failure. Different system architectures might provide different mechanisms to do so, \eg by time-outs. As it is the standard for the analysis of distributed algorithms, our approach allows to port the verified algorithms on different systems architectures, provided that the respective structure and its failure pattern preserves correctness of the considered properties.

The main difference between these two approaches is how they react to failures. In \cite{capecchi2016} \textsc{throw}-messages are propagated among nested \textsc{try}-and-\textsc{catch}-blocks to ensure that all participants are consistently informed about concurrent \textsc{throws} of exceptions. In distributed systems such a reaction towards a system error is unrealistic. Distributed processes usually do not have any method to observe an error on another system part and if a participant is crashed or a link fails permanently there is usually no way to inform a waiting communication partner. Instead abstractions (failure detectors) are used to model the detection of failures that can \eg be implemented by time-outs. Here it is crucial to mention that failure detectors are usually considered to be local and can not ensure global consistency. Distributed algorithms have to deal with the problem that some part of a system may consider a process/link as crashed, while at the same time the same process/link is regarded as correct by another part. This is one of the most challenging problems in the design and verification of distributed algorithms.

In the case of link failures, if a participant is directly influenced by a failure on some other system part (a receiver of a lost message) it will eventually abort the respective communication attempt. If a participant does not depend (the sender in an unreliable link) it may never know about the failure or its nature.
Distributed algorithms usually deal with unexpected failures that are hard to detect and often impossible to propagate. Generating correct algorithms for this scenario is difficult and error-prone, thus we need methods to verify them.

\paragraph{Contribution.}
We extend multi-party session types---first a basic version similar to \cite{BettiniAtall08, BocciAtall10} and then the type system of \cite{DemangeonHonda12}---by optional blocks, \ie protected parts of sessions that either yield a value to be used in the continuation or fail and return a former specified default value.
This simple but restricted mechanism allows us to model link failures of the underlying system and to model distributed algorithms on top of such an unreliable communication infrastructure.
Moreover the type system ensures that well-typed processes progress, \ie termination, for the case of arbitrary occurrences of link failures.
Our approach is limited with respect to two aspects: We only cover algorithms that
\begin{inparaenum}[(1)]
	\item allow us to specify default values for all unreliable communication steps and
	\item terminate despite arbitrary link failures.
\end{inparaenum}
Accordingly, this approach is only a first step towards the analysis of distributed algorithms with session types. It shows however that it is possible to analyse distributed algorithms with session types and how the latter can solve the otherwise often complicated and elaborate task of proving termination.
We show that our attempt respects two important aspects of fault-tolerant distributed algorithms:
\begin{inparaenum}[(1)]
	\item The modularity as \eg present in the concept of rounds in many algorithms can be expressed naturally, and
	\item the model respects the asynchronous nature of distributed systems such that messages are not necessarily delivered in the order they are send and the rounds may overlap.
\end{inparaenum}

\paragraph{Overview.}
We extend global types and restriction in \S\ref{sec:globalTypes}, local types and projection in \S\ref{sec:localTypes}, and introduce the extended session calculus in \S\ref{sec:calculus} with a mechanism to check types in \S\ref{sec:wellTypedness}.
We present two examples---that are different variants to consider the rotating coordinator algorithm---as running examples with the presented definitions. A third example is presented in \S\ref{sec:ExaSubSessionsWithinOptionalBlocks}.
In \S\ref{sec:properties} we analyse the properties of the extended type system and use them to prove termination despite link failures of our running example.
We conclude with \S\ref{sec:conclusions}.
This paper is an extended version of \cite{adameitPetersNestmann17}.

\subsection{Properties of Optional Blocks}
\label{sec:optionalBlocks}

We extend standard versions of session types---as given \eg in \cite{BettiniAtall08,BocciAtall10,DemangeonHonda12}---with optional blocks. An optional block is a simple construct that encapsulates and isolates a potentially unreliable part of a protocol.
Since we are interested in the proof of relatively strong system properties such as termination, we restrain the effect that the failure (and thus also the success) of the encapsulated part can impose on the remainder of the protocol. The encapsulated part can compute some values and has---for the case of failure---to provide some default values.
In order to use optional blocks to model failures and at the same time ensure certain system properties, we designed optional blocks such that they ensure the properties encapsulation, isolation, safety, and reliance:
\begin{description}
	\item[Encapsulation:] Optional blocks encapsulate a potentially unreliable part of a protocol such that unreliable parts and reliable parts of a system are clearly distinguished and cannot interfere except for values that might be computed differently in the case of a failure.
	\item[Isolation:] Communication from within an optional block is restricted to its participants.
	\item[Safety:] Regardless of success or failure, each participant of an optional block---if it does not loop forever---returns a (potentially empty) vector of values of the required kinds. In the case of success these return values can be computed using communications with the other participants of the block. In the case of failure the return values are the default values.
	\item[Reliance:] If the considered system can terminate, then there is also a way to successfully complete all optional blocks.
\end{description}
Here encapsulation and isolation result from the semantics of the newly introduced concepts, whereas safety and reliance are enforced by the type system.
Session types introduce three different layers of abstraction:
\begin{inparaenum}[(1)]
	\item The session calculus is a process calculus---usually a variant of the $ \pi $-calculus \cite{milnerParrowWalker92}---that allows to model systems. Usually it is designed (or adapted) to provide flexibility and an easy and intuitive syntax in order to support the designer. For this purpose session calculi usually reflect only the local views on the respective participants, \ie the overall system behaviour is represented by modelling the single participants and their abilities to interact. With that session calculi are relatively close to programming languages. As a consequence, session calculi themselves provide very little guarantees on the correctness of the modelled systems.
	\item Global types, on the other hand side, provide a global view on a system. They are used to specify and formalise the desired properties of the overall system.
	\item To mediate between these two points of view and to guarantee that an implementation in a session calculus of a specification given as a global type has the desired properties, session types introduce an intermediate layer called local types. Intuitively local types specify the consequences of a global specification on a single participant. Projection functions allow to automatically derive local types from global types and typing rules allow to automatically check whether an implementation in a session calculus satisfies the global specification, by comparing the process against the local type.
\end{inparaenum}
With that session types provide a static and thus very efficient way to analyse and guarantee different kinds of system properties.
When extending session types with optional blocks we follow these three layers, starting with global types.

\section{Global Types with Optional Blocks}
\label{sec:globalTypes}

We want to derive a correct implementation of the algorithm of our running example from its specification. Accordingly we start with the introduction of the type system.

Throughout the paper we use $ G $ for global types, $ T $ for local types, $ \Labe $ for communication labels, $ \Chan[s], \Chan[k] $ for session names, $ \Chan $ for shared channels, $ \Role $ for role identifiers, $ \Prot $ for protocol identifiers, and $ \Args[v] $ for values of a base type (\eg integer or string). $ \Args, \Args[y] $ are variables to represent \eg session names, shared channels, or values.
We formally distinguish between roles, labels, process variables, type variables, and names---additionally to identifiers for global/local types, protocols, processes, \ldots. Formally we do however not further distinguish between different kinds of names but use different identifiers ($ \Chan, \Chan[s], \Args[v], \ldots $) to provide hints on the main intended purpose at the respective occurrence. Roles and participants are used as synonyms.
To simplify the presentation, we adapt set-like notions for tuples. For example we write $ \Args_i \in \tilde{\Args} $ if $ \tilde{\Args} = \left( \Args_1, \ldots, \Args_n \right) $ and $ 1 \leq i \leq n $. We use $ \cdot $ to denote the empty tuple.

Global types describe protocols from a global point of view on systems by interactions between roles. They are used to formalise specifications that describe the desired properties of a system. We extend the basic global types as used \eg in \cite{BettiniAtall08,BocciAtall10} with a global type for optional blocks.

\begin{definition}[Global Types]
	\label{def:globalTypes}
	Global types with optional blocks are given by
	\begin{align*}
		G & \deffTerms \GTCom{\Role_1}{\Role_2}{\sum_{i \in \indexSet} \Set{\GTInp{\Labe_i}{\Typed{\tilde{\Args}_i}{\tilde{\Sort}_i}}{G_i}}}
		\sep \textcolor{blue}{ \GTOptBl{\widetilde{\Role, \Typed{\tilde{\Args}}{\tilde{\Sort}}}}{G}{G'} }\\
		& \sep \GTChoi{G_1}{\Role}{G_2} \sep \GTPar{G_1}{G_2} \sep \GTRec{\TermV}{G} \sep \TermV \sep \GTEnd
	\end{align*}
\end{definition}

\noindent
$ \GTCom{\Role_1}{\Role_2}{\sum_{i \in \indexSet} \Set{\GTInp{\Labe_i}{\Typed{\tilde{\Args}_i}{\tilde{\Sort}_i}}{G_i}}} $ is the standard way to specify a communication from role $ \Role_1 $ to role $ \Role_2 $, where $ \Role_1 $ has a direct choice between several labels $ \Labe_i $ proposed by $ \Role_2 $.
Each branch expects values $ \tilde{\Args}_i $ of sorts $ \tilde{\Sort}_i $ and executes the continuation $ G_i $.
When $ \indexSet $ is a singleton, we write $ \GTCom{\Role_1}{\Role_2}{\GTInpS{\Labe}{\Typed{\tilde{\Args}}{\tilde{\Sort}}}} $.
$ \GTChoi{G_1}{\Role}{G_2} $ introduces so-called located (or internal) choice: the choice for one role $ \Role $ between two distinct protocol branches.
The parallel composition $ \GTPar{G_1}{G_2} $ allows to specify independent parts of a protocol.
$ \GTRec{\TermV}{G} $ and $ \TermV $ are used to allow for recursion.
$ \GTEnd $ denotes the successful completion of a global type.
We often omit trailing $ \GTEnd $ clauses.

We add the primitive $ \GTOptBl{\widetilde{\Role, \Typed{\tilde{\Args}}{\tilde{\Sort}}}}{G}{G'} $ to describe an optional block between the roles $ \Role_1, \ldots, \Role_n $, where $ \widetilde{\Role, \Typed{\tilde{\Args}}{\tilde{\Sort}}} $ abbreviates the sequence $ \Role_1, \Typed{\tilde{\Args}_1}{\tilde{\Sort}_1}, \ldots, \Role_n, \Typed{\tilde{\Args}_n}{\tilde{\Sort}_n} $  for some natural number $ n $.
Here $ G $ is the protocol that is encapsulated by the optional block and the $ \tilde{\Args}_i $ are so-called default values that are used within the continuation $ G' $ of the surrounding parent session if the optional block fails.
There is one (possibly empty) vector of default values $ \tilde{\Args}_i $ for each role $ \Role_i $.
The inner part $ G $ of an optional block is a (part of a) protocol that (in the case of success) is used to compute the vectors of return values. The typing rules ensure that for each role $ \Role_i $ the type of the computed vector coincides with the type $ \tilde{\Sort}_i $ of the specified vector of default values $ \tilde{\Args}_i $. Intuitively, if the block does not fail, each participant can use its respective vector of computed values in the continuation $ G' $. Otherwise, the default values are used. An optional block can either be completed successfully or fail completely.

Optional blocks capture the main features of a failure very naturally: a part of a protocol either succeeds or fails. They also encapsulate the source and direct impact of the failure, which allows us to study their implicit effect---as \eg missing communication partners---on the overall behaviour of protocols. With that they help us to specify, implement, and verify failure-tolerant algorithms.

Using optional blocks we provide a natural and simple specification of an unreliable link $ c $ between the two roles $ \Role[src] $ and $ \Role[trg] $, where in the case of success the value $ \Args[v]_{\Role[src]} $ is transmitted and in the case of failure a default value $ \Args[v]_{\Role[trg]} $ is used by the receiver.

\begin{example}[Global Type of an Unreliable Link]
	\label{exa:GTunreliableLink}
	\begin{align*}
		\GUL{\Role[src]}{\Args[v]_{\Role[src]}}{\Role[trg]}{\Args[v]_{\Role[trg]}}
		= \GTOptBlS{\Role[src], \cdot, \Role[trg], \Typed{\Args[v]_{\Role[trg]}}{\Sort[V]}}{\left( \GTCom{\Role[src]}{\Role[trg]}{\GTInp{\Labe[c]}{\Typed{\Args[v]_{\Role[src]}}{\Sort[V]}}{\GTEnd}} \right)}
	\end{align*}
\end{example}

\noindent
Here we have a single communication step---to model the potential loss of a single message---that is covered within an optional block. In the term $ \GUL{\Role[src]}{\Args[v]_{\Role[src]}}{\Role[trg]}{\Args[v]_{\Role[trg]}}\!.G' $ the receiver $ \Role[trg] $ may use the transmitted value $ \Args[v]_{\Role[src]} $ in the continuation $ G' $ if the communication succeeds or else uses its default value $ \Args[v]_{\Role[trg]} $.
Note that the optional block above specifies the empty sequence of values as default values for the sending process $ \Role[src] $, \ie the sender needs no default values.

In the remaining text, we use $ \prod_{i = 1..n} G_i \deff \GTPar{G_1}{\GTPar{\cdots}{G_n}} $ to abbreviate parallel composition and $ \bigodot_{i = 1..n} G_i \deff G_1.\cdots.G_n $ likewise for sequential composition. We naturally adapt these notations to local types and processes that are introduced later.

Remember that global types specify a global point of view of the communication structure, whereas the pseudo code of Example~\ref{exa:RCAlgorithm} provides the local view for participant~$ i $ containing also the data flow.
Accordingly we obtain a global type for Example~\ref{exa:RCAlgorithm} by abstracting partly from the values; concentrating on the communications.

Let $ \Args[v]_{i, j} $ be the value of participant~$ i $ of Example~\ref{exa:RCAlgorithm} after round~$ j $ such that $ \Args[v]_{i, i} := \Args[v]_{i, i - 1} $ (the coordinator does not update its value) and assume a vector $ \left( \Args[v]_{1, 0}, \ldots, \Args[v]_{n, 0} \right) $ of initial values.
Here only the initial values $ \Args[v]_{i, 0} $ are actually values, the remaining $ \Args[v]_{i, j} $ are variables that are instantiated with values during runtime. Then, in
\begin{example}[Global Type for Rotating Coordinators]
	\label{exa:GTRC}
	\begin{align*}
		\GRC{n} ={} & \bigodot_{i = 1..n} \; \bigodot_{j = 1..n, j \neq i} \GUL{\Role[p]_i}{\Args[v]_{i, i - 1}}{\Role[p]_j}{\Args[v]_{j, i - 1}}
	\end{align*}
\end{example}
the index $ i $ is used to specify the number of the current round, while $ j $ iterates over potential communication partners in round $ i $.
From a global point of view, there are $ n $~rounds such that each participant is exactly once the coordinator $ \Role[p]_i $ and transmits its value to all other participants $ \Role[p]_j $ using an unreliable link.
This global type abstracts in particular from Line~$ 5 $ in Example~\ref{exa:RCAlgorithm}, since it does not specify that or how the values of the receivers are updated.
For simplicity we do not consider the Lines~1 and 7.

\subsection{Global Types with Optional Blocks and Sub-Sessions}

As it is the case for our running example, many distributed algorithms are organised in rounds or use similar concepts of modularisation.
We want to be able to directly mirror this modularity.
To do so, we make use of the extension of multi-party session types with nested sessions of \cite{DemangeonHonda12}.
These authors introduce two additional primitives for global types
\begin{align*}
	\sep \GTDecl{\Prot}{\tilde{\Role}_1}{\Typed{\tilde{\Args[y]}}{\tilde{\Sort}}}{\tilde{\Role}_2}{G}{G'}
	\sep \GTCall{\Role}{\Prot}{\tilde{\Role}}{\tilde{\Args[y]}}{G}
\end{align*}
to implement sub-sessions.
The type $ \GTDecl{\Prot}{\tilde{\Role}_{1}}{\Typed{\tilde{\Args[y]}}{\tilde{\Sort}}}{\tilde{\Role}_2}{G}{G'} $ describes the declaration of a sub-protocol $ G $ identified via $ \Prot $, to be called from within the main protocol $ G' $.
Here $ \tilde{\Role}_1 $, $ \tilde{\Args[y]} $, and $ \tilde{\Role}_2 $ are the internally invited participants, the arguments, and the externally invited participants of $ G $, respectively.
With the protocol call $ \GTCall{\Role}{\Prot}{\tilde{\Role}}{\tilde{\Args[y]}}{G} $ a formerly declared sub-protocol $ \Prot $ can be initialised, where $ \tilde{\Role} $ and $ \tilde{\Args[y]} $ specify the internally invited roles and the arguments of $ \Prot $. $ G $ is the remainder of the parent session.
The sub-protocols that are introduced by these two primitives allow to specify algorithms in a modular way.

\begin{definition}[Global Types with Sub-Sessions]
	\label{def:globalTypesWSS}
	\begin{align*}
		G & \deffTerms \GTCom{\Role_1}{\Role_2}{\sum_{i \in \indexSet} \Set{\GTInp{\Labe_i}{\Typed{\tilde{\Args}_i}{\tilde{\Sort}_i}}{G_i}}}
		\sep \textcolor{blue}{ \GTOptBl{\widetilde{\Role, \Typed{\tilde{\Args}}{\tilde{\Sort}}}}{G}{G'} }\\
		& \sep \GTDecl{\Prot}{\tilde{\Role}_1}{\Typed{\tilde{\Args[y]}}{\tilde{\Sort}}}{\tilde{\Role}_2}{G}{G'}
		\sep \GTCall{\Role}{\Prot}{\tilde{\Role}}{\tilde{\Args[y]}}{G}\\
		& \sep \GTChoi{G_1}{\Role}{G_2} \sep \GTPar{G_1}{G_2} \sep \GTRec{\TermV}{G} \sep \TermV \sep \GTEnd
	\end{align*}
\end{definition}

Here optional blocks can surround (a part of) a session that possibly contains sub-sessions or it may surround a part of a single sub-session.

We extend the global type of our running example.

\begin{example}[Rotating Coordinators with Sub-Sessions]
	\label{exa:GTRCWSS}
	\begin{align*}
		\GRCB{n}
		= \GTDecl{\Prot[R]_n}{\Role[src], \widetilde{\Role[trg]}}{\Typed{\Args[v]_{\Role[src]}}{\Sort[V]}}{\cdot}{\GR{n}}{\left( \bigodot_{i = 1..n} \GTCallS{\Role[p]_i}{\Prot[R]_n}{\overline{\Role[p]_i}}{\Args[v]_{i, i - 1}} \right)}
	\end{align*}
	where the global type of a round is given by:
	\begin{align*}
		\GR{n} = \bigodot_{j = 1..(n - 1)} \GUL{\Role[src]}{\Args[v]_{\Role[src]}}{\Role[trg]_j}{\Args[v]_{\Role[src]}}
	\end{align*}
\end{example}

\noindent
The type of a single round basically remains the same but is transferred into the sub-session $ \GR{n} $ such that each round corresponds to its own sub-session.
The overall session $ \GRCB{n} $ then consists of the declaration of this sub-protocol (using the $ \mathtt{let} $-construct) followed by the iteration over $ i $ of the rounds, where in each round the respective sub-session is called.
Here $\overline{\Role[p]_i}$ is used to abbreviate the reordering $ \Role[p]_i, \Role[p]_1, \ldots, \Role[p]_{i - 1}, \Role[p]_{i + 1}, \ldots \Role[p]_n $ of the vector $ \tilde{\Role[p]} $.
In $ \GR{n} $ we use $ \Args[v]_{\Role[src]} $ not only as transmitted value but also as the default value of the receiver. This violates our intuition of the algorithm.
Intuitively the default value should be the last known value of the receiver, \ie $ \Args[v]_{i, j - 1} $ for the receiver $ i $ in round~$ j $. The implementation in the session calculus will use this value.
However, since all $ \Args[v]_{i, j} $ are of the same type and because we consider (global) types here, we can use $ \Args[v]_{\Role[src]} $.

\subsection{Restriction}

$ \GRC{n} $ and $ \GRCB{n} $ have $ n $ roles: $ \Role[p]_1, \ldots, \Role[p]_n $.
\emph{Restriction} maps a global type to the parts that are relevant for a certain role.
We extend the restriction rules of \cite{DemangeonHonda12,Demangeon15} with a rule for optional blocks.

\begin{definition}[Restriction]
	\label{def:restriction}
	$ $\\
	The restriction operator $ \RestS{G}{\Role} $ goes inductively through all constructors except:
	\begin{align*}
		& \Rest{\GTDecl{\Prot}{\tilde{\Role}_1}{\tilde{\Args[y]}}{\tilde{\Role}_2}{G}{G'}}{\Role_0} = \GTDecl{\Prot}{\tilde{\Role}_1}{\tilde{\Args[y]}}{\tilde{\Role}_2}{G}{\!\Rest{G'}{\Role_0}}\\
		& \Rest{\GTCall{\Role}{\Prot}{\tilde{\Role}}{\tilde{\Args[y]}}{G}}{\Role_0} = \begin{cases} \GTCall{\Role}{\Prot}{\tilde{\Role}}{\tilde{\Args[y]}}{\!\Rest{G}{\Role_0}} & \text{if } \Role_0 = \Role \text{ or } \Role_0 \in \tilde{\Role}\\ \RestS{G}{\Role_0} & \text{else} \end{cases}\\
		& \Rest{\GTCom{\Role_1}{\Role_2}{\sum_{i \in \indexSet} \Set[]{\GTInp{\Labe_i}{\tilde{\Args}_i}{G_i}}}}{\Role_0} \!= \begin{cases} \GTCom{\Role_1}{\Role_2}{\sum_{i \in \indexSet} \Set[]{\GTInp{\Labe_i}{\tilde{\Args}_i}{\!\Rest{G_i}{\Role_0}}}} & \text{if } \Role_0 \in \Set[]{ \Role_1, \Role_2 }\\ \RestS{G_1}{\Role_0} & \text{else} \end{cases}\\
		& \textcolor{blue}{\Rest{\GTOptBl{\widetilde{\Role, \tilde{\Args}}}{G}{G'}}{\Role_0}} \textcolor{blue}{\ = \begin{cases} \GTOptBl{\widetilde{\Role, \tilde{\Args}}}{\left( \RestS{G}{\Role_0} \right)}{\left( \RestS{G'}{\Role_0} \right)} & \text{if } \Role_0 \in \tilde{\Role}\\ \RestS{G'}{\Role_0} & \text{else} \end{cases}}
	\end{align*}
\end{definition}

\noindent
For simplicity we abbreviate $ \Typed{\tilde{\Args[z]}}{\tilde{\Sort}} $ by $ \tilde{\Args[z]} $ for all vectors of values in this definition.
This definition also captures the definition of restriction for the smaller type system without sub-sessions---in this case the first two rules are superfluous.
Here, the restriction of an optional block on one of its participants results in the restriction of both, its inner part as well as its continuation, on that role. If we restrict an optional block on a role that does not participate, the result is the restriction of the continuation only.
Accordingly, the restriction of an unreliable link on a role $ \Role[p]_i $ results in the link itself if $ \Role[p]_i $ is either the source or the target of the unreliable link and else removes the unreliable link.

\begin{example}[Restriction on Participant~$ i $]
	\label{exa:RestRC}
	\begin{align*}
		\RestS{\GRC{n}}{\Role[p]_i} ={}
		& \bigodot_{j = 1..(i{-}1)} \GUL{\Role[p]_j}{\Args[v]_{j, j{-}1}}{\Role[p]_i}{\Args[v]_{i, j{-}1}}. & \tag{a}\label{exa:RestA}\\
		& \bigodot_{j = 1..n, j \neq i} \GUL{\Role[p]_i}{\Args[v]_{i, i{-}1}}{\Role[p]_j}{\Args[v]_{j, i{-}1}}. & \tag{b}\label{exa:RestB}\\
		& \bigodot_{j = (i + 1)..n} \GUL{\Role[p]_j}{\Args[v]_{j, j{-}1}}{\Role[p]_i}{\Args[v]_{i, j{-}1}} & \tag{c}\label{exa:RestC}
	\end{align*}
\end{example}

\noindent
Restricting $ \GRC{n} $ on participant~$ i $ reduces all rounds~$ j $ (except for round~$ j = i $) to the single communication step of round~$ j $ in that participant~$ i $ receives a value.
Accordingly, participant~$ i $ receives $ i{-}1 $ times a value in $ i{-}1 $ rounds in (\ref{exa:RestA}), then broadcasts its current value to all other participants (modelled by $ n{-}1 $ single communication steps) in (\ref{exa:RestB}), and then receives $ n{-}i $ times a value in the remaining rounds in (\ref{exa:RestC}).

For the extended Example~\ref{exa:GTRCWSS} we have $ \RestS{\GRCB{n}}{\Role[p]_i} = \GRCB{n} $, since each call of a sub-session refers to all roles.
The restriction of the protocol for rounds $ \GR{n} $ on the role that coordinates the respective round, \ie for $ \Role[p]_i = \Role[src] $, is (similar to (\ref{exa:RestB})):
\begin{align*}
	\RestS{\GR{n}}{\Role[p]_i} ={} & \bigodot_{j = 1..n, j \neq i} \GUL{\Role[p]_i}{\Args[v]_{i, i{-}1}}{\Role[p]_j}{\Args[v]_{i, i{-}1}}
\end{align*}
whereas its restriction on another role, \ie for $ \Role[p]_i \in \widetilde{\Role[trg]} $ and $ \Role[src] = \Role[p]_j $, leads (similar to (\ref{exa:RestA}) and (\ref{exa:RestC})) to:
\begin{align*}
	\RestS{\GR{n}}{\Role[p]_i} ={} & \GUL{\Role[p]_j}{\Args[v]_{j, j{-}1}}{\Role[p]_i}{\Args[v]_{j, j{-}1}}
\end{align*}

\subsection{Well-Formed Global Types}

Following \cite{DemangeonHonda12} we type all objects appearing in global types with \emph{kinds} (types for types) $ \Sort[K] \deffTerms \Sort[Role] \mid \Sort[Val] \mid \diamond \mid \left( \Sort[K]_1 \times \ldots \times \Sort[K]_n \right) \to \Sort[K] $.
$ \Sort[Val] $ are value-kinds, which are first-order types for values (like $ \mathbb{B} $ for boolean) or data types. $ \Sort[Role] $ is used for identifiers of roles.
We use $ \diamond $ to denote protocol types and $ \to $ to denote parametrisation.
We adopt the definition of \emph{well-kinded} global types from \cite{DemangeonHonda12} that basically ensures that all positions $ \Role, \Role_1, \Role_2, \tilde{\Role}_1, \tilde{\Role}_2 $ in global types can be instantiated only by objects of type $ \Sort[Role] $ and that the type of all sub-protocols in declarations and calls is of the form $ \Sort[K] \to \diamond $, where $ \Sort[K] $ is the product of the types of the internally invited roles and the arguments of the sub-protocol.

According to \cite{DemangeonHonda12} a global type $ G $ is \emph{projectable} if
\begin{inparaenum}[(1)]
	\item for each occurrence of $ \GTChoi{G_1}{\Role}{G_2} $ in the type and for any free role $ \Role' \neq \Role $ we have $ \RestS{G_1}{\Role'} = \RestS{G_2}{\Role'} $,
	\item for each occurrence of $ \GTCom{\Role_1}{\Role_2}{\sum_{i \in \indexSet} \Set[]{\GTInp{\Labe_i}{\tilde{\Args}_i}{G_i}}} $ in the type and for any free role $ \Role' \notin \Set{ \Role_1, \Role_2 } $ we have $ \RestS{G_i}{\Role'} = \RestS{G_j}{\Role'} $ for all $ i, j \in \indexSet $,
	\item for each occurrence of $ \GTPar{G_1}{G_2} $ the types $ G_1 $ and $ G_2 $ do not share the same free role.
\end{inparaenum}
To simplify the definition of projection, we write $ \Role \in G $ if $ \Role $ is a free role in the global type $ G $ and else $ \Role \notin G $.
Additionally we require (similar to sub-sessions in \cite{DemangeonHonda12}) for a global type $ G $ to be projectable that
\begin{inparaenum}[(1)]
	\item for each optional block $ \GTOptBl{\widetilde{\Role, \Typed{\tilde{\Args}}{\tilde{\Sort}}}}{G_1}{G_2} $ in $ G $, all roles in $ G_1 $ are either contained in $ \tilde{\Role} $ or are newly introduced by a sub-session and
	\item for each each $ \GTDecl{\Prot}{\tilde{\Role}_1}{\tilde{\Args[y]}}{\tilde{\Role}_2}{G_1}{G_2} $ in $ G $, all roles in $ G_1 $ are either contained in $ \tilde{\Role}_1 $ or $ \tilde{\Role}_2 $ or are newly introduced by a sub-session.
\end{inparaenum}

A global type is \emph{well-formed} when it is well-kinded and projectable, and satisfies the standard linearity condition \cite{BettiniAtall08}.
For more intuition on the notion of well-formedness and examples for non-well-formed protocols we refer to \cite{DemangeonHonda12}.
In the examples, we use $ \Sort[V] $ as the type of the values $ \Args[v]_{i, j} $. Clearly, $ \GRC{n} $ and $ \GRCB{n} $ are well-formed.

\section{Local Types with Optional Blocks}
\label{sec:localTypes}

Local types describe a local and partial point of view on a global communication protocol \wrt a single participant. They are used to validate and monitor distributed programs.
We extend the basic local types as used \eg in \cite{BettiniAtall08,BocciAtall10} with a local type for optional blocks.

\begin{definition}[Local Types]
	\label{def:localTypes}
	Local types with optional blocks are given by
	\begin{align*}
		T & \deffTerms \LTGet{\Role}{_{i \in \indexSet} \Set{ \LTLab{\Labe_i}{\Typed{\tilde{\Args}_i}{\tilde{\Sort}_i}}{T_i} }}
		\sep \LTSend{\Role}{_{i \in \indexSet} \Set{ \LTLab{\Labe_i}{\Typed{\tilde{\Args}_i}{\tilde{\Sort}_i}}{T_i} }} \sep \textcolor{blue}{\LTOpt{\tilde{\Role}}{T}{\Typed{\tilde{\Args}}{\tilde{\Sort}}}{T'}}\\
		& \sep \LTChoi{T_1}{T_2} \sep \LTPar{T_1}{T_2} \sep \LTRec{\TermV}{T} \sep \TermV \sep \LTEnd
	\end{align*}
\end{definition}

\noindent
The first two operators specify endpoint primitives for communications with $ \mathtt{get} $ for the receiver side---where $ \Role $ is the sender---and $ \mathtt{send} $ for the sender side---where $ \Role $ denotes the receiver. Accordingly, they introduce the two possible local views of a global type for communication.
$ \LTChoi{T_1}{T_2} $ is the local view of the global type $ \GTChoi{G_1}{\Role}{G_2} $ for a choice determined by the role $ \Role $ for which this local type is created.
$ \LTPar{T_1}{T_2} $ represents the local view of the global type for parallel composition, \ie describes independent parts of the protocol for the considered role.
Again $ \LTRec{\TermV}{T} $ and $ \TermV $ are used to introduce recursion and $ \LTEnd $ denotes the successful completion of a protocol.

We add the local type $ \LTOpt{\tilde{\Role}}{T}{\Typed{\tilde{\Args}}{\tilde{\Sort}}}{T'} $.
It initialises an optional block between the roles $ \tilde{\Role} $ around the local type $ T $, where the currently considered participant $ \Role $ (called \textit{owner}) is a participant of this block, \ie $ \Role \in \tilde{\Role} $. After the optional block the local type continues with $ T' $.
Again we usually omit trailing $ \LTEnd $ clauses.

\subsection{Local Types with Optional Blocks and Sub-Sessions}

Local types describe a local and partial point of view on a global communication protocol \wrt a single participant.
To obtain the local types that correspond to global types with sub-sessions we add the three operators of local types introduced by \cite{DemangeonHonda12}:
\begin{align*}
	& \sep \LTCall{\Prot}{G}{\tilde{\Args[v]}}{\Typed{\tilde{\Args[y]}}{\tilde{\Sort}}}{\tilde{\Role}_2}{T} \sep \LTEnt{\Prot}{\Role_1}{\tilde{\Args[v]}}{\Role_2}{T} \sep \LTReq{\Prot}{\Role_1}{\tilde{\Args[v]}}{\Role_2}{T}
\end{align*}
Sub-sessions are created with the $ \mathtt{call} $ operator; internal invitations are handled by the $ \mathtt{req} $-operator for requests and the $ \mathtt{ent} $-operator to accept invitations.
A $ \mathtt{call} $ creates a sub-session for protocol $ \Prot $ of the global type $ G $, where $ \tilde{\Args[v]} $ are value arguments handed to the protocol and $ \tilde{\Role}_2 $ are the external roles that are invited to this sub-session.
In $ \LTEnt{\Prot}{\Role_1}{\tilde{\Args[v]}}{\Role_2}{T} $ the role $ \Role_2 $ refers to the initiator of the sub-session, $ \Role_1 $ denotes the role in the sub-protocol the participant accepts to take, and $ \tilde{\Args[v]} $ are the arguments of the respective protocol.
Similarly, in $ \LTReq{\Prot}{\Role_1}{\tilde{\Args[v]}}{\Role_2}{T} $ the role $ \Role_1 $ is the role of the protocol the participant is invited to take, $ \tilde{\Args[v]} $ are the arguments of the protocol, and $ \Role_2 $ refers to the participant the invitation is directed to.

\begin{definition}[Local Types with Sub-Sessions]
	\label{def:localTypesWSS}
	\begin{align*}
		T & \deffTerms \LTGet{\Role}{_{i \in \indexSet} \Set{ \LTLab{\Labe_i}{\Typed{\tilde{\Args}_i}{\tilde{\Sort}_i}}{T_i} }}
		\sep \LTSend{\Role}{_{i \in \indexSet} \Set{ \LTLab{\Labe_i}{\Typed{\tilde{\Args}_i}{\tilde{\Sort}_i}}{T_i} }} \sep \textcolor{blue}{\LTOpt{\Role}{T}{\Typed{\tilde{\Args}}{\tilde{\Sort}}}{T'}}\\
		& \sep \LTCall{\Prot}{G}{\tilde{\Args[v]}}{\Typed{\tilde{\Args[y]}}{\tilde{\Sort}}}{\tilde{\Role}_2}{T} \sep \LTEnt{\Prot}{\Role_1}{\tilde{\Args[v]}}{\Role_2}{T} \sep \LTReq{\Prot}{\Role_1}{\tilde{\Args[v]}}{\Role_2}{T}\\
		& \sep \LTChoi{T_1}{T_2} \sep \LTPar{T_1}{T_2} \sep \LTRec{\TermV}{T} \sep \TermV \sep \LTEnd
	\end{align*}
\end{definition}

\subsection{Projection}

To ensure that a global type and its local types coincide, global types are projected to their local types.
In \cite{DemangeonHonda12} projection is defined \wrt a protocol environment $ \env $ that associates protocol identifiers to their contents and is updated in $ \mathtt{let} $ constructs. For global types without sub-sessions $ \env $ remains empty.
We inherit the rules to project global types on their local types from \cite{Demangeon15,DemangeonHonda12} and add a rule to cover global types of optional blocks.

\begin{figure}[tp]
	\begin{align*}
		& \Proj{\GTDecl{\Prot}{\tilde{\Role}_1}{\Typed{\tilde{\Args[y]}}{\tilde{\Sort}}}{\tilde{\Role}_2}{G}{G'}}{\env}{\Role_p} = \ProjS{G'}{\env, \PDec{\Prot}{\tilde{\Role}_1}{\Typed{\tilde{\Args[y]}\,}{\,\tilde{\Sort}}}{\tilde{\Role}_2}{G}}{\Role_p}\\
		& \Proj{\GTCall{\Role_A}{\Prot}{\tilde{\Role}}{\tilde{\Args[y]}}{G}}{\env, \PDec{\Prot}{\tilde{\Role}'}{\Typed{\tilde{\Args[v]}\,}{\,\tilde{\Sort}}}{\tilde{\Role}''}{G_{\Prot}}}{\Role_p} \\
		& \hspace{2em} = \begin{cases}
			\begin{array}{l}
				\!\!\LTCallS{\Prot}{G_{\Prot}}{\tilde{\Args[y]}}{\Typed{\tilde{\Args[v]}}{\tilde{\Sort}}}{\tilde{\Role}''}.\\
				\!\!\big( \LTPar{\left( \prod_{i = 1..n} {\LTReqS{\Prot}{\Role'_i}{\tilde{\Args[y]}}{\Role_i}} \right)}{\Proj{G}{\env, \PDec{\Prot}{\tilde{\Role}'}{\Typed{\tilde{\Args[v]}\,}{\,\tilde{\Sort}}}{\tilde{\Role}''}{G_{\Prot}}}{\Role_p}} \big)
			\end{array} & \begin{array}{l} \text{if } \Role_p = \Role_A\\ \text{ and } \Role_A \notin \tilde{\Role} \end{array}\\
			\\
			\begin{array}{l}
				\!\!\LTCallS{\Prot}{G_{\Prot}}{\tilde{\Args[y]}}{\Typed{\tilde{\Args[v]}}{\tilde{\Sort}}}{\tilde{\Role}''}.\\
				\!\!\big( \LTPar{\left( \prod_{i = 1..n} {\LTReqS{\Prot}{\Role'_i}{\tilde{\Args[y]}}{\Role_i}} \right)}{\LTPar{\LTEntS{\Prot}{\Role'_i}{\tilde{\Args[y]}}{\Role_A}}{}}\\
				\!\!\!\ProjS{G}{\env, \PDec{\Prot}{\tilde{\Role}'}{\Typed{\tilde{\Args[v]}\,}{\,\tilde{\Sort}}}{\tilde{\Role}''}{G_{\Prot}}}{\Role_p} \big)
			\end{array} & \begin{array}{l} \text{if } \Role_p = \Role_A\\ \text{ and } \Role_A = \Role_i \end{array}\\
			\\
			\LTEnt{\Prot}{\Role'_i}{\tilde{\Args[y]}}{\Role_A}{\Proj{G}{\env, \PDec{\Prot}{\tilde{\Role}'}{\Typed{\tilde{\Args[v]}\,}{\,\tilde{\Sort}}}{\tilde{\Role}''}{G_{\Prot}}}{\Role_p}} & \begin{array}{l} \text{if } \Role_p \neq \Role_A\\ \text{ and } \Role_p = \Role_i \end{array}\\
			\\
			\ProjS{G}{\env, \PDec{\Prot}{\tilde{\Role}'}{\Typed{\tilde{\Args[v]}\,}{\,\tilde{\Sort}}}{\tilde{\Role}''}{G_{\Prot}}}{\Role_p} & \text{else}
		\end{cases}\\
		& \Proj{\GTCom{\Role_1}{\Role_2}{\sum_{i \in \indexSet} \Set{ \GTInp{\Labe_i}{\Typed{\tilde{\Args}_i}{\tilde{\Sort}_i}}{G_i} }}}{\env}{\Role_p} =
		\begin{cases}
			\LTSend{\Role_2}{_{i \in \indexSet} \Set{ \LTLab{\Labe_i}{\Typed{\tilde{\Args}_i}{\tilde{\Sort}_i}}{\Proj{G_i}{\env}{\Role_p}} }} & \text{if } \Role_p = \Role_1\\
			\LTGet{\Role_1}{_{i \in \indexSet} \Set{ \LTLab{\Labe_i}{\Typed{\tilde{\Args}_i}{\tilde{\Sort}_i}}{\Proj{G_i}{\env}{\Role_p}} }} & \text{if } \Role_p = \Role_2\\
			\ProjS{G_1}{\env}{\Role_p} & \text{else}
		\end{cases}\\
		& \textcolor{blue}{\Proj{\GTOptBl{\widetilde{\Role, \Typed{\tilde{\Args}}{\tilde{\Sort}}}}{G}{G'}}{\env}{\Role_p}}
			\textcolor{blue}{\ = \begin{cases}
				\LTOpt{\tilde{\Role}}{\ProjS{G}{\env}{\Role_p}}{\Typed{\tilde{\Args}_i}{\tilde{\Sort}_i}}{\Proj{G'}{\env}{\Role_p}} & \text{if } \Role_p = \Role_i \in \tilde{\Role} \text{ and } \tilde{\Args}_i \neq \cdot\\
				\LTPar{\LTOptS{\tilde{\Role}}{\ProjS{G}{\env}{\Role_p}}{\cdot}}{\Proj{G'}{\env}{\Role_p}} & \text{if } \Role_p = \Role_i \in \tilde{\Role} \text{ and } \tilde{\Args}_i = \cdot\\
				\ProjS{G'}{\env}{\Role_p} & \text{else}
			\end{cases}}\\
		& \Proj{\GTChoi{G_1}{\Role}{G_2}}{\env}{\Role_p} = 
		\begin{cases}
			\LTChoi{\Proj{G_1}{\env}{\Role_p}}{\Proj{G_2}{\env}{\Role_p}} & \text{if } \Role_p = \Role\\
			\ProjS{G_1}{\env}{\Role_p} & \text{else}
		\end{cases}\\
		& \Proj{\GTPar{G_1}{G_2}}{\env}{\Role_p} =
		\begin{cases}
			\Proj{G_i}{\env}{\Role_p} & \text{if } \Role_p \in G_i \text{ and } \Role_p \notin G_j \text{ and } \Set{ i, j } = \Set{ 1, 2 }\\
			\LTEnd & \text{if } \Role_p \notin G_1 \text{ and } \Role_p \notin G_2
		\end{cases}\\
		& \Proj{\GTRec{\TermV}{G}}{\env}{\Role_p} = \LTRec{\TermV}{\Proj{G}{\env}{\Role_p}} \hspace{2em}
		\ProjS{\TermV}{\env}{\Role_p} = \TermV \hspace{2em}
		\ProjS{\GTEnd}{\env}{\Role_p} = \LTEnd
	\end{align*}
	\caption{Projection Rules}
	\label{fig:projectionRules}
\end{figure}

Figure~\ref{fig:projectionRules} contains all projection rules for both considered type system. Again there are rules---the rules to deal with sub-sessions---that are superfluous in the smaller type system.

The projection rule for optional blocks has three cases.
The last case is used to skip optional blocks when they are projected to roles that do not participate. The first two cases handle projection of optional blocks to one of its participants. A local optional block is generated with the projection of $ G $ as content.

The first two cases check whether the optional block indeed computes any values for the role we project onto. They differ only in the way that the continuation of the optional block and its inner part are connected. If the projected role does not specify default values---because no such values are required---the projected continuation $ \Proj{G'}{}{\Role_p} $ can be placed in parallel to the optional block (second case).
Otherwise, the continuation has to be guarded by the optional block and, thus, by the computation of the computed values (first case).

By distinguishing between these two first cases, we follow the same line of argument as used for sub-sessions in \cite{DemangeonHonda12}, where the projected continuation of a sub-session $ \texttt{call} $ is either in parallel to the projection of the $ \texttt{call} $ itself or connected sequentially.
Intuitively, whenever the continuation depends on the outcome of the optional block it has to be connected sequentially.

The observant reader may have recognised that the global and the local types do not specify any mechanism to install the value computed in a successful optional block in its continuation.
We do not want to restrict the way in which the result of an optional block is computed, except that it has to be derived from the knowledge of the owner together with communications with the other participant of the optional block.
Hence obtaining its values in the projection function is difficult.
For the global and the local type it is however not necessary to derive the correct value but only its kinds and these kinds have to coincide with the kinds of the default values.
We leave the computation of return values and thus the data flow of the algorithm to its actual implementation after introducing the session calculus. The type system will, however, ensure that for each optional block---if no block fails---exactly one vector of return values is computed within each optional block.

\begin{example}[Projection of Unreliable Links]
	\begin{align*}
		\ProjS{\GUL{\Role[src]}{\Args[v]_{\Role[src]}}{\Role[trg]}{\Args[v]_{\Role[trg]}}}{}{\Role[src]}
		& = \LULS{\Role[src]}{\Args[v]_{\Role[src]}}{\Role[trg]} = \LTOptS{\Role[scr], \Role[trg]}{\LTSend{\Role[trg]}{\LTLabS{\Labe[c]}{\Typed{\Args[v]_{\Role[src]}}{\Sort[V]}}}}{\cdot}\\
		\ProjS{\GUL{\Role[src]}{\Args[v]_{\Role[src]}}{\Role[trg]}{\Args[v]_{\Role[trg]}}}{}{\Role[trg]}
		& = \LULT{\Role[src]}{\Args[v]_{\Role[src]}}{\Role[trg]}{\Args[v]_{\Role[trg]}} = \LTOptS{\Role[src], \Role[trg]}{\LTGet{\Role[src]}{\LTLabS{\Labe[c]}{\Typed{\Args[v]_{\Role[src]}}{\Sort[V]}}}}{\Typed{\Args[v]_{\Role[trg]}}{\Sort[V]}}
	\end{align*}
\end{example}

\noindent
When projected onto its sender, the global type for a communication over an unreliable link of Example~\ref{exa:GTunreliableLink} results in the local type $ \LULS{\Role[src]}{\Args[v]_{\Role[src]}}{\Role[trg]} $ that consists of an optional block containing a send operation towards $ \Role[trg] $.
Since the optional block for the sender does not specify any default values, the local type $ \LULS{\Role[src]}{\Args[v]_{\Role[src]}}{\Role[trg]} $ will be placed in parallel to the projection of the continuation.
The projection onto the receiver results in the local type $ \LULT{\Role[src]}{\Args[v]_{\Role[src]}}{\Role[trg]}{\Args[v]_{\Role[trg]}} $ that consists of an optional block containing a receive operation from $ \Role[src] $.
Here a default value is necessary for the case that the message is lost.
So the type $ \LULT{\Role[src]}{\Args[v]_{\Role[src]}}{\Role[trg]}{\Args[v]_{\Role[trg]}} $ has to be composed sequentially with the projection of the continuation.

\begin{example}[Local Types for Rotating Coordinators]
	\label{exa:LTRC}
	\begin{align*}
		\ProjS{\GRC{n}}{}{\Role[p]_i}
		={} & \Proj{\RestS{\GRC{n}}{\Role[p]_i}}{}{\Role[p]_i}\\
		={} & \bigodot_{j = 1..(i{-}1)} \LULT{\Role[p]_j}{\Args[v]_{j, j{-}1}}{\Role[p]_i}{\Args[v]_{i, j{-}1}}.\\
		& \left( \LTPar{\left( \prod_{j = 1..n, j \neq i} \LULS{\Role[p]_i}{\Args[v]_{i, i{-}1}}{\Role[p]_j} \right)}{\left( \bigodot_{j = (i + 1)..n} \LULT{\Role[p]_j}{\Args[v]_{j, j{-}1}}{\Role[p]_i}{\Args[v]_{i, j{-}1}} \right)} \right)
	\end{align*}
\end{example}

\noindent
The projection of Example~\ref{exa:GTRC} onto participant $ \Role[p]_i $ consists of $ i{-}1 $ sequential receptions (rounds $ 1 $ to $ i{-}1 $), then $ n{-}1 $ parallel transmissions of the current value of $ \Role[p]_i $ to the remaining participants (round $ i $), and finally $ n{-}i $ more sequential receptions (round $ i {+} 1 $ to $ n $) in parallel to round $ i $.
Due to the different cases of the projection of optional blocks, we obtain parallel optional blocks in the local type although all blocks are sequential in the global type.

\begin{example}[Rotating Coordinators with Sub-Sessions]
	\label{exa:LTRCWSS}
	\begin{align*}
		\ProjS{\GRCB{n}}{}{\Role[p]_i} ={} & \ProjS{G'}{\PDec{\Prot[R]_n}{\Role[src], \widetilde{\Role[trg]}}{\Args[v]_{\Role[src]}}{\cdot}{\GR{n}}}{\Role[p]_i}\\
		={} & \bigodot_{j = 1..(i{-}1)} \LTEntS{\Prot[R]_n}{\Role[trg]_i}{\Args[v]_{j, j{-}1}}{\Role[p]_j}.\\
		& \LTCall{\Prot[R]}{\GR{n}}{\Args[v]_{i, i{-}1}}{\Args[v]_{\Role[src]}}{\cdot}{\big( }
		\LTPar{\LTEntS{\Prot[R]_n}{\Role[src]}{\Args[v]_{i, i{-}1}}{\Role[p]_i}}{\LTReqS{\Prot[R]_n}{\Role[src]}{\Args[v]_{i, i{-}1}}{\Role[p]_i}}\\
		& \quad \LTPar{}{\LTPar{\prod_{j = 1..i} \LTReqS{\Prot[R]_n}{\Role[trg]_j}{\Args[v]_{i, i{-}1}}{\Role[p]_j}}{\prod_{j = i..(n{-}1)} \LTReqS{\Prot[R]_n}{\Role[trg]_j}{\Args[v]_{i, i{-}1}}{\Role[p]_{j + 1}}}}\\
		& \quad \LTPar{}{\big( \bigodot_{j = (i + 1)..n}\!\!\! \LTEntS{\Prot[R]_n}{\Role[trg]_{i{-}1}}{\Args[v]_{j, j{-}1}}{\Role[p]_j} \big)} \big)
	\end{align*}
\end{example}

\noindent
To project the global type $ \GRCB{n} $ of Example~\ref{exa:GTRCWSS} to the local type of participant~$ i $, we first add the information about the declaration of the protocol $ \Prot[R] $ to the environment and then project the $ n $ rounds.
The first $ i{-}1 $ rounds and the last $ n{-}i $ rounds are projected to sequentially composed acceptance notifications $ \LTEntS{\Prot[R]_n}{\Role[trg]_i}{\Args[v]_{j, j{-}1}}{\Role[p]_j} $ to participate in the sub-session for the respective round as target, \ie receiver.
The projection of round $ i $ on the coordinator participant $ i $ initialises a sub-session using the $ \mathtt{call} $-operator followed by the acceptance notion of participant~$ i $ ($ \mathtt{ent} $) to participate as $ \Role[src] $ (sender) and the invitations for all participants ($ \mathtt{req} $).
Similar to Example~\ref{exa:LTRC}, the projections of the rounds $ j \neq i $ are composed sequentially, whereas the projection of round $ i $---consisting of the parallel composition of the respective invitations and the acceptance of $ \Role[p]_i $---is composed in parallel to the projection of round $ i {+} 1 $.

\section{A Session Calculus with Optional Blocks}
\label{sec:calculus}

Global types (and the local types that are derived from them) can be considered as specifications that describe the desired properties of the system we want to analyse. The process calculus, that we use to model/implement the system, is in the case of session types usually a variant of the $ \pi $-calculus \cite{milnerParrowWalker92}.
We extend a basic session-calculus as used \eg in \cite{BettiniAtall08,BocciAtall10} with two operators.

\begin{definition}[Processes]
	\label{def:processes}
	Processes are given by
	\begin{align*}
		P & \deffTerms \PInp{\Chan}{\tilde{\Args}}{P} \sep \POut{\Chan}{\tilde{\Chan[s]}}{P}
		\sep \PGet{\Chan[k]}{\Role_1}{\Role_2}{_{i \in \indexSet} \Set{ \PLab{\Labe_i}{\tilde{\Args}_i}{P_i} }} \sep \PSend{\Chan[k]}{\Role_1}{\Role_2}{\Labe}{\tilde{\Args[v]}}{P}\\
		& \sep \textcolor{blue}{\POpt{\Role}{\tilde{\Role}}{P}{\tilde{\Args}}{\tilde{\Args[v]}}{P'}} \sep \textcolor{blue}{\POptEnd{\Role}{\tilde{\Args[v]}}}\\
		& \sep \PRes{\Args}{P} \sep \PChoi{P_1}{P_2} \sep \PPar{P_1}{P_2}
		\sep \PRec{\TermV[X]}{P} \sep \PVar{\TermV[X]} \sep \PEnd
	\end{align*}
\end{definition}

The prefixes $ \PInp{\Chan}{\tilde{\Args}}{P} $ and $ \POut{\Chan}{\tilde{\Chan[s]}}{P} $ are inherited from the $ \pi $-calculus and are used for external invitations.
Using the shared channel $ \Chan $, an external participant can be invited with the output $ \POut{\Chan}{\tilde{\Chan[s]}}{P} $ transmitting the session channels $ \tilde{\Chan[s]} $ that are necessary to participate and the external participant can accept the invitation using the input $ \PInp{\Chan}{\tilde{\Args}}{P} $.
The following two operators introduce a branching input and the corresponding transmission on the session channel $ k $ from $ \Role_1 $ to $ \Role_2 $. These two operators correspond to the local types for $ \mathtt{get} $ and $ \mathtt{send} $.
Restriction $ \PRes{\Args}{P} $ allows to generate a fresh name that is not known outside of the scope of this operator unless it was explicitly communicated.
For simplicity and following \cite{Demangeon15} we assume that only shared channels $ \Chan $ for external invitations and session channels $ \Chan[s], \Chan[k] $ for not yet initialised sub-sessions are restricted, because this covers the interesting cases\footnote{Sometimes it might be useful to allow the restriction of values, \eg for security. For this case an additional restriction operator can be introduced.} and simplifies the typing rules in Figure~\ref{fig:typingRules}.
The term $ \PChoi{P_1}{P_2} $ either behaves as $ P_1 $ or $ P_2 $.
$ \PPar{P_1}{P_2} $ defines the parallel composition of the processes $ P_1 $ and $ P_2 $.
$ \PRec{\TermV[X]}{P} $ and $ \PVar{\TermV[X]} $ are used to introduce recursion.
$ \PEnd $ denotes the completion of a process.

To implement optional blocks, we add $ \POpt{\Role}{\tilde{\Role}}{P}{\tilde{\Args}}{\tilde{\Args[v]}_d}{P'} $ and $ \POptEnd{\Role}{\tilde{\Args[v]}} $.
The former defines an optional block between the roles $ \tilde{\Role} $ around the process $ P $ with the default values $ \tilde{\Args[v]}_d $. We require that the owner $ \Role $ of this block is one of its participants $ \tilde{\Role} $, \ie $ \Role \in \tilde{\Role} $.
In the case of success, $ \POptEnd{\Role}{\tilde{\Args[v]}} $ transmits the computed values $ \tilde{\Args[v]} $ from within the optional block to the continuation $ P' $ to be substituted for the variables $ \tilde{x} $ within $ P' $.
If the optional block fails the variables $ \tilde{x} $ of $ P' $ are replaced by the default values $ \tilde{\Args[v]}_d $ instead.
Without loss of generality we assume that the roles $ \tilde{\Role} $ of optional blocks are distinct. Since optional blocks can compute only values and their defaults need to be of the same kind, $ \POptEnd{\Role}{\tilde{\Args[v]}} $ and the defaults cannot carry session names, \ie names used as session channels.
The type system ensures that the inner part $ P $ of a successful optional block reaches some $ \POptEnd{\Role}{\tilde{\Args[v]}} $ and thus transmits computed values of the expected kinds in exactly one of its parallel branches. The semantics presented below ensures that every optional block can transmit at most one vector of computed values and has to fail otherwise.
Similarly optional blocks, that use roles in their inner part $ P $ that are different from $ \tilde{r} $ and are not newly introduced as part of a sub-session within $ P $, cannot be well-typed.
Since optional blocks open a context block around their inner part that separates $ P $ from the continuation $ P' $, scopes as introduced by input prefixes and restriction that are opened within $ P $ cannot cover parts of $ P' $.
If an optional block does not compute any values and consequently the vector of default values is empty, we abbreviate $ \POpt{\Role}{\tilde{\Role}}{P}{\cdot}{\cdot}{Q} $ by $ \POptNV{\Role}{\tilde{\Role}}{P}{Q} $.

Again we usually omit trailing $ \PEnd $.
In Definition~\ref{def:processes} all occurrences of $ \Args $, $ \tilde{\Args} $, and $ \tilde{\Args}_i $ refer to bound names of the respective operators. The set $ \FreeNames{P} $ of free names of $ P $ is the set of names of $ P $ that are not bound.
A substitution $ \Set[]{ \Subst{\Args[y]_1}{\Args_1}, \ldots, \Subst{\Args[y]_n}{\Args_n} } = \Set[]{ \Subst{\tilde{\Args[y]}}{\tilde{\Args}} } $ is a finite mapping from names to names, where the $ \tilde{\Args} $ are pairwise distinct. The application of a substitution on a term $ P\!\Set[]{ \Subst{\tilde{\Args[y]}}{\tilde{\Args}} } $ is defined as the result of simultaneously replacing all free occurrences of $ \Args_i $ by $ \Args[y]_i $, possibly applying alpha-conversion to avoid capture or name clashes. For all names $ n \notin \tilde{x} $ the substitution behaves as the identity mapping.
We use '$ . $' (as \eg in $ \PInp{\Chan}{\tilde{\Args}}{P} $) to denote sequential composition. In all operators the part before '$ . $' guards the continuation after the '$ . $', \ie the continuation cannot reduce before the guard was reduced.
A subprocess of a process is \emph{guarded} if it occurs after such a guard, \ie is the continuation (or part of the continuation) of a guard. Guarded subprocesses can be \emph{unguarded} by steps that remove the guard.

\begin{example}[Implementation of Unreliable Links]
	\label{exa:PUL}
	\begin{align*}
		\PULS{\Role[p]_1}{\Args[v]_1}{\Role[p]_2} &= \POptNVS{\Role[p]_1}{\Role[p]_1, \Role[p]_2}{\PSend{\Chan[s]}{\Role[p]_1}{\Role[p]_2}{\Labe[c]}{\Args[v]_1}{\POptEnd{\Role[p]_1}{\cdot}}}\\
		\PULT{\Role[p]_1}{\Role[p]_2}{\Args[v]_2} &= \POptS{\Role[p]_2}{\Role[p]_1, \Role[p]_2}{\PGet{\Chan[s]}{\Role[p]_1}{\Role[p]_2}{\PLab{\Labe[c]}{\Args}{\POptEnd{\Role[p]_2}{\Args}}}}{\Args[y]}{\Args[v]_2}
	\end{align*}
\end{example}

\noindent
$ \PULS{\Role[p]_1}{\Args[v]_1}{\Role[p]_2} $ is the implementation of a single send action on an unreliable link and $ \PULT{\Role[p]_1}{\Role[p]_2}{\Args[v]_2} $ the corresponding receive action. Here a continuation of the sender cannot gain any information from the modelled communication; not even whether it succeeded, whereas a continuation of the receiver in the case of success obtains the transmitted value $ \Args[v]_1 $ and else its own default value $ \Args[v]_2 $.

To implement the rotating coordinator algorithm of Example~\ref{exa:RCAlgorithm}, we replace the check '\texttt{if alive}$ (p_r) $' by an optional block for communications.
In the first $ i{-}1 $ and the last $ n{-}i $ rounds, participant~$ i $ either receives a value or (if the respective communication fails) uses as default value its value of the round before.
In round~$ i $, participant~$ i $ transmits its current value to each other participant.

\begin{example}[Rotating Coordinator Implementation]
	\label{exa:PRC}
	\begin{align*}
		\PRC{n} ={} & \prod_{i = 1..n} \left( \PPar{\POutS{\Chan_i}{\Chan[s]}}{\PInp{\Chan_i}{\Chan[s]}{\PIN{i}{n}}} \right)\\
		\PIN{i}{n} ={} & ( \bigodot_{j = 1..(i{-}1)} \POptS{\Role[p]_i}{\Role[p]_i, \Role[p]_j}{P_{j \to i, \downarrow}}{\Args[v]_{i, j}}{\Args[v]_{i, j{-}1}} ).\\
		& ( \prod_{j = 1..n, j \neq i} \POptNVS{\Role[p]_i}{\Role[p]_i, \Role[p]_j}{P_{i \to j, \uparrow}}
		\PPar{}{( \bigodot_{j = (i + 1)..n} \POptS{\Role[p]_i}{\Role[p]_i, \Role[p]_j}{P_{j \to i, \downarrow}}{\Args[v]_{i, j}}{\Args[v]_{i, j{-}1}} ))}\\
		P_{j \to i, \downarrow} ={} & \PGet{\Chan[s]}{\Role[p]_j}{\Role[p]_i}{\PLab{\Labe[c]}{\Args[v]_{j, j{-}1}}{\POptEnd{\Role[p]_i}{\Args[v]_{j, j{-}1}}}}\\
		P_{i \to j, \uparrow} ={} & \PSend{\Chan[s]}{\Role[p]_i}{\Role[p]_j}{\Labe[c]}{\Args[v]_{i, i{-}1}}{\POptEnd{\Role[p]_i}{\cdot}}
	\end{align*}
\end{example}

\noindent
The overall system $ \PRC{n} $ consists of the parallel composition of the $ n $ participants. The channel $ \Chan_i $ is used to distribute the initial session channel. Since these communications on $ \tilde{\Chan} $ are used to initialise the system and not to model the algorithm, we assume that they are reliable.

The term $ \PIN{i}{n} $ models participant $ i $. Each participant first optionally receives $ i{-}1 $ times a value from another participant. Therefore, an optional block surrounds the term $ P_{j \to i, \downarrow} $. If the optional block succeeds, then $ P_{j \to i, \downarrow} $ receives the value $ \Args[v]_{j, j{-}1} $ from the current leader of the round and finishes its optional block with the transmission of the computed value $ \POptEnd{\Role[p]_i}{\Args[v]_{j, j{-}1}} $.
In this case, the value $ \Args[v]_{i, j} $ is instantiated with the received value $ \Args[v]_{j, j{-}1} $. If the communication with the current leader $ \Role[p]_j $ fails, then the value $ \Args[v]_{i, j} $ is instantiated instead with the default value $ \Args[v]_{i, j{-}1} $, \ie the last value of participant~$ i $.
In these first $ i{-}1 $ (and also the last $ n{-}i $) rounds, participant~$ i $ is a receiver and consists of exactly one optional block per round.
The last $ n{-}i $ rounds of participant~$ i $ are similar.

In round~$ i $, participant~$ i $ is the sender and consists of $ n{-}1 $ optional blocks (second line of the definition of $ \PIN{i}{n} $), exactly one such block with each other participant.
Here these $ n{-}1 $ optional blocks do not need a default value and accordingly do not compute a value. Note that the continuation of all these $ n{-}1 $ blocks of the coordinator is $ \PEnd $.
For each other participant~$ j $ these blocks surround the term $ P_{i \to j, \uparrow} $ in which $ \Role[p]_i $ transmits its current value $ \Args[v]_{i, i{-}1} $ to $ \Role[p]_j $.
Note that, in round~$ i $, participant~$ i $ does not need to update its own value, since it gains no new information. Therefore, we assumed $ \Args[v]_{i, i} := \Args[v]_{i, i{-}1} $ for all $ i < n $ in the assumed vectors of values.

In Example~\ref{exa:PRC}, the $ n{-}1 $ optional blocks of round~$ i $ are pairwise in parallel and parallel to the optional block of round~$ i {+} 1 $, which guards the block of round~$ i {+} 2 $ and so forth.
This matches an intuitive understanding of this process in terms of asynchronous communications. The sending operations emit the respective value as soon as they are unguarded, but they syntactically remain part of the term until the (possibly later) reception of the respective message consumes it.
In fact, the presented session calculus is synchronous but the examples---including the examples with sub-sessions presented later---can be considered as distributed asynchronous processes, because they use neither choice nor output continuations different from $ \PEnd $ \cite{hondaTokoro91,boudol92,palamidessi03} and because optional blocks of senders have no default values and each send action matches exactly one receive action and vice versa \cite{fossacs12_pi}.

\subsection{A Session Calculus with Optional Blocks and Sub-Sessions}

Again we extend the session calculus, in order to obtain a mechanism to express modularity.
\cite{DemangeonHonda12} introduces three operators for this purpose.
\begin{align*}
	\sep \PDecl{\Chan[k]}{\Chan[s]}{\tilde{\Args[v]}}{\tilde{\Chan}}{\tilde{\Role}}{P}
	\sep \PEnt{\Chan[s]}{\Role_1}{\Role_2}{\Role_3}{\Args}{P} \sep \PReq{\Chan[s]}{\Role_1}{\Role_2}{\Role_3}{\Chan[k]}{P}
\end{align*}
$ \PDecl{\Chan[k]}{\Chan[s]}{\tilde{\Args[v]}}{\tilde{\Chan}}{\tilde{\Role}}{P} $ allows a process to create a sub-session $ \Chan[k] $, where $ \Chan[s] $ is the parent session, $ \tilde{\Args[v]} $ are arguments, $ \tilde{\Role} $ are external participants, and $ \tilde{\Chan} $ are the channels for external invitations.
Internal invitations are handled by $ \PEnt{\Chan[s]}{\Role_1}{\Role_2}{\Role_3}{\Args}{P} $ and $ \PReq{\Chan[s]}{\Role_1}{\Role_2}{\Role_3}{\Chan[k]}{P} $, where $ \Role_1 $ invites $ \Role_2 $ to play role $ \Role_3 $ in a sub-session.
Here $ \Args $ is a name that is bounded in $ P $ within the operator $ \PEnt{\Chan[s]}{\Role_1}{\Role_2}{\Role_3}{\Args}{P} $. All other names of these three operators are free. Again the '$ . $' is used to refer to sequential composition, \ie in all three operators the respective continuation $ P $ is guarded.

\begin{definition}[Processes with Sub-Sessions]
	\label{def:processesWSS}
	\begin{align*}
		P & \deffTerms \PInp{\Chan}{\tilde{\Args}}{P} \sep \POut{\Chan}{\tilde{\Chan[s]}}{P}
		\sep \PGet{\Chan[k]}{\Role_1}{\Role_2}{_{i \in \indexSet} \Set{ \PLab{\Labe_i}{\tilde{\Args}_i}{P_i} }} \sep \PSend{\Chan[k]}{\Role_1}{\Role_2}{\Labe}{\tilde{\Args[v]}}{P}\\
		& \sep \textcolor{blue}{\POpt{\Role}{\tilde{\Role}}{P}{\tilde{\Args}}{\tilde{\Args[v]}}{P'}} \sep \textcolor{blue}{\POptEnd{\Role}{\tilde{\Args[v]}}}\\
		& \sep \PDecl{\Chan[k]}{\Chan[s]}{\tilde{\Args[v]}}{\tilde{\Chan}}{\tilde{\Role}}{P} \sep \PEnt{\Chan[s]}{\Role_1}{\Role_2}{\Role_3}{\Args}{P} \sep \PReq{\Chan[s]}{\Role_1}{\Role_2}{\Role_3}{\Chan[k]}{P}\\
		& \sep \PRes{\Args}{P} \sep \PChoi{P_1}{P_2} \sep \PPar{P_1}{P_2}
		\sep \PRec{\TermV[X]}{P} \sep \PVar{\TermV[X]} \sep \PEnd
	\end{align*}
\end{definition}

Similar to Example~\ref{exa:PRC}, we present an implementation of the rotating coordinators with a sub-session for each round.

\begin{example}[Rotating Coordinators with Sub-Sessions]
	\label{exa:PRCWSS}
	\begin{align*}
		\PRCB{n} ={} & \prod_{i = 1..n} \left( \PPar{\POutS{\Chan_i}{\Chan[s]}}{\PInp{\Chan_i}{\Chan[s]}{\PINB{i}{n}}} \right)\\
		\PINB{i}{n} ={} & \bigodot_{j = 1..(i{-}1)} \PEnt{\Chan[s]}{\Role[p]_j}{\Role[p]_i}{\Role[trg]_i}{\Args}{P'_{\downarrow}\!\left( \Role[trg]_i, i, j \right)}.\\
		& \quad \big( \PRes{\Chan[k]}{\left( \PDecl{\Chan[k]}{\Chan[s]}{\Args[v]_{i, i{-}1}}{\cdot}{\cdot}{P'_{\uparrow}\!\left( i \right)} \right)}
		\PPar{}{\big( \bigodot_{j = (i + 1)..n} \PEnt{\Chan[s]}{\Role[p]_j}{\Role[p]_i}{\Role[trg]_{i{-}1}}{\Args}{P'_{\downarrow}\!\left( \Role[trg]_{i{-}1}, i, j \right)} \big)} \big)\\
		P'_{\downarrow}\!\left( \Role, i, j \right) ={} & \POptS{\Role}{\Role, \Role[src]}{\PGet{\Args}{\Role[src]}{\Role}{\PLab{\Labe[c]}{\Args[y]}{\POptEnd{\Role}{\Args[y]}}}}{\Args[v]_{i, j}}{\Args[v]_{i, j{-}1}}\\
		P'_{\uparrow}\!\left( i \right) ={} & \PPar{\PEnt{\Chan[s]}{\Role[p]_i}{\Role[p]_i}{\Role[src]}{\Args}{P'_{\Prot[R]}\!\left( i, \Args \right)}}{\PReqS{\Chan[s]}{\Role[p]_i}{\Role[p]_i}{\Role[src]}{\Chan[k]}}
		\PPar{}{\PPar{\prod_{j = 1..i} \PReqS{\Chan[k]}{\Role[p]_i}{\Role[p]_j}{\Role[trg]_j}{\Chan[k]}}{\prod_{j = i..(n{-}1)} \PReqS{\Chan[s]}{\Role[p]_i}{\Role[p]_{j + 1}}{\Role[trg]_j}{\Chan[k]}}}\\
		P'_{\Prot[R]}\!\left( i, \Args \right) ={} & \prod_{j = 1..(n{-}1)} \!\!\!\!\!\POptNVS{\Role[scr]}{\Role[src], \Role[trg]_j}{\PSend{\Args}{\Role[src]}{\Role[trg]_j}{\Labe[c]}{\Args[v]_{i, i{-}1}}{\POptEnd{\Role[src]}{\cdot}}}
	\end{align*}
\end{example}

\noindent
$ \PRCB{n} $ is the implementation of the rotating coordinator algorithm using a sub-session for each round in the global type. Accordingly the differences between the Examples~\ref{exa:PRC} and \ref{exa:PRCWSS} are due to the initialisation of the sub-sessions.
In each round in which participant~$ i $ is not the coordinator, participant~$ i $ first accepts the invitation of the current coordinator to participate as receiver in the sub-session of the round and then receives and updates its value similar to Example~\ref{exa:PRC}.
If participant~$ i $ is itself the coordinator, then it initialises a new session~$ k $ for the round and then, in parallel, invites all processes (including itself) to participate and accepts to participate in this sub-session as sender followed by the $ n{-}1 $ optional blocks to transmit its value similar to Example~\ref{exa:PRC}.

\subsection{Reduction Semantics}

We identify processes up structural congruence, where structural congruence is defined by the rules:
\begin{center}
\begin{tabular}{c}
	$ \PPar{P}{\PEnd} \equiv P $ \quad\quad
	$ \PPar{P_1}{P_2} \equiv \PPar{P_2}{P_1} $ \quad\quad
	$ \PPar{P_1}{\left( \PPar{P_2}{P_3} \right)} \equiv \PPar{\left( \PPar{P_1}{P_2} \right)}{P_3} $ \vspace{0.3em}\\
	$ \PRec{\TermV[X]}{P} \equiv P\!\Set[]{ \Subst{\PRec{\TermV[X]}{P}}{\TermV[X]} } $ \quad\quad
	$ \PChoi{P_1}{P_2} \equiv \PChoi{P_2}{P_1} $ \vspace{0.3em}\\
	$ \PRes{\Args}{\PEnd} \equiv \PEnd $ \quad\quad
	$ \PRes{\Args}{\PRes{\Args[y]}{P}} \equiv \PRes{\Args[y]}{\PRes{\Args}{P}} $ \quad\quad
	$ \PRes{\Args}{\left( \PPar{P_1}{P_2} \right)} \equiv \PPar{P_1}{\PRes{\Args}{P_2}} $ \; if $ \Args \notin \FreeNames{P_1} $
\end{tabular}
\end{center}

In \cite{DemangeonHonda12} the semantics is given by a set of reduction rules that are defined \wrt evaluation contexts.
We extend them with optional blocks.

\begin{definition}
	\label{def:evaluationContexts}
	$ \EC \deffTerms \ECHole \sep \ECPar{P}{\EC} \sep \ECRes{\Args}{\EC} \sep \POpt{\Role}{\tilde{\Role}}{\EC}{\tilde{\Args}}{\tilde{\Args[v]}}{P'} $
\end{definition}

\noindent
Intuitively an evaluation context is a term with a single hole that is not guarded.
Additionally, we introduce two variants of evaluation contexts and a context for blocks that are used to simplify the presentation of our new rules.

\begin{definition}
	$ \ECR \deffTerms \ECHole \sep \ECPar{P}{\ECR} \sep \POpt{\Role}{\tilde{\Role}}{\ECR}{\tilde{\Args}}{\tilde{\Args[v]}}{P'} $\\
	$ \ECO \deffTerms \POpt{\Role}{\tilde{\Role}}{\ECP}{\tilde{\Args}}{\tilde{\Args[v]}}{P'} $,  where $ \ECP \deffTerms \ECHole \sep \ECPar{P}{\ECP} $
\end{definition}

\noindent
Accordingly, a $ \ECO $-context consists of exactly one optional block that contains an $ \ECP $-context, \ie a single hole that can occur within the parallel composition of arbitrary processes.
We define the function $ \RolesOf{\POpt{\Role}{\tilde{\Role}}{\ECP}{\tilde{\Args}}{\tilde{\Args[v]}}{P'}} \deff \tilde{\Role} $, to return the roles of the optional block of a $ \ECO $-context, and the function $ \OwnerOf{\POpt{\Role}{\tilde{\Role}}{\ECP}{\tilde{\Args}}{\tilde{\Args[v]}}{P'}} \deff \Role $, to return its owner.

Figure~\ref{fig:reductionRules} presents all reduction rules for the two introduced versions of session calculi: both come with optional blocks, but the first one without sub-sessions, while the second one including sub-sessions.
For the simpler session calculus we need the Rules~$ (\mathsf{comS}) $, $ (\mathsf{choice}) $, and $ (\mathsf{comC}) $ to deal with the standard operators for communication, choice, and external invitations to sessions, respectively.
Since evaluation contexts $ \EC $ contain optional blocks, these rules allow for steps within a single optional block.
To capture optional blocks for this first session calculus, we introduce the new Rules~$ (\mathsf{fail}) $, $ (\mathsf{succ}) $, $ (\mathsf{cSO}) $, and $ (\mathsf{cCO}) $.
For the second session calculus \cite{DemangeonHonda12} add the Rules~$ (\mathsf{subs}) $ and $ (\mathsf{join}) $ to deal with sub-sessions and we introduce the new Rule~$ (\mathsf{jO}) $ to capture sub-sessions within optional blocks.

Here $ \dot{=} $ means that the two compared vectors contain the same roles but not necessarily in the same order, \ie $ \dot{=} $ checks whether the set of participants of two optional blocks are the same.

\begin{figure*}[tp]
	\[ \begin{array}{c}
		(\mathsf{comS}) \dfrac{j \in \indexSet}{\AEC{\PPar{\PSend{\Chan[k]}{\Role_1}{\Role_2}{\Labe_j}{\tilde{\Args[v]}}{P}}{\PGet{\Chan[k]}{\Role_1}{\Role_2}{_{i \in \indexSet} \Set{ \PLab{\Labe_i}{\tilde{\Args}_i}{P_i} }}}} \longmapsto \AEC{\PPar{P}{P_j\!\Set[]{ \Subst{\tilde{\Args[v]}}{\tilde{\Args}_j} }}}} \hspace{2em}
		(\mathsf{choice}) \dfrac{P_i \longmapsto P_i'}{\AEC{\PChoi{P_1}{P_2}} \longmapsto \AEC{P_i'}} \vspace*{0.75em}\\
		(\mathsf{comC}) \dfrac{}{\AEC{\PPar{\POut{\Chan}{\tilde{\Chan[s]}}{P_1}}{\PInp{\Chan}{\tilde{\Args}}{P_2}}} \longmapsto \AEC{\PPar{P_1}{P_2\!\Set[]{ \Subst{\tilde{\Chan[s]}}{\tilde{\Args}} }}}} \vspace*{0.75em}\\
		(\mathsf{subs}) \dfrac{\tilde{\Role} = \left( \Role_1, \ldots, \Role_n \right) \quad \tilde{\Chan} = \left( \Chan_1, \ldots, \Chan_n \right)}{\AEC{\PDecl{\Chan[k]}{\Chan[s]}{\tilde{\Args[v]}}{\tilde{\Chan}}{\tilde{\Role}}{P}} \longmapsto \AEC{\PPar{P}{\PPar{\POutS{\Chan_1}{\Chan[s]}}{\PPar{\ldots}{\POutS{\Chan_n}{\Chan[s]}}}}}} \vspace*{0.75em}\\
		(\mathsf{join}) \dfrac{}{\AEC{\PPar{\PReq{\Chan[s]}{\Role_1}{\Role_2}{\Role_3}{\Chan[k]}{P_1}}{\PEnt{\Chan[s]}{\Role_1}{\Role_2}{\Role_3}{\Args}{P_2}}} \longmapsto \AEC{\PPar{P_1}{P_2\!\Set[]{ \Subst{\Chan[k]}{\Args} }}}} \vspace*{0.75em}\\
		\textcolor{blue}{(\mathsf{fail}) \dfrac{}{\AEC{\POpt{\Role}{\tilde{\Role}}{P}{\tilde{\Args}}{\tilde{\Args[v]}_d}{P'}} \longmapsto \AEC{P'\!\Set[]{ \Subst{\tilde{\Args[v]}_d}{\tilde{\Args}} }}}} \hspace{2em}
		\textcolor{blue}{(\mathsf{succ}) \dfrac{}{\AEC{\POpt{\Role}{\tilde{\Role}}{\POptEnd{\Role}{\tilde{\Args[v]}}}{\tilde{\Args}}{\tilde{\Args[v]}_d}{P}} \longmapsto \AEC{P\!\Set[]{ \Subst{\tilde{\Args[v]}}{\tilde{\Args}} }}}} \vspace*{0.75em}\\
		\textcolor{blue}{(\mathsf{cSO}) \dfrac{j \in \indexSet \quad \RolesOf{\ECO} \dot{=} \RolesOf{\ECO[C']} \quad \OwnerOf{\ECO} = \Role_1 \quad \OwnerOf{\ECO[C']} = \Role_2}{\begin{array}{c} \AEC{\PPar{\AECR{\AECO{\PSend{\Chan[k]}{\Role_1}{\Role_2}{\Labe_j}{\tilde{\Args[v]}}{P}}}}{\AECR[E']{\AECO[C']{\PGet{\Chan[k]}{\Role_1}{\Role_2}{_{i \in \indexSet} \Set{ \PLab{\Labe_i}{\tilde{\Args}_i}{P_i} }}}}}}\\ \longmapsto \AEC{\PPar{\AECR{\AECO{P}}}{\AECR[E']{\AECO[C']{P_j\!\Set[]{ \Subst{\tilde{\Args[v]}}{\tilde{\Args}_j} }}}}} \end{array}}} \vspace*{0.75em}\\
		\textcolor{blue}{(\mathsf{cCO}) \dfrac{\RolesOf{\ECO} \dot{=} \RolesOf{\ECO[C']}}{\AEC{\PPar{\AECR{\AECO{\POut{\Chan}{\tilde{\Chan[s]}}{P_1}}}}{\AECR[E']{\AECO[C']{\PInp{\Chan}{\tilde{\Args}}{P_2}}}}} \longmapsto \AEC{\PPar{\AECR{\AECO{P_1}}}{\AECR[E']{\AECO[C']{P_2\!\Set[]{ \Subst{\tilde{\Chan[s]}}{\tilde{\Args}} }}}}}}} \vspace*{0.75em}\\
		\textcolor{blue}{(\mathsf{jO}) \dfrac{\RolesOf{\ECO} \dot{=} \RolesOf{\ECO[C']} \quad \OwnerOf{\ECO} = \Role_1 \quad \OwnerOf{\ECO[C']} = \Role_2}{\AEC{\PPar{\AECR{\AECO{\PReq{\Chan[s]}{\Role_1}{\Role_2}{\Role_3}{\Chan[k]}{P_1}}}}{\AECR[E']{\AECO[C']{\PEnt{\Chan[s]}{\Role_1}{\Role_2}{\Role_3}{\Args}{P_2}}}}} \longmapsto \AEC{\PPar{\AECR{\AECO{P_1}}}{\AECR[E']{\AECO[C']{P_2\!\Set[]{ \Subst{\Chan[k]}{\Args} }}}}}}}
	\end{array} \]
	\caption{Reduction Rules}
	\label{fig:reductionRules}
\end{figure*}

The Rules~(\textsf{comS}), (\textsf{comC}), and (\textsf{join}) represent three different kinds of communication. They define communications within a session, external session invitations, and internal session invitations, respectively.
In all three cases communication is an axiom that requires the occurrence of two matching counterparts of communication primitives (of the respective kind) to be placed in parallel within an evaluation context. As a consequence of the respective communication step the continuations of the communication primitives are unguarded and the values transmitted in the communication step are instantiated (substituted) in the receiver continuation.
(\textsf{choice}) allows the reduction of either side of a choice, if the respective side can perform a step.
(\textsf{subs}) initialises a sub-session by transmitting external invitations.

The two rules (\textsf{succ}) and (\textsf{fail}) describe the main features of optional blocks, namely how they succeed (\textsf{succ}) and what happens if they fail (\textsf{fail}).
(\textsf{fail}) aborts an optional block, \ie removes it and unguards its continuation instantiated with the default values.
This rule can be applied whenever an optional block is unguarded, \ie there is no way to ensure that an optional block does indeed perform any step (or terminates after successfully doing some of its steps).
In combination with (\textsf{succ}), it introduces the non-determinism that is used to express the random nature in that system errors may occur.
If we use optional blocks the cover a single transmission over an unreliable link, each use of the Rule~(\textsf{fail}) refers to a single link failure.

(\textsf{succ}) is the counterpart of (\textsf{fail}); it removes a successfully completed optional block and unguards its continuation instantiated with the computed results.
To successfully complete an optional block, we require that its content has to reduce to a single occurrence of $ \POptEnd{\Role}{\tilde{\Args[v]}} $, where $ \Role $ is the owner of the block and accordingly one of the participating roles.
Since (\textsf{succ}) and (\textsf{fail}) are the only ways to reduce $ \POptEnd{\Role}{\tilde{\Args[v]}} $, this ensures that a successful optional block can compute only a single vector of return values. Other parallel branches in the inner part of an optional block have to terminate with $ \PEnd $.
This ensures that no confusion can arise from the computation of different values in different parallel branches. Since at the process-level an optional block covers only a single participant, this limitation does not restrict the expressive power of the considered processes.
If the content of an optional block cannot reduce to $ \POptEnd{\Role}{\tilde{\Args[v]}} $ the optional block is doomed to fail.

The remaining rules describe how different optional blocks can interact.
Here, we need to ensure that communication from within an optional block ensures isolation, \ie that such communications are restricted to the encapsulated parts of other optional blocks.
The $ \ECR $-contexts allow for two such blocks to be nested within different optional blocks.
The exact definition of such a communication rule depends on the semantics of the considered calculi and their communication rules. Here there are the Rules (\textsf{cSO}), (\textsf{cCO}), and (\textsf{jO}). They are the counterparts of (\textsf{comS}), (\textsf{comC}), and (\textsf{join}) and accordingly allow for the respective kind of communication step. As an example consider Rule (\textsf{cSO}). In comparison to (\textsf{comS}), Rule~(\textsf{cSO}) ensures that communications involving the content of an optional block are limited to two such contents of optional blocks with the same participants.
This ensures that optional blocks describe the local view-points of the encapsulated protocol.

Optional blocks do not allow for scope extrusion of restricted names, \ie a name restricted within an optional block cannot be transmitted nor can an optional block successfully be terminated if the computed result values are subject to a restriction from the content of the optional block.
Also values that are communicated between optional blocks can be used only by the continuation of the optional block and only if the optional block was completed successfully.
If an optional block fails while another process is still waiting for a communication within its optional block, the latter optional block is doomed to fail.
Note that the semantics of optional blocks is inherently synchronous, since an optional sending operation can realise the failing of its matching receiver (\eg by $ \POpt{\Role_1}{\Role_2}{\ldots\POptEnd{\Role_1}{\Args[ok]}}{\Args}{\Args[fail]}{P} $).

Let $ \longmapsto^+ $ denote the transitive closure of $ \longmapsto $ and let $ \longmapsto^* $ denote the reflexive and transitive closure of $ \longmapsto $, respectively.

\subsubsection{Reaching Consensus Despite Crash Failures}

To illustrate the semantics of optional blocks and our implementation of the rotating coordinator algorithm (Example~\ref{exa:PRC}), we present one execution for the case of $ n = 3 $. Assume that $ \Args[v]_{1, 0} = 0 $, $ \Args[v]_{2, 0} = 1 = \Args[v]_{3, 0} $, and, since the coordinator does not update its value, $ \Args[v]_{i, i} = \Args[v]_{i, i-1} $ for all $ 1 \leq i \leq 3 $. Moreover, assume that participant~$ 1 $ crashes in round~$ 1 $ after delivering its value to participant~$ 3 $ but before participant~$ 2 $ obtains the value.
\begin{align*}
	\PRC{3} &= \PPar{\left( \PPar{\POutS{\Chan_1}{\Chan[s]}}{\PInp{\Chan_1}{\Chan[s]}{\PIN{1}{3}}} \right)}{\PPar{\left( \PPar{\POutS{\Chan_2}{\Chan[s]}}{\PInp{\Chan_2}{\Chan[s]}{\PIN{2}{3}}} \right)}{\left( \PPar{\POutS{\Chan_3}{\Chan[s]}}{\PInp{\Chan_3}{\Chan[s]}{\PIN{3}{3}}} \right)}}\\
	&\longmapsto^3 \PPar{\PIN{1}{3}}{\PPar{\PIN{2}{3}}{\PIN{3}{3}}}
\end{align*}
where
\begin{align*}
	\PIN{i}{n} ={} & ( \bigodot_{j = 1..(i{-}1)} \POptS{\Role[p]_i}{\Role[p]_i, \Role[p]_j}{P_{j \to i, \downarrow}}{\Args[v]_{i, j}}{\Args[v]_{i, j{-}1}} ).\\
	& ( \prod_{j = 1..n, j \neq i} \POptNVS{\Role[p]_i}{\Role[p]_i, \Role[p]_j}{P_{i \to j, \uparrow}}
	\PPar{}{( \bigodot_{j = (i + 1)..n} \POptS{\Role[p]_i}{\Role[p]_i, \Role[p]_j}{P_{j \to i, \downarrow}}{\Args[v]_{i, j}}{\Args[v]_{i, j{-}1}} ))}\\
	P_{j \to i, \downarrow} ={} & \PGet{\Chan[s]}{\Role[p]_j}{\Role[p]_i}{\PLab{\Labe[c]}{\Args[v]_{j, j{-}1}}{\POptEnd{\Role[p]_i}{\Args[v]_{j, j{-}1}}}}\\
	P_{i \to j, \uparrow} ={} & \PSend{\Chan[s]}{\Role[p]_i}{\Role[p]_j}{\Labe[c]}{\Args[v]_{i, i{-}1}}{\POptEnd{\Role[p]_i}{\cdot}}
\end{align*}
The first three steps initialise the session using three times Rule~$ (\mathsf{comC}) $. We assume here that these steps belong to the environment and do never fail. If one of these steps fails, the respective participant does not know the global session channel and thus cannot participate in the algorithm, \ie is crashed from the beginning.
After the initialisation all participants consist of sequential and parallel optional blocks. Each of these optional blocks can fail any time. The coordinator of the first round $ \Role[p]_1 $ can transmit its value to the other two participants.
Since the respective two blocks of the sender and each of the matching blocks of the two receivers are all in parallel, both communications can happen. Intuitively, by unguarding a sending operation, we can consider the message as already being emitted by the sender. It remains syntactically present until the receiver captures it to complete the transmission. Accordingly, $ \Role[p]_1 $ directly moves to round~$ 2 $.

\begin{align*}
	\longmapsto^3 \PPar{\PIN{1}{3}'}{\PPar{\PIN{2}{3}}{\PIN{3}{3}'}}
\end{align*}
where
\begin{align*}
	\PIN{1}{3}' &= \PPar{\POptNVS{\Role[p]_1}{\Role[p]_1, \Role[p]_2}{P_{1 \to 2, \uparrow}}}{\POpt{\Role[p]_1}{\Role[p]_1, \Role[p]_2}{P_{2 \to 1, \downarrow}}{\Args[v]_{1, 2}}{0}{\POptS{\Role[p]_1}{\Role[p]_1, \Role[p]_3}{P_{3 \to 1, \downarrow}}{\Args[v]_{1, 3}}{\Args[v]_{1, 2}}}}\\
	\PIN{3}{3}' &= \POpt{\Role[p]_3}{\Role[p]_3, \Role[p]_2}{P_{2 \to 3, \downarrow}}{\Args[v]_{3, 2}}{0}{\left( \PPar{\POptNVS{\Role[p]_3}{\Role[p]_3, \Role[p]_1}{P_{3 \to 1, \uparrow}}}{\POptNVS{\Role[p]_3}{\Role[p]_3, \Role[p]_2}{P_{3 \to 2, \uparrow}}} \right)}
\end{align*}
Next $ \Role[p]_3 $ receives the value~$ 0 $ from $ \Role[p]_1 $ using Rule~$ (\mathsf{cSO}) $. Then both optional blocks that participate in this communication are completed successfully using Rule~$ (\mathsf{succ}) $ such that $ \Role[p]_3 $ updates its current value to the received $ 0 $. This completes the first round for $ \Role[p]_3 $ and it moves to round~$ 2 $.
The remainder of $ \Role[p]_1 $, \ie $ \PIN{1}{3}' = \PPar{\POptNVS{\Role[p]_1}{\Role[p]_1, \Role[p]_2}{\ldots}}{\POpt{\Role[p]_1}{\Role[p]_1, \Role[p]_2}{\ldots}{\Args[v]_{1, 2}}{0}{\ldots}} $, consists of the remaining optional block towards $ \Role[p]_2 $ in parallel with the optional block for round~$ 2 $ of $ \Role[p]_1 $ that guards the optional block for round~$ 3 $.
$ \PIN{3}{3}' = \POpt{\Role[p]_3}{\Role[p]_3, \Role[p]_2}{\ldots}{\Args[v]_{3,2}}{0}{\ldots} $ is guarded by its optional block to receive in round~$ 2 $, where its current value $ \Args[v]_{3, 1} $ was instantiated with $ 0 $ received from $ \Role[p]_1 $.

\begin{align*}
	\longmapsto^5 \PPar{\PIN{2}{3}'}{\PIN{3}{3}'}
\end{align*}
where
\begin{align*}
	\PIN{2}{3}' = \PPar{\POptNVS{\Role[p]_2}{\Role[p]_2, \Role[p]_3}{P_{2 \to 3, \uparrow}}}{\POptS{\Role[p]_2}{\Role[p]_2, \Role[p]_3}{P_{3 \to 2, \downarrow}}{\Args[v]_{2, 3}}{1}}
\end{align*}
Now $ \Role[p]_1 $ crashes, \ie its three remaining optional blocks are removed using Rule~$ (\mathsf{fail}) $ three times. Since $ \Role[p]_2 $ cannot receive a value from $ \Role[p]_1 $ after that we also remove its optional block of round~$ 1 $ using the default value $ \Args[v]_{2, 0} = 1 $ to instantiate $ \Args[v]_{2, 1} = \Args[v]_{2, 2} $.
With that $ \Role[p]_2 $ moves to round~$ 2 $, unguards its two optional blocks to transmit its current value, \ie $ \Args[v]_{2,0} = \Args[v]_{2, 1} = 1 $, and, by unguarding also its optional block of round~$ 3 $, directly moves forward to round~$ 3 $.
Since $ \Role[p]_1 $ is crashed, the first block of $ \Role[p]_2 $ in round~$ 2 $ towards $ \Role[p]_1 $ is doomed to fail causing another application of Rule~$ (\mathsf{fail}) $.
As result we obtain $ \PIN{2}{3}' = \PPar{\POptNVS{\Role[p]_2}{\Role[p]_2, \Role[p]_3}{\ldots}}{\POptS{\Role[p]_2}{\Role[p]_2, \Role[p]_3}{\ldots}{\Args[v]_{2,3}}{1}} $.

\begin{align*}
	\longmapsto^3 \PPar{\PIN{2}{3}''}{\PIN{3}{3}''}
\end{align*}
where
\begin{align*}
	\PIN{2}{3}'' &= \POptS{\Role[p]_2}{\Role[p]_2, \Role[p]_3}{P_{3 \to 2, \downarrow}}{\Args[v]_{2, 3}}{1}\\
	\PIN{3}{3}'' &= \PPar{\POptNVS{\Role[p]_3}{\Role[p]_3, \Role[p]_1}{P_{3 \to 1, \uparrow}}}{\POptNVS{\Role[p]_3}{\Role[p]_3, \Role[p]_2}{P_{3 \to 2, \uparrow}}}
\end{align*}
$ \Role[p]_3 $ completes round~$ 2 $ by receiving the value $ 1 $ that was transmitted by $ \Role[p]_2 $ in round~$ 2 $ using Rule~$ (\mathsf{cSO}) $. After this communication the respective two optional blocks are resolved by Rule~$ (\mathsf{succ}) $ that also updates the current value of $ \Role[p]_3 $ to $ \Args[v]_{3, 2} = 1 $. With that $ \Role[p]_3 $ moves to round~$ 3 $. As results we obtain $ \PIN{2}{3}'' = \POptS{\Role[p]_2}{\Role[p]_2, \Role[p]_3}{\ldots}{\Args[v]_{2,3}}{1} $ for $ \Role[p]_2 $ and for $ \Role[p]_3 $ we obtain $ \PIN{3}{3}'' = \PPar{\POptNVS{\Role[p]_3}{\Role[p]_3, \Role[p]_1}{\ldots}}{\POptNVS{\Role[p]_3}{\Role[p]_3, \Role[p]_2}{\ldots}} $.

\begin{align*}
	\longmapsto^4 \PEnd
\end{align*}
Similarly, round~$ 3 $ is completed by the reception of $ 1 $ by $ \Role[p]_2 $ from $ \Role[p]_3 $ and two steps to resolve the optional blocks. Additionally the remaining block of $ \Role[p]_3 $ towards the crashed $ \Role[p]_1 $ is removed using Rule~$ (\mathsf{fail}) $. The last values of $ \Role[p]_2 $ and $ \Role[p]_3 $ were $ 1 $. Since $ \Role[p]_1 $ (that still holds $ 0 $) crashed, this solves Consensus (although we abstract from the outputs of the results).

This example visualises how the rotating coordinator algorithm allows processes to reach Consensus despite crash failures.
Please observe that, due to the asynchronous nature of the processes, rounds can overlap, \ie there are derivatives in which the participants are situated in different rounds.
Overlapping rounds are an important property of round-based distributed algorithms that significantly complicate their analysis. Hence it is important to model them properly.
We gain overlapping rounds by
\begin{inparaenum}[(1)]
	\item placing the optional blocks of senders in parallel and
	\item (for the case of sub-sessions as visualised in the next section) using the sub-sessions of \cite{DemangeonHonda12} that similarly place acceptance notions (and thus the content of sub-sessions) and the continuation of this sub-session in parallel.
\end{inparaenum}

\subsubsection{Reaching Consensus with Sub-Sessions}

To show that the overlapping of rounds for these cases is the same, we map the above reduction of Example~\ref{exa:PRC} on Example~\ref{exa:PRCWSS}.

In contrast to the first example the second wraps each round within a sub-session.
We start again with
\begin{align*}
	\PRCB{3} \longmapsto^3{} & \PPar{\PIN{1}{3}}{\PPar{\PIN{2}{3}}{\PIN{3}{3}}}
\end{align*}
to initialise the parent session using three times Rule~$ (\mathsf{comC}) $.
This unguards the first sub-session call that is performed by the first co-ordinator $ \Role[p]_i $ to initialise a sub-session for round~$ 1 $.
Accordingly, in the following four steps
\begin{align*}
	\longmapsto^4{} & \PPar{\PINB{1}{3}}{\PPar{\PINB{2}{3}}{\PINB{3}{3}}}
\end{align*}
$ \Role[p]_1 $ calls the sub-session using Rule~$ (\mathsf{subs}) $ and unguards the corresponding three internal session invitations $ \PReqS{\Chan[s]}{\Role[p]_1}{\Role[p]_1}{\Role[scr]}{\Chan[k]} $, $ \PReqS{\Chan[s]}{\Role[p]_1}{\Role[p]_2}{\Role[trg]_1}{\Chan[k]} $, and $ \PReqS{\Chan[s]}{\Role[p]_1}{\Role[p]_3}{\Role[trg]_2}{\Chan[k]} $ that are answered using three applications of Rule~$ (\mathsf{join}) $.
Therefore the acceptance notification $ \PEntS{\Chan[s]}{\Role[p]_1}{\Role[p]_1}{\Role[src]}{\Args} $ of $ \Role[p]_1 $ is unguarded by Rule~$ (\mathsf{subs}) $ and the other two acceptance notifications $ \PEntS{\Chan[s]}{\Role[p]_2}{\Role[p]_1}{\Role[trg]_1}{\Args} $ and $ \PEntS{\Chan[s]}{\Role[p]_3}{\Role[p]_1}{\Role[trg]_2}{\Args} $ are unguarded in the first three steps.
Again $ \Role[p]_1 $ directly moves to round~$ 2 $.

The following three steps
\begin{align*}
	\longmapsto^3{} & \PPar{\PIN[P'']{1}{3}}{\PPar{\PIN[P']{2}{3}}{\PIN[P'']{3}{3}}}
\end{align*}
are similar to the first example.
$ \Role[p]_3 $ receives the value~$ 0 $ from $ \Role[p]_1 $ using Rule~$ (\mathsf{cSO}) $. Then both optional blocks that participate in this communication are completed successfully using Rule~$ (\mathsf{succ}) $ such that $ \Role[p]_3 $ updates its current value to the received $ 0 $. This completes the first round for $ \Role[p]_3 $ and it moves to round~$ 2 $.

Now $ \Role[p]_1 $ crashes, \ie its three remaining optional blocks will be removed using Rule~$ (\mathsf{fail}) $ as soon as they are unguarded but $ \Role[p]_1 $ still answers session invitations. One optional block of $ \Role[p]_1 $ is already unguarded and thus removed.
Since $ \Role[p]_2 $ cannot receive a value from $ \Role[p]_1 $ after that we also remove its optional block of round~$ 1 $.
\begin{align*}
	\longmapsto^2{} & \PPar{\PIN[P''']{1}{3}}{\PPar{\PIN[P'']{2}{3}}{\PIN[P'']{3}{3}}}
\end{align*}
With that $ \Role[p]_2 $ moves to round~$ 2 $.

Next $ \Role[p]_2 $ initialises the sub-session for round~$ 2 $ in the same way as $ \Role[p]_1 $ did for round~$ 1 $ and $ \Role[p]_1 $ removes the next optional block using Rule~$ (\mathsf{fail}) $
\begin{align*}
	\longmapsto^5{} & \PPar{\PIN[P'''']{1}{3}}{\PPar{\PIN[P''']{2}{3}}{\PIN[P''']{3}{3}}}
\end{align*}
The session initialisation unguards its two optional blocks of $ \Role[p]_2 $ as co-ordinator to transmit its current value, \ie $ \Args[v]_{2,0} = \Args[v]_{2, 1} = 1 $.
$ \Role[p]_3 $ directly moves forward to round~$ 3 $.

Since $ \Role[p]_1 $ is crashed, the first block of $ \Role[p]_2 $ in round~$ 2 $ towards $ \Role[p]_1 $ is doomed to fail causing another application of Rule~$ (\mathsf{fail}) $
\begin{align*}
	\longmapsto & \PPar{\PIN[P'''']{1}{3}}{\PPar{\PIN[P'''']{2}{3}}{\PIN[P''']{3}{3}}}
\end{align*}

$ \Role[p]_3 $ completes round~$ 2 $ by receiving the value $ 1 $ that was transmitted by $ \Role[p]_2 $ in round~$ 2 $ using Rule~$ (\mathsf{cSO}) $. After this communication the respective two optional blocks are resolved by Rule~$ (\mathsf{succ}) $ that also updates the current value of $ \Role[p]_3 $ to $ \Args[v]_{3,2} = 1 $
\begin{align*}
	\longmapsto^3 \PPar{\PIN[P'''']{1}{3}}{\PPar{\PIN[P''''']{2}{3}}{\PIN[P'''']{3}{3}}}
\end{align*}
With that $ \Role[p]_3 $ moves to round~$ 3 $.

The last sub-session for round~$ 3 $ is initialised and $ \Role[p]_1 $ removes it last optional block
\begin{align*}
	\longmapsto^5 \PPar{\PIN[P'''''']{2}{3}}{\PIN[P''''']{3}{3}}
\end{align*}
With that $ \Role[p]_1 $ is reduced to $ \PEnd $ and $ \Role[p]_3 $ finishes its last round.

The last steps
\begin{align*}
	\longmapsto^4 \PEnd
\end{align*}
are used to remove the optional block of $ \Role[p]_3 $ towards the crashed $ \Role[p]_1 $, to complete the reception of the value from $ \Role[p]_3 $ by $ \Role[p]_2 $, and to resolve the remaining optional blocks.
The last values of $ \Role[p]_2 $ and $ \Role[p]_3 $ were $ 1 $. As above, since $ \Role[p]_1 $ (that still holds $ 0 $) crashed, this solves Consensus.

\section{Well-Typed Processes}
\label{sec:wellTypedness}

In the Sections~\ref{sec:globalTypes} and \ref{sec:localTypes} we provided the types of our type system. Now we connect types with processes from Section~\ref{sec:calculus} by the notion of well-typedness.
A process $ P $ is \emph{well-typed} if it satisfies a typing judgement of the form $ \Gamma \vdash P \triangleright \Delta $, \ie under the \emph{global environment}~$ \Gamma $, $ P $ is validated by the \emph{session environment}~$ \Delta $.
We extend environments defined in \cite{DemangeonHonda12} with a primitive for session environments.

\begin{definition}[Environments]
	\begin{align*}
		\Gamma & \deffTerms \emptyset \sep \Gamma, \Typed{\Chan}{\AT{T}{\Role}} \sep \Gamma, \TypedProt{\Prot}{\tilde{\Role}_1}{\tilde{\Args[y]}}{\tilde{\Role}_2}{G} \sep \Gamma, \Typed{\Chan[s]}{G}\\
		\Delta & \deffTerms \emptyset \sep \Delta, \Typed{\AT{\Chan[s]}{\Role}}{T} \sep \Delta, \Typed{\ATE{\Chan[s]}{\Role}}{T} \sep \Delta, \Typed{\ATI{\Chan[s]}{\Role}}{T}
		\sep \textcolor{blue}{\Delta, \Typed{\Role}{\OV{\tilde{\Sort}}}}
	\end{align*}
\end{definition}

\noindent
The global environment $ \Gamma $ relates shared channels to the type of the invitation they carry, protocol names to their code, and session channels $ \Chan[s] $ to the global type $ G $ they implement.
$ \Typed{\Chan}{\AT{T}{\Role}} $ means that $ \Chan $ is used to send and receive invitations to play role $ \Role $ with local type $ T $.
In $ \TypedProt{\Prot}{\tilde{\Role}_1}{\tilde{\Args[y]}}{\tilde{\Role}_2}{G} $, $ \Prot $ is a protocol of the global type $ G $ with the internal (external) participants $ \tilde{\Role}_1 $ ($ \tilde{\Role}_2 $) and the arguments $ \tilde{\Args[y]} $.

The session environment $ \Delta $ relates pairs of session channels $ \Chan[s] $ and roles $ \Role $ to local types $ T $. We use $ \Typed{\ATE{\Chan[s]}{\Role}}{T} $ ($ \Typed{\ATI{\Chan[s]}{\Role}}{T} $) to denote the capability to invite externally (internally) someone to play role $ \Role $ in $ \Chan[s] $.

We add the declaration $ \Typed{\Role}{\OV{\tilde{\Sort}}} $, to cover the kinds of the return values of an optional block of the owner $ \Role $. A session environment is \emph{closed} if it does not contain declarations $ \Typed{\Role}{\OV{\tilde{\Sort}}} $.
We assume that initially session environments do not contain declarations $ \Typed{\Role}{\OV{\tilde{\Sort}}} $, \ie are closed. Such declarations are introduced while typing the content of an optional block. Whereby the typing rules ensure that environments can never contain more than one declaration $ \Typed{\Role}{\OV{\tilde{\Sort}}} $.

Let $ \left( \Delta, \Typed{\AT{\Chan[s]}{\Role}}{\LTEnd} \right) = \Delta $.
Let $ \AT{\Chan[s]}{\Role}^- $ denote either $ \ATE{\Chan[s]}{\Role} $ or $ \ATI{\Chan[s]}{\Role} $ or $ \AT{\Chan[s]}{\Role} $.
If $ \AT{\Chan[s]}{\Role}^- $ does not appear in $ \Delta $, we write $ \GetType[\Delta]{\AT{\Chan[s]}{\Role}} = 0 $.
Following \cite{DemangeonHonda12} we assume an operator $ \otimes $ such that
\begin{enumerate}[(1)]
	\item $ \Delta \otimes \emptyset = \Delta $,
	\item $ \Delta_1 \otimes \Delta_2 = \Delta_2 \otimes \Delta_1 $,
	\item $ \Delta_1 \otimes \left( \Delta_2, \Typed{\Role}{\OV{\tilde{\Sort}}} \right) = \left( \Delta_1, \Typed{\Role}{\OV{\tilde{\Sort}}} \right) \otimes \Delta_2 $,
	\item $ \Delta_1 \otimes \left( \Delta_2, \Typed{\AT{\Chan[s]}{\Role}^-}{T} \right) = \left( \Delta_1, \Typed{\AT{\Chan[s]}{\Role}^-}{T} \right) \otimes \Delta_2 $ if $ \GetType[\Delta_1]{\AT{\Chan[s]}{\Role}} = 0 = \GetType[\Delta_2]{\AT{\Chan[s]}{\Role}} $, and
	\item $ \left( \Delta_1, \Typed{\AT{\Chan[s]}{\Role}^-}{T_1} \right) \otimes \left( \Delta_2, \Typed{\AT{\Chan[s]}{\Role}^-}{T_2} \right) = \left( \Delta_1, \Typed{\AT{\Chan[s]}{\Role}^-}{\LTPar{T_1}{T_2}} \right) \otimes \Delta_2 $.
\end{enumerate}
Thus $ \otimes $ allows to split parallel parts of local types.
We write $ \vdash \Args[v] : \Sort $ if value $ \Args[v] $ is of kind $ \Sort $.

\begin{figure*}[tp]
	\[ \begin{array}{c}
		(\mathsf{I}) \dfrac{\Gamma \vdash P \triangleright \Delta, \Typed{\AT{\Args}{\Role}}{T} \quad \GetType{\Chan} = \AT{T}{\Role}}{\Gamma \vdash \PInp{\Chan}{\Args}{P} \triangleright \Delta}
		\hspace*{1.5em}
		(\mathsf{O}) \dfrac{\Gamma \vdash P \triangleright \Delta \quad \GetType{\Chan} = \AT{T}{\Role}}{\Gamma \vdash \POut{\Chan}{\Chan[s]}{P} \triangleright \Delta, \Typed{\ATE{\Chan[s]}{\Role}}{T}}
		\hspace*{1.5em}
		(\mathsf{N}) \dfrac{}{\Gamma \vdash \PEnd \triangleright \emptyset}
		\vspace*{0.75em}\\
		(\mathsf{C}) \dfrac{\left( \Gamma \vdash P_i \triangleright \Delta, \Typed{\AT{\Chan[k]}{\Role_2}}{T_i} \quad \vdash \Typed{\tilde{\Args[y]}_i}{\tilde{\Sort}_i} \right)_{i \in \indexSet}}{\Gamma \vdash \PGet{\Chan[k]}{\Role_1}{\Role_2}{_{i \in \indexSet} \Set{ \PLab{\Labe_i}{\tilde{\Args[y]}_i}{P_i} }} \triangleright \Delta, \Typed{\AT{\Chan[k]}{\Role_2}}{\LTGet{\Role_1}{_{i \in \indexSet{}} \Set{ \LTLab{\Labe_i}{\Typed{\tilde{\Args}_i}{\tilde{\Sort}_i}}{T_i} }}}}
		\hspace*{1.5em}
		(\mathsf{R}) \dfrac{\Gamma, \Typed{\Args}{\AT{T}{\Role}} \vdash P \triangleright \Delta}{\Gamma \vdash \PRes{\Args}{P} \triangleright \Delta}
		\vspace*{0.75em}\\
		(\mathsf{S}) \dfrac{\Gamma \vdash P \triangleright \Delta, \Typed{\AT{\Chan[k]}{\Role_1}}{T_j} \quad \vdash \Typed{\tilde{\Args[v]}}{\tilde{\Sort}_j} \quad j \in \indexSet}{\Gamma \vdash \PSend{\Chan[k]}{\Role_1}{\Role_2}{\Labe_j}{\tilde{\Args[v]}}{P} \triangleright \Delta, \Typed{\AT{\Chan[k]}{\Role_1}}{\LTSend{\Role_2}{_{i \in \indexSet} \Set{ \LTLab{\Labe_i}{\Typed{\tilde{\Args}_i}{\tilde{\Sort}_i}}{T_i} }}}}
		\hspace*{1.5em}
		(\mathsf{Pa}) \dfrac{\Gamma \vdash P_1 \triangleright \Delta_1 \quad \Gamma \vdash P_2 \triangleright \Delta_2}{\Gamma \vdash \PPar{P_1}{P_2} \triangleright \Delta_1 \otimes \Delta_2}
		\vspace*{0.75em}\\
		(\mathsf{S1}) \dfrac{\Gamma \vdash P_1 \triangleright \Delta, \Typed{\AT{\Chan[s]}{\Role}}{T_1} \quad \Gamma \vdash P_2 \triangleright \Delta, \Typed{\AT{\Chan[s]}{\Role}}{T_2}}{\Gamma \vdash \PChoi{P_1}{P_2} \triangleright \Delta, \Typed{\AT{\Chan[s]}{\Role}}{T_1 \oplus T_2}}
		\hspace*{1.5em}
		(\mathsf{S2}) \dfrac{\Gamma \vdash P \triangleright \Delta, \Typed{\AT{\Chan[s]}{\Role}}{T_i} \quad i \in \Set[]{ 1, 2 }}{\Gamma \vdash P \triangleright \Delta, \Typed{\AT{\Chan[s]}{\Role}}{T_1 \oplus T_2}}
		\vspace*{0.75em}\\
		(\mathsf{J}) \dfrac{\Gamma \vdash P \triangleright \Delta, \Typed{\AT{\Chan[s]}{\Role_2}}{T_2}, \Typed{\AT{\Args}{\Role_3}}{T_3} \quad \GetType{\Prot} = \TypeOfProt{\tilde{\Role}_4}{\tilde{\Args[y]}}{\tilde{\Role}_5}{G} \quad \ProjS{G\!\Set[]{ \Subst{\tilde{\Args[v]}}{\tilde{\Args[y]}} }}{}{\Role_3} = T_3}{\Gamma \vdash \PEnt{\Chan[s]}{\Role_1}{\Role_2}{\Role_3}{\Args}{P} \triangleright \Delta, \Typed{\AT{\Chan[s]}{\Role_2}}{\LTEnt{\Prot}{\Role_3}{\tilde{\Args[v]}}{\Role_1}{T_2}}}
		\vspace*{0.75em}\\
		(\mathsf{P}) \dfrac{\Gamma \vdash P \triangleright \Delta, \Typed{\AT{\Chan[s]}{\Role_1}}{T_1} \quad \GetType{\Prot} = \TypeOfProt{\tilde{\Role}_4}{\tilde{\Args[y]}}{\tilde{\Role}_5}{G} \quad \ProjS{G\!\Set[]{ \Subst{\tilde{\Args[v]}}{\tilde{\Args[y]}} }}{}{\Role_3} = T_3}{\Gamma \vdash \PReq{\Chan[s]}{\Role_1}{\Role_2}{\Role_3}{\Chan[k]}{P} \triangleright \Delta, \Typed{\AT{\Chan[s]}{\Role_1}}{\LTReq{\Prot}{\Role_3}{\tilde{\Args[v]}}{\Role_2}{T_1}}, \Typed{\ATI{\Chan[k]}{\Role_3}}{T_3}}
		\vspace*{0.75em}\\
		(\mathsf{New}) \dfrac{\begin{array}{c} \Gamma \vdash P \triangleright \Delta, \Typed{\AT{\Chan[s]}{\Role''}}{T}, \Typed{\ATI{\Chan[k]}{\Role_1}}{T'_1}, \ldots, \Typed{\ATI{\Chan[k]}{\Role_n}}{T'_n}, \Typed{\ATE{\Chan[k]}{\Role'_1}}{T'_{n + 1}}, \ldots, \Typed{\ATE{\Chan[k]}{\Role'_m}}{T'_{n + m}} \quad \GetType{\Prot} = \TypeOfProt{\tilde{\Role}}{\tilde{\Args[y]}}{\tilde{\Role}'}{G}\\ \forall i \logdot \GetType{\Chan_i} = \AT{T'_{i + n}}{\Role'_{i + n}} \quad \forall i \logdot \ProjS{G\!\Set[]{ \Subst{\tilde{\Args[v]}}{\tilde{\Args[y]}} }}{}{\Role_i} = T'_i \quad \forall j \logdot \ProjS{G\!\Set[]{ \Subst{\tilde{\Args[v]}}{\tilde{\Args[y]}} }}{}{\Role'_j} = T'_{j + n} \quad \vdash \Typed{\tilde{\Args[v]}}{\tilde{\Sort}} \quad \GetType{\Chan[k]} = {\Prot\!\Set[]{ \Subst{\tilde{\Args[v]}}{\tilde{\Args[y]}} }} \end{array}}{\Gamma \vdash \PDecl{\Chan[k]}{\Chan[s]}{\tilde{\Args[v]}}{\tilde{\Chan}}{\tilde{\Role}'}{P} \triangleright \Delta, \Typed{\AT{\Chan[s]}{\Role''}}{\LTCall{\Prot}{G}{\tilde{\Args[v]}}{\Typed{\tilde{\Args[y]}}{\tilde{\Sort}}}{\tilde{\Role}'}{T}}}
		\vspace*{0.75em}\\
		\textcolor{blue}{(\mathsf{OptE}) \dfrac{\vdash \Typed{\tilde{\Args[v]}}{\tilde{\Sort}}}{\Gamma \vdash \POptEnd{\Role}{\tilde{\Args[v]}} \triangleright \Typed{\Role}{\OV{\tilde{\Sort}}}}}
		\hspace*{1.5em}
		\textcolor{blue}{(\mathsf{Opt}) \dfrac{\begin{array}{c} \tilde{\Role} \ \dot{=} \ \tilde{\Role}' \quad \Gamma \vdash P \triangleright \Delta, \Typed{\AT{\Chan[s]}{\Role_1}}{T}, \Typed{\Role_1}{\OV{\tilde{\Sort}}} \quad \nexists \Role'', \tilde{\Sort[K]} \logdot \Typed{\Role''}{\OV{\tilde{\Sort[K]}}} \in \Delta\\ \Gamma \vdash P' \triangleright \Delta', \Typed{\AT{\Chan[s]}{\Role_1}}{T'} \quad \vdash \Typed{\tilde{\Args}}{\tilde{\Sort}} \quad \vdash \Typed{\tilde{\Args[v]}}{\tilde{\Sort}} \end{array}}{\Gamma \vdash \POpt{\Role_1}{\tilde{\Role}}{P}{\tilde{\Args}}{\tilde{\Args[v]}}{P'} \triangleright \Delta \otimes \Delta', \Typed{\AT{\Chan[s]}{\Role_1}}{\LTOpt{\tilde{\Role}'}{T}{\Typed{\tilde{\Args[y]}}{\tilde{\Sort}}}{T'}}}}
	\end{array} \]
	\caption{Typing Rules}
	\label{fig:typingRules}
\end{figure*}

In Figure~\ref{fig:typingRules} we extend the typing rules of \cite{DemangeonHonda12} with the Rules~(\textsf{Opt}) and (\textsf{OptE}) for optional blocks.
(\textsf{Opt}) ensures that
\begin{inparaenum}[(1)]
	\item the process and the local type specify the same set of roles $ \tilde{\Role} \ \dot{=} \ \tilde{\Role}' $ as participants of the optional block,
	\item the kinds of the default values $ \tilde{\Args[v]} $, the arguments $ \tilde{\Args} $ of the continuation $ P' $, and the respective variables $ \tilde{\Args[y]} $ in the local type coincide,
	\item the continuation $ P' $ is well-typed \wrt the part $ \Delta' $ of the current session environment and the remainder $ T' $ of the local type of $ \AT{\Chan[s]}{\Role_1} $,
	\item the content $ P $ of the block is well-typed \wrt the session environment $ \Delta, \Typed{\AT{\Chan[s]}{\Role_1}}{T}, \Typed{\Role_1}{\OV{\tilde{\Sort}}} $, where $ \Typed{\Role_1}{\OV{\tilde{\Sort}}} $ ensures that $ P $ computes return values of the kinds $ \tilde{\Sort} $ if no failure occurs, and
	\item the return values of a surrounding optional block cannot be returned in a nested block, because of the condition $ \nexists \Role'', \tilde{\Sort[K]} \logdot \Typed{\Role''}{\OV{\tilde{\Sort[K]}}} \in \Delta $.
\end{inparaenum}
(\textsf{OptE}) ensures that the kinds of the values computed by a successful completion of an optional block match the kinds of the respective default values. Apart from that this rule is similar to (\textsf{N}) in Figure~\ref{fig:typingRules}.
Since (\textsf{OptE}) is the only way to consume an instance of $ \Typed{\Role}{\OV{\tilde{\Sort}}} $, this rule checks that---ignoring the possibility to fail---the content of an optional block reduces to $ \POptEnd{\Role}{\tilde{\Args[v]}} $, if the corresponding local type requires it to do so.
Combining these rules, (\textsf{Opt}) introduces exactly one occurrence of $ \Typed{\Role}{\OV{\tilde{\Sort}}} $ in the session environment, the function $ \otimes $ in (\textsf{Pa}) for parallel processes in Figure~\ref{fig:typingRules} ensures that this occurrence reaches exactly one of the parallel branches of the content of the optional block, and finally only (\textsf{OptE}) allows to terminate a branch with this occurrence. This ensures that---ignoring the possibility to fail---each block computes exactly one vector of return values $ \POptEnd{\Role}{\tilde{\Args[v]}} $ (or, more precisely, one such vector for each choice-branch).
For an explanation of the remaining rules we refer to \cite{BettiniAtall08,BocciAtall10} and \cite{DemangeonHonda12}. Instead we present the derivation of the type judgements of some examples starting with Example~\ref{exa:PRC}.

Applying these typing rules is elaborate but straightforward and can be automated easily, since for all processes except choice exactly one rule applies and all parameters except for restriction are determined by the respective process and the given type environments. Thus, the number of different proof-trees is determined by the number of choices and the type of restricted channels can be derived using back-tracking.

\subsection{Our Implementation of the Rotating Co-ordinators is Well-Typed}

When testing the type of a process, the first step is to choose a suitable global and local environment for the type judgement.

The global environment initially contains
\begin{inparaenum}[(1)]
	\item the channels for the invitations to the session to that the projection of the global type on the respective participant is assigned,
	\item the session channel to that the complete global type is assigned, and
	\item the global types of all sub-sessions.
\end{inparaenum}
For Example~\ref{exa:PRC} this means
\begin{align*}
	\Gamma = \Typed{\Chan_1}{\ProjS{\GRC{n}}{}{\Role[p]_1}}, \ldots, \Typed{\Chan_n}{\ProjS{\GRC{n}}{}{\Role[p]_n}}, \Typed{\Chan[s]}{\GRC{n}}
\end{align*}
where Example~\ref{exa:GTRC} provides $ \GRC{n} $ and its projection $ \ProjS{\GRC{n}}{}{\Role[p]_i} $ is given in Example~\ref{exa:LTRC}.

The session environment initially maps the session channel for each participant to the local type of the respective participant. Here we have:
\begin{align*}
	\Delta = \Typed{\ATE{\Chan[s]}{\Role[p]_1}}{\ProjS{\GRC{n}}{}{\Role[p]_1}}, \ldots, \Typed{\ATE{\Chan[s]}{\Role[p]_n}}{\ProjS{\GRC{n}}{}{\Role[p]_n}}
\end{align*}
Notice that $ \Delta $ is closed.
We have to prove $ \Gamma \vdash \PRC{n} \triangleright \Delta $, where $ \PRC{n} $ is given in Example~\ref{exa:PRC}.

First we apply Rule~$ (\mathsf{Pa}) $ $ n $ times to separate the $ n $ participants $ \PPar{\POutS{\Chan_i}{\Chan[s]}}{\PInp{\Chan_i}{\Chan[s]}{\PIN{i}{n}}} $, whereby we split $ \Delta $ into $ \Delta_1 \otimes \ldots \otimes \Delta_n $ with $ \Delta_i = \Typed{\ATE{\Chan[s]}{\Role[p]_i}}{\ProjS{\GRC{n}}{}{\Role[p]_i}} $. For each participant we separate the output $ \POutS{\Chan_i}{\Chan[s]} $ and $ \PInp{\Chan_i}{\Chan[s]}{\PIN{i}{n}} $ by another application of $ (\mathsf{Pa}) $.
\begin{align*}
	\dfrac{\dfrac{}{\Gamma \vdash \PEnd \triangleright \emptyset} (\mathsf{N}) \quad \Gamma(\Chan_i) = \ProjS{\GRC{n}}{}{\Role[p]_i}}{\Gamma \vdash \POutS{\Chan_i}{\Chan[s]} \triangleright \Typed{\ATE{\Chan[s]}{\Role[p]_i}}{\ProjS{\GRC{n}}{}{\Role[p]_i}}} (\mathsf{O})
\end{align*}
\begin{align*}
	\dfrac{\Gamma \vdash \PIN{i}{n} \triangleright \Typed{\AT{\Chan[s]}{\Role[p]_i}}{\ProjS{\GRC{n}}{}{\Role[p]_i}} \quad \Gamma(\Chan_i) = \ProjS{\GRC{n}}{}{\Role[p]_i}}{\Gamma \vdash \PInp{\Chan_i}{\Chan[s]}{\PIN{i}{n}} \triangleright \emptyset} (\mathsf{I})
\end{align*}
It remains to show that $ \Gamma \vdash \PIN{i}{n} \triangleright \Typed{\AT{\Chan[s]}{\Role[p]_i}}{\ProjS{\GRC{n}}{}{\Role[p]_i}} $.

$ \PIN{i}{n} $ starts with $ i{-}1 $ sequentially composed optional blocks $ \POptS{\Role[p]_i}{\Role[p]_i, \Role[p]_j}{P_{j \to i, \downarrow}}{\Args[v]_{i, j}}{\Args[v]_{i, j{-}1}} $ and similarly $ \ProjS{\GRC{n}}{}{\Role[p]_i} $ with $ i{-}1 $ sequential local types $ \LULT{\Role[p]_j}{\Args[v]_{j, j{-}1}}{\Role[p]_i}{\Args[v]_{i, j{-}1}} = \LTOptS{\Role[p]_i, \Role[p]_j}{\LTGet{\Role[p]_j}{\LTLabS{\Labe[c]}{\Typed{\Args[v]_{j, j{-}1}}{\Sort[V]}}}}{\Typed{\Args[v]_{i, j{-}1}}{\Sort[V]}} $.
For each of these blocks we apply Rule~$ (\mathsf{Opt}) $ and have to show:
\begin{enumerate}
	\item $ \Gamma \vdash P_{j \to i, \downarrow} \triangleright \Typed{\AT{\Chan[s]}{\Role[p]_i}}{\LTGet{\Role[p]_j}{\LTLabS{\Labe[c]}{\Typed{\Args[v]_{j, j{-}1}}{\Sort[V]}}}}, \Typed{\Role[p]_i}{\OV{\Sort[V]}} $
	\item $ \Gamma \vdash P' \triangleright \Typed{\AT{\Chan[s]}{\Role[p]_i}}{T'} $, where $ P' $ is the continuation of the optional block and $ T' $ the continuation of the local type
	\item $ \vdash \Typed{\Args[v]_{i, j}}{\Sort[V]} $ and $ \vdash \Typed{\Args[v]_{i, j{-}1}}{\Sort[V]} $
\end{enumerate}
The third condition checks whether the variable $ \Args[v]_{i, j} $ and the default value $ \Args[v]_{i, j{-}1} $ of the optional block are values and of the same type as the value $ \Typed{\Args[v]_{i, j{-}1}}{\Sort[V]} $ of the local type. Since all $ \Args[v]_{k, l} $ are of kind $ \Sort[V] $, this condition is satisfied.
The first condition checks the type of the content of the optional block, where $ P_{j \to i, \downarrow} = \PGet{\Chan[s]}{\Role[p]_j}{\Role[p]_i}{\PLab{\Labe[c]}{\Args[v]_{j, j{-}1}}{\POptEnd{\Role[p]_i}{\Args[v]_{j, j{-}1}}}} $ and
\begin{align*}
	\dfrac{\dfrac{\vdash \Typed{\Args[v]_{j, j{-}1}}{\Sort[V]}}{\Gamma \vdash \POptEnd{\Role[p]_i}{\Args[v]_{j, j{-}1}} \triangleright \Typed{\Role[p]_i}{\OV{\Sort[V]}}} (\mathsf{OptE}) \quad \vdash \Typed{\Args[v]_{j, j{-}1}}{\Sort[V]}}{\Gamma \vdash P_{j \to i, \downarrow} \triangleright \Typed{\AT{\Chan[s]}{\Role[p]_i}}{\LTGet{\Role[p]_j}{\LTLabS{\Labe[c]}{\Typed{\Args[v]_{j, j{-}1}}{\Sort[V]}}}}, \Typed{\Role[p]_i}{\OV{\Sort[V]}}} (\mathsf{C})
\end{align*}

The second condition refers to the respective next part of the process and the local type.
After the first $ i{-}1 $ sequential optional blocks, $ P' $ is the parallel composition of $ n{-}1 $ optional blocks $ \POptNVS{\Role[p]_i}{\Role[p]_i, \Role[p]_j}{P_{i \to j, \uparrow}} $ and the remaining sequential blocks:
\begin{align*}
	P'' = \bigodot_{j = (i + 1)..n} \POptS{\Role[p]_i}{\Role[p]_i, \Role[p]_j}{P_{j \to i, \downarrow}}{\Args[v]_{i, j}}{\Args[v]_{i, j{-}1}}
\end{align*}
Similarly, $ T' $ is the parallel composition of the $ \LULS{\Role[p]_i}{\Args[v]_{i, i{-}1}}{\Role[p]_j} $ and the remaining (sequentially composed) $ \LULT{\Role[p]_j}{\Args[v]_{j, j{-}1}}{\Role[p]_i}{\Args[v]_{i, j{-}1}} $.
We use $ (\mathsf{Pa}) $ to separate the parallel components in both the process and the local type.
Thus we have to show $ \Gamma \vdash \POptNVS{\Role[p]_i}{\Role[p]_i, \Role[p]_j}{P_{i \to j, \uparrow}} \triangleright \Typed{\AT{\Chan[s]}{\Role[p]_i}}{\LULS{\Role[p]_i}{\Args[v]_{i, i{-}1}}{\Role[p]_j}} $ for $ j = 1..n, i \neq j $ and $ \Gamma \vdash P'' \triangleright \Typed{\AT{\Chan[s]}{\Role[p]_i}}{\LULT{\Role[p]_j}{\Args[v]_{j, j{-}1}}{\Role[p]_i}{\Args[v]_{i, j{-}1}}.T''} $.
The proof of the last typing judgement for the $ n{-}i $ last sequential blocks is very similar to the proof for the first $ i{-}1 $ sequential blocks with an application of Rule~$ (\mathsf{N}) $ in the end.

For each $ j \in \Set[]{ 1, \ldots, i{-}1, i + 1, \ldots, n } $ we apply Rule~$ (\mathsf{Opt}) $ and have to show:
\begin{enumerate}
	\item $ \Gamma \vdash P_{i \to j, \uparrow} \triangleright \Typed{\AT{\Chan[s]}{\Role[p]_i}}{\LTSend{\Role[p]_j}{\LTLabS{\Labe[c]}{\Typed{\Args[v]_{i, i{-}1}}{\Sort[V]}}}}, \Typed{\Role[p]_i}{\OV{\cdot}} $
	\item $ \Gamma \vdash \PEnd \triangleright \emptyset $
\end{enumerate}
There are no default values and thus the conditions $ \vdash \Typed{\tilde{\Args}}{\tilde{\Sort}} $ and $ \vdash \Typed{\tilde{\Args[v]}}{\tilde{\Sort}} $ of Rule~$ (\mathsf{Opt}) $ hold trivially.
The second condition, for the continuation of the optional blocks, is in all cases of $ j $ the same and follows from Rule~$ (\mathsf{N}) $.
For the first condition we have $ P_{i \to j, \uparrow} = \PSend{\Chan[s]}{\Role[p]_i}{\Role[p]_j}{\Labe[c]}{\Args[v]_{i, i{-}1}}{\POptEnd{\Role[p]_i}{\cdot}} $ and thus
\begin{align*}
	\dfrac{\dfrac{\vdash \Typed{\cdot}{\cdot}}{\Gamma \vdash \POptEnd{\Role[p]_i}{\cdot} \triangleright \Typed{\Role[p]_i}{\OV{\cdot}}} (\mathsf{OptE}) \quad \vdash \Typed{\Args[v]_{i, i{-}1}}{\Sort[V]}}{\Gamma \vdash \PSend{\Chan[s]}{\Role[p]_i}{\Role[p]_j}{\Labe[c]}{\Args[v]_{i, i{-}1}}{\POptEnd{\Role[p]_i}{\cdot}} \triangleright \Typed{\AT{\Chan[s]}{\Role[p]_i}}{\LTSend{\Role[p]_j}{\LTLabS{\Labe[c]}{\Typed{\Args[v]_{i, i{-}1}}{\Sort[V]}}}}, \Typed{\Role[p]_i}{\OV{\cdot}}} (\mathsf{S})
\end{align*}

\subsection{Our Implementation with Sub-Sessions is Well-Typed}

To check the type of $ \PRCB{n} $ of Example~\ref{exa:PRCWSS} we need to add the type of the protocol to the global environment used for $ \PRC{n} $:
\begin{align*}
	\Gamma = \Typed{\Chan_1}{\ProjS{\GRCB{n}}{}{\Role[p]_1}}, \ldots, \Typed{\Chan_n}{\ProjS{\GRCB{n}}{}{\Role[p]_n}}, \Typed{\Chan[s]}{\GRCB{n}}, \TypedProt{\Prot[R]_n}{\Role[scr], \tilde{\Role[trg]}}{\Args[v]_{\Role[scr]}}{\cdot}{\GR{n}}
\end{align*}

The session environment initially is the same as for the first example:
\begin{align*}
	\Delta = \Typed{\ATE{\Chan[s]}{\Role[p]_1}}{\ProjS{\GRCB{n}}{}{\Role[p]_1}}, \ldots, \Typed{\ATE{\Chan[s]}{\Role[p]_n}}{\ProjS{\GRCB{n}}{}{\Role[p]_n}}
\end{align*}
Again $ \Delta $ is closed.
We have to prove $ \Gamma \vdash \PRCB{n} \triangleright \Delta $.
Also this type derivation is very similar to our first example.
We provide one derivation for the sub-session of a round, to demonstrate the additional steps.

\begin{align*}
	\dfrac{\begin{array}{c} \Gamma \vdash P' \triangleright \Typed{\AT{\Chan[s]}{\Role[p]_i}}{T}, \Delta_k \quad \GetType{\Prot[R]_n} = \TypedProt{\Prot[R]_n}{\Role[scr], \tilde{\Role[trg]}}{\Args[v]_{\Role[scr]}}{\cdot}{\GR{n}}\\ \ProjS{\GR{n}\!\Set[]{ \Subst{\Args[v]_{i, i - 1}}{\Args[v]_{\Role[scr]}} }}{}{\Role[scr]} = T_{\Role[scr]} \quad \forall j < n \logdot \ProjS{\GR{n}\!\Set[]{ \Subst{\Args[v]_{i, i - 1}}{\Args[v]_{\Role[scr]}} }}{}{\Role[trg]_j} = T_{\Role[trg], j}\\ \vdash \Typed{\Args[v]_{i, i - 1}}{\Sort[V]} \quad \GetType{\Chan[k]} = \GR{n}\!\Set[]{ \Subst{\Args[v]_{i, i - 1}}{\Args[v]_{\Role[scr]}} } \end{array}}{\Gamma \vdash \PDecl{\Chan[k]}{\Chan[s]}{\Args[v]_{i, i - 1}}{\cdot}{\cdot}{P'} \triangleright \; \Typed{\AT{\Chan[s]}{\Role[p]_i}}{\LTCall{\Prot[R]_n}{\GR{n}}{\Args[v]_{i, i - 1}}{\Typed{\Args[v]_{\Role[scr]}}{\!\Sort[V]}}{\cdot}{T}}}(\mathsf{New})
\end{align*}
and
\begin{align*}
	\Delta_k ={} & \Typed{\ATI{\Chan[k]}{\Role[scr]}}{\LTSend{\Role[trg]}{\LTLabS{\Labe[bc]}{\Args[v]_{i, i - 1}}}}, \Typed{\ATI{\Chan[k]}{\Role[trg]_1}}{\LTGet{\Role[scr]}{\LTLabS{\Labe[bc]}{\Args[v]_{i, i - 1}}}}, \ldots, \Typed{\ATI{\Chan[k]}{\Role[trg]_{i - 1}}}{\LTGet{\Role[scr]}{\LTLabS{\Labe[bc]}{\Args[v]_{i, i - 1}}}}
\end{align*}
Then we have to show that $ \Gamma \vdash P' \triangleright \Typed{\AT{\Chan[s]}{\Role[p]_i}}{T}, \Delta_k $.

The internal session invitations within $ P' $ are handled by the Rules~$ (\mathsf{P}) $ and $ (\mathsf{J}) $ similar to
\begin{align*}
	\dfrac{\dfrac{}{\Gamma \vdash \PEnd \triangleright \emptyset}(\mathsf{N}) \quad \GetType{\Prot[R]_n} = \TypedProt{\Prot[R]_n}{\Role[scr], \tilde{\Role[trg]}}{\Args[v]_{\Role[scr]}}{\cdot}{\GR{n}} \quad \ProjS{\GR{n}\!\Set[]{ \Subst{\Args[v]_{i, i - 1}}{\Args[v]_{\Role[scr]}} }}{}{\Role[scr]} = T_{\Role[scr]}}{\Gamma'' \vdash \PReqS{\Chan[s]}{\Role[p]_i}{\Role[p]_i}{\Role[scr]}{\Chan[k]} \triangleright \Typed{\AT{\Chan[s]}{\Role[p]_i}}{\LTReqS{\Prot[R]_n}{\Role[src]}{\Args[v]_{i, i - 1}}{\Role[p]_i}}, \Typed{\ATI{\Chan[k]}{\Role[scr]}}{T_{\Role[scr]}}} (\mathsf{P})
\end{align*}
and
\begin{align*}
	\dfrac{\Gamma \vdash P'' \triangleright \Typed{\AT{\Args[z]}{\Role[scr]}}{T_{\Role[scr]}'} \quad \GetType{\Prot[R]_n} = \TypedProt{\Prot[R]_n}{\Role[scr], \tilde{\Role[trg]}}{\Args[v]_{\Role[scr]}}{\cdot}{\GR{n}} \quad \ProjS{\GR{n}\!\Set[]{ \Subst{\Args[v]_{i, i}}{\Args[v]_{\Role[scr]}} }}{}{\Role[scr]} = T_{\Role[scr]}'}{\Gamma \vdash \PEnt{\Chan[s]}{\Role[p]_i}{\Role[p]_i}{\Role[scr]}{\Args[z]}{P''} \triangleright \Typed{\AT{\Chan[s]}{\Role[p]_i}}{\LTEnt{\Prot[R]_n}{\Role[scr]}{\Args[v]_{i, i - 1}}{\Role[p]_i}{T_{\Role[scr]}'}}} (\mathsf{J})
\end{align*}
for the coordinator inviting himself.

\section{An Example with Sub-Sessions within Optional Blocks}
\label{sec:ExaSubSessionsWithinOptionalBlocks}

We present a third example---again a variant of the rotating coordinator algorithm in Example~\ref{exa:RCAlgorithm}---to demonstrate the use of sub-sessions within optional blocks.

\subsection{Global and Local Types}

\begin{example}[Global Type for Rotating Coordinators]
	\label{exa:globalType}
	\begin{align*}
		\GRC{n} ={} & \GTDecl{\Prot[C]}{\Role[scr],\Role[trg]}{\Typed{\Args[val]}{\Sort[V]}}{\cdot}{G_{\Prot[C]}}{} \bigodot_{i=1..n} G_{\text{Round}}(i, n)\\
		G_{\Prot[C]} ={} & \GTCom{\Role[scr]}{\Role[trg]}{\GTInpS{\Labe[bc]}{\Typed{\Args[val]}{\Sort[V]}}}\\
		G_{\text{Round}}(i, n) ={} & \bigodot_{j = 1..n, \; j \neq i} {\GTOptBlS{\Role[p]_i, \cdot, \Role[p]_j, \Typed{\Args[v]_{j, i - 1}}{\Sort[V]}}{\left( \GTCallS{\Role[p]_i}{\Prot[C]}{\Role[p]_i, \Role[p]_j}{\Args[v]_{i, i - 1}} \right)}}
	\end{align*}
\end{example}

\noindent
$ \GRC{n} $ first declares a sub-protocol and then performs the $ n $ rounds of the algorithm sequentially.
The sub-protocol $ G_{\Prot[C]} $, identified with $ \Prot[C] $, specifies a single communication as part of the broadcast in Line~4 of Example~\ref{exa:RCAlgorithm}.
This communication step covers the transmission of the value $ \Args[val] $ (under the label $ \Labe[bc] $ for broadcast) from $ \Role[scr] $ to $ \Role[trg] $.
In each round the current coordinator participant $ \Role[p]_i $ calls this sub-protocol sequentially $ n - 1 $ times, in order to transmit its current value $ \Args[v]_{i, i - 1} $ to all other participants.
To simulate link failures we capture each communication of the broadcast in a single optional block.
Here the sender $ \Role[p]_i $ does not need to specify a default value, whereas the continuation of the receiver $ \Role[p]_j $ uses its last known value $ \Args[v]_{j, i - 1} $ if the communication fails.
Since global types describe a global point of view, the communication steps modelled above also cover the reception of values in Line~5.
For simplicity we omit the Lines~1 and 7 from our consideration.

Restricting $ \GRC{n} $ on participant~$ i $ reduces all rounds~$ j $ (except for round~$ j = i $) to the single communication step of round~$ j $ in that participant~$ i $ receives a value.
Accordingly, participant~$ i $ receives $ i - 1 $ times a value in $ i - 1 $ rounds, then broadcasts its current value to all other participants (modelled by $ n - 1 $ single communication steps), and then receives $ n - i $ times a value in the remaining $ n - i $ rounds.

\begin{example}[Restriction on Participant~$ i $]
	\label{exa:restriction}
	\begin{align*}
		\RestS{\GRC{n}}{\Role[p]_i} ={} & \GTDecl{\Prot[C]}{\Role[scr],\Role[trg]}{\Typed{\Args[val]}{\!\Sort[V]}}{\cdot}{G_{\Prot[C]}}{\!\RestS{G'}{\Role[p]_i}}\\
		\RestS{G'}{\Role[p]_i} ={} & \left( \bigodot_{j = 1..(i - 1)} \GTOptBlS{\Role[p]_j, \cdot, \Role[p]_i, \Typed{\Args[v]_{i, j - 1}}{\!\Sort[V]}}{\left( \GTCallS{\Role[p]_j}{\Prot[C]}{\Role[p]_j, \Role[p]_i}{\Args[v]_{j, j - 1}} \right)} \right).G_{\text{Round}}(i, n).\\
		& \bigodot_{j = (i + 1)..n} \GTOptBlS{\Role[p]_j, \cdot, \Role[p]_i, \Typed{\Args[v]_{i, j - 1}}{\!\Sort[V]}}{\left( \GTCallS{\Role[p]_j}{\Prot[C]}{\Role[p]_j, \Role[p]_i}{\Args[v]_{j, j - 1}} \right)}
	\end{align*}
\end{example}

\begin{example}[Projection to the Local Type of Participant~$ i $]
	\label{exa:localType}
	\begin{align*}
		\ProjS{G_{\text{RC}}(n)}{}{\Role[p]_i} ={} & \ProjS{( \bigodot_{i=1..n} G_{\text{Round}}(i, n) )}{\PDec{\Prot[C]}{\Role[scr],\Role[trg]}{\Typed{\Args[val]\,}{\Sort[V]}}{\cdot}{G_{\Prot[C]}}}{\Role[p]_i}\\
		={} & \ProjS{( \RestS{G'}{\Role[p]_i} )}{\PDec{\Prot[C]}{\Role[scr],\Role[trg]}{\Typed{\Args[val]\,}{\Sort[V]}}{\cdot}{G_{\Prot[C]}}}{\Role[p]_i}\\
		={} & ( \bigodot_{j = 1..(i - 1)} \LTOptS{\Role[p]_i, \Role[p]_j}{T_{j \to i}}{\Typed{\Args[v]_{i, j - 1}}{\Sort[V]}} ).\\
		& ( \prod_{j = 1..n, j \neq i} \LTOptS{\Role[p]_i, \Role[p]_j}{T_{i \text{ calls } j}}{\cdot} \LTPar{}{( \bigodot_{j = (i + 1)..n} \LTOptS{\Role[p]_i, \Role[p]_j}{T_{j \to i}}{\Typed{\Args[v]_{i, j - 1}}{\Sort[V]}} )} )\\
		T_{j \to i} ={} & \LTEntS{\Prot[C]}{\Role[trg]}{\Args[v]_{j, j - 1}}{\Role[p]_j}\\
		T_{i \text{ calls } j} ={} & \LTCall{\Prot[C]}{G_{\Prot[C]}}{\Args[v]_{i, i -1}}{\Typed{\Args[val]}{\!\Sort[V]}}{\cdot}{T_{i \to j}}\\
		T_{i \to j} ={} & \LTPar{\LTReqS{\Prot[C]}{\Role[src]}{\Args[v]_{i, i - 1}}{\!\Role[p]_i}}{\LTReqS{\Prot[C]}{\Role[trg]}{\Args[v]_{i, i - 1}}{\!\Role[p]_j}} \LTPar{}{\LTEntS{\Prot[C]}{\Role[scr]}{\Args[v]_{i, i - 1}}{\Role[p]_i}}
	\end{align*}
\end{example}

\noindent
To project the global type $ G_{\text{RC}}(n) $ on the local type of participant~$ i $ we first add the information about the declaration of the protocol $ \Prot[C] $ to the environment and then project the $ n $ rounds.
Since the projection of optional blocks to a role that does not participate in that optional block simply removes the respective block, the projection of the $ n $ rounds on the local type of participant~$ i $ is same as the projection of the restriction $ \Rest{G'}{\Role[p]_i} $ (see Example~\ref{exa:restriction}) on the local type of participant~$ i $.
Accordingly participant~$ i $ is $ i - 1 $ times the target $ \Role[trg] $ of the protocol $ \Prot[C] $ (\cf $ T_{j \to i} $), \ie optionally receives $ n - i $ values, then initiates round~$ i $ (\cf $ T_{i \text{ calls } j} $) and broadcasts its current value by calling the protocol $ \Prot[C] $ $ n - 1 $ times as source $ \Role[scr] $ (\cf $ T_{i \to j} $), and finally participant~$ i $ is $ n - i $ more times the target $ \Role[trg] $ of $ \Prot[C] $ (\cf $ T_{j \to i} $).
Observe that, due to the different cases of the projection of optional blocks, receiving values from rounds different from $ i $ guards the continuation of participant~$ i $ while broadcasting its own value is performed in parallel (although in the global type all optional blocks guard the respective continuation).

\subsection{Implementation}

Based on the local type of participant~$ i $ in Example~\ref{exa:localType} we provide an implementation of the rotating coordinator. Therefore we replace the check '\texttt{if alive}$ (p_r) $' by an optional block for communications.
In the first $ i - 1 $ and the last $ n - i $ rounds participant~$ i $ either receives a value or (if the respective communication fails) uses as default value its value from the former round.
In round~$ i $ participant~$ i $ initiates $ n - 1 $ new sub-sessions---each covered within an optional block---to transmit its current value to each other participant.

\begin{example}[Rotating Coordinator Implementation]
	\label{exa:process}
	\begin{align*}
		P_{\text{RC}}(n) ={} & \prod_{i = 1..n} \left( \PPar{\POutS{\Chan_i}{\Chan[s]}}{\PInp{\Chan_i}{\Chan[s]}{P_i}} \right)\\
		P_i ={} & ( \bigodot_{j = 1..(i - 1)} \POptS{\Role[p]_i}{\Role[p]_i, \Role[p]_j}{P_{j \to i}}{\Args[v]_{i, j}}{\Args[v]_{i, j - 1}} ).\\
		& ( \prod_{j = 1..n, j \neq i} \PRes{\Chan[k]}{\POptNVS{\Role[p]_i}{\Role[p]_i, \Role[p]_j}{P_{i \text{ calls } j}}} \PPar{}{( \bigodot_{j = (i + 1)..n} \POptS{\Role[p]_i}{\Role[p]_i, \Role[p]_j}{P_{j \to i}}{\Args[v]_{i, j}}{\Args[v]_{i, j - 1}} ))}\\
		P_{j \to i} ={} & \PEnt{\Chan[s]}{\Role[p]_j}{\Role[p]_i}{\Role[trg]}{\Args}{\PGet{\Args}{\Role[scr]}{\Role[trg]}{\PLab{\Labe[bc]}{\Args[v]}{\POptEnd{\Role[p]_i}{\Args[v]}}}}\\
		P_{i \text{ calls } j} ={} & \PDecl{\Chan[k]}{\Chan[s]}{\Args[v]_{i, i - 1}}{\cdot}{\cdot}{P_{i \to j, \Chan[k]}}\\
		P_{i \to j, \Chan[k]} ={} & \PPar{\PPar{\PReqS{\Chan[s]}{\Role[p]_i}{\Role[p]_i}{\Role[scr]}{\Chan[k]}}{\PReqS{\Chan[s]}{\Role[p]_i}{\Role[p]_j}{\Role[trg]}{\Chan[k]}}}{\PEnt{\Chan[s]}{\Role[p]_i}{\Role[p]_i}{\Role[scr]}{\Args[z]}{\PSend{\Args[z]}{\Role[scr]}{\Role[trg]}{\Labe[bc]}{\Args[v]_{i, i - 1}}{\POptEnd{\Role[p]_i}{\cdot}}}}
	\end{align*}
\end{example}

\noindent
The overall system $ P_{\text{RC}}(n) $ consists of the parallel composition of the $ n $ participants. The channel $ \Chan_i $ is used to distribute the initial session channel. Since these communications on $ \tilde{\Chan} $ are used to initialise the system and not to model the algorithm, we assume that they are reliable.

The term $ P_i $ models participant $ i $. Each participant first optionally receives $ i - 1 $ times a value from another participant. Therefore an optional block surrounds the term $ P_{j \to i} $ that first answers the sub-session request of participant~$ j $, then receives (as target) in the respective sub-session a value from participant~$ j $ (the source), and finally outputs $ \POptEnd{\Role[p]_i}{\Args[v]} $. This last output terminates the optional block and transmits the received value to its continuation. If this communication succeeds, the respective optional block succeeds, and the received value replaces the current value of participant~$ i $. Otherwise the default value of the former round is used, \ie the current value of participant~$ i $ remains unchanged.

In round~$ i $ participant~$ i $ initiates $ n - 1 $ parallel sub-sessions; one for each other participant. For each sub-session a private version of the sub-session channel $ \Chan[k] $ is restricted and an optional block is created. Within the optional block, the term $ P_{i \text{ calls } j} $ creates a sub-session between $ i $ (source) and $ j $ (target). This sub-session $ P_{i \to j, k} $ consists of the parallel composition of the invitations of the source and the target to participate in the sub-session using $ \Chan[k] $, and the session acceptance of the source (participant~$ i $) followed by the transmission of its current value towards the target (participant~$ j $) and the empty transmission $ \POptEnd{\Role[p]_i}{\cdot} $ that terminates the optional block of participant~$ i $.

Finally participant~$ i $ optionally receives $ n - i $ more values from other participants in the same way as in its first $ i - 1 $ rounds.

The sub-sessions initiated by participant~$ i $ for the broadcast are in parallel to the reception of the value for round~$ i + 1 $. The remaining rounds are composed sequentially.
In round~$ i $ the value $ \Args[v]_{i, i - 1} $ is emitted, \ie the (initial value or) last value that is received in the sequential $ n - 1 $ rounds that guard the parallel composition of the sub-sessions to transmit this value.
This matches an intuitive understanding of this process in terms of asynchronous communications. The sending operations emit the respective value as soon as they are unguarded but they syntactically remain part of the term until the (possibly later) reception of the respective message consumes it.
In fact the presented session calculus is synchronous but, since Example~\ref{exa:process} uses neither choice nor output continuations different from $ \PEnd $ or $ \POptEnd{\Role[p]_i}{\cdot} $, the process in Example~\ref{exa:process} can be considered as an asynchronous process \cite{hondaTokoro91, boudol92, palamidessi03, fossacs12_pi}.

\subsection{Reaching Consensus}

The first three steps initialise the outermost session using three times Rule~$ (\mathsf{comC}) $. We assume here that these steps belong to the environment and do never fail. If one of these steps fails, the respective participant does not know the global session channel and thus cannot participate in the algorithm, \ie is crashed from the beginning.
\begin{align*}
	P_{\text{RC}}(3) &
	\longmapsto^3 \PPar{P_1}{\PPar{P_2}{P_3}}
	\intertext{After the initialisation all participants consist of sequential and parallel optional blocks. Each of these optional blocks can fail any time. $ \Role[p]_1 $ can initialise one of its two sub-sessions to transmit its value to one of the other participants.
		Since these two blocks are in parallel, $ \Role[p]_1 $ can start with either of them. We assume however that in the next two steps it successfully initialise both sub-sessions within its unguarded optional blocks using two times Rule~$ (\mathsf{subs}) $. There is no external partner to invite, so the initialisation of the sub-sessions does not generate output messages but only unguards $ P_{1 \to 2, \Chan[k]} $ and $ P_{1 \to 3, \Chan[k]} $.}
	& \longmapsto^2 \PRes{\Chan[k]}{\POptNVS{\Role[p]_1}{\Role[p]_1, \Role[p]_2}{P_{1 \to 2, \Chan[k]}}}
	\PPar{\PPar{}{\PRes{\Chan[k]}{\POptNVS{\Role[p]_1}{\Role[p]_1, \Role[p]_3}{P_{1 \to 3, \Chan[k]}}}}}{\PPar{P_1'}{\PPar{P_2}{P_3}}}
	\intertext{Next $ \Role[p]_1 $ accepts the invitation to its own session---using Rule~$ (\mathsf{join}) $---within the second sub-session ($ P_{1 \to 3}, \Chan[k] $) and $ \Role[p]_3 $ accepts the invitation from $ \Role[p]_1 $---using Rule~$ (\mathsf{jO}) $.}
	& \longmapsto^2 \PPar{\PRes{\Chan[k]}{(\POptNVS{\Role[p]_1}{\Role[p]_1, \Role[p]_2}{P_{1 \to 2, \Chan[k]}})}}{\PPar{P_1'}{P_2}}\\
	& \hspace{2em} \PPar{}{\PRes{\Chan[k]}{( \POptNVS{\Role[p]_1}{\Role[p]_1, \Role[p]_3}{\PSend{\Chan[k]}{\Role[scr]}{\Role[trg]}{\Labe[bc]}{0}{\POptEnd{\Role[p]_1}{\cdot}}} }}
	\PPar{}{\POpt{\Role[p]_3}{\Role[p]_3, \Role[p]_1}{\PGet{\Chan[k]}{\Role[scr]}{\Role[trg]}{\PLab{\Labe[bc]}{\Args[v]}}{\POptEnd{\Role[p]_3}{\Args[v]}}}{\Args[v]_{3, 1}}{1}{P_3'}})
	\intertext{After transmitting its value to $ \Role[p]_3 $---using Rule~$ (\mathsf{cSO}) $---the content of the second optional block (from $ \Role[p]_1 $ to $ \Role[p]_3 $) is reduced to $ \POptEnd{\Role[p]_1}{\cdot} $ and the block can be removed by Rule~$ (\mathsf{succ}) $.}
	& \longmapsto^2 \PPar{\PRes{\Chan[k]}{(\POptNVS{\Role[p]_1}{\Role[p]_1, \Role[p]_2}{P_{1 \to 2, \Chan[k]}})}}{\PPar{P_1'}{P_2}}
	\PPar{}{\POpt{\Role[p]_3}{\Role[p]_3, \Role[p]_1}{\POptEnd{\Role[p]_3}{0}}{\Args[v]_{3, 1}}{1}{P_3'}})
	\intertext{Finally, $ \Role[p]_3 $ replaces its own value by the value $ 0 $ it received from $ \Role[p]_1 $---using Rule~$ (\mathsf{succ}) $. With that $ \Role[p]_3 $ finishes round~$ 1 $ and moves to round~$ 2 $, whereas the other two participants still remain in round~$ 1 $.}
	& \longmapsto \PPar{\PRes{\Chan[k]}{(\POptNVS{\Role[p]_1}{\Role[p]_1, \Role[p]_2}{P_{1 \to 2, \Chan[k]}})}}{\PPar{P_1'}{\PPar{P_2}{P_3'\Set[]{\Subst{0}{\Args[v]_{3, 1}}}}}}
	\intertext{$ \Role[p]_1 $ and $ \Role[p]_3 $ are waiting for a value from $ \Role[p]_2 $ and $ \Role[p]_2 $ is waiting for a value from $ \Role[p]_1 $. $ \Role[p]_1 $ fails to deliver its value to $ \Role[p]_2 $. It does not really matter whether it crashes while the invitations for the sub-session are accepted or before sending the value. The result is the same: The first optional block of $ \Role[p]_1 $ is removed by Rule~$ (\mathsf{fail}) $. After that there is no other optional block with matching roles for the first optional block of $ \Role[p]_2 $ and thus no communication can take place. Hence it has to be aborted as well. With that all three participants move to round~$ 2 $. Moreover, since we assume that $ \Role[p]_1 $ is crashed, we also abort the remaining two optional blocks of $ \Role[p]_1 $ in $ P_1' $ and $ \Role[p]_1 $ completes round~$ 3 $.}
	& \longmapsto^4 \PPar{P_2'}{P_3'\Set[]{\Subst{0}{\Args[v]_{3, 1}}}}
	\intertext{$ \Role[p]_2 $ in $ P_2' $ holds the value $ 1 $ (its initial value) and $ \Role[p]_3 $ in $ P_3'\!\Set[]{ \Subst{0}{\Args[v]_{3, 1}} } $ holds the value $ 0 $. To complete round~$ 2 $ Participant~$ 2 $ aborts its attempt to send to $ \Role[p]_1 $ and (successfully) completes the sub-session with $ \Role[p]_3 $.
	Round~$ 3 $ is completed in the same way and finally the two remaining participants both hold the value $ 1 $.}
	& \longmapsto^7 \PPar{P_2''}{P_3''\!\Set[]{ \Subst{1}{\Args[v]_{3, 2}} }}
	\longmapsto^7 \PPar{P_2'''\!\Set[]{ \Subst{1}{\Args[v]_{2, 3}} }}{P_3'''\!\Set[]{ \Subst{1}{\Args[v]_{3, 2}} }}
\end{align*}

\subsection{Well-Typed Processes}

Let
\begin{align*}
	\Gamma = \Typed{\Chan[a]_1}{\Rest{G_{\text{RC}}^n}{\Role[p]_1}}, \ldots, \Typed{\Chan[a]_n}{\Rest{G_{\text{RC}}^n}{\Role[p]_n}}, \TypedProt{\Prot[C]}{\Role[scr], \Role[trg]}{\Args[val]}{\cdot}{G_{\Prot[C]}}, \Typed{\Chan[s]}{G_{\text{RC}}^n}
\end{align*}
where $ G_{\text{RC}}^n $, $ G_{\Prot[C]} $, and $ \Rest{G_{\text{RC}}^n}{\Role[p]_i} $ are provided by the Examples~\ref{exa:globalType} and \ref{exa:restriction}.
Similarly, let
\begin{align*}
	\Delta = \Typed{\ATE{\Chan[s]}{\Role[p]_1}}{\ProjS{G_{\text{RC}}^n}{}{\Role[p]_1}}, \ldots, \Typed{\ATE{\Chan[s]}{\Role[p]_1}}{\ProjS{G_{\text{RC}}^n}{}{\Role[p]_1}}
\end{align*}
where $ \ProjS{G_{\text{RC}}^n}{}{\Role[p]_i} $ is provided by Example~\ref{exa:localType}.
We notice that $ \Delta $ is closed.
We first apply the Rule~$ (\mathsf{Pa}) $ $ n $ times to separate $ P_{\text{RC}}^n = \prod_{i = 1..n} \left( \PPar{\POutS{\Chan}{\Chan[s]}}{\PInp{\Chan}{\Chan[s]}{P_i}} \right) $ into $ n $ participants $ \PPar{\POutS{\Chan}{\Chan[s]}}{\PInp{\Chan}{\Chan[s]}{P_i}} $, whereby we split $ \Delta $ into $ \Delta_1 \otimes \ldots \otimes \Delta_n $ with $ \Delta_i = \Typed{\ATE{\Chan[s]}{\Role[p]_i}}{\ProjS{G_{\text{RC}}^n}{}{\Role[p]_i}} $.

Since $ \Delta_i \otimes \emptyset = \Delta_i = \Typed{\ATE{\Chan[s]}{\Role[p]_i}}{\ProjS{G_{\text{RC}}^n}{}{\Role[p]_i}} $ and $ \ProjS{G_{\text{RC}}^n}{}{\Role[p]_i} = \Proj{\Rest{G_{\text{RC}}^n}{\Role[p]_i}}{}{\Role[p]_i} $, we have:
\begin{align*}
	\hspace*{-1em}\dfrac{\dfrac{\dfrac{}{\Gamma \vdash \PEnd \triangleright \emptyset}(\mathsf{N}) \quad \GetType{\Chan_i} = \Rest{G_{\text{RC}}^n}{\Role[p]_i}}{\Gamma \vdash \POutS{\Chan_i}{\Chan[s]} \triangleright \Delta_i}(\mathsf{O}) \quad \dfrac{\Gamma \vdash P_i \triangleright \Typed{\AT{\Chan[s]}{\Role[p]_i}}{\ProjS{G_{\text{RC}}^n}{}{\Role[p]_i}} \quad \GetType{\Chan_i} = \Rest{G_{\text{RC}}^n}{\Role[p]_i}}{\Gamma \vdash \PInp{\Chan_i}{\Chan[s]}{P_i} \triangleright \emptyset}(\mathsf{I})}{\Gamma \vdash \PPar{\POutS{\Chan_i}{\Chan[s]}}{\PInp{\Chan_i}{\Chan[s]}{P_i}} \triangleright \Delta_i}(\mathsf{Pa})
\end{align*}
It remains to prove that $ \Gamma \vdash P_i \triangleright \Typed{\AT{\Chan[s]}{\Role[p]_i}}{\ProjS{G_{\text{RC}}^n}{}{\Role[p]_i}} $.

By Example~\ref{exa:process}, $ P_i $ starts with $ i - 1 $ sequential optional blocks and, by Example~\ref{exa:localType}, $ \ProjS{G_{\text{RC}}^n}{}{\Role[p]_i} $ similarly starts with $ i - 1 $ sequential local types of optional blocks. For each of theses blocks
\begin{align*}
	\dfrac{\begin{array}{c} \Role[p]_i, \Role[p]_j \ \dot{=} \ \Role[p]_i, \Role[p]_j \quad \Gamma \vdash P_{j \to i} \triangleright \Typed{\AT{\Chan[s]}{\Role[p]_i}}{T_{j \to i}}, \Typed{\Role[p]_i}{\OV{\Sort[V]}} \quad \nexists \Role, \tilde{\Sort[K]} \logdot \Typed{\Role}{\OV{\tilde{\Sort[K]}}} \in \emptyset\\ \Gamma \vdash P_{j \to i}' \triangleright \Typed{\AT{\Chan[s]}{\Role[p]_i}}{T_{j \to i}'} \quad \vdash \Typed{\Args[v]_{i, j}}{\Sort[V]} \quad \vdash \Typed{\Args[v]_{i, j - 1}}{\Sort[V]} \end{array}}{\Gamma \vdash \POpt{\Role[p]_i}{\Role[p]_i, \Role[p]_j}{P_{j \to i}}{\Args[v]_{i, j}}{\Args[v]_{i, j - 1}}{P_{j \to i}'} \triangleright \Typed{\AT{\Chan[s]}{\Role[p]_i}}{\LTOpt{\Role[p]_i, \Role[p]_j}{T_{j \to i}}{\Typed{\Args[v]_{i, j - 1}}{\!\Sort[V]}}{T_{j \to i}'}}}(\mathsf{Opt})
\end{align*}
with $ P_{j \to i} = \PEnt{\Chan[s]}{\Role[p]_j}{\Role[p]_i}{\Role[trg]}{\Args}{\PGet{\Args}{\Role[scr]}{\Role[trg]}{\PLab{\Labe[bc]}{\Args[v]}{\POptEnd{\Role[p]_i}{\Args[v]}}}} $, the type $ T_{j \to i} = \LTEntS{\Prot[C]}{\Role[trg]}{\Args[v]_{j, j - 1}}{\Role[p]_j} $ and where $ P_{j \to i}' $ and $ T_{j \to i}' $ are the respective continuations. For the content of the optional blocks we have to check the type of $ P_{j \to i} $.
\begin{align*}
	\hspace{-0.2em}\hspace{-1.5em}\dfrac{D_1 \quad \GetType{\Prot[C]} = \TypeOfProt{\Role[scr], \Role[trg]}{\Args[val]}{\cdot}{G_{\Prot[C]}} \quad \ProjS{G_{\Prot[C]}\!\Set[]{ \Subst{\Args[v]_{j, j - 1}}{\Args[val]} }}{}{\Role[trg]} = \LTGet{\Role[scr]}{\LTLabS{\Labe[bc]}{\Typed{\Args[v]_{j, j - 1}}{\Sort[V]}}}}{\Gamma \vdash \PEnt{\Chan[s]}{\Role[p]_j}{\Role[p]_i}{\Role[trg]}{\Args}{\PGet{\Args}{\Role[scr]}{\Role[trg]}{\PLab{\Labe[bc]}{\Args[v]}{\POptEnd{\Role[p]_i}{\Args[v]}}}} \triangleright \Typed{\AT{\Chan[s]}{\Role[p]_i}}{\LTEntS{\Prot[C]}{\Role[trg]}{\Args[v]_{j, j - 1}}{\Role[p]_j}}, \Typed{\Role[p]_i}{\OV{\Sort[V]}}}(\mathsf{J})
\end{align*}
with
\begin{align*}
	D_1 = \dfrac{\dfrac{\vdash \Typed{\Args[v]}{\Sort[V]}}{\Gamma \vdash \POptEnd{\Role[p]_i}{\Args[v]} \triangleright \Typed{\Role[p]_i}{\OV{\Sort[V]}}}(\mathsf{OptE}) \quad \vdash \Typed{\Args[v]}{\Sort[V]}}{\Gamma \vdash \PGet{\Args}{\Role[scr]}{\Role[trg]}{\PLab{\Labe[bc]}{\Args[v]}{\POptEnd{\Role[p]_i}{\Args[v]}}} \triangleright \Typed{\AT{\Args[x]}{\Role[trg]}}{\LTGet{\Role[scr]}{\LTLabS{\Labe[bc]}{\Typed{\Args[v]_{j, j - 1}}{\Sort[V]}}}}, \Typed{\Role[p]_i}{\OV{\Sort[V]}}}(\mathsf{C})
\end{align*}

After removing $ i - 1 $ optional blocks this way, $ P_{i - 1 \to i}' $ and $ T_{i - 1 \to i}' $ consist of $ n $ parallel components, respectively. We use the Rule~$ (\mathsf{Pa}) $ $ n $ times to separate these components. The $ n $'th component, we obtain this way, consists of $ n - 1 $ sequential optional blocks, respectively. Their type is checked similar to the first $ i - 1 $ such sequential optional blocks with a derivation for $ \PEnd $ using Rule~$ (\mathsf{N}) $ in the end.

It remains to show that $ \Gamma \vdash \PRes{\Chan[k]}{\POptNVS{\Role[p]_i}{\Role[p]_i, \Role[p]_j}{P_{i \text{ calls } j}}} \triangleright \Typed{\AT{\Chan[s]}{\Role[p]_i}}{\LTOptS{\Role[p]_i, \Role[p]_j}{T_{i \text{ calls } j}}{\cdot}} $ holds for all $ j = 1..n $ with $ j \neq i $, where
\begin{align*}
	P_{i \text{ calls } j} &= \PDecl{\Chan[k]}{\Chan[s]}{\Args[v]_{i, i - 1}}{\cdot}{\cdot}{P_{i \to j, \Chan[k]}}\\
	T_{i \text{ calls } j} &= \LTCall{\Prot[C]}{G_{\Prot[C]}}{\Args[v]_{i, i - 1}}{\Typed{\Args[val]}{\!\Sort[V]}}{\cdot}{T_{i \to j}}
\end{align*}
Let $ \Gamma' = \Gamma, \Typed{\Chan[k]}{G_{\Prot[C]}\!\Set[]{ \Subst{\Args[v]_{i, i - 1}}{\Args[val]} }} $.
\begin{align*}
	\dfrac{\dfrac{\Role[p]_i, \Role[p]_j \ \dot{=} \ \Role[p]_i, \Role[p]_j \quad \Gamma' \vdash P_{i \text{ calls } j} \triangleright \Typed{\AT{\Chan[s]}{\Role[p]_i}}{T_{i \text{ calls } j}}, \Typed{\Role[p]_i}{\OV{\cdot}} \quad \nexists \Role, \tilde{\Sort[K]} \logdot \Typed{\Role}{\OV{\tilde{\Sort[K]}}} \in \emptyset \quad \dfrac{}{\Gamma' \vdash \PEnd \triangleright \emptyset}(\mathsf{N}) \quad \vdash \Typed{\cdot}{\cdot} \quad \vdash \Typed{\cdot}{\cdot}}{\Gamma' \vdash \POptNVS{\Role[p]_i}{\Role[p]_i, \Role[p]_j}{P_{i \text{ calls } j}} \triangleright \Typed{\AT{\Chan[s]}{\Role[p]_i}}{\LTOptS{\Role[p]_i, \Role[p]_j}{T_{i \text{ calls } j}}{\cdot}}}(\mathsf{Opt})}{\Gamma \vdash \PRes{\Chan[k]}{\POptNVS{\Role[p]_i}{\Role[p]_i, \Role[p]_j}{P_{i \text{ calls } j}}} \triangleright \Typed{\AT{\Chan[s]}{\Role[p]_i}}{\LTOptS{\Role[p]_i, \Role[p]_j}{T_{i \text{ calls } j}}{\cdot}}}(\mathsf{R})
\end{align*}
with $ G_{\Prot[C]}\!\Set[]{ \Subst{\Args[v]_{i, i - 1}}{\Args[val]} } = \GTCom{\Role[scr]}{\Role[trg]}{\GTInpS{\Labe[bc]}{\Typed{\Args[v]_{i, i - 1}}{\!\Sort[V]}}} $,
\begin{align*}
	\dfrac{\begin{array}{c} \Gamma' \vdash P_{i \to j, \Chan[k]} \triangleright \Typed{\AT{\Chan[s]}{\Role[p]_i}}{T_{i \to j}}, \Delta_k, \Typed{\Role[p]_i}{\OV{\cdot}}\\ \GetType[\Gamma']{\Prot[C]} = \TypedProt{\Prot[C]}{\Role[scr], \Role[trg]}{\Args[val]}{\cdot}{G_{\Prot[C]}} \quad \ProjS{G_{\Prot[C]}\!\Set[]{ \Subst{\Args[v]_{i, i - 1}}{\Args[val]} }}{}{\Role[scr]} = \LTSend{\Role[trg]}{\LTLabS{\Labe[bc]}{\Args[v]_{i, i - 1}}}\\ \ProjS{G_{\Prot[C]}\!\Set[]{ \Subst{\Args[v]_{i, i - 1}}{\Args[val]} }}{}{\Role[trg]} = \LTGet{\Role[scr]}{\LTLabS{\Labe[bc]}{\Args[v]_{i, i - 1}}} \quad \vdash \Typed{\Args[v]_{i, i - 1}}{\Sort[V]} \quad \GetType[\Gamma']{\Chan[k]} = G_{\Prot[C]}\!\Set[]{ \Subst{\Args[v]_{i, i - 1}}{\Args[val]} } \end{array}}{\Gamma' \vdash \PDecl{\Chan[k]}{\Chan[s]}{\Args[v]_{i, i - 1}}{\cdot}{\cdot}{P_{i \to j, \Chan[k]}} \triangleright \Typed{\AT{\Chan[s]}{\Role[p]_i}}{\LTCall{\Prot[C]}{G_{\Prot[C]}}{\Args[v]_{i, i - 1}}{\Typed{\Args[val]}{\!\Sort[V]}}{\cdot}{T_{i \to j}}}, \Typed{\Role[p]_i}{\OV{\cdot}}}(\mathsf{New})
\end{align*}
and $ \Delta_k = \Typed{\ATI{\Chan[k]}{\Role[scr]}}{\LTSend{\Role[trg]}{\LTLabS{\Labe[bc]}{\Args[v]_{i, i - 1}}}}, \Typed{\ATI{\Chan[k]}{\Role[trg]}}{\LTGet{\Role[scr]}{\LTLabS{\Labe[bc]}{\Args[v]_{i, i - 1}}}} $.

It remains to show that $ \Gamma' \vdash P_{i \to j, \Chan[k]} \triangleright \Typed{\AT{\Chan[s]}{\Role[p]_i}}{T_{i \to j}}, \Delta_k, \Typed{\Role[p]_i}{\OV{\cdot}} $.
By Example~\ref{exa:process},
\begin{align*}
	P_{i \to j, \Chan[k]} = \PPar{\PReqS{\Chan[s]}{\Role[p]_i}{\Role[p]_i}{\Role[scr]}{\Chan[k]}}{\PReqS{\Chan[s]}{\Role[p]_i}{\Role[p]_j}{\Role[trg]}{\Chan[k]}}\PPar{}{\PEnt{\Chan[s]}{\Role[p]_i}{\Role[p]_i}{\Role[scr]}{\Args[z]}{\PSend{\Args[z]}{\Role[scr]}{\Role[trg]}{\Labe[bc]}{\Args[v]_{i, i - 1}}{\POptEnd{\Role[p]_i}{\cdot}}}}
\end{align*}
and, by Example~\ref{exa:localType},
\begin{align*}
	T_{i \to j} = \LTPar{\LTReqS{\Prot[C]}{\Role[src]}{\Args[v]_{i, i - 1}}{\!\Role[p]_i}}{\LTReqS{\Prot[C]}{\Role[trg]}{\Args[v]_{i, i - 1}}{\!\Role[p]_j}} \LTPar{}{\LTEntS{\Prot[C]}{\Role[scr]}{\Args[v]_{i, i - 1}}{\Role[p]_i}}
\end{align*}
We apply the Rule~$ (\mathsf{Pa}) $ two times such that it remains to show:
\begin{enumerate}[(1)]
	\item $ \Gamma' \vdash \PReqS{\Chan[s]}{\Role[p]_i}{\Role[p]_i}{\Role[scr]}{\Chan[k]} \triangleright \Typed{\AT{\Chan[s]}{\Role[p]_i}}{\LTReqS{\Prot[C]}{\Role[src]}{\Args[v]_{i, i - 1}}{\!\Role[p]_i}}, \Typed{\ATI{\Chan[k]}{\Role[scr]}}{\LTSend{\Role[trg]}{\LTLabS{\Labe[bc]}{\Args[v]_{i, i - 1}}}} $
	\item $ \Gamma' \vdash \PReqS{\Chan[s]}{\Role[p]_i}{\Role[p]_j}{\Role[trg]}{\Chan[k]} \triangleright \Typed{\AT{\Chan[s]}{\Role[p]_i}}{\LTReqS{\Prot[C]}{\Role[trg]}{\Args[v]_{i, i - 1}}{\!\Role[p]_j}}, \Typed{\ATI{\Chan[k]}{\Role[trg]}}{\LTGet{\Role[scr]}{\LTLabS{\Labe[bc]}{\Args[v]_{i, i - 1}}}} $
	\item $ \Gamma' \vdash \PEnt{\Chan[s]}{\Role[p]_i}{\Role[p]_i}{\Role[scr]}{\Args[z]}{\PSend{\Args[z]}{\Role[scr]}{\Role[trg]}{\Labe[bc]}{\Args[v]_{i, i - 1}}{\POptEnd{\Role[p]_i}{\cdot}}} \triangleright \Typed{\AT{\Chan[s]}{\Role[p]_i}}{\LTEntS{\Prot[C]}{\Role[scr]}{\Args[v]_{i, i - 1}}{\Role[p]_i}}, \Typed{\Role[p]_i}{\OV{\cdot}} $
\end{enumerate}

For the first case we have:
\begin{align*}
	\dfrac{\dfrac{}{\Gamma' \vdash \PEnd \triangleright \emptyset}(\mathsf{N}) \quad \GetType{\Prot[C]} = \TypedProt{\Prot[C]}{\Role[scr], \Role[trg]}{\Args[val]}{\cdot}{G_{\Prot[C]}} \quad \ProjS{G_{\Prot[C]}\!\Set[]{ \Subst{\Args[v]_{i, i}}{\Args[val]} }}{}{\Role[scr]} = \LTSend{\Role[trg]}{\LTLabS{\Labe[bc]}{\Args[v]_{i, i - 1}}}}{\Gamma' \vdash \PReqS{\Chan[s]}{\Role[p]_i}{\Role[p]_i}{\Role[scr]}{\Chan[k]} \triangleright \Typed{\AT{\Chan[s]}{\Role[p]_i}}{\LTReqS{\Prot[C]}{\Role[src]}{\Args[v]_{i, i - 1}}{\!\Role[p]_i}}, \Typed{\ATI{\Chan[k]}{\Role[scr]}}{\LTSend{\Role[trg]}{\LTLabS{\Labe[bc]}{\Args[v]_{i, i - 1}}}}}(\mathsf{P})
\end{align*}

The second case is similar. We have:
\begin{align*}
	\dfrac{\dfrac{}{\Gamma' \vdash \PEnd \triangleright \emptyset}(\mathsf{N}) \quad \GetType{\Prot[C]} = \TypedProt{\Prot[C]}{\Role[scr], \Role[trg]}{\Args[val]}{\cdot}{G_{\Prot[C]}} \quad \ProjS{G_{\Prot[C]}\!\Set[]{ \Subst{\Args[v]_{i, i}}{\Args[val]} }}{}{\Role[trg]} = \LTGet{\Role[scr]}{\LTLabS{\Labe[bc]}{\Args[v]_{i, i - 1}}}}{\Gamma' \vdash \PReqS{\Chan[s]}{\Role[p]_i}{\Role[p]_j}{\Role[trg]}{\Chan[k]} \triangleright \Typed{\AT{\Chan[s]}{\Role[p]_i}}{\LTReqS{\Prot[C]}{\Role[trg]}{\Args[v]_{i, i - 1}}{\!\Role[p]_j}}, \Typed{\ATI{\Chan[k]}{\Role[trg]}}{\LTGet{\Role[scr]}{\LTLabS{\Labe[bc]}{\Args[v]_{i, i - 1}}}}}(\mathsf{P})
\end{align*}

For the third case we have:
\begin{align*}
	\dfrac{\begin{array}{c}
		\dfrac{\dfrac{\vdash \Typed{\cdot}{\cdot}}{\Gamma' \vdash \POptEnd{\Role[p]_i}{\cdot} \triangleright \Typed{\Role[p]_i}{\OV{\cdot}}}(\mathsf{OptE}) \quad \vdash \Typed{\Args[v]_{i, i - 1}}{\Sort[V]}}{\Gamma' \vdash \PSend{\Args[z]}{\Role[scr]}{\Role[trg]}{\Labe[bc]}{\Args[v]_{i, i - 1}}{\POptEnd{\Role[p]_i}{\cdot}} \triangleright \Typed{\AT{\Args[z]}{\Role[scr]}}{\LTSend{\Role[trg]}{\LTLabS{\Labe[bc]}{\Args[v]_{i, i - 1}}}}, \Typed{\Role[p]_i}{\OV{\cdot}}}(\mathsf{S})\\
		\GetType[\Gamma']{\Prot[C]} = \TypedProt{\Prot[C]}{\Role[scr], \Role[trg]}{\Args[val]}{\cdot}{G_{\Prot[C]}} \quad
		\ProjS{G_{\Prot[C]}\!\Set[]{ \Subst{\Args[v]_{i, i}}{\Args[val]} }}{}{\Role[scr]} = \LTSend{\Role[trg]}{\LTLabS{\Labe[bc]}{\Args[v]_{i, i - 1}}}
	\end{array}}{\Gamma' \vdash \PEnt{\Chan[s]}{\Role[p]_i}{\Role[p]_i}{\Role[scr]}{\Args[z]}{\PSend{\Args[z]}{\Role[scr]}{\Role[trg]}{\Labe[bc]}{\Args[v]_{i, i - 1}}{\POptEnd{\Role[p]_i}{\cdot}}} \triangleright \Typed{\AT{\Chan[s]}{\Role[p]_i}}{\LTEntS{\Prot[C]}{\Role[scr]}{\Args[v]_{i, i - 1}}{\Role[p]_i}}, \Typed{\Role[p]_i}{\OV{\cdot}}}(\mathsf{J})
\end{align*}

We conclude that $ \Gamma \vdash P_{\text{RC}}^n \triangleright \Delta $ holds.

\section{Properties of the Type Systems}
\label{sec:properties}

In the following we analyse the properties of the (two versions of the) type systems.

We formally distinguish between the following sets:
\begin{compactitem}
	\item The set $ \nameSet $ of names that captures all kinds of channel names, session names, and names for values. We often use different identifiers to hint on the different purpose of a name, \eg we use $ \Chan $ for shared channels, $ \Chan[s], \Chan[k] $ for session names, and $ \Args[v] $ for values. We do however not formally distinguish between these different kinds of names but formally distinguish names from the following sets.
	\item The set $ \roleSet $ of roles, usually identified by $ \Role[r], \Role[r]', \Role[r]_i, \ldots $ (in the examples we used the roles $ \Role[p]_1, \ldots, \Role[p]_n, \Role[scr] $, and $ \Role[trg] $).
	\item The set $ \labelSet $ of labels, usually identified by $ \Labe, \Labe', \Labe_i, \ldots $ (in the examples we used the labels $ \Labe[c] $ and $ \Labe[bc] $ for communication).
	\item The set $ \procVarSet $ of process variables, usually identified by $ \TermV[X] $.
	\item The set $ \typeVarSet $ of type variables, usually identified by $ \TermV $.
\end{compactitem}
Moreover notice that kinds, usually identified by $ \Sort[S], \Sort[S]', \Sort[S]_, \ldots $ (and $ \Sort[V] $ in our examples), are neither global nor local types and can be formally distinguished from every global or local type.
Because of that, a statement $ \vdash \Typed{\Args[v]}{\Sort[S]} $ tells us that $ \Args[v] $ is a value that is different from all names that are used \eg as (either shared or session) channel.
Remember the type environments---$ \Gamma $ for the global types and $ \Delta $ for the session types---cannot contain multiple type statements for the same name or the same combination of a name and a role, respectively.

\subsection{Structural Congruence, Substitution, and Evaluation Contexts}

We start with a few auxiliary results.
The first Lemma tells us, that the property of being well-typed is preserved by structural congruence.
Note that following \cite{DemangeonHonda12} we handle recursion implicitly using the rules
\[ \begin{array}{c}
	\GTPar{G_1}{G_2} \equiv \GTPar{G_2}{G_1} \hspace*{1.4em}
	\GTPar{G_1}{\left( \GTPar{G_2}{G_3} \right)} \equiv \GTPar{\left( \GTPar{G_1}{G_2} \right)}{G_3} \hspace*{1.4em}
	\GTChoi{G_1}{\Role}{G_2} \equiv \GTChoi{G_2}{\Role}{G_1} \hspace*{1.4em}
	\GTRec{\TermV}{G} \equiv G\!\Set{ \Subst{\GTRec{\TermV}{P}}{\TermV} }
\end{array} \]
for global types and the rules
\[ \begin{array}{c}
	\LTPar{T_1}{T_2} \equiv \LTPar{T_2}{T_1} \hspace*{1.4em}
	\LTPar{T_1}{\left( \LTPar{T_2}{T_3} \right)} \equiv \LTPar{\left( \LTPar{T_1}{T_2} \right)}{T_3} \hspace*{1.4em}
	\LTChoi{T_1}{T_2} \equiv \LTChoi{T_2}{T_1} \hspace*{1.4em}
	\LTRec{\TermV}{T} \equiv T\!\Set{ \Subst{\LTRec{\TermV}{P}}{\TermV} }
\end{array} \]
for local types and usually equate structural equivalent types and processes.

\begin{lemma}
	\label{lem:typedStructuralCongruence}
	For both type systems:
	If $ \Gamma \vdash P \triangleright \Delta $ and $ P \equiv P' $ then $ \Gamma \vdash P' \triangleright \Delta $.
\end{lemma}

\begin{proof}
	We start with the larger type system, \ie the session types with optional blocks and sub-sessions.
	The proof is by induction on the structural congruence $ \equiv $ between processes.
	\begin{description}
		\item[Case $ \PPar{P}{\PEnd} \equiv P $:]
			Assume $ \Gamma \vdash \PPar{P}{\PEnd} \triangleright \Delta $.
			Then, by the typing rules of Figure~\ref{fig:typingRules}, the proof of this judgement has (modulo applications of Rule~(\textsf{S2}) that can be moved towards the type check of $ P $) to start with
			\begin{align*}
				\dfrac{\Gamma \vdash P \triangleright \Delta_{P} \quad \dfrac{}{\Gamma \vdash \PEnd \triangleright \emptyset}(\mathsf{N})}{\Gamma \vdash \PPar{P}{\PEnd} \triangleright \Delta}(\mathsf{Pa})
			\end{align*}
			where $ \Delta = \Delta_P \otimes \emptyset $ and thus $ \Delta_P = \Delta $.
			Then also $ \Gamma \vdash P \triangleright \Delta $.

			Assume $ \Gamma \vdash P \triangleright \Delta $. With the Rules~$ (\mathsf{N}) $ and $ (\mathsf{Pa}) $ and, because $ \Delta \otimes \emptyset = \Delta $, we then have
			\begin{align*}
				\dfrac{\Gamma \vdash P \triangleright \Delta \quad \dfrac{}{\Gamma \vdash \PEnd \triangleright \emptyset}(\mathsf{N})}{\Gamma \vdash \PPar{P}{\PEnd} \triangleright \Delta}(\mathsf{Pa})
			\end{align*}
			Hence also $ \Gamma \vdash \PPar{P}{\PEnd} \triangleright \Delta $.
		\item[Case $ \PPar{P_1}{P_2} \equiv \PPar{P_2}{P_1} $:]
			Assume $ \Gamma \vdash \PPar{P_1}{P_2} \triangleright \Delta $.
			Then, by the typing rules of Figure~\ref{fig:typingRules}, the proof of this judgement has to start with a number of applications of (\textsf{S2}) that reduce $ \Gamma \vdash \PPar{P_1}{P_2} \triangleright \Delta $ to $ \Gamma \vdash \PPar{P_1}{P_2} \triangleright \Delta' $ for some $ \Delta' $ such that
			\begin{align*}
				\dfrac{\Gamma \vdash P_1 \triangleright \Delta_{P1} \quad \Gamma \vdash P_2 \triangleright \Delta_{P2}}{\Gamma \vdash \PPar{P_1}{P_2} \triangleright \Delta'}(\mathsf{Pa})
			\end{align*}
			where $ \Delta' = \Delta_{P1} \otimes \Delta_{P2} $.
			With Rule~$ (\mathsf{Pa}) $ and, because $ \Delta_{P1} \otimes \Delta_{P2} = \Delta_{P2} \otimes \Delta_{P1} $, then also $ \Gamma \vdash \PPar{P_2}{P_1} \triangleright \Delta' $. We use the same applications of (\textsf{S2}) to derive $ \Gamma \vdash \PPar{P_2}{P_1} \triangleright \Delta $.

			The other direction is similar.
		\item[Case $ \PPar{P_1}{\left( \PPar{P_2}{P_3} \right)} \equiv \PPar{\left( \PPar{P_1}{P_2} \right)}{P_3} $:]
			Assume $ \Gamma \vdash \PPar{P_1}{\left( \PPar{P_2}{P_3} \right)} \triangleright \Delta $.
			Then, by the typing rules of Figure~\ref{fig:typingRules}, the proof of this judgement has to start with a number of applications of (\textsf{S2}) that reduce $ \Gamma \vdash \PPar{P_1}{\left( \PPar{P_2}{P_3} \right)} \triangleright \Delta $ to $ \Gamma \vdash \PPar{P_1}{\left( \PPar{P_2}{P_3} \right)} \triangleright \Delta' $ for some $ \Delta' $ such that
			\begin{align*}
				\dfrac{\Gamma \vdash P_1 \triangleright \Delta_{P1} \quad \dfrac{\Gamma \vdash P_2 \triangleright \Delta_{P2} \quad \Gamma \vdash P_3 \triangleright \Delta_{P3}}{\Gamma \vdash \PPar{P_2}{P_3} \triangleright \Delta_{P2 - 3}}(\mathsf{Pa})}{\Gamma \vdash \PPar{P_1}{\left( \PPar{P_2}{P_3} \right)} \triangleright \Delta'}(\mathsf{Pa})
			\end{align*}
			where $ \Delta' = \Delta_{P1} \otimes \Delta_{P2 - 3} $ and $ \Delta_{P2  - 3} = \Delta_{P2} \otimes \Delta_{P3} $.
			With Rule~$ (\mathsf{Pa}) $ and, because $ \Delta_{P1} \otimes \left( \Delta_{P2} \otimes \Delta_{P3} \right) = \left( \Delta_{P1} \otimes \Delta_{P2} \right) \otimes \Delta_{P3} $, we have
			\begin{align*}
				\dfrac{\dfrac{\Gamma \vdash P_1 \triangleright \Delta_{P1} \quad \Gamma \vdash P_2 \triangleright \Delta_{P2}}{\Gamma \vdash \PPar{P_1}{P_2} \triangleright \Delta_{P1} \otimes \Delta_{P2}}(\mathsf{Pa}) \quad \Gamma \vdash P_3 \triangleright \Delta_{P3}}{\Gamma \vdash \PPar{\left( \PPar{P_1}{P_2} \right)}{P_3} \triangleright \Delta'}(\mathsf{Pa})
			\end{align*}
			Using the same applications of (\textsf{S2}) we obtain $ \Gamma \vdash \PPar{\left( \PPar{P_1}{P_2} \right)}{P_3} \triangleright \Delta $.

			The other direction is similar.
		\item[Case $ \PRes{\Args}{\PEnd} \equiv \PEnd $:]
			Assume $ \Gamma \vdash \PRes{\Chan}{\PEnd} \triangleright \Delta $.
			Then, by the typing rules of Figure~\ref{fig:typingRules}, the proof of this judgement has to start with a number of applications of (\textsf{S2}) that reduce $ \Gamma \vdash \PRes{\Chan}{\PEnd} \triangleright \Delta $ to $ \Gamma \vdash \PRes{\Chan}{\PEnd} \triangleright \Delta' $ for some $ \Delta' $ such that
			\begin{align*}
				\dfrac{\dfrac{}{\Gamma, \Typed{\Args}{\AT{T}{\Role}} \vdash \PEnd \triangleright \Delta'}(\mathsf{N})}{\Gamma \vdash \PRes{\Args}{\PEnd} \triangleright \Delta'}(\mathsf{R})
			\end{align*}
			where $ \Delta' = \emptyset $.
			By Rule~$ (\mathsf{N}) $ and because $ \Delta' = \emptyset $, we have $ \Gamma \vdash \PEnd \triangleright \Delta' $.
			Using the same applications of (\textsf{S2}) we obtain $ \Gamma \vdash \PEnd \triangleright \Delta $.

			Assume $ \Gamma \vdash \PEnd \triangleright \Delta $.
			Then, by the typing rules of Figure~\ref{fig:typingRules}, the proof of this judgement has to start with a number of applications of (\textsf{S2}) followed by one application of (\textsf{N}), where for the last step the session environment has to be empty.
			With the Rules~$ (\mathsf{N}) $ and $ (\mathsf{R}) $ we then have
			\begin{align*}
				\dfrac{\dfrac{}{\Gamma, \Typed{\Args}{\AT{T'}{\Role'}} \vdash \PEnd \triangleright \emptyset}(\mathsf{N})}{\Gamma \vdash \PRes{\Args}{\PEnd} \triangleright \emptyset}(\mathsf{R})
			\end{align*}
			Using the same applications of (\textsf{S2}) we obtain $ \Gamma \vdash \PRes{\Args}{\PEnd} \triangleright \Delta $.
		\item[Case $ \PRes{\Args}{\PRes{\Args[y]}{P}} \equiv \PRes{\Args[y]}{\PRes{\Args}{P}} $:]
			Assume $ \Gamma \vdash \PRes{\Args}{\PRes{\Args[y]}{P}} \triangleright \Delta $.
			Then, by the typing rules of Figure~\ref{fig:typingRules}, the proof of this judgement has to start with a number of applications of (\textsf{S2}) that reduce $ \Gamma \vdash \PRes{\Args}{\PRes{\Args[y]}{P}} \triangleright \Delta $ to $ \Gamma \vdash \PRes{\Args}{\PRes{\Args[y]}{P}} \triangleright \Delta' $ for some $ \Delta' $ such that
			\begin{align*}
				\dfrac{\dfrac{\Gamma, \Typed{\Args}{\AT{T}{\Role}}, \Typed{\Args[y]}{\AT{T'}{\Role'}} \vdash P \triangleright \Delta'}{\Gamma, \Typed{\Args}{\AT{T}{\Role}} \vdash \PRes{\Args[y]}{P} \triangleright \Delta'}(\mathsf{R})}{\Gamma \vdash \PRes{\Args}{\PRes{\Args[y]}{P}} \triangleright \Delta'}(\mathsf{R})
			\end{align*}
			By $ \Gamma, \Typed{\Args}{\AT{T}{\Role}}, \Typed{\Args[y]}{\AT{T'}{\Role'}} \vdash P \triangleright \Delta' $, Rule~$ (\mathsf{R}) $ and because $ \Gamma, \Typed{\Args}{\AT{T}{\Role}}, \Typed{\Args[y]}{\AT{T'}{\Role'}} = \Gamma, \Typed{\Args[y]}{\AT{T'}{\Role'}}, \Typed{\Args}{\AT{T}{\Role}} $, we have
			\begin{align*}
				\dfrac{\dfrac{\Gamma, \Typed{\Args}{\AT{T}{\Role}}, \Typed{\Args[y]}{\AT{T'}{\Role'}} \vdash P \triangleright \Delta'}{\Gamma, \Typed{\Args[y]}{\AT{T'}{\Role'}} \vdash \PRes{\Args}{P} \triangleright \Delta'}(\mathsf{R})}{\Gamma \vdash \PRes{\Args[y]}{\PRes{\Args}{P}} \triangleright \Delta'}(\mathsf{R})
			\end{align*}
			Using the same applications of (\textsf{S2}) we obtain $ \Gamma \vdash \PRes{\Args[y]}{\PRes{\Args}{P}} \triangleright \Delta $.

			The other direction is similar.
		\item[Case $ \PRes{\Args}{\left( \PPar{P_1}{P_2} \right)} \equiv \PPar{P_1}{\PRes{\Args}{P_2}} $ if $ \Args \notin \FreeNames{P_1} $:]
			Assume $ \Gamma \vdash \PRes{\Args}{\left( \PPar{P_1}{P_2} \right)} \triangleright \Delta $.
			Then, by the typing rules of Figure~\ref{fig:typingRules}, the proof of this judgement has to start with a number of applications of (\textsf{S2}) that reduce $ \Gamma \vdash \PRes{\Args}{\left( \PPar{P_1}{P_2} \right)} \triangleright \Delta $ to $ \Gamma \vdash \PRes{\Args}{\left( \PPar{P_1}{P_2} \right)} \triangleright \Delta' $ for some $ \Delta' $ such that
			\begin{align*}
				\dfrac{\dfrac{\Gamma, \Typed{\Args}{\AT{T}{\Role}} \vdash P_1 \triangleright \Delta_{P1} \quad \Gamma, \Typed{\Args}{\AT{T}{\Role}} \vdash P_2 \triangleright \Delta_{P2}}{\Gamma, \Typed{\Args}{\AT{T}{\Role}} \vdash \PPar{P_1}{P_2} \triangleright \Delta'}(\mathsf{Pa})}{\Gamma \vdash \PRes{\Args}{\left( \PPar{P_1}{P_2} \right)} \triangleright \Delta'}(\mathsf{R})
			\end{align*}
			where $ \Delta' = \Delta_{P1} \otimes \Delta_{P2} $.
			The only rules that makes use of type declarations of channels from the global environment are the Rules~$ (\mathsf{New}) $, $ (\mathsf{I}) $, and $ (\mathsf{O}) $.
			Since $ \Args \notin \FreeNames{P_1} $, the Rules~$ (\mathsf{New}) $, $ (\mathsf{I}) $, and $ (\mathsf{O}) $---even if they are used---will not check for the type of $ \Args $ in the judgement $ \Gamma, \Typed{\Args}{\AT{T}{\Role}} \vdash P_1 \triangleright \Delta_{P1} $.
			Hence also $ \Gamma \vdash P_1 \triangleright \Delta_{P1} $.
			With $ \Gamma, \Typed{\Args}{\AT{T}{\Role}} \vdash P_2 \triangleright \Delta_{P2} $ and the Rules~$ (\mathsf{Pa}) $ and $ (\mathsf{R}) $ then
			\begin{align*}
				\dfrac{\Gamma \vdash P_1 \triangleright \Delta_{P1} \quad \dfrac{\Gamma, \Typed{\Args}{\AT{T}{\Role}} \vdash P_2 \triangleright \Delta_{P2}}{\Gamma \vdash \PRes{\Args}{P_2} \triangleright \Delta_{P2}}(\mathsf{R})}{\Gamma \vdash \PPar{P_1}{\PRes{\Args}{P_2}} \triangleright \Delta'}(\mathsf{Pa})
			\end{align*}
			Using the same applications of (\textsf{S2}) we obtain $ \Gamma \vdash \PPar{P_1}{\PRes{\Args}{P_2}} \triangleright \Delta $.

			Assume $ \Gamma \vdash \PPar{P_1}{\PRes{\Args}{P_2}} \triangleright \Delta $.
			Then, by the typing rules of Figure~\ref{fig:typingRules}, the proof of this judgement has to start with a number of applications of (\textsf{S2}) that reduce $ \Gamma \vdash \PPar{P_1}{\PRes{\Args}{P_2}} \triangleright \Delta $ to $ \Gamma \vdash \PPar{P_1}{\PRes{\Args}{P_2}} \triangleright \Delta' $ for some $ \Delta' $ such that
			\begin{align*}
				\dfrac{\Gamma \vdash P_1 \triangleright \Delta_{P1} \quad \dfrac{\Gamma, \Typed{\Args}{\AT{T}{\Role}} \vdash P_2 \triangleright \Delta_{P2}}{\Gamma \vdash \PRes{\Args}{P_2} \triangleright \Delta_{P2}}(\mathsf{R})}{\Gamma \vdash \PPar{P_1}{\PRes{\Args}{P_2}} \triangleright \Delta'}(\mathsf{Pa})
			\end{align*}
			where $ \Delta' = \Delta_{P1} \otimes \Delta_{P2} $.
			Since $ \Args \notin \FreeNames{P_1} $, no rule will check for the type of $ \Args $ in the judgement $ \Gamma \vdash P_1 \triangleright \Delta_{P1} $.
			Hence also $ \Gamma, \Typed{\Args}{\AT{T}{\Role}} \vdash P_1 \triangleright \Delta_{P1} $.
			With $ \Gamma, \Typed{\Args}{\AT{T}{\Role}} \vdash P_2 \triangleright \Delta_{P2} $ and the Rules~$ (\mathsf{Pa}) $ and $ (\mathsf{R}) $ then
			\begin{align*}
				\dfrac{\dfrac{\Gamma, \Typed{\Args}{\AT{T}{\Role}} \vdash P_1 \triangleright \Delta_{P1} \quad \Gamma, \Typed{\Args}{\AT{T}{\Role}} \vdash P_2 \triangleright \Delta_{P2}}{\Gamma, \Typed{\Args}{\AT{T}{\Role}} \vdash \PPar{P_1}{P_2} \triangleright \Delta'}(\mathsf{Pa})}{\Gamma \vdash \PRes{\Args}{\left( \PPar{P_1}{P_2} \right)} \triangleright \Delta'}(\mathsf{R})
			\end{align*}
			Using the same applications of (\textsf{S2}) we obtain $ \Gamma \vdash \PRes{\Args}{\left( \PPar{P_1}{P_2} \right)} \triangleright \Delta $.
		\item[Case $ \PChoi{P_1}{P_2} \equiv \PChoi{P_2}{P_1} $:]
			Assume $ \Gamma \vdash \PChoi{P_1}{P_2} \triangleright \Delta $.
			Then, by the typing rules of Figure~\ref{fig:typingRules}, the proof of this judgement has to start with a number of applications of (\textsf{S2}) that reduce $ \Gamma \vdash \PChoi{P_1}{P_2} \triangleright \Delta $ to $ \Gamma \vdash \PChoi{P_1}{P_2} \triangleright \Delta' $ for some $ \Delta' $ such that
			\begin{align*}
				\dfrac{\Gamma \vdash P_1 \triangleright \Delta'', \Typed{\AT{\Chan[s]}{\Role}}{T_1} \quad \Gamma \vdash P_2 \triangleright \Delta'', \Typed{\AT{\Chan[s]}{\Role}}{T_2}}{\Gamma \vdash \PChoi{P_1}{P_2} \triangleright \Delta'}(\mathsf{S1})
			\end{align*}
			where $ \Delta' = \Delta'', \Typed{\AT{\Chan[s]}{\Role}}{T_1 \oplus T_2} $.
			With Rule~$ (\mathsf{S1}) $ and, because $ T_1 \oplus T_2 = T_2 \oplus T_1 $, then also $ \Gamma \vdash \PChoi{P_2}{P_1} \triangleright \Delta' $.
			Using the same applications of (\textsf{S2}) we obtain $ \Gamma \vdash \PChoi{P_2}{P_1} \triangleright \Delta $.

			The other direction is similar.
		\item[Case $ \PRec{\TermV[X]}{P} \equiv P\!\Set{ \Subst{\PRec{\TermV[X]}{P}}{\TermV[X]} } $:]
			Assume $ \Gamma \vdash \PRec{\TermV[X]}{P} \triangleright \Delta $.
			Then, by the typing rules of Figure~\ref{fig:typingRules}, $ \Delta $ contains (modulo some applications of (\textsf{S2})) some $ \Typed{\AT{\Chan[s]}{\Role}}{\LTRec{\TermV}{T}} $ to check the type of $ \PRec{\TermV[X]}{P} $.
			Because of $ \LTRec{\TermV}{T} \equiv T\!\Set[]{ \Subst{\LTRec{\TermV}{T}}{\TermV} } $, then also $ \Gamma \vdash P\!\Set{ \Subst{\PRec{\TermV[X]}{P}}{\TermV[X]} } \triangleright \Delta $.

			The other direction is similar.
	\end{description}
	Since the rules of structural congruence $ \equiv $ are the same for both type systems and because none of the above cases relies on one of the Rules~$ (\mathsf{P}) $, $ (\mathsf{J}) $, or $ (\mathsf{New}) $, this lemma also holds for the session types with optional blocks but without sub-sessions.
\end{proof}

The next Lemma allows us to substitute (session) names within type judgements if the session environment is adapted accordingly.

\begin{lemma}
	\label{lem:typeSubstA}
	For both type systems:
	If $ \Gamma \vdash P \triangleright \Delta, \Typed{\AT{\Args}{\Role}}{T} $ then $ \Gamma \vdash P\!\Set[]{ \Subst{\Chan[s]}{\Args} } \triangleright \Delta, \Typed{\AT{\Chan[s]}{\Role}}{T} $.
\end{lemma}

\begin{proof}
	We start with the larger type system, \ie the session types with optional blocks and sub-sessions.
	Assume $ \Gamma \vdash P \triangleright \Delta, \Typed{\AT{\Args}{\Role}}{T} $.
	We perform an induction on the derivation of this judgement from the typing rules of Figure~\ref{fig:typingRules}.
	Note that the Rules~$ (\mathsf{N}) $ and $ (\mathsf{OptE}) $ refer to base cases, while the remaining rules refer to the induction steps of the induction.
	\begin{description}
		\item[Case Rule~$ (\mathsf{N}) $:]
			In this case $ P = \PEnd $ and $ \Delta, \Typed{\AT{\Args}{\Role}}{T} = \emptyset $.
			This is a contradiction.
			Hence the implication holds trivially.
		\item[Case Rule~$ (\mathsf{OptE}) $:]
			In this case $ P = \POptEnd{\Role'}{\tilde{\Args[v]}} $ and $ \Delta, \Typed{\AT{\Args}{\Role}}{T} = \Typed{\Role'}{\OV{\tilde{\Sort}}} $.
			Again this is a contradiction.
		\item[Case Rule~$ (\mathsf{I}) $:]
			In this case $ P = \PInp{\Chan}{\Args'}{P'} $ and $ \Delta, \Typed{\AT{\Args}{\Role}}{T} = \Delta' $ and we have $ \Gamma \vdash P' \triangleright \Delta', \Typed{\AT{\Args'}{\Role'}}{T'} $ and $ \GetType{\Chan} = \AT{T'}{\Role'} $.
			Using alpha-conversion before Rule~$ (\mathsf{I}) $ we can ensure that $ \Args \neq \Args' $.
			Then, by the induction hypothesis, $ \Gamma \vdash P' \triangleright \Delta, \Typed{\AT{\Args}{\Role}}{T}, \Typed{\AT{\Args'}{\Role'}}{T'} $ implies $ \Gamma \vdash P'\!\Set[]{ \Subst{\Chan[s]}{\Args} } \triangleright \Delta, \Typed{\AT{\Chan[s]}{\Role}}{T}, \Typed{\AT{\Args'}{\Role'}}{T'} $.
			With Rule~$ (\mathsf{I}) $ and $ \GetType{\Chan} = \AT{T'}{\Role'} $ we have $ \Gamma \vdash \PInp{\Chan}{\Args'}{\left( P'\!\Set[]{ \Subst{\Chan[s]}{\Args} } \right)} \triangleright \Delta, \Typed{\AT{\Chan[s]}{\Role}}{T} $.
			Since $ \Args \neq \Args' $, then $ \Gamma \vdash P\!\Set[]{ \Subst{\Chan[s]}{\Args} } \triangleright \Delta, \Typed{\AT{\Chan[s]}{\Role}}{T} $.
		\item[Case Rule~$ (\mathsf{O}) $:]
			In this case $ P = \POut{\Chan}{\Chan[s]'}{P'} $ and $ \Delta, \Typed{\AT{\Args}{\Role}}{T} = \Delta', \Typed{\ATE{\Chan[s]'}{\Role'}}{T'} $ and we have $ \Gamma \vdash P' \triangleright \Delta' $ and $ \GetType{\Chan} = \AT{T'}{\Role'} $.
			Hence $ \Args \neq \Chan[s]' $, $ \Delta = \Delta'', \Typed{\ATE{\Chan[s]'}{\Role'}}{T'} $, and $ \Delta' = \Delta'', \Typed{\AT{\Args}{\Role}}{T} $.
			By the induction hypothesis, $ \Gamma \vdash P' \triangleright \Delta'', \Typed{\AT{\Args}{\Role}}{T} $ implies $ \Gamma \vdash P'\!\Set[]{ \Subst{\Chan[s]}{\Args} } \triangleright \Delta'', \Typed{\AT{\Chan[s]}{\Role}}{T} $.
			With Rule~$ (\mathsf{O}) $ and $ \GetType{\Chan} = \AT{T'}{\Role'} $ we have $ \Gamma \vdash \POut{\Chan}{\Chan[s]'}{\left( P'\!\Set[]{ \Subst{\Chan[s]}{\Args} } \right)} \triangleright \Delta'', \Typed{\AT{\Chan[s]}{\Role}}{T}, \Typed{\ATE{\Chan[s]'}{\Role'}}{T'} $.
			Since $ \Args \neq \Chan[s]' $, then $ \Gamma \vdash P\!\Set[]{ \Subst{\Chan[s]}{\Args} } \triangleright \Delta, \Typed{\AT{\Chan[s]}{\Role}}{T} $.
		\item[Case Rule~$ (\mathsf{C}) $:]
			In this case $ P = \PGet{\Chan[k]}{\Role_1}{\Role_2}{_{i \in \indexSet} \Set{ \PLab{\Labe_i}{\tilde{\Args[y]}_i}{P_i} }} $ and $ \Delta, \Typed{\AT{\Args}{\Role}}{T} = \Delta', \Typed{\AT{\Chan[k]}{\Role_2}}{\LTGet{\Role_1}{_{i \in \indexSet{}} \Set{ \LTLab{\Labe_i}{\Typed{\tilde{\Args}_i}{\tilde{\Sort}_i}}{T_i} }}} $ and we have $ \Gamma \vdash P_i \triangleright \Delta', \Typed{\AT{\Chan[k]}{\Role_2}}{T_i} $ and $ \vdash \Typed{\tilde{\Args[y]}_i}{\tilde{\Sort}_i} $ for all $ i \in \indexSet $.
			Using alpha-conversion before Rule~$ (\mathsf{C}) $ we can ensure that $ \Args \notin \tilde{\Args[y]}_i $ for all $ i \in \indexSet $.
			We distinguish between the cases (1)~$ \Args = \Chan[k] $ and (2)~$ \Args \neq \Chan[k] $.
			\begin{enumerate}[(1)]
				\item Then $ \Delta = \Delta' $, $ \Role = \Role_2 $, and $ T = \LTGet{\Role_1}{_{i \in \indexSet{}} \Set{ \LTLab{\Labe_i}{\Typed{\tilde{\Args}_i}{\tilde{\Sort}_i}}{T_i} }} $.
					By the induction hypothesis, $ \Gamma \vdash P_i \triangleright \Delta, \Typed{\AT{\Args}{\Role}}{T_i} $ implies $ \Gamma \vdash P_i\!\Set[]{ \Subst{\Chan[s]}{\Args} } \triangleright \Delta, \Typed{\AT{\Chan[s]}{\Role}}{T_i} $ for all $ i \in \indexSet $.
					With Rule~$ (\mathsf{C}) $ and $ \vdash \Typed{\tilde{\Args[y]}_i}{\tilde{\Sort}_i} $ for all $ i \in \indexSet $ we have $ \Gamma \vdash \PGet{\Chan[s]}{\Role_1}{\Role}{_{i \in \indexSet} \Set{ \PLab{\Labe_i}{\tilde{\Args[y]}_i}{\left( P_i\!\Set[]{ \Subst{\Chan[s]}{\Args} } \right)} }} \triangleright \Delta, \Typed{\AT{\Chan[s]}{\Role}}{\LTGet{\Role_1}{_{i \in \indexSet{}} \Set{ \LTLab{\Labe_i}{\Typed{\tilde{\Args}_i}{\tilde{\Sort}_i}}{T_i} }}} $.
					Since $ \Args \notin \tilde{\Args[y]}_i $ for all $ i \in \indexSet $, $ \Gamma \vdash P\!\Set[]{ \Subst{\Chan[s]}{\Args} } \triangleright \Delta, \Typed{\AT{\Chan[s]}{\Role}}{T} $.
				\item Then $ \Delta = \Delta'', \Typed{\AT{\Chan[k]}{\Role_2}}{\LTGet{\Role_1}{_{i \in \indexSet{}} \Set{ \LTLab{\Labe_i}{\Typed{\tilde{\Args}_i}{\tilde{\Sort}_i}}{T_i} }}} $ and $ \Delta' = \Delta'', \Typed{\AT{\Args}{\Role}}{T} $.
					By the induction hypothesis, $ \Gamma \vdash P_i \triangleright \Delta'', \Typed{\AT{\Args}{\Role}}{T}, \Typed{\AT{\Chan[k]}{\Role_2}}{T_i} $ implies $ \Gamma \vdash P_i\!\Set[]{ \Subst{\Chan[s]}{\Args} } \triangleright \Delta'', \Typed{\AT{\Chan[s]}{\Role}}{T}, \Typed{\AT{\Chan[k]}{\Role_2}}{T_i} $ for all $ i \in \indexSet $.
					With Rule~$ (\mathsf{C}) $ and $ \vdash \Typed{\tilde{\Args[y]}_i}{\tilde{\Sort}_i} $ for all $ i \in \indexSet $ we have $ \Gamma \vdash \PGet{\Chan[k]}{\Role_1}{\Role_2}{_{i \in \indexSet} \Set{ \PLab{\Labe_i}{\tilde{\Args[y]}_i}{\left( P_i\!\Set[]{ \Subst{\Chan[s]}{\Args} } \right)} }} \triangleright \Delta'', \Typed{\AT{\Chan[s]}{\Role}}{T}, \Typed{\AT{\Chan[k]}{\Role_2}}{\LTGet{\Role_1}{_{i \in \indexSet{}} \Set{ \LTLab{\Labe_i}{\Typed{\tilde{\Args}_i}{\tilde{\Sort}_i}}{T_i} }}} $.
					Since $ \Args \notin \tilde{\Args[y]}_i $ for all $ i \in \indexSet $, $ \Gamma \vdash P\!\Set[]{ \Subst{\Chan[s]}{\Args} } \triangleright \Delta, \Typed{\AT{\Chan[s]}{\Role}}{T} $.
			\end{enumerate}
		\item[Case Rule~$ (\mathsf{S}) $:]
			In this case $ P = \PSend{\Chan[k]}{\Role_1}{\Role_2}{\Labe_j}{\tilde{\Args[v]}}{P'} $ and $ \Delta, \Typed{\AT{\Args}{\Role}}{T} = \Delta', \Typed{\AT{\Chan[k]}{\Role_1}}{\LTSend{\Role_2}{_{i \in \indexSet} \Set{ \LTLab{\Labe_i}{\Typed{\tilde{\Args_i}}{\tilde{\Sort}_i}}{T_i} }}} $ and we have $ \Gamma \vdash P' \triangleright \Delta', \Typed{\AT{\Chan[k]}{\Role_1}}{T_j} $, $ \vdash \Typed{\tilde{\Args[v]}}{\tilde{\Sort}_j} $, and $ \Args \notin \tilde{\Args[v]} $.
			We distinguish between the cases (1)~$ \Args = \Chan[k] $ and (2)~$ \Args \neq \Chan[k] $.
			\begin{enumerate}[(1)]
				\item Then $ \Delta = \Delta' $, $ \Role = \Role_1 $, and $ T = \LTSend{\Role_2}{_{i \in \indexSet} \Set{ \LTLab{\Labe_i}{\Typed{\tilde{\Args_i}}{\tilde{\Sort_i}}}{T_i} }} $.
					By the induction hypothesis, $ \Gamma \vdash P' \triangleright \Delta, \Typed{\AT{\Args}{\Role}}{T_j} $ implies $ \Gamma \vdash P'\!\Set[]{ \Subst{\Chan[s]}{\Args} } \triangleright \Delta, \Typed{\AT{\Chan[s]}{\Role}}{T_j} $.
					With Rule~$ (\mathsf{S}) $ and $ \vdash \Typed{\tilde{\Args[v]}}{\tilde{\Sort}_j} $ we have $ \Gamma \vdash \PSend{\Chan[s]}{\Role}{\Role_2}{\Labe_j}{\tilde{\Args[v]}}{\left( P'\!\Set[]{ \Subst{\Chan[s]}{\Args} } \right)} \triangleright \Delta, \Typed{\AT{\Chan[s]}{\Role}}{\LTSend{\Role_2}{_{i \in \indexSet} \Set{ \LTLab{\Labe_i}{\Typed{\tilde{\Args_i}}{\tilde{\Sort_i}}}{T_i} }}} $.
					Since $ \Args \notin \tilde{\Args[v]} $, then $ \Gamma \vdash P\!\Set[]{ \Subst{\Chan[s]}{\Args} } \triangleright \Delta, \Typed{\AT{\Chan[s]}{\Role}}{T} $.
				\item Then $ \Delta = \Delta'', \Typed{\AT{\Chan[k]}{\Role_1}}{\LTSend{\Role_2}{_{i \in \indexSet} \Set{ \LTLab{\Labe_i}{\Typed{\tilde{\Args_i}}{\tilde{\Sort_i}}}{T_i} }}} $ and $ \Delta' = \Delta'', \Typed{\AT{\Args}{\Role}}{T} $.
					By the induction hypothesis, $ \Gamma \vdash P' \triangleright \Delta'', \Typed{\AT{\Args}{\Role}}{T}, \Typed{\AT{\Chan[k]}{\Role_1}}{T_j} $ implies $ \Gamma \vdash P'\!\Set[]{ \Subst{\Chan[s]}{\Args} } \triangleright \Delta'', \Typed{\AT{\Chan[s]}{\Role}}{T}, \Typed{\AT{\Chan[k]}{\Role_1}}{T_j} $.
					With Rule~$ (\mathsf{S}) $ and $ \vdash \Typed{\tilde{\Args[v]}}{\tilde{\Sort}_j} $ we have $ \Gamma \vdash \PSend{\Chan[k]}{\Role_1}{\Role_2}{\Labe_j}{\tilde{\Args[v]}}{\left( P'\!\Set[]{ \Subst{\Chan[s]}{\Args} } \right)} \triangleright \Delta'', \Typed{\AT{\Chan[s]}{\Role}}{T}, \Typed{\AT{\Chan[k]}{\Role_1}}{\LTSend{\Role_2}{_{i \in \indexSet} \Set{ \LTLab{\Labe_i}{\Typed{\tilde{\Args_i}}{\tilde{\Sort_i}}}{T_i} }}} $.
					Since $ \Args \notin \tilde{\Args[v]} $, then $ \Gamma \vdash P\!\Set[]{ \Subst{\Chan[s]}{\Args} } \triangleright \Delta, \Typed{\AT{\Chan[s]}{\Role}}{T} $.
			\end{enumerate}
		\item[Case Rule~$ (\mathsf{R}) $:]
			In this case $ P = \PRes{\Args'}{P'} $ and $ \Delta, \Typed{\AT{\Args}{\Role}}{T} = \Delta' $ and we have $ \Gamma, \Typed{\Args'}{\AT{T'}{\Role'}} \vdash P' \triangleright \Delta' $.
			Using alpha-conversion before Rule~$ (\mathsf{R}) $ we can ensure that $ \Args \neq \Args' $.
			By the induction hypothesis, $ \Gamma, \Typed{\Args'}{\AT{T'}{\Role'}} \vdash P' \triangleright \Delta, \Typed{\AT{\Args}{\Role}}{T} $ implies $ \Gamma, \Typed{\Args'}{\AT{T'}{\Role'}} \vdash P'\!\Set[]{ \Subst{\Chan[s]}{\Args} } \triangleright \Delta, \Typed{\AT{\Chan[s]}{\Role}}{T} $.
			With Rule~$ (\mathsf{R}) $ we have $ \Gamma \vdash \PRes{\Args'}{\left( P'\!\Set[]{ \Subst{\Chan[s]}{\Args} } \right)} \triangleright \Delta, \Typed{\AT{\Chan[s]}{\Role}}{T} $.
			Since $ \Args \neq \Args' $, we have $ \Gamma \vdash P\!\Set[]{ \Subst{\Chan[s]}{\Args} } \triangleright \Delta, \Typed{\AT{\Chan[s]}{\Role}}{T} $.
		\item[Case Rule~$ (\mathsf{P}) $:]
			In this case $ P = \PReq{\Chan[s]'}{\Role_1}{\Role_2}{\Role_3}{\Chan[k]}{P'} $ and $ \Delta, \Typed{\AT{\Args}{\Role}}{T} = \Delta', \Typed{\AT{\Chan[s]'}{\Role_1}}{\LTReq{\Prot}{\Role_3}{\tilde{\Args[v]}}{\Role_2}{T_1}}, \Typed{\ATI{\Chan[k]}{\Role_3}}{T_3} $ and we have $ \Gamma \vdash P' \triangleright \Delta', \Typed{\AT{\Chan[s]'}{\Role_1}}{T_1} $, $ \GetType{\Prot} = \TypeOfProt{\tilde{\Role}_4}{\tilde{\Args[y]}}{\tilde{\Role}_5}{G} $, and $ \ProjS{G\!\Set[]{ \Subst{\tilde{\Args[v]}}{\tilde{\Args[y]}} }}{}{\Role_3} = T_3 $.
			Hence $ \Args \neq \Chan[k] $.
			We distinguish between the cases (1)~$ \Args = \Chan[s]' $ and (2)~$ \Args \neq \Chan[s]' $.
			\begin{enumerate}[(1)]
				\item Then $ \Delta = \Delta', \Typed{\ATI{\Chan[k]}{\Role_3}}{T_3} $, $ \Role = \Role_1 $, and $ T = \LTReq{\Prot}{\Role_3}{\tilde{\Args[v]}}{\Role_2}{T_1} $.
					By the induction hypothesis, $ \Gamma \vdash P' \triangleright \Delta', \Typed{\AT{\Args}{\Role}}{T_1} $ implies $ \Gamma \vdash P'\!\Set[]{ \Subst{\Chan[s]}{\Args} } \triangleright \Delta', \Typed{\AT{\Chan[s]}{\Role}}{T_1} $.
					With Rule~$ (\mathsf{P}) $, $ \GetType{\Prot} = \TypeOfProt{\tilde{\Role}_4}{\tilde{\Args[y]}}{\tilde{\Role}_5}{G} $, and $ \ProjS{G\!\Set[]{ \Subst{\tilde{\Args[v]}}{\tilde{\Args[y]}} }}{}{\Role_3} = T_3 $ we have $ \Gamma \vdash \PReq{\Chan[s]}{\Role}{\Role_2}{\Role_3}{\Chan[k]}{\left( P'\!\Set[]{ \Subst{\Chan[s]}{\Args} } \right)} \triangleright \Delta', \Typed{\AT{\Chan[s]}{\Role}}{\LTReq{\Prot}{\Role_3}{\tilde{\Args[v]}}{\Role_2}{T_1}}, \Typed{\ATI{\Chan[k]}{\Role_3}}{T_3} $.
					Hence $ \Gamma \vdash P\!\Set[]{ \Subst{\Chan[s]}{\Args} } \triangleright \Delta, \Typed{\AT{\Chan[s]}{\Role}}{T} $.
				\item Then $ \Delta = \Delta'', \Typed{\AT{\Chan[s]'}{\Role^1}}{\LTReq{\Prot}{\Role_3}{\tilde{\Args[v]}}{\Role_2}{T_1}}, \Typed{\ATI{\Chan[k]}{\Role_3}}{T_3} $ and $ \Delta' = \Delta'', \Typed{\AT{\Args}{\Role}}{T} $.
					By the induction hypothesis, $ \Gamma \vdash P' \triangleright \Delta'', \Typed{\AT{\Args}{\Role}}{T}, \Typed{\AT{\Chan[s]'}{\Role_1}}{T_1} $ implies $ \Gamma \vdash P'\!\Set[]{ \Subst{\Chan[s]}{\Args} } \triangleright \Delta'', \Typed{\AT{\Chan[s]}{\Role}}{T}, \Typed{\AT{\Chan[s]'}{\Role_1}}{T_1} $.
					With Rule~$ (\mathsf{P}) $, $ \GetType{\Prot} = \TypeOfProt{\tilde{\Role}_4}{\tilde{\Args[y]}}{\tilde{\Role}_5}{G} $, and $ \ProjS{G\!\Set[]{ \Subst{\tilde{\Args[v]}}{\tilde{\Args[y]}} }}{}{\Role_3} = T_3 $ we have $ \Gamma \vdash \PReq{\Chan[s]'}{\Role_1}{\Role_2}{\Role_3}{\Chan[k]}{\left( P'\!\Set[]{ \Subst{\Chan[s]}{\Args} } \right)} \triangleright \Delta'', \Typed{\AT{\Chan[s]}{\Role}}{T}, \Typed{\AT{\Chan[s]'}{\Role_1}}{\LTReq{\Prot}{\Role_3}{\tilde{\Args[v]}}{\Role_2}{T_1}}, \Typed{\ATI{\Chan[k]}{\Role_3}}{T_3} $.
					Hence $ \Gamma \vdash P\!\Set[]{ \Subst{\Chan[s]}{\Args} } \triangleright \Delta, \Typed{\AT{\Chan[s]}{\Role}}{T} $.
			\end{enumerate}
		\item[Case Rule~$ (\mathsf{J}) $:]
			In this case we have $ P = \PEnt{\Chan[s]'}{\Role_1}{\Role_2}{\Role_3}{\Args'}{P'} $ and $ \Delta, \Typed{\AT{\Args}{\Role}}{T} = \Delta', \Typed{\AT{\Chan[s]'}{\Role_2}}{\LTEnt{\Prot}{\Role_3}{\tilde{\Args[v]}}{\Role_1}{T_2}} $ and we have $ \Gamma \vdash P' \triangleright \Delta', \Typed{\AT{\Chan[s]'}{\Role_2}}{T_2}, \Typed{\AT{\Args'}{\Role_3}}{T_3} $, $ \GetType{\Prot} = \TypeOfProt{\tilde{\Role}_4}{\tilde{\Args[y]}}{\tilde{\Role}_5}{G} $, and $ \ProjS{G\!\Set[]{ \Subst{\tilde{\Args[v]}}{\tilde{\Args[y]}} }}{}{\Role_3} = T_3 $.
			Using alpha-conversion before Rule~$ (\mathsf{J}) $ we can ensure that $ \Args \neq \Args' $.
			We distinguish between the cases (1)~$ \Args = \Chan[s]' $ and (2)~$ \Args \neq \Chan[s]' $.
			\begin{enumerate}[(1)]
				\item Then $ \Delta = \Delta' $, $ \Role = \Role_2 $, and $ T = \LTEnt{\Prot}{\Role_3}{\tilde{\Args[v]}}{\Role_1}{T_2} $.
					By the induction hypothesis, $ \Gamma \vdash P' \triangleright \Delta, \Typed{\AT{\Args}{\Role}}{T_2}, \Typed{\AT{\Args'}{\Role_3}}{T_3} $ implies $ \Gamma \vdash P'\!\Set[]{ \Subst{\Chan[s]}{\Args} } \triangleright \Delta, \Typed{\AT{\Chan[s]}{\Role}}{T_2}, \Typed{\AT{\Args'}{\Role_3}}{T_3} $.
					With Rule~$ (\mathsf{J}) $, $ \GetType{\Prot} = \TypeOfProt{\tilde{\Role}_4}{\tilde{\Args[y]}}{\tilde{\Role}_5}{G} $, and $ \ProjS{G\!\Set[]{ \Subst{\tilde{\Args[v]}}{\tilde{\Args[y]}} }}{}{\Role_3} = T_3 $ we have $ \Gamma \vdash \PEnt{\Chan[s]}{\Role_1}{\Role}{\Role_3}{\Args'}{\left( P'\!\Set[]{ \Subst{\Chan[s]}{\Args} } \right)} \triangleright \Delta, \Typed{\AT{\Chan[s]}{\Role}}{\LTEnt{\Prot}{\Role_3}{\tilde{\Args[v]}}{\Role_1}{T_2}} $.
					Since $ \Args \neq \Args' $, we have $ \Gamma \vdash P\!\Set[]{ \Subst{\Chan[s]}{\Args} } \triangleright \Delta, \Typed{\AT{\Chan[s]}{\Role}}{T} $.
				\item Then $ \Delta = \Delta'', \Typed{\AT{\Chan[s]'}{\Role_2}}{\LTEnt{\Prot}{\Role_3}{\tilde{\Args[v]}}{\Role_1}{T_2}} $ and $ \Delta' = \Delta'', \Typed{\AT{\Args}{\Role}}{T} $.
					By the induction hypothesis,
					\begin{align*}
						\Gamma \vdash P' \triangleright \Delta'', \Typed{\AT{\Args}{\Role}}{T}, \Typed{\AT{\Chan[s]'}{\Role_2}}{T_2}, \Typed{\AT{\Args'}{\Role_3}}{T_3}
					\end{align*}
					implies $ \Gamma \vdash P'\!\Set[]{ \Subst{\Chan[s]}{\Args} } \triangleright \Delta'', \Typed{\AT{\Chan[s]}{\Role}}{T}, \Typed{\AT{\Chan[s]'}{\Role_2}}{T_2}, \Typed{\AT{\Args'}{\Role_3}}{T_3} $.
					With Rule~$ (\mathsf{J}) $, $ \GetType{\Prot} = \TypeOfProt{\tilde{\Role}_4}{\tilde{\Args[y]}}{\tilde{\Role}_5}{G} $, and $ \ProjS{G\!\Set[]{ \Subst{\tilde{\Args[v]}}{\tilde{\Args[y]}} }}{}{\Role_3} = T_3 $ we have:
					\begin{align*}
						\Gamma \vdash \PEnt{\Chan[s]'}{\Role_1}{\Role_2}{\Role_3}{\Args'}{\left( P'\!\Set[]{ \Subst{\Chan[s]}{\Args} } \right)} \triangleright \Delta'', \Typed{\AT{\Chan[s]}{\Role}}{T}, \Typed{\AT{\Chan[s]'}{\Role_2}}{\LTEnt{\Prot}{\Role_3}{\tilde{\Args[v]}}{\Role_1}{T_2}}
					\end{align*}
					Since $ \Args \neq \Args' $, we have $ \Gamma \vdash P\!\Set[]{ \Subst{\Chan[s]}{\Args} } \triangleright \Delta, \Typed{\AT{\Chan[s]}{\Role}}{T} $.
			\end{enumerate}
		\item[Case Rule~$ (\mathsf{New}) $:]
			In this case $ P = \PDecl{\Chan[k]}{\Chan[s]'}{\tilde{\Args[v]}}{\tilde{\Chan}}{\tilde{\Role}'''}{P'} $ and $ \Delta, \Typed{\AT{\Args}{\Role}}{T} = \Delta', \Typed{\AT{\Chan[s]'}{\Role'}}{\LTCall{\Prot}{G}{\tilde{\Args[v]}}{\Typed{\tilde{\Args[y]}}{\tilde{\Sort}}}{\tilde{\Role}'''}{T'}} $ and we have $ \Gamma \vdash P' \triangleright \Delta', \Typed{\AT{\Chan[s]'}{\Role'}}{T'}, \Typed{\ATI{\Chan[k]}{\Role''_1}}{T'_1}, \ldots, \Typed{\ATI{\Chan[k]}{\Role''_n}}{T'_n}, \Typed{\ATE{\Chan[k]}{\Role'''_1}}{T'_{n + 1}}, \ldots, \Typed{\ATE{\Chan[k]}{\Role'''_m}}{T'_{n + m}} $, $ \GetType{\Prot} = \TypeOfProt{\tilde{\Role}''}{\tilde{\Args[y]}}{\tilde{\Role}'''}{G} $, $ \forall i \logdot \GetType{\Chan_i} = \AT{T'_{i + n}}{\Role'''_{i + n}} $, $ \forall i \logdot \ProjS{G\!\Set[]{ \Subst{\tilde{\Args[v]}}{\tilde{\Args[y]}} }}{}{\Role''_i} = T'_i $, $ \forall j \logdot \ProjS{G\!\Set[]{ \Subst{\tilde{\Args[v]}}{\tilde{\Args[y]}} }}{}{\Role'''_j} = T'_{j + n} $, $ \vdash \Typed{\tilde{\Args[v]}}{\tilde{\Sort}} $, and $ \GetType{\Chan[k]} = {\Prot\!\Set[]{ \Subst{\tilde{\Args[v]}}{\tilde{\Args[y]}} }} $.
			Since $ \vdash \Typed{\tilde{\Args[v]}}{\tilde{\Sort}} $, we have $ \Args \notin \tilde{\Args[v]} $.
			We distinguish between the cases (1)~$ \Args = \Chan[s]' $ and (2)~$ \Args \neq \Chan[s]' $.
			\begin{enumerate}[(1)]
				\item Then $ \Delta = \Delta' $, $ \Role = \Role' $, and $ T = \LTCall{\Prot}{G}{\tilde{\Args[v]}}{\Typed{\tilde{\Args[y]}}{\tilde{\Sort}}}{\tilde{\Role}'''}{T'} $.
					By the induction hypothesis,
					\begin{align*}
						\Gamma \vdash P' \triangleright & \Delta, \Typed{\AT{\Args}{\Role}}{T'}, \Typed{\ATI{\Chan[k]}{\Role''_1}}{T'_1}, \ldots, \Typed{\ATI{\Chan[k]}{\Role''_n}}{T'_n}, \Typed{\ATE{\Chan[k]}{\Role'''_1}}{T'_{n + 1}}, \ldots, \Typed{\ATE{\Chan[k]}{\Role'''_m}}{T'_{n + m}}
					\end{align*}
					implies
					\begin{align*}
						\Gamma \vdash P'\!\Set[]{ \Subst{\Chan[s]}{\Args} } \triangleright & \Delta, \Typed{\AT{\Chan[s]}{\Role}}{T'}, \Typed{\ATI{\Chan[k]}{\Role''_1}}{T'_1}, \ldots, \Typed{\ATI{\Chan[k]}{\Role''_n}}{T'_n}, \Typed{\ATE{\Chan[k]}{\Role'''_1}}{T'_{n + 1}}, \ldots, \Typed{\ATE{\Chan[k]}{\Role'''_m}}{T'_{n + m}}
					\end{align*}
					With Rule~$ (\mathsf{New}) $ and $ \GetType{\Prot} = \TypeOfProt{\tilde{\Role}''}{\tilde{\Args[y]}}{\tilde{\Role}'''}{G} $ and $ \forall i \logdot \GetType{\Chan_i} = \AT{T'_{i + n}}{\Role'''_{i + n}} $, $ \forall i \logdot \ProjS{G\!\Set[]{ \Subst{\tilde{\Args[v]}}{\tilde{\Args[y]}} }}{}{\Role''_i} = T'_i $ and $ \forall j \logdot \ProjS{G\!\Set[]{ \Subst{\tilde{\Args[v]}}{\tilde{\Args[y]}} }}{}{\Role'''_j} = T'_{j + n} $, $ \vdash \Typed{\tilde{\Args[v]}}{\tilde{\Sort}} $ and $ \GetType{\Chan[k]} = {\Prot\!\Set[]{ \Subst{\tilde{\Args[v]}}{\tilde{\Args[y]}} }} $ we have
					\begin{align*}
						& \Gamma \vdash \PDecl{\Chan[k]}{\Chan[s]'}{\tilde{\Args[v]}}{\tilde{\Chan}}{\tilde{\Role}'''}{\left( P'\!\Set[]{ \Subst{\Chan[s]}{\Args} } \right)} \triangleright \; \Delta, \Typed{\AT{\Chan[s]}{\Role}}{\LTCall{\Prot}{G}{\tilde{\Args[v]}}{\Typed{\tilde{\Args[y]}}{\tilde{\Sort}}}{\tilde{\Role}'''}{T'}}
					\end{align*}
					Since $ \Args \notin \tilde{\Args[v]} $, we have $ \Gamma \vdash P\!\Set[]{ \Subst{\Chan[s]}{\Args} } \triangleright \Delta, \Typed{\AT{\Chan[s]}{\Role}}{T} $.
				\item Then $ \Delta = \Delta'', \Typed{\AT{\Chan[s]'}{\Role'}}{\LTCall{\Prot}{G}{\tilde{\Args[v]}}{\Typed{\tilde{\Args[y]}}{\tilde{\Sort}}}{\tilde{\Role}'''}{T'}} $ and $ \Delta' = \Delta'', \Typed{\AT{\Args}{\Role}}{T} $.
					By the induction hypothesis, $ \Gamma \vdash P' \triangleright \Delta'', \Typed{\AT{\Args}{\Role}}{T}, \Typed{\AT{\Chan[s]'}{\Role'}}{T'}, \Typed{\ATI{\Chan[k]}{\Role''_1}}{T'_1}, \ldots, \Typed{\ATI{\Chan[k]}{\Role''_n}}{T'_n}, \Typed{\ATE{\Chan[k]}{\Role'''_1}}{T'_{n + 1}}, \ldots, \Typed{\ATE{\Chan[k]}{\Role'''_m}}{T'_{n + m}} $ implies $ \Gamma \vdash P'\!\Set[]{ \Subst{\Chan[s]}{\Args} } \triangleright \Delta'', \Typed{\AT{\Chan[s]}{\Role}}{T}, \Typed{\AT{\Chan[s]'}{\Role'}}{T'}, \Typed{\ATI{\Chan[k]}{\Role''_1}}{T'_1}, \ldots, \Typed{\ATI{\Chan[k]}{\Role''_n}}{T'_n}, \Typed{\ATE{\Chan[k]}{\Role'''_1}}{T'_{n + 1}}, \ldots, \Typed{\ATE{\Chan[k]}{\Role'''_m}}{T'_{n + m}} $.
					With Rule~$ (\mathsf{New}) $ and $ \GetType{\Prot} = \TypeOfProt{\tilde{\Role}''}{\tilde{\Args[y]}}{\tilde{\Role}'''}{G} $ and $ \forall i \logdot \GetType{\Chan_i} = \AT{T'_{i + n}}{\Role'''_{i + n}} $ and $ \forall i \logdot \ProjS{G\!\Set[]{ \Subst{\tilde{\Args[v]}}{\tilde{\Args[y]}} }}{}{\Role''_i} = T'_i $ and $ \forall j \logdot \ProjS{G\!\Set[]{ \Subst{\tilde{\Args[v]}}{\tilde{\Args[y]}} }}{}{\Role'''_j} = T'_{j + n} $ and $ \vdash \Typed{\tilde{\Args[v]}}{\tilde{\Sort}} $ and $ \GetType{\Chan[k]} = {\Prot\!\Set[]{ \Subst{\tilde{\Args[v]}}{\tilde{\Args[y]}} }} $ we have
					\begin{align*}
						& \Gamma \vdash \PDecl{\Chan[k]}{\Chan[s]'}{\tilde{\Args[v]}}{\tilde{\Chan}}{\tilde{\Role}'''}{\left( P'\!\Set[]{ \Subst{\Chan[s]}{\Args} } \right)} \triangleright \; \Delta'', \Typed{\AT{\Chan[s]}{\Role}}{T}, \Typed{\AT{\Chan[s]'}{\Role'}}{\LTCall{\Prot}{G}{\tilde{\Args[v]}}{\Typed{\tilde{\Args[y]}}{\tilde{\Sort}}}{\tilde{\Role}'''}{T'}}
					\end{align*}
					Since $ \Args \notin \tilde{\Args[v]} $, we have $ \Gamma \vdash P\!\Set[]{ \Subst{\Chan[s]}{\Args} } \triangleright \Delta, \Typed{\AT{\Chan[s]}{\Role}}{T} $.
			\end{enumerate}
		\item[Case Rule~$ (\mathsf{S1}) $:]
			In this case $ P = \PChoi{P_1}{P_2} $ and $ \Delta, \Typed{\AT{\Args}{\Role}}{T} = \Delta', \Typed{\AT{\Chan[s]'}{\Role'}}{T_1 \oplus T_2} $ and we have $ \Gamma \vdash P_1 \triangleright \Delta', \Typed{\AT{\Chan[s]'}{\Role'}}{T_1} $ and $ \Gamma \vdash P_2 \triangleright \Delta', \Typed{\AT{\Chan[s]'}{\Role'}}{T_2} $.
			We distinguish between the cases (1)~$ \Args = \Chan[s]' $ and (2)~$ \Args \neq \Chan[s]' $.
			\begin{enumerate}[(1)]
				\item Then $ \Delta = \Delta' $, $ \Role = \Role' $, and $ T = T_1 \oplus T_2 $.
					By the induction hypothesis, $ \Gamma \vdash P_1 \triangleright \Delta, \Typed{\AT{\Args}{\Role}}{T_1} $ and $ \Gamma \vdash P_2 \triangleright \Delta, \Typed{\AT{\Args}{\Role}}{T_2} $ imply $ \Gamma \vdash P_1\!\Set[]{ \Subst{\Chan[s]}{\Args} } \triangleright \Delta, \Typed{\AT{\Chan[s]}{\Role}}{T_1} $ and $ \Gamma \vdash P_2\!\Set[]{ \Subst{\Chan[s]}{\Args} } \triangleright \Delta, \Typed{\AT{\Chan[s]}{\Role}}{T_2} $.
					With Rule~$ (\mathsf{S1}) $ we have $ \Gamma \vdash \PChoi{\left( P_1\!\Set[]{ \Subst{\Chan[s]}{\Args} } \right)}{\left( P_2\!\Set[]{ \Subst{\Chan[s]}{\Args} } \right)} \triangleright \Delta, \Typed{\AT{\Chan[s]}{\Role}}{\LTChoi{T_1}{T_2}} $.
					Hence $ \Gamma \vdash P\!\Set[]{ \Subst{\Chan[s]}{\Args} } \triangleright \Delta, \Typed{\AT{\Chan[s]}{\Role}}{T} $.
				\item Then $ \Delta = \Delta'', \Typed{\AT{\Chan[s]'}{\Role'}}{T_1 \oplus T_2} $ and $ \Delta' = \Delta'', \Typed{\AT{\Args}{\Role}}{T} $.
					By the induction hypothesis, $ \Gamma \vdash P_1 \triangleright \Delta'', \Typed{\AT{\Args}{\Role}}{T}, \Typed{\AT{\Chan[s]'}{\Role'}}{T_1} $ and $ \Gamma \vdash P_2 \triangleright \Delta'', \Typed{\AT{\Args}{\Role}}{T}, \Typed{\AT{\Chan[s]'}{\Role'}}{T_2} $ imply $ \Gamma \vdash P_1\!\Set[]{ \Subst{\Chan[s]}{\Args} } \triangleright \Delta'', \Typed{\AT{\Chan[s]}{\Role}}{T}, \Typed{\AT{\Chan[s]'}{\Role'}}{T_1} $ and $ \Gamma \vdash P_2\!\Set[]{ \Subst{\Chan[s]}{\Args} } \triangleright \Delta'', \Typed{\AT{\Chan[s]}{\Role}}{T}, \Typed{\AT{\Chan[s]'}{\Role'}}{T_2} $.
					With Rule~$ (\mathsf{S1}) $ we have $ \Gamma \vdash \PChoi{\left( P_1\!\Set[]{ \Subst{\Chan[s]}{\Args} } \right)}{\left( P_2\!\Set[]{ \Subst{\Chan[s]}{\Args} } \right)} \triangleright \Delta'', \Typed{\AT{\Chan[s]}{\Role}}{T}, \Typed{\AT{\Chan[s]'}{\Role'}}{\LTChoi{T_1}{T_2}} $.
					Hence $ \Gamma \vdash P\!\Set[]{ \Subst{\Chan[s]}{\Args} } \triangleright \Delta, \Typed{\AT{\Chan[s]}{\Role}}{T} $.
			\end{enumerate}
		\item[Case Rule~$ (\mathsf{S2}) $:]
			In this case $ \Delta, \Typed{\AT{\Args}{\Role}}{T} = \Delta', \Typed{\AT{\Chan[s]'}{\Role'}}{T_1 \oplus T_2} $ and we have $ \Gamma \vdash P \triangleright \Delta', \Typed{\AT{\Chan[s]'}{\Role'}}{T_i} $ with $ i \in \Set[]{ 1, 2 } $.
			We distinguish between the cases (1)~$ \Args = \Chan[s]' $ and (2)~$ \Args \neq \Chan[s]' $.
			\begin{enumerate}[(1)]
				\item Then $ \Delta = \Delta' $, $ \Role = \Role' $, and $ T = T_1 \oplus T_2 $.
					By the induction hypothesis, $ \Gamma \vdash P \triangleright \Delta, \Typed{\AT{\Args}{\Role}}{T_i} $ implies $ \Gamma \vdash P\!\Set[]{ \Subst{\Chan[s]}{\Args} } \triangleright \Delta, \Typed{\AT{\Chan[s]}{\Role}}{T_i} $.
					With Rule~$ (\mathsf{S2}) $ we have $ \Gamma \vdash P\!\Set[]{ \Subst{\Chan[s]}{\Args} } \triangleright \Delta, \Typed{\AT{\Chan[s]}{\Role}}{\LTChoi{T_1}{T_2}} $.
					Hence $ \Gamma \vdash P\!\Set[]{ \Subst{\Chan[s]}{\Args} } \triangleright \Delta, \Typed{\AT{\Chan[s]}{\Role}}{T} $.
				\item Then $ \Delta = \Delta'', \Typed{\AT{\Chan[s]'}{\Role'}}{T_1 \oplus T_2} $ and $ \Delta' = \Delta'', \Typed{\AT{\Args}{\Role}}{T} $.
					By the induction hypothesis, $ \Gamma \vdash P \triangleright \Delta'', \Typed{\AT{\Args}{\Role}}{T}, \Typed{\AT{\Chan[s]'}{\Role'}}{T_i} $ implies $ \Gamma \vdash P\!\Set[]{ \Subst{\Chan[s]}{\Args} } \triangleright \Delta'', \Typed{\AT{\Chan[s]}{\Role}}{T}, \Typed{\AT{\Chan[s]'}{\Role'}}{T_i} $.
					With Rule~$ (\mathsf{S2}) $ we have $ \Gamma \vdash P\!\Set[]{ \Subst{\Chan[s]}{\Args} } \triangleright \Delta'', \Typed{\AT{\Chan[s]}{\Role}}{T}, \Typed{\AT{\Chan[s]'}{\Role'}}{\LTChoi{T_1}{T_2}} $.
					Hence $ \Gamma \vdash P\!\Set[]{ \Subst{\Chan[s]}{\Args} } \triangleright \Delta, \Typed{\AT{\Chan[s]}{\Role}}{T} $.
			\end{enumerate}
		\item[Case Rule~$ (\mathsf{Pa}) $:]
			In this case $ P = \PPar{P_1}{P_2} $ and $ \Delta, \Typed{\AT{\Args}{\Role}}{T} = \Delta_1 \otimes \Delta_2 $ and we have $ \Gamma \vdash P_1 \triangleright \Delta_1 $ and $ \Gamma \vdash P_2 \triangleright \Delta_2 $.
			Hence $ \Typed{\AT{\Args}{\Role}}{T_1} \in \Delta_1 $, $ \Typed{\AT{\Args}{\Role}}{T_2} \in \Delta_2 $, and $ T = \LTPar{T_1}{T_2} $.
			Then $ \Delta_1 = \Delta_1', \Typed{\AT{\Args}{\Role}}{T_1} $, $ \Delta_2 = \Delta_2', \Typed{\AT{\Args}{\Role}}{T_2} $, and $ \Delta = \Delta_1' \otimes \Delta_2' $.
			By the induction hypothesis, $ \Gamma \vdash P_1 \triangleright \Delta_1', \Typed{\AT{\Args}{\Role}}{T_1} $ and $ \Gamma \vdash P_2 \triangleright \Delta_2', \Typed{\AT{\Args}{\Role}}{T_2} $ imply $ \Gamma \vdash P_1\!\Set[]{ \Subst{\Chan[s]}{\Args} } \triangleright \Delta_1', \Typed{\AT{\Chan[s]}{\Role}}{T_1} $ and $ \Gamma \vdash P_2\!\Set[]{ \Subst{\Chan[s]}{\Args} } \triangleright \Delta_2', \Typed{\AT{\Chan[s]}{\Role}}{T_2} $.
			With Rule~$ (\mathsf{Pa}) $ we have $ \Gamma \vdash \PPar{\left( P_1\!\Set[]{ \Subst{\Chan[s]}{\Args} } \right)}{\left( P_2\!\Set[]{ \Subst{\Chan[s]}{\Args} } \right)} \triangleright \left( \Delta_1', \Typed{\AT{\Chan[s]}{\Role}}{T_1} \right) \otimes \left( \Delta_2', \Typed{\AT{\Chan[s]}{\Role}}{T_2} \right) $.
			Hence $ \Gamma \vdash P\!\Set[]{ \Subst{\Chan[s]}{\Args} } \triangleright \Delta, \Typed{\AT{\Chan[s]}{\Role}}{T} $.
		\item[Case Rule~$ (\mathsf{Opt}) $:]
			Here $ P = \POpt{\Role_1}{\tilde{\Role}}{P_1}{\tilde{\Args}'}{\tilde{\Args[v]}}{P_2} $ and $ \Delta, \Typed{\AT{\Args}{\Role}}{T} = \Delta_1 \otimes \Delta_2, \Typed{\AT{\Chan[s]'}{\Role_1}}{\LTOpt{\tilde{\Role}}{T_1}{\Typed{\tilde{\Args[y]}}{\tilde{\Sort}}}{T_2}} $, and we have $ \Gamma \vdash P_1 \triangleright \Delta_1, \Typed{\AT{\Chan[s]'}{\Role_1}}{T_1}, \Typed{\Role_1}{\OV{\tilde{\Sort}}} $, $ \nexists \Role', \Sort[K] \logdot \Typed{\Role'}{\OV{\tilde{\Sort[K]}}} \in \Delta_1 $, $ \Gamma \vdash P_2 \triangleright \Delta_2, \Typed{\AT{\Chan[s]'}{\Role_1}}{T_2} $, $ \vdash \Typed{\tilde{\Args}'}{\tilde{\Sort}} $, and $ \vdash \Typed{\tilde{\Args[v]}}{\tilde{\Sort}} $.
			Hence $ \Args \notin \tilde{\Args[v]} $ and $ \Args \notin \tilde{\Args[y]} $.
			Using alpha-conversion before Rule~$ (\mathsf{Opt}) $ we can ensure that $ \Args \notin \tilde{\Args}' $.
			We distinguish between the cases (1)~$ \Args = \Chan[s]' $ and (2)~$ \Args \neq \Chan[s]' $.
			\begin{enumerate}[(1)]
				\item Then $ \Delta = \Delta_1 \otimes \Delta_2 $, $ \Role = \Role_1 $, and $ T = \LTOpt{\tilde{\Role}}{T_1}{\Typed{\tilde{\Args[y]}}{\tilde{\Sort}}}{T_2} $.
					By the induction hypothesis, $ \Gamma \vdash P_1 \triangleright \Delta_1, \Typed{\AT{\Args}{\Role}}{T_1}, \Typed{\Role_1}{\OV{\tilde{\Sort}}} $ and $ \Gamma \vdash P_2 \triangleright \Delta_2, \Typed{\AT{\Args}{\Role}}{T_2} $ imply $ \Gamma \vdash P_1\!\Set[]{ \Subst{\Chan[s]}{\Args} } \triangleright \Delta_1, \Typed{\AT{\Chan[s]}{\Role}}{T_1}, \Typed{\Role_1}{\OV{\tilde{\Sort}}} $ and $ \Gamma \vdash P_2\!\Set[]{ \Subst{\Chan[s]}{\Args} } \triangleright \Delta_2, \Typed{\AT{\Chan[s]}{\Role}}{T_2} $.
					With Rule~$ (\mathsf{Opt}) $, $ \nexists \Role', \Sort[K] \logdot \Typed{\Role'}{\OV{\tilde{\Sort[K]}}} \in \Delta_1 $, $ \vdash \Typed{\tilde{\Args}'}{\tilde{\Sort}} $, and $ \vdash \Typed{\tilde{\Args[v]}}{\tilde{\Sort}} $ we have $ \Gamma \vdash \POpt{\Role_1}{\tilde{\Role}}{P_1\!\Set[]{ \Subst{\Chan[s]}{\Args} }}{\tilde{\Args}'}{\tilde{\Args[v]}}{\left( P_2\!\Set[]{ \Subst{\Chan[s]}{\Args} } \right)} \triangleright \Delta_1 \otimes \Delta_2, \Typed{\AT{\Chan[s]}{\Role}}{\LTOpt{\tilde{\Role}}{T_1}{\Typed{\tilde{\Args[y]}}{\tilde{\Sort}}}{T_2}} $.
					Since $ \Args \notin \tilde{\Args}' $ and $ \Args \notin \tilde{\Args[v]} $, we have $ \Gamma \vdash P\!\Set[]{ \Subst{\Chan[s]}{\Args} } \triangleright \Delta, \Typed{\AT{\Chan[s]}{\Role}}{T} $.
				\item Then $ \Delta = \Delta_1' \otimes \Delta_2', \Typed{\AT{\Chan[s]'}{\Role_1}}{\LTOpt{\tilde{\Role}}{T_1}{\Typed{\tilde{\Args[y]}}{\tilde{\Sort}}}{T_2}} $, $ \Delta_1 = \Delta_1', \Typed{\AT{\Args}{\Role}}{T_1'} $, $ \Delta_2 = \Delta_2', \Typed{\AT{\Args}{\Role}}{T_2'} $, and $ T = \LTPar{T_1'}{T_2'} $.
					By the induction hypothesis, $ \Gamma \vdash P_1 \triangleright \Delta_1', \Typed{\AT{\Args}{\Role}}{T_1'}, \Typed{\AT{\Chan[s]'}{\Role_1}}{T_1}, \Typed{\Role_1}{\OV{\tilde{\Sort}}} $ and $ \Gamma \vdash P_2 \triangleright \Delta_2', \Typed{\AT{\Args}{\Role}}{T_2'}, \Typed{\AT{\Chan[s]'}{\Role_1}}{T_2} $ imply $ \Gamma \vdash P_1\!\Set[]{ \Subst{\Chan[s]}{\Args} } \triangleright \Delta_1', \Typed{\AT{\Chan[s]}{\Role}}{T_1'}, \Typed{\AT{\Chan[s]'}{\Role_1}}{T_1}, \Typed{\Role_1}{\OV{\tilde{\Sort}}} $ and $ \Gamma \vdash P_2\!\Set[]{ \Subst{\Chan[s]}{\Args} } \triangleright \Delta_2', \Typed{\AT{\Chan[s]}{\Role}}{T_2'}, \Typed{\AT{\Chan[s]'}{\Role_1}}{T_2} $.
					With Rule~$ (\mathsf{Opt}) $, $ \nexists \Role', \Sort[K] \logdot \Typed{\Role'}{\OV{\tilde{\Sort[K]}}} \in \Delta_1 $, $ \vdash \Typed{\tilde{\Args}'}{\tilde{\Sort}} $, and $ \vdash \Typed{\tilde{\Args[v]}}{\tilde{\Sort}} $ we have that $ \Gamma \vdash \POpt{\Role_1}{\tilde{\Role}}{P_1\!\Set[]{ \Subst{\Chan[s]}{\Args} }}{\tilde{\Args}'}{\tilde{\Args[v]}}{\left( P_2\!\Set[]{ \Subst{\Chan[s]}{\Args} } \right)} \triangleright \left( \Delta_1', \Typed{\AT{\Chan[s]}{\Role}}{T_1'} \right) \otimes \left( \Delta_2', \Typed{\AT{\Chan[s]}{\Role}}{T_2'} \right), \Typed{\AT{\Chan[s]'}{\Role_1}}{\LTOpt{\tilde{\Role}}{T_1}{\Typed{\tilde{\Args[y]}}{\tilde{\Sort}}}{T_2}} $.
					Since $ \Args \notin \tilde{\Args}' $ and $ \Args \notin \tilde{\Args[v]} $, we have $ \Gamma \vdash P\!\Set[]{ \Subst{\Chan[s]}{\Args} } \triangleright \Delta, \Typed{\AT{\Chan[s]}{\Role}}{T} $.
			\end{enumerate}
	\end{description}
	The proof for the session types with optional blocks but without sub-sessions is similar but omits the cases for the Rules~$ (\mathsf{P}) $, $ (\mathsf{J}) $, and $ (\mathsf{New}) $. The remaining cases do not rely on the Rules~$ (\mathsf{P}) $, $ (\mathsf{J}) $, or $ (\mathsf{New}) $.
\end{proof}

Moreover, values of the same kind can be substituted in the process without changing the session environment.

\begin{lemma}
	\label{lem:typeSubstB}
	In both type systems:
	If $ \Gamma \vdash P \triangleright \Delta $, $ \vdash \Typed{\Args[y]}{\Sort} $, and $ \vdash \Typed{\Args[v]}{\Sort} $ then $ \Gamma \vdash P\!\Set[]{ \Subst{\Args[v]}{\Args[y]} } \triangleright \Delta $.
\end{lemma}

\begin{proof}
	We start with the larger type system, \ie the session types with optional blocks and sub-sessions.
	Assume $ \Gamma \vdash P \triangleright \Delta $, $ \vdash \Typed{\Args[y]}{\Sort} $, and $ \vdash \Typed{\Args[v]}{\Sort} $.
	We perform an induction on the derivation of the judgement from the typing rules of Figure~\ref{fig:typingRules}.
	Note that the Rules~$ (\mathsf{N}) $ and $ (\mathsf{OptE}) $ refer to base cases, while the remaining rules refer to the induction steps of the induction.
	Also note that the only rules with free values that can be substituted are $ (\mathsf{OptE}) $, $ (\mathsf{S}) $, $ (\mathsf{New}) $, and $ (\mathsf{Opt}) $.
	For these rules we have to check that the kind of values is respected and that a substitution of a value in the process does not conflict with the session environment required by this rule.
	We avoid the substitution of bound names explicitly using alpha-conversion.
	\begin{description}
		\item[Case Rule~$ (\mathsf{N}) $:]
			In this case $ P = \PEnd $ and $ \Delta = \emptyset $.
			By Rule~$ (\mathsf{N}) $, we have $ \Gamma \vdash \PEnd \triangleright \emptyset $.
			Hence $ \Gamma \vdash P\!\Set[]{ \Subst{\Args[v]}{\Args[y]} } \triangleright \Delta $.
		\item[Case Rule~$ (\mathsf{OptE}) $:]
			In this case $ P = \POptEnd{\Role}{\tilde{\Args[v]}'} $ and $ \Delta = \Typed{\Role}{\OV{\tilde{\Sort}'}} $ and we have $ \vdash \Typed{\tilde{\Args[v]}'}{\tilde{\Sort}'} $.
			Hence, if $ \tilde{\Args[v]}'_i = \Args[y] $, then $ \tilde{\Sort}'_i = \Sort $ and thus the kinds of $ \tilde{\Args[v]}'\!\Set[]{ \Subst{\Args[v]}{\Args[y]} } $ and $ \tilde{\Args[y]}' $ coincide.
			Thus $ \vdash \Typed{\tilde{\Args[v]}'}{\tilde{\Sort}'} $ implies $ \vdash \Typed{\tilde{\Args[v]}'\!\Set[]{ \Subst{\Args[v]}{\Args[y]} }}{\tilde{\Sort}'} $.
			With Rule~$ (\mathsf{OptE}) $ we have $ \Gamma \vdash \POptEnd{\Role}{\tilde{\Args[v]}'} \!\Set[]{ \Subst{\Args[v]}{\Args[y]} } \triangleright \Typed{\Role}{\OV{\tilde{\Sort}'}} $.
			Hence $ \Gamma \vdash P\!\Set[]{ \Subst{\Args[v]}{\Args[y]} } \triangleright \Delta $.
		\item[Case Rule~$ (\mathsf{I}) $:]
			In this case $ P = \PInp{\Chan}{\Args}{P'} $ and we have $ \Gamma \vdash P' \triangleright \Delta, \Typed{\AT{\Args}{\Role}}{T} $ and $ \GetType{\Chan} = \AT{T}{\Role} $.
			Using alpha-conversion before Rule~$ (\mathsf{I}) $ we can ensure that $ \Args \notin \Set[]{ \Args[v], \Args[y] } $.
			Because of $ \vdash \Typed{\Args[y]}{\Sort} $ and $ \vdash \Typed{\Args[v]}{\Sort} $, we have $ \Chan \notin \Set[]{ \Args[v], \Args[y] } $.
			Then, by the induction hypothesis, $ \Gamma \vdash P' \triangleright \Delta, \Typed{\AT{\Args}{\Role}}{T} $ implies $ \Gamma \vdash P'\!\Set[]{ \Subst{\Args[v]}{\Args[y]} } \triangleright \Delta, \Typed{\AT{\Args}{\Role}}{T} $.
			With Rule~$ (\mathsf{I}) $ and $ \GetType{\Chan} = \AT{T}{\Role} $ we have $ \Gamma \vdash \PInp{\Chan}{\Args}{\left( P'\!\Set[]{ \Subst{\Args[v]}{\Args[y]} } \right)} \triangleright \Delta $.
			Since $ \Chan, \Args \notin \Set[]{ \Args[v], \Args[y] } $, then $ \Gamma \vdash P\!\Set[]{ \Subst{\Args[v]}{\Args[y]} } \triangleright \Delta $.
		\item[Case Rule~$ (\mathsf{O}) $:]
			In this case $ P = \POut{\Chan}{\Chan[s]}{P'} $ and $ \Delta = \Delta', \Typed{\ATE{\Chan[s]}{\Role}}{T} $ and we have $ \Gamma \vdash P' \triangleright \Delta' $ and $ \GetType{\Chan} = \AT{T}{\Role} $.
			Because of $ \vdash \Typed{\Args[y]}{\Sort} $ and $ \vdash \Typed{\Args[v]}{\Sort} $, we have $ \Chan, \Chan[s] \notin \Set[]{ \Args[v], \Args[y] } $.
			By the induction hypothesis, $ \Gamma \vdash P' \triangleright \Delta' $ implies $ \Gamma \vdash P'\!\Set[]{ \Subst{\Args[v]}{\Args[y]} } \triangleright \Delta' $.
			With Rule~$ (\mathsf{O}) $ and $ \GetType{\Chan} = \AT{T}{\Role} $ we have $ \Gamma \vdash \POut{\Chan}{\Chan[s]}{\left( P'\!\Set[]{ \Subst{\Args[v]}{\Args[y]} } \right)} \triangleright \Delta', \Typed{\ATE{\Chan[s]}{\Role}}{T} $.
			Since $ \Chan, \Chan[s] \notin \Set[]{ \Args[v], \Args[y] } $, we have $ \Gamma \vdash P\!\Set[]{ \Subst{\Args[v]}{\Args[y]} } \triangleright \Delta $.
		\item[Case Rule~$ (\mathsf{C}) $:]
			In this case we have $ P = \PGet{\Chan[k]}{\Role_1}{\Role_2}{_{i \in \indexSet} \Set{ \PLab{\Labe_i}{\tilde{\Args[y]}'_i}{P_i} }} $, $ \Delta = \Delta', \Typed{\AT{\Chan[k]}{\Role_2}}{\LTGet{\Role_1}{_{i \in \indexSet{}} \Set{ \LTLab{\Labe_i}{\Typed{\tilde{\Args}_i}{\tilde{\Sort}'_i}}{T_i} }}} $ and we have $ \Gamma \vdash P_i \triangleright \Delta', \Typed{\AT{\Chan[k]}{\Role_2}}{T_i} $ and $ \vdash \Typed{\tilde{\Args[y]}'_i}{\tilde{\Sort}'_i} $ for all $ i \in \indexSet $.
			Using alpha-conversion before Rule~$ (\mathsf{C}) $ we can ensure that $ \Args[v], \Args[y] \notin \tilde{\Args[y]}'_i $ and $ \Args[v], \Args[y] \notin \tilde{\Args}_i $ for all $ i \in \indexSet $.
			Because of $ \vdash \Typed{\Args[y]}{\Sort} $ and $ \vdash \Typed{\Args[v]}{\Sort} $, we have $ \Chan[k] \notin \Set[]{ \Args[v], \Args[y] } $.
			By the induction hypothesis, $ \Gamma \vdash P_i \triangleright \Delta', \Typed{\AT{\Chan[k]}{\Role_2}}{T_i} $ implies $ \Gamma \vdash P_i\!\Set[]{ \Subst{\Args[v]}{\Args[y]} } \triangleright \Delta', \Typed{\AT{\Chan[k]}{\Role_2}}{T_i} $ for all $ i \in \indexSet $.
			With Rule~$ (\mathsf{C}) $ and $ \vdash \Typed{\tilde{\Args[y]}'_i}{\tilde{\Sort}'_i} $ for all $ i \in \indexSet $ we have $ \Gamma \vdash \PGet{\Chan[k]}{\Role_1}{\Role_2}{_{i \in \indexSet} \Set{ \PLab{\Labe_i}{\tilde{\Args[y]}'_i}{\left( P_i\!\Set[]{ \Subst{\Chan[s]}{\Args} } \right)} }} \triangleright \Delta', \Typed{\AT{\Chan[k]}{\Role_2}}{\LTGet{\Role_1}{_{i \in \indexSet{}} \Set{ \LTLab{\Labe_i}{\Typed{\tilde{\Args}_i}{\tilde{\Sort}'_i}}{T_i} }}} $.
			Since $ \Chan[k] \notin \Set[]{ \Args[v], \Args[y] } $ and $ \Args[v], \Args[y] \notin \tilde{\Args[y]}'_i $ for all $ i \in \indexSet $, we have $ \Gamma \vdash P\!\Set[]{ \Subst{\Args[v]}{\Args[y]} } \triangleright \Delta $.
		\item[Case Rule~$ (\mathsf{S}) $:]
			In this case we have $ P = \PSend{\Chan[k]}{\Role_1}{\Role_2}{\Labe_j}{\tilde{\Args[v]}'}{P'} $, $ \Delta = \Delta', \Typed{\AT{\Chan[k]}{\Role_1}}{\LTSend{\Role_2}{_{i \in \indexSet} \Set{ \LTLab{\Labe_i}{\Typed{\tilde{\Args}_i}{\tilde{\Sort}'_i}}{T_i} }}} $ and we have $ \Gamma \vdash P' \triangleright \Delta', \Typed{\AT{\Chan[k]}{\Role_1}}{T_j} $ and $ \vdash \Typed{\tilde{\Args[v]}'}{\tilde{\Sort}'_j} $.
			Hence, if $ \tilde{\Args[v]}'_i = \Args[y] $, then $ \tilde{\Sort}'_i = \Sort $ and thus the kinds of $ \tilde{\Args[v]}'\!\Set[]{ \Subst{\Args[v]}{\Args[y]} } $ and $ \tilde{\Args} $ coincide.
			Thus $ \vdash \Typed{\tilde{\Args[v]}'}{\tilde{\Sort}_j'} $ implies $ \vdash \Typed{\tilde{\Args[v]}'\!\Set[]{ \Subst{\Args[v]}{\Args[y]} }}{\tilde{\Sort}_j'} $.
			Because of $ \vdash \Typed{\Args[y]}{\Sort} $ and $ \vdash \Typed{\Args[v]}{\Sort} $, we have $ \Chan[k] \notin \Set[]{ \Args[v], \Args[y] } $.
			By the induction hypothesis, $ \Gamma \vdash P' \triangleright \Delta', \Typed{\AT{\Chan[k]}{\Role_1}}{T_j} $ implies $ \Gamma \vdash P'\!\Set[]{ \Subst{\Args[v]}{\Args[y]} } \triangleright \Delta', \Typed{\AT{\Chan[k]}{\Role_1}}{T_j} $.
			With Rule~$ (\mathsf{S}) $ and $ \vdash \Typed{\tilde{\Args[v]}'\!\Set[]{ \Subst{\Args[v]}{\Args[y]} }}{\tilde{\Sort}_j'} $ we have $ \Gamma \vdash \PSend{\Chan[k]}{\Role_1}{\Role_2}{\Labe_j}{\tilde{\Args[v]}'\!\Set[]{ \Subst{\Args[v]}{\Args[y]} }}{\left( P'\!\Set[]{ \Subst{\Args[v]}{\Args[y]} } \right)} \triangleright \Delta', \Typed{\AT{\Chan[k]}{\Role_1}}{\LTSend{\Role_2}{_{i \in \indexSet} \Set{ \LTLab{\Labe_i}{\Typed{\tilde{\Args}_i}{\tilde{\Sort}'_i}}{T_i} }}} $.
			Since $ \Chan[k] \notin \Set[]{ \Args[v], \Args[y] } $, we have $ \Gamma \vdash P\!\Set[]{ \Subst{\Args[v]}{\Args[y]} } \triangleright \Delta $.
		\item[Case Rule~$ (\mathsf{R}) $:]
			In this case $ P = \PRes{\Args}{P'} $ and we have $ \Gamma, \Typed{\Args}{\AT{T}{\Role}} \vdash P' \triangleright \Delta $.
			Using alpha-conversion before Rule~$ (\mathsf{R}) $ we can ensure that $ \Args \notin \Set[]{ \Args[v], \Args[y] } $.
			By the induction hypothesis, $ \Gamma, \Typed{\Args}{\AT{T}{\Role}} \vdash P' \triangleright \Delta $ implies $ \Gamma, \Typed{\Args}{\AT{T}{\Role}} \vdash P'\!\Set[]{ \Subst{\Args[v]}{\Args[y]} } \triangleright \Delta $.
			With Rule~$ (\mathsf{R}) $ we have $ \Gamma \vdash \PRes{\Args}{\left( P'\!\Set[]{ \Subst{\Args[v]}{\Args[y]} } \right)} \triangleright \Delta $.
			Since $ \Args \notin \Set[]{ \Args[v], \Args[y] } $, we have $ \Gamma \vdash P\!\Set[]{ \Subst{\Args[v]}{\Args[y]} } \triangleright \Delta $.
		\item[Case Rule~$ (\mathsf{P}) $:]
			In this case we have $ P = \PReq{\Chan[s]}{\Role_1}{\Role_2}{\Role_3}{\Chan[k]}{P'} $ and $ \Delta = \Delta', \Typed{\AT{\Chan[s]}{\Role_1}}{\LTReq{\Prot}{\Role_3}{\tilde{\Args[v]}'}{\Role_2}{T_1}}, \Typed{\ATI{\Chan[k]}{\Role_3}}{T_3} $, $ \Gamma \vdash P' \triangleright \Delta', \Typed{\AT{\Chan[s]}{\Role_1}}{T_1} $, $ \GetType{\Prot} = \TypeOfProt{\tilde{\Role}_4}{\tilde{\Args[y]}'}{\tilde{\Role}_5}{G} $, and $ \ProjS{G\!\Set[]{ \Subst{\tilde{\Args[v]}'}{\tilde{\Args[y]}'} }}{}{\Role_3} = T_3 $.
			Because of $ \vdash \Typed{\Args[y]}{\Sort} $ and $ \vdash \Typed{\Args[v]}{\Sort} $, we have $ \Chan[s], \Chan[k] \notin \Set[]{ \Args[v], \Args[y] } $.
			By the induction hypothesis, $ \Gamma \vdash P' \triangleright \Delta', \Typed{\AT{\Chan[s]}{\Role_1}}{T_1} $ implies $ \Gamma \vdash P'\!\Set[]{ \Subst{\Args[v]}{\Args[y]} } \triangleright \Delta', \Typed{\AT{\Chan[s]}{\Role_1}}{T_1} $.
			With Rule~$ (\mathsf{P}) $, $ \GetType{\Prot} = \TypeOfProt{\tilde{\Role}_4}{\tilde{\Args[y]}'}{\tilde{\Role}_5}{G} $, and $ \ProjS{G\!\Set[]{ \Subst{\tilde{\Args[v]}'}{\tilde{\Args[y]}'} }}{}{\Role_3} = T_3 $ we have $ \Gamma \vdash \PReq{\Chan[s]}{\Role_1}{\Role_2}{\Role_3}{\Chan[k]}{\left( P'\!\Set[]{ \Subst{\Args[v]}{\Args[y]} } \right)} \triangleright \Delta', \Typed{\AT{\Chan[s]}{\Role_1}}{\LTReq{\Prot}{\Role_3}{\tilde{\Args[v]}'}{\Role_2}{T_1}}, \Typed{\ATI{\Chan[k]}{\Role_3}}{T_3} $.
			Since $ \Chan[s], \Chan[k] \notin \Set[]{ \Args[v], \Args[y] } $, we have $ \Gamma \vdash P\!\Set[]{ \Subst{\Args[v]}{\Args[y]} } \triangleright \Delta $.
		\item[Case Rule~$ (\mathsf{J}) $:]
			In this case we have $ P = \PEnt{\Chan[s]}{\Role_1}{\Role_2}{\Role_3}{\Args}{P'} $ and $ \Delta = \Delta', \Typed{\AT{\Chan[s]}{\Role_2}}{\LTEnt{\Prot}{\Role_3}{\tilde{\Args[v]}'}{\Role_1}{T_2}} $, $ \Gamma \vdash P' \triangleright \Delta', \Typed{\AT{\Chan[s]}{\Role_2}}{T_2}, \Typed{\AT{\Args}{\Role_3}}{T_3} $, $ \GetType{\Prot} = \TypeOfProt{\tilde{\Role}_4}{\tilde{\Args[y]}'}{\tilde{\Role}_5}{G} $, and $ \ProjS{G\!\Set[]{ \Subst{\tilde{\Args[v]}'}{\tilde{\Args[y]}'} }}{}{\Role_3} = T_3 $.
			Using alpha-conversion before Rule~$ (\mathsf{J}) $ we can ensure that $ \Args \notin \Set[]{ \Args[v], \Args[y] } $.
			Because of $ \vdash \Typed{\Args[y]}{\Sort} $ and $ \vdash \Typed{\Args[v]}{\Sort} $, we have $ \Chan[s] \notin \Set[]{ \Args[v], \Args[y] } $.
			By the induction hypothesis, $ \Gamma \vdash P' \triangleright \Delta', \Typed{\AT{\Chan[s]}{\Role_2}}{T_2}, \Typed{\AT{\Args}{\Role_3}}{T_3} $ implies $ \Gamma \vdash P'\!\Set[]{ \Subst{\Args[v]}{\Args[y]} } \triangleright \Delta', \Typed{\AT{\Chan[s]}{\Role_2}}{T_2}, \Typed{\AT{\Args}{\Role_3}}{T_3} $.
			With Rule~$ (\mathsf{J}) $, $ \GetType{\Prot} = \TypeOfProt{\tilde{\Role}_4}{\tilde{\Args[y]}'}{\tilde{\Role}_5}{G} $, and $ \ProjS{G\!\Set[]{ \Subst{\tilde{\Args[v]}'}{\tilde{\Args[y]}'} }}{}{\Role_3} = T_3 $ we have $ \Gamma \vdash \PEnt{\Chan[s]}{\Role_1}{\Role_2}{\Role_3}{\Args}{\left( P'\!\Set[]{ \Subst{\Args[v]}{\Args[y]} } \right)} \triangleright \Delta', \Typed{\AT{\Chan[s]}{\Role_2}}{\LTEnt{\Prot}{\Role_3}{\tilde{\Args[v]}'}{\Role_1}{T_2}} $.
			Since $ \Chan[s], \Args \notin \Set[]{ \Args[v], \Args[y] } $, we have $ \Gamma \vdash P\!\Set[]{ \Subst{\Args[v]}{\Args[y]} } \triangleright \Delta $.
		\item[Case Rule~$ (\mathsf{New}) $:]
			In this case $ P = \PDecl{\Chan[k]}{\Chan[s]}{\tilde{\Args[v]}'}{\tilde{\Chan}}{\tilde{\Role}'''}{P'} $ and
			\begin{align*}
				\Delta = \Delta', \Typed{\AT{\Chan[s]}{\Role}}{\LTCall{\Prot}{G}{\tilde{\Args[v]}'}{\Typed{\tilde{\Args[y]}'}{\tilde{\Sort}'}}{\tilde{\Role}'''}{T'}}
			\end{align*}
			and we have
			\begin{align*}
				\Gamma \vdash P' \triangleright \; \Delta', \Typed{\AT{\Chan[s]}{\Role}}{T}, \Typed{\ATI{\Chan[k]}{\Role''_1}}{T'_1}, \ldots, \Typed{\ATI{\Chan[k]}{\Role''_n}}{T'_n}, \Typed{\ATE{\Chan[k]}{\Role'''_1}}{T'_{n + 1}}, \ldots, \Typed{\ATE{\Chan[k]}{\Role'''_m}}{T'_{n + m}}
			\end{align*}
			and $ \GetType{\Prot} = \TypeOfProt{\tilde{\Role}''}{\tilde{\Args[y]}'}{\tilde{\Role}'''}{G} $ and $ \forall i \logdot \GetType{\Chan_i} = \AT{T'_{i + n}}{\Role'''_{i + n}} $ and $ \forall i \logdot \ProjS{G\!\Set[]{ \Subst{\tilde{\Args[v]}'}{\tilde{\Args[y]}'} }}{}{\Role''_i} = T'_i $ and $ \forall j \logdot \ProjS{G\!\Set[]{ \Subst{\tilde{\Args[v]}'}{\tilde{\Args[y]}'} }}{}{\Role'''_j} = T'_{j + n} $ and $ \vdash \Typed{\tilde{\Args[v]}'}{\tilde{\Sort}'} $ and $ \GetType{\Chan[k]} = {\Prot\!\Set[]{ \Subst{\tilde{\Args[v]}'}{\tilde{\Args[y]}'} }} $.
			Hence, if $ \tilde{\Args[v]}'_i = \Args[y] $, then $ \tilde{\Sort}'_i = \Sort $ and thus the kinds of $ \tilde{\Args[v]}'\!\Set[]{ \Subst{\Args[v]}{\Args[y]} } $ and $ \tilde{\Args[y]}' $ coincide.
			Because of $ \vdash \Typed{\Args[y]}{\Sort} $ and $ \vdash \Typed{\Args[v]}{\Sort} $, we have $ \Chan[s], \Chan[k] \notin \Set[]{ \Args[v], \Args[y] } $ and $ \Args[v], \Args[y] \notin \tilde{\Chan} $.
			By the induction hypothesis,
			\begin{align*}
				\Gamma \vdash P' \triangleright \Delta', \Typed{\AT{\Chan[s]}{\Role}}{T}, \Typed{\ATI{\Chan[k]}{\Role''_1}}{T'_1}, \ldots, \Typed{\ATI{\Chan[k]}{\Role''_n}}{T'_n}, \Typed{\ATE{\Chan[k]}{\Role'''_1}}{T'_{n + 1}}, \ldots, \Typed{\ATE{\Chan[k]}{\Role'''_m}}{T'_{n + m}}
			\end{align*}
			implies
			\begin{align*}
				\Gamma \vdash P'\!\Set[]{ \Subst{\Args[v]}{\Args[y]} } \triangleright \; \Delta', \Typed{\AT{\Chan[s]}{\Role}}{T}, \Typed{\ATI{\Chan[k]}{\Role''_1}}{T'_1}, \ldots, \Typed{\ATI{\Chan[k]}{\Role''_n}}{T'_n}, \Typed{\ATE{\Chan[k]}{\Role'''_1}}{T'_{n + 1}}, \ldots, \Typed{\ATE{\Chan[k]}{\Role'''_m}}{T'_{n + m}}
			\end{align*}
			With Rule~$ (\mathsf{New}) $ and $ \GetType{\Prot} = \TypeOfProt{\tilde{\Role}''}{\tilde{\Args[y]}'}{\tilde{\Role}'''}{G} $ and $ \forall i \logdot \GetType{\Chan_i} = \AT{T'_{i + n}}{\Role'''_{i + n}} $ and $ \forall i \logdot \ProjS{G\!\Set[]{ \Subst{\tilde{\Args[v]}'}{\tilde{\Args[y]}'} }}{}{\Role''_i} = T'_i $ and $ \forall j \logdot \ProjS{G\!\Set[]{ \Subst{\tilde{\Args[v]}'}{\tilde{\Args[y]}'} }}{}{\Role'''_j} = T'_{j + n} $ and $ \vdash \Typed{\tilde{\Args[v]}'}{\tilde{\Sort}'} $ and $ \GetType{\Chan[k]} = {\Prot\!\Set[]{ \Subst{\tilde{\Args[v]}'}{\tilde{\Args[y]}'} }} $ we have
			\begin{align*}
				\Gamma \vdash \PDecl{\Chan[k]}{\Chan[s]}{\tilde{\Args[v]}'}{\tilde{\Chan}}{\tilde{\Role}'''}{\left( P'\!\Set[]{ \Subst{\Args[v]}{\Args[y]} } \right)} \triangleright \Delta', \Typed{\AT{\Chan[s]}{\Role}}{\LTCall{\Prot}{G}{\tilde{\Args[v]}'}{\Typed{\tilde{\Args[y]}'}{\tilde{\Sort}'}}{\tilde{\Role}'''}{T}}
			\end{align*}
			Since $ \Chan[s], \Chan[k] \notin \Set[]{ \Args[v], \Args[y] } $ and $ \Args[v], \Args[y] \notin \tilde{\Chan} $, we have $ \Gamma \vdash P\!\Set[]{ \Subst{\Args[v]}{\Args[y]} } \triangleright \Delta $.
		\item[Case Rule~$ (\mathsf{S1}) $:]
			In this case $ P = \PChoi{P_1}{P_2} $ and $ \Delta = \Delta', \Typed{\AT{\Chan[s]}{\Role}}{T_1 \oplus T_2} $ and we have $ \Gamma \vdash P_1 \triangleright \Delta', \Typed{\AT{\Chan[s]}{\Role}}{T_1} $ and $ \Gamma \vdash P_2 \triangleright \Delta', \Typed{\AT{\Chan[s]}{\Role}}{T_2} $.
			Because of $ \vdash \Typed{\Args[y]}{\Sort} $ and $ \vdash \Typed{\Args[v]}{\Sort} $, we have $ \Chan[s] \notin \Set[]{ \Args[v], \Args[y] } $.
			By the induction hypothesis, $ \Gamma \vdash P_1 \triangleright \Delta', \Typed{\AT{\Chan[s]}{\Role}}{T_1} $ and $ \Gamma \vdash P_2 \triangleright \Delta', \Typed{\AT{\Chan[s]}{\Role}}{T_2} $ imply $ \Gamma \vdash P_1\!\Set[]{ \Subst{\Args[v]}{\Args[y]} } \triangleright \Delta', \Typed{\AT{\Chan[s]}{\Role}}{T_1} $ and $ \Gamma \vdash P_2\!\Set[]{ \Subst{\Args[v]}{\Args[y]} } \triangleright \Delta', \Typed{\AT{\Chan[s]}{\Role}}{T_2} $.
			With Rule~$ (\mathsf{S1}) $ we have $ \Gamma \vdash \PChoi{\left( P_1\!\Set[]{ \Subst{\Args[v]}{\Args[y]} } \right)}{\left( P_2\!\Set[]{ \Subst{\Args[v]}{\Args[y]} } \right)} \triangleright \Delta', \Typed{\AT{\Chan[s]}{\Role}}{\LTChoi{T_1}{T_2}} $.
			Hence $ \Gamma \vdash P\!\Set[]{ \Subst{\Args[v]}{\Args[y]} } \triangleright \Delta $.
		\item[Case Rule~$ (\mathsf{S2}) $:]
			In this case $ \Delta = \Delta', \Typed{\AT{\Chan[s]}{\Role}}{T_1 \oplus T_2} $ and we have $ \Gamma \vdash P \triangleright \Delta', \Typed{\AT{\Chan[s]}{\Role}}{T_i} $ with $ i \in \Set[]{ 1, 2 } $.
			Because of $ \vdash \Typed{\Args[y]}{\Sort} $ and $ \vdash \Typed{\Args[v]}{\Sort} $, we have $ \Chan[s] \notin \Set[]{ \Args[v], \Args[y] } $.
			By the induction hypothesis, $ \Gamma \vdash P \triangleright \Delta', \Typed{\AT{\Chan[s]}{\Role}}{T_i} $ implies $ \Gamma \vdash P\!\Set[]{ \Subst{\Args[v]}{\Args[y]} } \triangleright \Delta', \Typed{\AT{\Chan[s]}{\Role}}{T_i} $.
			With Rule~$ (\mathsf{S2}) $ we have $ \Gamma \vdash P\!\Set[]{ \Subst{\Args[v]}{\Args[y]} } \triangleright \Delta', \Typed{\AT{\Chan[s]}{\Role}}{\LTChoi{T_1}{T_2}} $.
			Hence $ \Gamma \vdash P\!\Set[]{ \Subst{\Args[v]}{\Args[y]} } \triangleright \Delta $.
		\item[Case Rule~$ (\mathsf{Pa}) $:]
			In this case $ P = \PPar{P_1}{P_2} $ and $ \Delta = \Delta_1 \otimes \Delta_2 $ and we have $ \Gamma \vdash P_1 \triangleright \Delta_1 $ and $ \Gamma \vdash P_2 \triangleright \Delta_2 $.
			By the induction hypothesis, $ \Gamma \vdash P_1 \triangleright \Delta_1 $ and $ \Gamma \vdash P_2 \triangleright \Delta_2 $ imply $ \Gamma \vdash P_1\!\Set[]{ \Subst{\Args[v]}{\Args[y]} } \triangleright \Delta_1 $ and $ \Gamma \vdash P_2\!\Set[]{ \Subst{\Args[v]}{\Args[y]} } \triangleright \Delta_2 $.
			With Rule~$ (\mathsf{Pa}) $ we have $ \Gamma \vdash \PPar{\left( P_1\!\Set[]{ \Subst{\Args[v]}{\Args[y]} } \right)}{\left( P_2\!\Set[]{ \Subst{\Args[v]}{\Args[y]} } \right)} \triangleright \Delta_1 \otimes \Delta_2 $.
			Hence $ \Gamma \vdash P\!\Set[]{ \Subst{\Args[v]}{\Args[y]} } \triangleright \Delta $.
		\item[Case Rule~$ (\mathsf{Opt}) $:]
			In this case $ P = \POpt{\Role_1}{\tilde{\Role}}{P_1}{\tilde{\Args}}{\tilde{\Args[v]}'}{P_2} $ and $ \Delta = \Delta_1 \otimes \Delta_2, \Typed{\AT{\Chan[s]}{\Role_1}}{\LTOpt{\tilde{\Role}}{T_1}{\Typed{\tilde{\Args[y]}'}{\tilde{\Sort}'}}{T_2}} $ and we have $ \Gamma \vdash P_1 \triangleright \Delta_1, \Typed{\AT{\Chan[s]}{\Role_1}}{T_1}, \Typed{\Role_1}{\OV{\tilde{\Sort}'}} $, $ \nexists \Role', \Sort[K] \logdot \Typed{\Role'}{\OV{\tilde{\Sort[K]}}} \in \Delta_1 $, $ \Gamma \vdash P_2 \triangleright \Delta_2, \Typed{\AT{\Chan[s]}{\Role_1}}{T_2} $, $ \vdash \Typed{\tilde{\Args}}{\tilde{\Sort}'} $, and $ \vdash \Typed{\tilde{\Args[v]}'}{\tilde{\Sort}'} $.
			Hence, if $ \tilde{\Args[v]}'_i = \Args[y] $, then $ \tilde{\Sort}'_i = \Sort $ and thus the kinds of $ \tilde{\Args[v]}'\!\Set[]{ \Subst{\Args[v]}{\Args[y]} } $ and $ \tilde{\Args[y]}' $ coincide.
			Thus $ \vdash \Typed{\tilde{\Args[v]}'}{\tilde{\Sort}'} $ implies $ \vdash \Typed{\tilde{\Args[v]}'\!\Set[]{ \Subst{\Args[v]}{\Args[y]} }}{\tilde{\Sort}'} $.
			Using alpha-conversion before Rule~$ (\mathsf{Opt}) $ we can ensure that $ \Args[v], \Args[y] \notin \tilde{\Args} $.
			Because of $ \vdash \Typed{\Args[y]}{\Sort} $ and $ \vdash \Typed{\Args[v]}{\Sort} $, we have $ \Chan[s] \notin \Set[]{ \Args[v], \Args[y] } $.
			By the induction hypothesis, $ \Gamma \vdash P_1 \triangleright \Delta_1, \Typed{\AT{\Chan[s]}{\Role_1}}{T_1}, \Typed{\Role_1}{\OV{\tilde{\Sort}'}} $ and $ \Gamma \vdash P_2 \triangleright \Delta_2, \Typed{\AT{\Chan[s]}{\Role_1}}{T_2} $ imply $ \Gamma \vdash P_1\!\Set[]{ \Subst{\Args[v]}{\Args[y]} } \triangleright \Delta_1, \Typed{\AT{\Chan[s]}{\Role_1}}{T_1}, \Typed{\Role_1}{\OV{\tilde{\Sort}'}} $ and $ \Gamma \vdash P_2\!\Set[]{ \Subst{\Args[v]}{\Args[y]} } \triangleright \Delta_2, \Typed{\AT{\Chan[s]}{\Role_1}}{T_2} $.
			With Rule~$ (\mathsf{Opt}) $, $ \nexists \Role', \Sort[K] \logdot \Typed{\Role'}{\OV{\tilde{\Sort[K]}}} \in \Delta_1 $, $ \vdash \Typed{\tilde{\Args}}{\tilde{\Sort}'} $, and $ \vdash \Typed{\tilde{\Args[v]}'\!\Set[]{ \Subst{\Args[v]}{\Args[y]} }}{\tilde{\Sort}'} $ we have $ \Gamma \vdash \POpt{\Role_1}{\tilde{\Role}}{P_1\!\Set[]{ \Subst{\Args[v]}{\Args[y]} }}{\tilde{\Args}}{\tilde{\Args[v]}'\!\Set[]{ \Subst{\Args[v]}{\Args[y]} }}{\left( P_2\!\Set[]{ \Subst{\Args[v]}{\Args[y]} } \right)} \triangleright \Delta_1 \otimes \Delta_2, \Typed{\AT{\Chan[s]}{\Role_1}}{\LTOpt{\tilde{\Role}}{T_1}{\Typed{\tilde{\Args[y]}'}{\tilde{\Sort}'}}{T_2}} $.
			Since $ \Args[v], \Args[y] \notin \tilde{\Args} $, we have $ \Gamma \vdash P\!\Set[]{ \Subst{\Args[v]}{\Args[y]} } \triangleright \Delta $.
	\end{description}
	The proof for the session types with optional blocks but without sub-sessions is similar but omits the cases for the Rules~$ (\mathsf{P}) $, $ (\mathsf{J}) $, and $ (\mathsf{New}) $. The remaining cases do not rely on the Rules~$ (\mathsf{P}) $, $ (\mathsf{J}) $, or $ (\mathsf{New}) $.
\end{proof}

The next lemma deals with evaluation contexts in typing judgements (compare to \cite{Demangeon15}). If a process $ P $ is well-typed within an evaluation context then
\begin{inparaenum}[(1)]
	\item the process $ P $ is well-typed itself and
	\item any other process that is well-typed \wrt to the same global environment as $ P $ is also well-typed within the evaluation context.
\end{inparaenum}

\begin{lemma}
	\label{lem:typeEC}
	For both type systems: If $ \Gamma \vdash \AEC{P} \triangleright \Delta $ then:
	\begin{enumerate}
		\item There exist $ \Delta_1, \Delta' $, and $ \Gamma \subseteq \Gamma' $ such that $ \Gamma' \vdash P \triangleright \Delta_1 $ and $ \Delta = \Delta' \otimes \Delta_1 $.
		\item For all $ P_2, \Delta_2 $ such that $ \Gamma' \vdash P_2 \triangleright \Delta_2 $ and $ \left( \forall \Role, \tilde{\Sort} \logdot \Typed{\Role}{\OV{\tilde{\Sort}}} \in \Delta_1 \text{ iff } \Typed{\Role}{\OV{\tilde{\Sort}}} \in \Delta_2 \right) $, we have $ \Gamma \vdash \AEC{P_2} \triangleright \Delta_2 \otimes \Delta' $.
	\end{enumerate}
\end{lemma}

\begin{proof}
	By the typing rules of Figure~\ref{fig:typingRules}, the derivation of $ \Gamma \vdash \AEC{P} \triangleright \Delta $ is a tree containing a derivation of $ \Gamma' \vdash P \triangleright \Delta_1 $ for some $ \Gamma', \Delta_1 $ as subtree.
	Since no rule removes elements of the global environment (but Rule~$ (\mathsf{R}) $ might add elements), $ \Gamma \subseteq \Gamma' $.
	By the definition of evaluation contexts, the only rules that can be used in the part of the derivation of $ \Gamma \vdash \AEC{P} \triangleright \Delta $ that is below the subtree $ \Gamma' \vdash P \triangleright \Delta_1 $ are the Rules~$ (\mathsf{R}) $, $ (\mathsf{Pa}) $, and $ (\mathsf{Opt}) $.
	Rule~$ (\mathsf{R}) $ does not change the session environment.
	The Rules~$ (\mathsf{Pa}) $ and $ (\mathsf{Opt}) $ split the session environment of the original judgement into two session environments using the operator $ \otimes $ such that each of the two subtrees generated by these rules obtains one part of the session environment.
	Hence, moving downwards from $ \Gamma' \vdash P \triangleright \Delta_1 $ in the derivation of $ \Gamma \vdash \AEC{P} \triangleright \Delta $, we can collect all session environments that were split from $ \Delta $ and combine them with $ \otimes $---possibly adding $ \emptyset $---to obtain $ \Delta' $ such that $ \Delta = \Delta' \otimes \Delta_1 $. Rule~$ (\mathsf{Opt}) $ additionally adds an assignment $ \Typed{\Role}{\OV{\tilde{\Sort}}} $ to cover the type of the return values to the part of the session environment that is used in $ \EC $ for the position that contains the hole. In this case either $ P $ contains $ \POptEnd{\Role}{\tilde{\Args[v]}} $ and $ \Typed{\Role}{\OV{\tilde{\Sort}}} \in \Delta_1 $ for some $ \Role $ and $ \Typed{\tilde{\Args[v]}}{\tilde{\Sort}} $, or the $ \Typed{\Role}{\OV{\tilde{\Sort}}} $ was split from $ \Delta $ into the part $ \Delta' $. The typing rules ensure that $ \Delta_1 $ can contain at most one assignment of the form $ \Typed{\Role}{\OV{\tilde{\Sort}}} $.

	Moreover, if we have $ \Gamma' \vdash P_2 \triangleright \Delta_2 $ and $ \left( \forall \Role, \tilde{\Sort} \logdot \Typed{\Role}{\OV{\tilde{\Sort}}} \in \Delta_1 \text{ iff } \Typed{\Role}{\OV{\tilde{\Sort}}} \in \Delta_2 \right) $, we can replace the subtree for $ \Gamma' \vdash P \triangleright \Delta_1 $ in the derivation of $ \Gamma \vdash \AEC{P} \triangleright \Delta $---while substituting all occurrences of $ \Delta_1 $ by $ \Delta_2 $ below the subtree---and obtain a derivation for $ \Gamma \vdash \AEC{P_2} \triangleright \Delta_2 \otimes \Delta' $. Here, the second condition ensures, that the type of the return values is checked for $ P_1 $ if and only if it is checked for $ P_2 $. This ensures that this property holds in case the hole of the context covers the enclosed part of an optional block.
	
	Since both type systems use evaluation contexts and the Rules~$ (\mathsf{R}) $, $ (\mathsf{Pa}) $, and $ (\mathsf{Opt}) $, to type them, in the same way, both type systems fulfil this property.
\end{proof}

Note that, since the contexts $ \ECR $, $ \ECO $, and $ \ECP $ are strict sub-contexts of evaluation contexts $ \EC $, \ie each context of one of the former kinds is also an evaluation context, the above lemma holds for all four kinds of contexts.

\subsection{Subject Reduction}

\emph{Subject reduction} is a basic property of each type system. It is this property that allows us to reason statically about terms, by ensuring that whenever a process its well-typed then all its derivatives are well-typed as well.
Hence, for all properties the type system ensures for well-typed terms, it is not necessary to compute executions but only to test for well-typedness of the original term.
We use a strong variant of subject reduction that additionally involves the condition $ \Delta \mapsto \Delta' $, in order to capture how the local types evolve alongside the reduction of processes.
More precisely, $ \Delta' $ is the session environment we obtain for the derivative $ P' $ of a process $ P $ with respect to a step $ P \longmapsto P' $.
Therefore the effect of reductions on processes on the corresponding local types is captured within the relation $ \mapsto $.

\begin{figure*}[tp]
	\[ \begin{array}{c}
		(\mathsf{comS}') \dfrac{j \in \indexSet}{\Delta, \Typed{\AT{\Chan[k]}{\Role_1}}{\LTSend{\Role_2}{_{i \in \indexSet} \Set{ \LTLab{\Labe_i}{\Typed{\tilde{\Args}_i}{\tilde{\Sort}_i}}{T_i} }}}, \Typed{\AT{\Chan[k]}{\Role_2}}{\LTGet{\Role_1}{_{i \in \indexSet{}} \Set{ \LTLab{\Labe_i}{\Typed{\tilde{\Args}_i'}{\tilde{\Sort}_i}}{T_i'} }}} \mapsto \Delta, \Typed{\AT{\Chan[k]}{\Role_1}}{T_j}, \Typed{\AT{\Chan[k]}{\Role_2}}{T_j'}}
		\vspace*{0.75em}\\
		(\mathsf{choice}') \dfrac{\Delta, \Typed{\AT{\Chan[s]}{\Role}}{T_i} \mapsto \Delta', \Typed{\AT{\Chan[s]}{\Role}}{T_i'} \quad i \in \Set[]{1, 2}}{\Delta, \Typed{\AT{\Chan[s]}{\Role}}{T_1 \oplus T_2} \mapsto \Delta', \Typed{\AT{\Chan[s]}{\Role}}{T_i'}}
		\hspace{2em}
		(\mathsf{comC}') \dfrac{}{\Delta, \Typed{\ATE{\Chan[s]}{\Role}}{T} \mapsto \Delta, \Typed{\AT{\Chan[s]}{\Role}}{T}}
		\vspace*{0.75em}\\
		(\mathsf{par}) \dfrac{\Delta_1 \mapsto \Delta_1'}{\Delta_1 \otimes \Delta_2 \mapsto \Delta_1' \otimes \Delta_2}
		\hspace{2em}
		(\mathsf{subs}') \dfrac{\forall i \logdot \ProjS{G\!\Set[]{ \Subst{\tilde{\Args[v]}}{\tilde{\Args[y]}} }}{}{\Role_i} = T'_i \quad \forall j \logdot \ProjS{G\!\Set[]{ \Subst{\tilde{\Args[v]}}{\tilde{\Args[y]}} }}{}{\Role'_j} = T'_{j + n}}{
			\begin{matrix}
				\Delta, \Typed{\AT{\Chan[s]}{\Role''}}{\LTCall{\Prot}{G}{\tilde{\Args[v]}}{\Typed{\tilde{\Args[y]}}{\tilde{\Sort}}}{\tilde{\Role}'}{T}} \mapsto\\
				\Delta, \Typed{\AT{\Chan[s]}{\Role''}}{T}, \Typed{\ATI{\Chan[k]}{\Role_1}}{T'_1}, \ldots, \Typed{\ATI{\Chan[k]}{\Role_n}}{T'_n}, \Typed{\ATE{\Chan[k]}{\Role'_1}}{T'_{n + 1}}, \ldots, \Typed{\ATE{\Chan[k]}{\Role'_m}}{T'_{n + m}}
			\end{matrix}
		} \vspace*{0.75em}\\
		(\mathsf{join}') \dfrac{}{\Delta, \Typed{\AT{\Chan[s]}{\Role_1}}{\LTReq{\Prot}{\Role_3}{\tilde{\Args[v]}}{\Role_2}{T_1}}, \Typed{\ATI{\Chan[k]}{\Role_3}}{T_3}, \Typed{\AT{\Chan[s]}{\Role_2}}{\LTEnt{\Prot}{\Role_3}{\tilde{\Args[v]}'}{\Role_1}{T_2}} \mapsto \Delta, \Typed{\AT{\Chan[s]}{\Role_1}}{T_1}, \Typed{\AT{\Chan[s]}{\Role_2}}{T_2}, \Typed{\AT{\Chan[k]}{\Role_3}}{T_3}} \vspace*{0.75em}\\
		(\mathsf{opt}') \dfrac{\Delta, \Typed{\AT{\Chan[s]}{\Role_1}}{T_1} \mapsto \Delta', \Typed{\AT{\Chan[s]}{\Role_1}}{T_1'}}{\Delta, \Typed{\AT{\Chan[s]}{\Role_1}}{\LTOpt{\tilde{\Role}}{T_1}{\Typed{\tilde{\Args[y]}}{\tilde{\Sort}}}{T_1''}} \mapsto \Delta', \Typed{\AT{\Chan[s]}{\Role_1}}{\LTOpt{\tilde{\Role}}{T_1'}{\Typed{\tilde{\Args[y]}}{\tilde{\Sort}}}{T_1''}}} \vspace*{0.75em}\\
		(\mathsf{optCom}) \dfrac{\Delta, \Typed{\AT{\Chan[s]}{\Role_1}}{T_1}, \Typed{\AT{\Chan[s]}{\Role_2}}{T_2} \mapsto \Delta', \Typed{\AT{\Chan[s]}{\Role_1}}{T_1'}, \Typed{\AT{\Chan[s]}{\Role_2}}{T_2'}}{
			\begin{matrix}
				\Delta, \Typed{\AT{\Chan[s]}{\Role_1}}{\LTOpt{\tilde{\Role}}{T_1}{\Typed{\tilde{\Args[y]}_1}{\tilde{\Sort}_1}}{T_1''}}, \Typed{\AT{\Chan[s]}{\Role_2}}{\LTOpt{\tilde{\Role}}{T_2}{\Typed{\tilde{\Args[y]}_2}{\tilde{\Sort}_2}}{T_2''}} \mapsto\\
				\Delta', \Typed{\AT{\Chan[s]}{\Role_1}}{\LTOpt{\tilde{\Role}}{T_1'}{\Typed{\tilde{\Args[y]}_1}{\tilde{\Sort}_1}}{T_1''}}, \Typed{\AT{\Chan[s]}{\Role_2}}{\LTOpt{\tilde{\Role}}{T_2'}{\Typed{\tilde{\Args[y]}_2}{\tilde{\Sort}_2}}{T_2''}}
			\end{matrix}
		} \vspace*{0.75em}\\
		(\mathsf{fail}') \dfrac{\exists \Gamma, P, \tilde{\Sort} \logdot \Gamma \vdash P \triangleright \Delta, \Typed{\AT{\Chan[s]}{\Role_1}}{T}, \Typed{\Role_1}{\OV{\tilde{\Sort}}}}{\Delta \otimes \Delta', \Typed{\AT{\Chan[s]}{\Role_1}}{\LTOpt{\tilde{\Role}}{T}{\Typed{\tilde{\Args[y]}}{\tilde{\Sort}}}{T'}} \mapsto \Delta', \Typed{\AT{\Chan[s]}{\Role_1}}{T'}}
		\hspace*{2em}
		(\mathsf{succ}') \dfrac{}{\Delta, \Typed{\AT{\Chan[s]}{\Role_1}}{\LTOpt{\tilde{\Role}}{\LTEnd}{\Typed{\tilde{\Args[y]}}{\tilde{\Sort}}}{T'}} \mapsto \Delta, \Typed{\AT{\Chan[s]}{\Role_1}}{T'}}
	\end{array} \]
	\caption{Reduction Rules for Session Environments}
	\label{fig:sessionTypeReductions}
\end{figure*}

In Figure~\ref{fig:sessionTypeReductions} we derive from the interplay of the reduction rules of processes in Figure~\ref{fig:reductionRules} and the typing rules in Figure~\ref{fig:typingRules} the rules for the evolution of session environments following the reductions of a process.
Note that Rule~(\textsf{succ}') is a special case of Rule~(\textsf{fail}'). The difference between a successful completion and the abortion of an optional block cannot be observed from the session environment.
Also note that the rules of Figure~\ref{fig:sessionTypeReductions} do not replace the rules for process reductions or type checks. They are used here as an auxiliary tool to simplify the argumentation about completion.
Figure~\ref{fig:sessionTypeReductions} contains all rules for both considered type systems. For the smaller type system with optional blocks but without sub-sessions the Rules~$ (\mathsf{subs}') $ and $ (\mathsf{join}') $ are superfluous.
Based on the rules of Figure~\ref{fig:sessionTypeReductions}, we add the condition $ \Delta \mapsto \Delta' $ to the formulation of subject reduction.
Obviously this extension results into a strictly stronger requirement that naturally implies the former statement.
The proof of subject reduction is by induction over the derivation of a single reduction step of the processes, \ie over the reduction rules.
For each reduction rule we have to prove how the proof of well-typedness of the process can be adapted to show that the derivative is also well-typed.
Thus we have to relate the reduction rules and the typing rules.

\begin{theorem}[Subject Reduction]
	\label{thm:subjectReduction}
	$ $\\
	For both type systems:
	If $ \Gamma \vdash P \triangleright \Delta $ and $ P \longmapsto P' $ then there exists $ \Delta' $ such that $ \Gamma \vdash P' \triangleright \Delta' $ and $ \Delta \mapsto \Delta' $.
\end{theorem}

\begin{proof}
	Again we consider the larger type system first.
	Assume $ \Gamma \vdash P \triangleright \Delta $ and $ P \longmapsto P' $.
	We perform an induction over the rules used to derive $ P \longmapsto P' $ with a case analysis over the rules of Figure~\ref{fig:typingRules}.
	\begin{description}
		\item[Cases $ (\mathsf{comS}) $:]
			In this case we have
			\begin{align*}
				P = \AEC{\PPar{\PSend{\Chan[k]}{\Role_1}{\Role_2}{\Labe_j}{\tilde{\Args[v]}}{P^*}}{\PGet{\Chan[k]}{\Role_1}{\Role_2}{_{i \in \indexSet} \Set{ \PLab{\Labe_i}{\tilde{\Args}_i}{P_i} }}}} \quad \text{ and } \quad P' = \AEC{\PPar{P^*}{P_j\!\Set[]{ \Subst{\tilde{\Args[v]}}{\tilde{\Args}_j} }}}
			\end{align*}
			With $ \Gamma \vdash P \triangleright \Delta $ and Lemma~\ref{lem:typeEC}~(1), there exist $ \Delta_P, \Delta_{\EC}, \Gamma' $ such that $ \Gamma \subseteq \Gamma' $, $ \Delta = \Delta_{\EC} \otimes \Delta_P $, and $ \Gamma' \vdash \PPar{\PSend{\Chan[k]}{\Role_1}{\Role_2}{\Labe_j}{\tilde{\Args[v]}}{P^*}}{\PGet{\Chan[k]}{\Role_1}{\Role_2}{_{i \in \indexSet} \Set{ \PLab{\Labe_i}{\tilde{\Args}_i}{P_i} }}} \triangleright \Delta_P $.
			By the rules in Figure~\ref{fig:typingRules} the proof of the judgement has to start (modulo Rule~(\textsf{S2})) as follows
			\begin{align*}
				\dfrac{\dfrac{\Gamma' \vdash P^* \triangleright \Delta_{P1}, \Typed{\AT{\Chan[k]}{\Role_1}}{T_j^*} \quad \vdash \Typed{\tilde{\Args[v]}}{\tilde{\Sort}_j}}{\Gamma' \vdash \PSend{\Chan[k]}{\Role_1}{\Role_2}{\Labe_j}{\tilde{\Args[v]}}{P^*} \triangleright \Delta_{P1}, \Typed{\AT{\Chan[k]}{\Role_1}}{T_{\text{send}}}} (\mathsf{S}) \quad \dfrac{\left( \Gamma' \vdash P_i \triangleright \Delta_{P2}, \Typed{\AT{\Chan[k]}{\Role_2}}{T_i} \quad \vdash \Typed{\tilde{\Args}_i}{\tilde{\Sort}_i} \right)_{i \in \indexSet}}{\Gamma' \vdash \PGet{\Chan[k]}{\Role_1}{\Role_2}{_{i \in \indexSet} \Set{ \PLab{\Labe_i}{\tilde{\Args}_i}{P_i} }} \triangleright \Delta_{P2}, \Typed{\AT{\Chan[k]}{\Role_2}}{T_{\text{get}}}} (\mathsf{C})}{\Gamma' \vdash \PPar{\PSend{\Chan[k]}{\Role_1}{\Role_2}{\Labe_j}{\tilde{\Args[v]}}{P^*}}{\PGet{\Chan[k]}{\Role_1}{\Role_2}{_{i \in \indexSet} \Set{ \PLab{\Labe_i}{\tilde{\Args}_i}{P_i} }}} \triangleright \Delta_P} (\mathsf{Pa})
			\end{align*}
			where $ T_{\text{send}} = \LTSend{\Role_2}{_{i \in \indexSet} \Set{ \LTLab{\Labe_i}{\Typed{\tilde{\Args[z]}_i}{\tilde{\Sort}_i}}{T^*_i} }} $, $ T_{\text{get}} = \LTGet{\Role_1}{_{i \in \indexSet{}} \Set{ \LTLab{\Labe_i}{\Typed{\tilde{\Args[y]}_i}{\tilde{\Sort}_i}}{T_i} }} $, and the session environment $ \Delta_P = \Delta_{P1}, \Typed{\AT{\Chan[k]}{\Role_1}}{T_{\text{send}}} \otimes \Delta_{P2}, \Typed{\AT{\Chan[k]}{\Role_2}}{T_{\text{get}}} $.

			By Lemma~\ref{lem:typeSubstB}, $ \Gamma' \vdash P_j \triangleright \Delta_{P2}, \Typed{\AT{\Chan[k]}{\Role_2}}{T_j} $, $ \vdash \Typed{\tilde{\Args}_j}{\tilde{\Sort}_j} $, and $ \vdash \Typed{\tilde{\Args[v]}}{\tilde{\Sort}_j} $ imply $ \Gamma' \vdash P_j\!\Set[]{\Subst{\tilde{\Args[v]}}{\tilde{\Args}_j}} \triangleright \Delta_{P2}, \Typed{\AT{\Chan[k]}{\Role_2}}{T_j} $.
			With $ \Gamma' \vdash P^* \triangleright \Delta_{P1}, \Typed{\AT{\Chan[k]}{\Role_1}}{T_j^*} $ and since $ \Delta_{P1} \otimes \Delta_{P2} $ is defined, we obtain
			\begin{align*}
				\dfrac{\Gamma' \vdash P^* \triangleright \Delta_{P1}, \Typed{\AT{\Chan[k]}{\Role_1}}{T_j^*} \quad \Gamma' \vdash P_j\!\Set[]{\Subst{\tilde{\Args[v]}}{\tilde{\Args}_j}} \triangleright \Delta_{P2}, \Typed{\AT{\Chan[k]}{\Role_2}}{T_j}}{\Gamma' \vdash \PPar{P^*}{P_j\!\Set[]{\Subst{\tilde{\Args[v]}}{\tilde{\Args}_j}}} \triangleright \Delta_{P1}, \Typed{\AT{\Chan[k]}{\Role_1}}{T_j^*} \otimes \Delta_{P2}, \Typed{\AT{\Chan[k]}{\Role_2}}{T_j}} (\mathsf{Pa})
			\end{align*}
			Note that $ \Delta_P $ contains some $ \Typed{\Role[p]}{\OV{\tilde{\Sort[K]}}} $ if and only if $ \Delta_{P1}, \Typed{\AT{\Chan[k]}{\Role_1}}{T_j^*} \otimes \Delta_{P2}, \Typed{\AT{\Chan[k]}{\Role_2}}{T_j} $ contains the same assignment $ \Typed{\Role[p]}{\OV{\tilde{\Sort[K]}}} $.
			With Lemma~\ref{lem:typeEC}~(2), we have $ \Gamma \vdash P' \triangleright \left( \Delta_{P1}, \Typed{\AT{\Chan[k]}{\Role_1}}{T_j^*} \otimes \Delta_{P2}, \Typed{\AT{\Chan[k]}{\Role_2}}{T_j} \right) \otimes \Delta_{\EC} $.

			It remains to show that $ \Delta_{\EC} \otimes \Delta_P \mapsto \left( \Delta_{P1}, \Typed{\AT{\Chan[k]}{\Role_1}}{T_j^*} \otimes \Delta_{P2}, \Typed{\AT{\Chan[k]}{\Role_2}}{T_j} \right) \otimes \Delta_{\EC} $.
			Because $ j \in \indexSet $ and
			\begin{align*}
				\Delta_P ={} & \Delta_{P1} \otimes \Delta_{P2}, \Typed{\AT{\Chan[k]}{\Role_1}}{\LTSend{\Role_2}{_{i \in \indexSet} \Set{ \LTLab{\Labe_i}{\Typed{\tilde{\Args[z]}_i}{\tilde{\Sort}_i}}{T^*_i} }}}, \Typed{\AT{\Chan[k]}{\Role_2}}{\LTGet{\Role_1}{_{i \in \indexSet{}} \Set{ \LTLab{\Labe_i}{\Typed{\tilde{\Args[y]}_i}{\tilde{\Sort}_i}}{T_i} }}}
			\end{align*}
			we obtain
			\begin{align*}
				\dfrac{\dfrac{}{\Delta_P \mapsto \Delta_{P1}, \Typed{\AT{\Chan[k]}{\Role_1}}{T_j^*} \otimes \Delta_{P2}, \Typed{\AT{\Chan[k]}{\Role_2}}{T_j}} (\mathsf{comS}')}{\Delta_{\EC} \otimes \Delta_P \mapsto \left( \Delta_{P1}, \Typed{\AT{\Chan[k]}{\Role_1}}{T_j^*} \otimes \Delta_{P2}, \Typed{\AT{\Chan[k]}{\Role_2}}{T_j} \right) \otimes \Delta_{\EC}} (\mathsf{par})
			\end{align*}
		\item[Case $ (\mathsf{choice}) $:]
			In this case we have $ P_i \longmapsto P_i' $ and
			\begin{align*}
				P = \AEC{\PChoi{P_1}{P_2}} \quad \text{ and } \quad P' = \AEC{P_i'}
			\end{align*}
			With $ \Gamma \vdash P \triangleright \Delta $ and Lemma~\ref{lem:typeEC}~(1), there exist $ \Delta_P, \Delta_{\EC}, \Gamma' $ such that $ \Gamma \subseteq \Gamma' $, $ \Delta = \Delta_{\EC} \otimes \Delta_P $, and $ \Gamma' \vdash \PChoi{P_1}{P_2} \triangleright \Delta_P $.
			By the rules in Figure~\ref{fig:typingRules} the proof of the judgement has to start (modulo Rule~(\textsf{S2})) as follows
			\begin{align*}
				\dfrac{\Gamma' \vdash P_1 \triangleright \Delta_{P}', \Typed{\AT{\Chan[s]}{\Role}}{T_1} \quad \Gamma' \vdash P_1 \triangleright \Delta_{P}', \Typed{\AT{\Chan[s]}{\Role}}{T_2}}{\Gamma' \vdash \PChoi{P_1}{P_2} \triangleright \Delta_P} (\mathsf{S1})
			\end{align*}
			for some $ \Chan[s] $ and $ \Role $, where $ \Delta_P = \Delta_{P}', \Typed{\AT{\Chan[s]}{\Role}}{\LTChoi{T_1}{T_2}} $.

			By the induction hypothesis, $ \PChoi{P_1}{P_2} \longmapsto P_i' $, $ \Gamma' \vdash P_1 \triangleright \Delta_{P}', \Typed{\AT{\Chan[s]}{\Role}}{T_1} $ and $ \Gamma' \vdash P_1 \triangleright \Delta_{P}', \Typed{\AT{\Chan[s]}{\Role}}{T_2} $ imply $ \Gamma' \vdash P_i' \triangleright \Delta_{P}'' $ for some $ \Delta_{P}'' $ such that $ \Delta_{P}', \Typed{\AT{\Chan[s]}{\Role}}{T_i} \mapsto \Delta_P'' $.
			Because of $ \Delta_P = \Delta_{P}', \Typed{\AT{\Chan[s]}{\Role}}{\LTChoi{T_1}{T_2}} $, $ \Delta_{P}', \Typed{\AT{\Chan[s]}{\Role}}{T_i} \mapsto \Delta_P'' $ and since the rules of Figure~\ref{fig:sessionTypeReductions} neither remove nor add assignments of the form $ \Typed{\Role[p]}{\OV{\tilde{\Sort[K]}}} $, we have $ \Typed{\Role[p]}{\OV{\tilde{\Sort[K]}}} \in \Delta_P $ iff $ \Typed{\Role[p]}{\OV{\tilde{\Sort[K]}}} \in \Delta_P'' $.
			Finally, with Lemma~\ref{lem:typeEC}~(2), we have $ \Gamma \vdash P' \triangleright \Delta_P'' \otimes \Delta_{\EC} $.

			It remains to show that $ \Delta_{\EC} \otimes \Delta_P \mapsto \Delta_P'' \otimes \Delta_{\EC} $.
			Because $ \Delta_{P}', \Typed{\AT{\Chan[s]}{\Role}}{T_i} \mapsto \Delta_P'' $ and $ \Delta_P = \Delta_{P}', \Typed{\AT{\Chan[s]}{\Role}}{\LTChoi{T_1}{T_2}} $, we obtain
			\begin{align*}
				\dfrac{\dfrac{\Delta_{P}', \Typed{\AT{\Chan[s]}{\Role}}{T_i} \mapsto \Delta_P''}{\Delta_P \mapsto \Delta_P''} (\mathsf{choice}')}{\Delta_{\EC} \otimes \Delta_P \mapsto \Delta_P'' \otimes \Delta_{\EC}} (\mathsf{par})
			\end{align*}
		\item[Cases $ (\mathsf{subs}) $:]
			In this case we have
			\begin{align*}
				P = \AEC{\PDecl{\Chan[k]}{\Chan[s]}{\tilde{\Args[v]}}{\tilde{\Chan}}{\tilde{\Role}}{P^*}} \quad \text{ and } \quad P' = \AEC{\PPar{P^*}{\PPar{\POutS{\Chan_1}{\Chan[k]}}{\PPar{\ldots}{\POutS{\Chan_m}{\Chan[s]}}}}}
			\end{align*}
			With $ \Gamma \vdash P \triangleright \Delta $ and Lemma~\ref{lem:typeEC}~(1), there exist $ \Delta_P, \Delta_{\EC}, \Gamma' $ such that $ \Gamma \subseteq \Gamma' $, $ \Delta = \Delta_{\EC} \otimes \Delta_P $, and $ \Gamma' \vdash \PDecl{\Chan[k]}{\Chan[k]}{\tilde{\Args[v]}}{\tilde{\Chan}}{\tilde{\Role}}{P^*} \triangleright \Delta_P $.
			By the rules in Figure~\ref{fig:typingRules} the proof of the judgement has to start (modulo Rule~(\textsf{S2})) as follows
			\begin{align*}
				\dfrac{\begin{array}{c} \Gamma' \vdash P^* \triangleright \Delta_P', \Typed{\AT{\Chan[s]}{\Role''}}{T}, \Typed{\ATI{\Chan[k]}{\Role_1}}{T'_1}, \ldots, \Typed{\ATI{\Chan[k]}{\Role_n}}{T'_n}, \Typed{\ATE{\Chan[k]}{\Role'_1}}{T'_{n + 1}}, \ldots, \Typed{\ATE{\Chan[k]}{\Role'_m}}{T'_{n + m}} \quad \GetType[\Gamma']{\Prot} = \TypeOfProt{\tilde{\Role}}{\tilde{\Args[y]}}{\tilde{\Role}'}{G}\\ \forall i \logdot \GetType[\Gamma']{\Chan_i} = \AT{T'_{i + n}}{\Role'_{i + n}} \quad \forall i \logdot \ProjS{G\!\Set[]{ \Subst{\tilde{\Args[v]}}{\tilde{\Args[y]}} }}{}{\Role_i} = T'_i \quad \forall j \logdot \ProjS{G\!\Set[]{ \Subst{\tilde{\Args[v]}}{\tilde{\Args[y]}} }}{}{\Role'_j} = T'_{j + n} \quad \vdash \Typed{\tilde{\Args[v]}}{\tilde{\Sort}} \quad \GetType[\Gamma']{\Chan[k]} = {\Prot\!\Set[]{ \Subst{\tilde{\Args[v]}}{\tilde{\Args[y]}} }} \end{array}}{\Gamma' \vdash \PDecl{\Chan[k]}{\Chan[s]}{\tilde{\Args[v]}}{\tilde{\Chan}}{\tilde{\Role}}{P^*} \triangleright \Delta_P} (\mathsf{New})
			\end{align*}
			where $ \Delta_P = \Delta_P', \Typed{\AT{\Chan[s]}{\Role''}}{\LTCall{\Prot}{G}{\tilde{\Args[v]}}{\Typed{\tilde{\Args[y]}}{\tilde{\Sort}}}{\tilde{\Role}'}{T}} $.

			Note that Rule~$ (\mathsf{subs}) $ replaces $ P^* $ by $ \PPar{P^*}{\PPar{\POutS{\Chan_1}{\Chan[k]}}{\PPar{\ldots}{\POutS{\Chan_m}{\Chan[k]}}}} $.
			To obtain a type derivation for this term, we first apply Rule~$ (\mathsf{Pa}) $ $ n $ times to split the parallel components of the process.
			Thereby $ \Delta_P', \Typed{\AT{\Chan[s]}{\Role''}}{T}, \Typed{\ATI{\Chan[k]}{\Role_1}}{T'_1}, \ldots, \Typed{\ATI{\Chan[k]}{\Role_n}}{T'_n}, \Typed{\ATE{\Chan[k]}{\Role'_1}}{T'_{n + 1}}, \ldots, \Typed{\ATE{\Chan[k]}{\Role'_m}}{T'_{n + m}} $ is split up into $ \Delta_P', \Typed{\AT{\Chan[s]}{\Role''}}{T}, \Typed{\ATI{\Chan[k]}{\Role_1}}{T'_1}, \ldots, \Typed{\ATI{\Chan[k]}{\Role_n}}{T'_n} $ and $ m $ instances of $ \Typed{\ATE{\Chan[k]}{\Role'_i}}{T'_{n + i}} $.
			Because of $ \forall i \logdot \GetType[\Gamma']{\Chan_i} = \AT{T'_{i + n}}{\Role'_{i + n}} $, for each $ \POutS{\Chan_i}{\Chan[k]} $ we have
			\begin{align*}
				\dfrac{\dfrac{}{\Gamma' \vdash \PEnd \triangleright \emptyset} (\mathsf{N}) \quad \GetType[\Gamma']{\Chan_i} = \AT{T'_{i + n}}{\Role'_{i + n}}}{\Gamma' \vdash \POutS{\Chan_i}{\Chan[k]} \triangleright \Typed{\ATE{\Chan[k]}{\Role'_i}}{T'_{n + i}}} (\mathsf{O})
			\end{align*}
			Finally, with Lemma~\ref{lem:typeEC}~(2), we have
			\begin{align*}
				\Gamma \vdash P' \triangleright \left( \Delta_P', \Typed{\AT{\Chan[s]}{\Role''}}{T}, \Typed{\ATI{\Chan[k]}{\Role_1}}{T'_1}, \ldots, \Typed{\ATI{\Chan[k]}{\Role_n}}{T'_n}, \Typed{\ATE{\Chan[k]}{\Role'_1}}{T'_{n + 1}}, \ldots, \Typed{\ATE{\Chan[k]}{\Role'_m}}{T'_{n + m}} \right) \otimes \Delta_{\EC}
			\end{align*}

			It remains to show that:
			\begin{align*}
				\Delta_{\EC} \otimes \Delta_P \mapsto \left( \Delta_P', \Typed{\AT{\Chan[s]}{\Role''}}{T}, \Typed{\ATI{\Chan[k]}{\Role_1}}{T'_1}, \ldots, \Typed{\ATI{\Chan[k]}{\Role_n}}{T'_n}, \Typed{\ATE{\Chan[k]}{\Role'_1}}{T'_{n + 1}}, \ldots, \Typed{\ATE{\Chan[k]}{\Role'_m}}{T'_{n + m}} \right) \otimes \Delta_{\EC}
			\end{align*}
			Because of $ \forall i \logdot \ProjS{G\!\Set[]{ \Subst{\tilde{\Args[v]}}{\tilde{\Args[y]}} }}{}{\Role_i} = T'_i $ and $ \forall j \logdot \ProjS{G\!\Set[]{ \Subst{\tilde{\Args[v]}}{\tilde{\Args[y]}} }}{}{\Role'_j} = T'_{j + n} $, this follows from Rule~$ (\mathsf{subs}') $.
		\item[Case $ (\mathsf{comC}) $:]
			In this case we have
			\begin{align*}
				P = \AEC{\PPar{\POut{\Chan}{\tilde{\Chan[s]}}{P_1}}{\PInp{\Chan}{\tilde{\Args}}{P_2}}} \quad \text{ and } \quad P' = \AEC{\PPar{P_1}{P_2\!\Set[]{ \Subst{\tilde{\Chan[s]}}{\tilde{\Args}} }}}
			\end{align*}
			With $ \Gamma \vdash P \triangleright \Delta $ and Lemma~\ref{lem:typeEC}~(1), there exist $ \Delta_P, \Delta_{\EC}, \Gamma' $ such that $ \Gamma \subseteq \Gamma' $, $ \Delta = \Delta_{\EC} \otimes \Delta_P $, and $ \Gamma' \vdash \PPar{\POut{\Chan}{\tilde{\Chan[s]}}{P_1}}{\PInp{\Chan}{\tilde{\Args}}{P_2}} \triangleright \Delta_P $.
			By the rules in Figure~\ref{fig:typingRules} the proof of the judgement has to start (modulo Rule~(\textsf{S2})) as follows
			\begin{align*}
				\dfrac{\dfrac{\Gamma' \vdash P_1 \triangleright \Delta_{P1} \quad \GetType[\Gamma']{\Chan} = \AT{T}{\Role}}{\Gamma' \vdash \POut{\Chan}{\tilde{\Chan[s]}}{P_1} \triangleright \Delta_{P1}, \Typed{\ATE{\Chan[s]}{\Role}}{T}} (\mathsf{O}) \quad \dfrac{\Gamma' \vdash P_2 \triangleright \Delta_{P2}, \Typed{\AT{\Args}{\Role}}{T} \quad \GetType[\Gamma']{\Chan}{\AT{T}{\Role}}}{\Gamma' \vdash \PInp{\Chan}{\tilde{\Args}}{P_2} \triangleright \Delta_{P2}} (\mathsf{I})}{\Gamma' \vdash \PPar{\POut{\Chan}{\tilde{\Chan[s]}}{P_1}}{\PInp{\Chan}{\tilde{\Args}}{P_2}} \triangleright \Delta_P} (\mathsf{Pa})
			\end{align*}
			where $ \Delta_P = \Delta_{P1} \otimes \Delta_{P2}, \Typed{\ATE{\Chan[s]}{\Role}}{T} $.

			By Lemma~\ref{lem:typeSubstA}, $ \Gamma' \vdash P_2 \triangleright \Delta_{P2}, \Typed{\AT{\Args}{\Role}}{T} $ implies $ \Gamma' \vdash P_2\!\Set[]{\Subst{\Chan[s]}{\Args}} \triangleright \Delta_{P2}, \Typed{\AT{\Chan[s]}{\Role_2}}{T} $.
			With $ \Gamma' \vdash P_1 \triangleright \Delta_{P1} $ and since $ \Delta_{P1} \otimes \Delta_{P2} $ is defined and there is no type for $ \AT{\Chan[s]}{\Role} $ in $ \Delta_{P1} $, we obtain
			\begin{align*}
				\dfrac{\Gamma' \vdash P_1 \triangleright \Delta_{P1} \quad \Gamma' \vdash P_2\!\Set[]{\Subst{\Chan[s]}{\Args}} \triangleright \Delta_{P2}, \Typed{\AT{\Chan[s]}{\Role}}{T}}{\Gamma' \vdash \PPar{P_1}{P_2\!\Set[]{\Subst{\Chan[s]}{\Args}}} \triangleright \Delta_{P1} \otimes \Delta_{P2}, \Typed{\AT{\Chan[s]}{\Role}}{T}} (\mathsf{Pa})
			\end{align*}
			Finally, with Lemma~\ref{lem:typeEC}~(2), we have $ \Gamma \vdash P' \triangleright \left( \Delta_{P1} \otimes \Delta_{P2}, \Typed{\AT{\Chan[s]}{\Role}}{T} \right) \otimes \Delta_{\EC} $.

			It remains to show that $ \Delta_{\EC} \otimes \Delta_P \mapsto \left( \Delta_{P1} \otimes \Delta_{P2}, \Typed{\AT{\Chan[s]}{\Role}}{T} \right) \otimes \Delta_{\EC} $.
			Because $ \Delta_P = \Delta_{P1} \otimes \Delta_{P2}, \Typed{\ATE{\Chan[s]}{\Role}}{T} $, we obtain
			\begin{align*}
				\dfrac{\dfrac{}{\Delta_P \mapsto \Delta_{P1} \otimes \Delta_{P2}, \Typed{\AT{\Chan[s]}{\Role}}{T}} (\mathsf{comC}')}{\Delta_{\EC} \otimes \Delta_P \mapsto \left( \Delta_{P1} \otimes \Delta_{P2}, \Typed{\AT{\Chan[s]}{\Role}}{T} \right) \otimes \Delta_{\EC}} (\mathsf{par})
			\end{align*}
		\item[Case $ (\mathsf{join}) $:]
			In this case we have
			\begin{align*}
				P = \AEC{\PPar{\PReq{\Chan[s]}{\Role_1}{\Role_2}{\Role_3}{\Chan[k]}{P_1}}{\PEnt{\Chan[s]}{\Role_1}{\Role_2}{\Role_3}{\Args}{P_2}}} \quad \text{ and } \quad P' = \AEC{\PPar{P_1}{P_2\!\Set[]{ \Subst{\Chan[k]}{\Args} }}}
			\end{align*}
			With $ \Gamma \vdash P \triangleright \Delta $ and Lemma~\ref{lem:typeEC}~(1), there exist $ \Delta_P, \Delta_{\EC}, \Gamma' $ such that $ \Gamma \subseteq \Gamma' $, $ \Delta = \Delta_{\EC} \otimes \Delta_P $, and $ \Gamma' \vdash \PPar{\PReq{\Chan[s]}{\Role_1}{\Role_2}{\Role_3}{\Chan[k]}{P_1}}{\PEnt{\Chan[s]}{\Role_1}{\Role_2}{\Role_3}{\Args}{P_2}} \triangleright \Delta_P $.
			By the rules in Figure~\ref{fig:typingRules} the proof of the judgement has to start (modulo Rule~(\textsf{S2})) as follows
			\begin{align*}
				\dfrac{\dfrac{\Gamma' \vdash P_1 \triangleright \Delta_{P1}, \Typed{\AT{\Chan[s]}{\Role_1}}{T_1} \quad \GetType[\Gamma']{\Prot} = \TypeOfProt{\tilde{\Role}_4}{\tilde{\Args[y]}}{\tilde{\Role}_5}{G} \quad \ProjS{G\!\Set[]{ \Subst{\tilde{\Args[v]}}{\tilde{\Args[y]}} }}{}{\Role_3} = T_3}{\Gamma' \vdash \PReq{\Chan[s]}{\Role_1}{\Role_2}{\Role_3}{\Chan[k]}{P_1} \triangleright \Delta_{P1}, \Typed{\AT{\Chan[s]}{\Role_1}}{\LTReq{\Prot}{\Role_3}{\tilde{\Args[v]}}{\Role_2}{T_1}}, \Typed{\ATI{\Chan[k]}{\Role_3}}{T_3}} (\mathsf{P}) \quad D}{\Gamma' \vdash \PPar{\PReq{\Chan[s]}{\Role_1}{\Role_2}{\Role_3}{\Chan[k]}{P_1}}{\PEnt{\Chan[s]}{\Role_1}{\Role_2}{\Role_3}{\Args}{P_2}} \triangleright \Delta_P} (\mathsf{Pa})
			\end{align*}
			with $ D = $
			\begin{align*}
				\dfrac{\Gamma' \vdash P_2 \triangleright \Delta_{P2}, \Typed{\AT{\Chan[s]}{\Role_2}}{T_2}, \Typed{\AT{\Args}{\Role_3}}{T_3} \quad \GetType[\Gamma']{\Prot} = \TypeOfProt{\tilde{\Role}_4}{\tilde{\Args[y]}}{\tilde{\Role}_5}{G} \quad \ProjS{G\!\Set[]{ \Subst{\tilde{\Args[v]}}{\tilde{\Args[y]}} }}{}{\Role_3} = T_3}{\Gamma' \vdash \PEnt{\Chan[s]}{\Role_1}{\Role_2}{\Role_3}{\Args}{P_2} \triangleright \Delta_{P2}, \Typed{\AT{\Chan[s]}{\Role_2}}{\LTEnt{\Prot}{\Role_3}{\tilde{\Args[v]}}{\Role_1}{T_2}}} (\mathsf{J})
			\end{align*}
			where $ \Delta_P = \Delta_{P1}, \Typed{\AT{\Chan[s]}{\Role_1}}{\LTReq{\Prot}{\Role_3}{\tilde{\Args[v]}}{\Role_2}{T_1}}, \Typed{\ATI{\Chan[k]}{\Role_3}}{T_3} \otimes \Delta_{P2}, \Typed{\AT{\Chan[s]}{\Role_2}}{\LTEnt{\Prot}{\Role_3}{\tilde{\Args[v]}}{\Role_1}{T_2}} $.

			By Lemma~\ref{lem:typeSubstA}, $ \Gamma' \vdash P_2 \triangleright \Delta_{P2}, \Typed{\AT{\Chan[s]}{\Role_2}}{T_2}, \Typed{\AT{\Args}{\Role_3}}{T_3} $ implies $ \Gamma' \vdash P_2\!\Set[]{\Subst{\Chan[k]}{\Args}} \triangleright \Delta_{P2}, \Typed{\AT{\Chan[s]}{\Role_2}}{T_2}, \Typed{\AT{\Chan[k]}{\Role_3}}{T_3} $.
			With $ \Gamma' \vdash P_1 \triangleright \Delta_{P1}, \Typed{\AT{\Chan[s]}{\Role_1}}{T_1} $ and since $ \Delta_{P1} \otimes \Delta_{P2} $ is defined, we obtain
			\begin{align*}
				\dfrac{\Gamma' \vdash P_1 \triangleright \Delta_{P1}, \Typed{\AT{\Chan[s]}{\Role_1}}{T_1} \quad \Gamma' \vdash P_2\!\Set[]{\Subst{\Chan[k]}{\Args}} \triangleright \Delta_{P2}, \Typed{\AT{\Chan[s]}{\Role_2}}{T_2}, \Typed{\AT{\Chan[k]}{\Role_3}}{T_3}}{\Gamma' \vdash \PPar{P_1}{P_2\!\Set[]{\Subst{\Chan[k]}{\Args}}} \triangleright \Delta_{P1}, \Typed{\AT{\Chan[s]}{\Role_1}}{T_1} \otimes \Delta_{P2}, \Typed{\AT{\Chan[s]}{\Role_2}}{T_2}, \Typed{\AT{\Chan[k]}{\Role_3}}{T_3}} (\mathsf{Pa})
			\end{align*}
			Finally, with Lemma~\ref{lem:typeEC}~(2), we have $ \Gamma \vdash P' \triangleright \left( \Delta_{P1}, \Typed{\AT{\Chan[s]}{\Role_1}}{T_1} \otimes \Delta_{P2}, \Typed{\AT{\Chan[s]}{\Role_2}}{T_2}, \Typed{\AT{\Chan[k]}{\Role_3}}{T_3} \right) \otimes \Delta_{\EC} $.

			It remains to show that:
			\begin{align*}
				\Delta_{\EC} \otimes \Delta_P \mapsto \left( \Delta_{P1}, \Typed{\AT{\Chan[s]}{\Role_1}}{T_1} \otimes \Delta_{P2}, \Typed{\AT{\Chan[s]}{\Role_2}}{T_2}, \Typed{\AT{\Chan[k]}{\Role_3}}{T_3} \right) \otimes \Delta_{\EC}
			\end{align*}
			Because $ \Delta_P = \Delta_{P1}, \Typed{\AT{\Chan[s]}{\Role_1}}{\LTReq{\Prot}{\Role_3}{\tilde{\Args[v]}}{\Role_2}{T_1}}, \Typed{\ATI{\Chan[k]}{\Role_3}}{T_3} \otimes \Delta_{P2}, \Typed{\AT{\Chan[s]}{\Role_2}}{\LTEnt{\Prot}{\Role_3}{\tilde{\Args[v]}}{\Role_1}{T_2}} $, we obtain
			\begin{align*}
				\dfrac{\dfrac{}{\Delta_P \mapsto \Delta_{P1} \otimes \Delta_{P2}, \Typed{\AT{\Chan[s]}{\Role_1}}{T_1}, \Typed{\AT{\Chan[s]}{\Role_2}}{T_2}, \Typed{\AT{\Chan[k]}{\Role_3}}{T_3}} (\mathsf{join}')}{\Delta_{\EC} \otimes \Delta_P \mapsto \left( \Delta_{P1}, \Typed{\AT{\Chan[s]}{\Role_1}}{T_1} \otimes \Delta_{P2}, \Typed{\AT{\Chan[s]}{\Role_2}}{T_2}, \Typed{\AT{\Chan[k]}{\Role_3}}{T_3} \right) \otimes \Delta_{\EC}} (\mathsf{par})
			\end{align*}
		\item[Case $ (\mathsf{fail}) $:]
			In this case we have
			\begin{align*}
				P = \AEC{\POpt{\Role_1}{\tilde{\Role}}{P_1}{\tilde{\Args}}{\tilde{\Args[v]}}{P_2}} \quad \text{ and } \quad P' = \AEC{P_2\!\Set[]{ \Subst{\tilde{\Args[v]}}{\tilde{\Args}} }}
			\end{align*}
			Because of $ \Gamma \vdash P \triangleright \Delta $ and Lemma~\ref{lem:typeEC}~(1), there exist $ \Delta_P, \Delta_{\EC}, \Gamma' $ such that $ \Gamma \subseteq \Gamma' $, $ \Delta = \Delta_{\EC} \otimes \Delta_P $, and $ \Gamma' \vdash \POpt{\Role_1}{\tilde{\Role}}{P_1}{\tilde{\Args}}{\tilde{\Args[v]}}{P_2} \triangleright \Delta_P $.
			By the rules in Figure~\ref{fig:typingRules} the proof of the judgement has to start (modulo Rule~(\textsf{S2})) as follows:
			\begin{align*}
				\dfrac{\Gamma' \vdash P_1 \triangleright \Delta_{P1}, \Typed{\AT{\Chan[s]}{\Role_1}}{T_1}, \Typed{\Role_1}{\OV{\tilde{\Sort}}} \quad \nexists \Role', \tilde{K} \logdot \Typed{\Role'}{\OV{\tilde{K}}} \in \Delta_{P1} \quad \Gamma' \vdash P_2 \triangleright \Delta_{P2}, \Typed{\AT{\Chan[s]}{\Role_1}}{T_1'} \quad \vdash \Typed{\tilde{\Args}}{\tilde{\Sort}} \quad \vdash \Typed{\tilde{\Args[v]}}{\tilde{\Sort}}}{\Gamma' \vdash \POpt{\Role_1}{\tilde{\Role}}{P_1}{\tilde{\Args}}{\tilde{\Args[v]}}{P_2} \triangleright \Delta_P}(\mathsf{Opt})
			\end{align*}
			where $ \Delta_P = \Delta_{P1} \otimes \Delta_{P2}, \Typed{\AT{\Chan[s]}{\Role_1}}{\LTOpt{\tilde{\Role}'}{T_1}{\Typed{\tilde{\Args[y]}}{\tilde{\Sort}}}{T_1'}} $ and $ \tilde{\Role} \ \dot{=} \ \tilde{\Role}' $.
			By $ \Gamma' \vdash P_2 \triangleright \Delta_{P2}, \Typed{\AT{\Chan[s]}{\Role_1}}{T_1'} $, $ \vdash \Typed{\tilde{\Args}}{\tilde{\Sort}} $, $ \vdash \Typed{\tilde{\Args[v]}}{\tilde{\Sort}} $, and Lemma~\ref{lem:typeSubstB}, we have $ \Gamma' \vdash P_2\!\Set[]{ \Subst{\tilde{\Args[v]}}{\tilde{\Args}} } \triangleright \Delta_{P2}, \Typed{\AT{\Chan[s]}{\Role_1}}{T_1'} $.
			Because of $ \nexists \Role', \tilde{K} \logdot \Typed{\Role'}{\OV{\tilde{K}}} \in \Delta_{P1} $, we have $ \Typed{\Role'}{\OV{\tilde{K}}} \in \Delta_P $ iff $ \Typed{\Role'}{\OV{\tilde{K}}} \in \Delta_{P2} $ iff $ \Typed{\Role'}{\OV{\tilde{K}}} \in \Delta_{P2}, \Typed{\AT{\Chan[s]}{\Role_1}}{T_1'} $.
			With Lemma~\ref{lem:typeEC}~(2) then $ \Gamma \vdash P' \triangleright \left( \Delta_{P2}, \Typed{\AT{\Chan[s]}{\Role_1}}{T_1'} \right) \otimes \Delta_{\EC} $.

			Because $ \Delta_P = \Delta_{P1} \otimes \Delta_{P2}, \Typed{\AT{\Chan[s]}{\Role_1}}{\LTOpt{\Role_2}{T_1}{\Typed{\tilde{\Args[y]}}{\tilde{\Sort}}}{T_1'}} $ and $ \Gamma' \vdash P_1 \triangleright \Delta_{P1}, \Typed{\AT{\Chan[s]}{\Role_1}}{T_1}, \Typed{\Role_1}{\OV{\tilde{\Sort}}} $, we obtain
			\begin{align*}
				\dfrac{\dfrac{}{\Delta_P \mapsto \Delta_{P2}, \Typed{\AT{\Chan[s]}{\Role_1}}{T_1'}} (\mathsf{fail}')}{\Delta_{\EC} \otimes \Delta_P \mapsto \left( \Delta_{P2}, \Typed{\AT{\Chan[s]}{\Role_1}}{T_1'} \right) \otimes \Delta_{\EC}} (\mathsf{par})
			\end{align*}
		\item[Case $ (\mathsf{succ}) $:]
			In this case we have
			\begin{align*}
				P = \AEC{\POpt{\Role_1}{\tilde{\Role}}{\POptEnd{\Role_1}{\tilde{\Args[v]}_1}}{\tilde{\Args}}{\tilde{\Args[v]}_2}{P_1}} \; \text{ and } \; P' = \AEC{P_1\!\Set[]{ \Subst{\tilde{\Args[v]}_1}{\tilde{\Args}} }}
			\end{align*}
			With $ \Gamma \vdash P \triangleright \Delta $ and Lemma~\ref{lem:typeEC}~(1), there exist $ \Delta_P, \Delta_{\EC}, \Gamma' $ such that $ \Gamma \subseteq \Gamma' $, $ \Delta = \Delta_{\EC} \otimes \Delta_P $, and $ \Gamma' \vdash \POpt{\Role_1}{\tilde{\Role}}{\POptEnd{\Role_1}{\tilde{\Args[v]}_1}}{\tilde{\Args}}{\tilde{\Args[v]}_2}{P_1} \triangleright \Delta_P $.
			By the rules in Figure~\ref{fig:typingRules} the proof of the judgement has to start (modulo Rule~(\textsf{S2})) as follows:
			\begin{align*}
				\dfrac{\dfrac{\vdash \Typed{\tilde{\Args[v]}_1}{\tilde{\Sort}}}{\Gamma' \vdash \POptEnd{\Role_1}{\tilde{\Args[v]}_1} \triangleright \Typed{\Role_1}{\OV{\tilde{\Sort}}}}(\mathsf{OptE}) \quad \Gamma' \vdash P_1 \triangleright \Delta_{P1}, \Typed{\AT{\Chan[s]}{\Role_1}}{T_1'} \quad \vdash \Typed{\tilde{\Args}}{\tilde{\Sort}} \quad \vdash \Typed{\tilde{\Args[v]}_2}{\tilde{\Sort}}}{\Gamma' \vdash \POpt{\Role_1}{\tilde{\Role}}{\POptEnd{\Role_1}{\tilde{\Args[v]}_1}}{\tilde{\Args}}{\tilde{\Args[v]}_2}{P_1} \triangleright \Delta_P}(\mathsf{Opt})
			\end{align*}
			where $ \Delta_P = \Delta_{P1}, \Typed{\AT{\Chan[s]}{\Role_1}}{\LTOpt{\tilde{\Role}'}{\LTEnd}{\Typed{\tilde{\Args[z]}}{\tilde{\Sort}}}{T_1'}} $ and $ \tilde{\Role} \ \dot{=} \ \tilde{\Role}' $.
			Since the global environment cannot contain two different declarations of output values for the same pair $ \AT{\Chan[s]}{\Role_1} $, the kinds of the values $ \tilde{\Args[v]}_1 $ and $ \tilde{\Args[v]}_2 $ have to be the same, \ie $ \vdash \Typed{\tilde{\Args[v]}_1}{\tilde{\Sort}} $ and $ \vdash \Typed{\tilde{\Args[v]}_2}{\tilde{\Sort}} $.
			Because of $ \Gamma' \vdash P_1 \triangleright \Delta_{P1}, \Typed{\AT{\Chan[s]}{\Role_1}}{T_1'} $, $ \vdash \Typed{\tilde{\Args[v]}_1}{\tilde{\Sort}} $, $ \vdash \Typed{\tilde{\Args}}{\tilde{\Sort}} $, and Lemma~\ref{lem:typeSubstB}, we have $ \Gamma' \vdash P_1\!\Set[]{ \Subst{\tilde{\Args[v]}_1}{\tilde{\Args}} } \triangleright \Delta_{P1}, \Typed{\AT{\Chan[s]}{\Role_1}}{T_1'} $.
			With Lemma~\ref{lem:typeEC}~(2) then $ \Gamma \vdash P' \triangleright \left( \Delta_{P1}, \Typed{\AT{\Chan[s]}{\Role_1}}{T_1'} \right) \otimes \Delta_{\EC} $.

			It remains to show that $ \Delta_{\EC} \otimes \Delta_P \mapsto \left( \Delta_{P1}, \Typed{\AT{\Chan[s]}{\Role_1}}{T_1'} \right) \otimes \Delta_{\EC} $.

			Because $ \Delta_P = \Delta_{P1}, \Typed{\AT{\Chan[s]}{\Role_1}}{\LTOpt{\tilde{\Role}}{\LTEnd}{\Typed{\tilde{\Args[z]}}{\tilde{\Sort}}}{T_1'}} $, we obtain
			\begin{align*}
				\dfrac{\dfrac{}{\Delta_P \mapsto \Delta_{P1}, \Typed{\AT{\Chan[s]}{\Role_1}}{T_1'}} (\mathsf{succ}')}{\Delta_{\EC} \otimes \Delta_P \mapsto \left( \Delta_{P1}, \Typed{\AT{\Chan[s]}{\Role_1}}{T_1'} \right) \otimes \Delta_{\EC}} (\mathsf{par})
			\end{align*}
		\item[Case $ (\mathsf{cCO}) $:]
			In this case we have
			\begin{align*}
				P & = \AEC{\PPar{\AECR{\POpt{\Role_1}{\tilde{\Role}}{\PPar{\POut{\Chan}{\Chan[s]}{P_1}}{P_2}}{\tilde{\Args}_1}{\tilde{\Args[v]}_1}{P_3}}}{\AECR[E']{\POpt{\Role_2}{\tilde{\Role}}{\PPar{\PInp{\Chan}{\Args}{P_4}}{P_5}}{\tilde{\Args}_2}{\tilde{\Args[v]}_2}{P_6}}}}\\
				P' & = \AEC{\PPar{\AECR{\POpt{\Role_1}{\tilde{\Role}}{\PPar{P_1}{P_2}}{\tilde{\Args}_1}{\tilde{\Args[v]}_1}{P_3}}}{\AECR[E']{\POpt{\Role_2}{\tilde{\Role}}{\PPar{P_4\!\Set[]{ \Subst{\Chan[s]}{\Args} }}{P_5}}{\tilde{\Args}_2}{\tilde{\Args[v]}_2}{P_6}}}}
			\end{align*}
			With $ \Gamma \vdash P \triangleright \Delta $ and Lemma~\ref{lem:typeEC}~(1), there exist $ \Delta_P, \Delta_{\EC}, \Gamma' $ such that $ \Gamma \subseteq \Gamma' $, $ \Delta = \Delta_{\EC} \otimes \Delta_P $, and
			\begin{align*}
				\Gamma' \vdash \PPar{\AECR{\POpt{\Role_1}{\tilde{\Role}}{\PPar{\POut{\Chan}{\Chan[s]}{P_1}}{P_2}}{\tilde{\Args}_1}{\tilde{\Args[v]}_1}{P_3}}}{\AECR[E']{\POpt{\Role_2}{\tilde{\Role}}{\PPar{\PInp{\Chan}{\Args}{P_4}}{P_5}}{\tilde{\Args}_2}{\tilde{\Args[v]}_2}{P_6}}} \triangleright \Delta_P
			\end{align*}
			By the rules in Figure~\ref{fig:typingRules} the proof of the judgement has to start (modulo Rule~(\textsf{S2})) with Rule~$ (\mathsf{Pa}) $, that splits $ \Delta_P $ such that $ \Delta_P = \Delta_{\ECR, P1-3} \otimes \Delta_{\ECR', P4-6} $. Again by Lemma~\ref{lem:typeEC}~(1), there exist $ \Delta_{P1-3} $, $ \Delta_{P4-6} $, $ \Delta_{\ECR} $, $ \Delta_{\ECR'} $, $ \Gamma_1 $, and $ \Gamma_2 $ such that $ \Gamma' \subseteq \Gamma_1 $, $ \Gamma' \subseteq \Gamma_2 $, $ \Delta_{\ECR, P1-3} = \Delta_{\ECR} \otimes \Delta_{P1-3} $, $ \Delta_{\ECR', P4-6} = \Delta_{\ECR'} \otimes \Delta_{P4-6} $, $ \Gamma_1 \vdash \POpt{\Role_1}{\tilde{\Role}}{\PPar{\POut{\Chan}{\Chan[s]}{P_1}}{P_2}}{\tilde{\Args}_1}{\tilde{\Args[v]}_1}{P_3} \triangleright \Delta_{P1-3} $, and $ \Gamma_2 \vdash \POpt{\Role_2}{\tilde{\Role}}{\PPar{\PInp{\Chan}{\Args}{P_4}}{P_5}}{\tilde{\Args}_2}{\tilde{\Args[v]}_2}{P_6} \triangleright \Delta_{P4-6} $. Then:
			\begin{align*}
				\hspace*{-1em}\dfrac{\dfrac{\begin{array}{l} \dfrac{\Gamma_1 \vdash P_1 \triangleright \Delta_{P1} \quad \GetType[\Gamma_1]{\Chan} = \AT{T_1}{\Role_3}}{\Gamma_1 \vdash \POut{\Chan}{\Chan[s]}{P_1} \triangleright \Delta_{P1}, \Typed{\ATE{\Chan[s]}{\Role_3}}{T_1}}(\mathsf{O}) \quad \Gamma_1 \vdash P_2 \triangleright \Delta_{P2} \end{array}}{\Gamma_1 \vdash \PPar{\POut{\Chan}{\Chan[s]}{P_1}}{P_2} \triangleright \left( \Delta_{P1}, \Typed{\ATE{\Chan[s]}{\Role_3}}{T_1} \right) \otimes \Delta_{P2}}(\mathsf{Pa}) \begin{array}{l} \Gamma_1 \vdash P_3 \triangleright \Delta_{P3}, \Typed{\AT{\Chan[k]_1}{\Role_1}}{T_1'}\\ \vdash \Typed{\tilde{\Args}_1}{\tilde{\Sort}} \quad \vdash \Typed{\tilde{\Args[v]}_1}{\tilde{\Sort}} \end{array}}{\Gamma_1 \vdash \POpt{\Role_1}{\tilde{\Role}}{\PPar{\POut{\Chan}{\Chan[s]}{P_1}}{P_2}}{\tilde{\Args}_1}{\tilde{\Args[v]}_1}{P_3} \triangleright \Delta_{P1-3}}(\mathsf{Opt})
			\end{align*}
			where $ \left( \Delta_{P1}, \Typed{\ATE{\Chan[s]}{\Role_3}}{T_1} \right) \otimes \Delta_{P2} = \Delta_{P1-2}, \Typed{\AT{\Chan[k]_1}{\Role_1}}{T_1}, \Typed{\Role_1}{\OV{\tilde{\Sort}}} $ and $ \nexists \Role', \tilde{\Sort[K]} \logdot \Typed{\Role'}{\OV{\tilde{\Sort[K]}}} \in \Delta_{P1-2} $ and $ \Delta_{P1-3} = \Delta_{P1-2} \otimes \Delta_{P3}, \Typed{\AT{\Chan[k]_1}{\Role_1}}{\LTOpt{\Role_2}{T_1}{\Typed{\tilde{\Args[y]}_1}{\tilde{\Sort}}}{T_1'}} $.
			Since $ \left( \Delta_{P1}, \Typed{\ATE{\Chan[s]}{\Role_3}}{T_1} \right) \otimes \Delta_{P2} $ is defined and because $ \Chan[s] \neq \Chan[k] $, we obtain
			\begin{align*}
				\dfrac{\dfrac{\begin{array}{l} \Gamma_1 \vdash P_1 \triangleright \Delta_{P1} \quad \Gamma_1 \vdash P_2 \triangleright \Delta_{P2} \end{array}}{\Gamma_1 \vdash \PPar{P_1}{P_2} \triangleright \Delta_{P1} \otimes \Delta_{P2}}(\mathsf{Pa}) \quad \begin{array}{l} \Gamma_1 \vdash P_3 \triangleright \Delta_{P3}, \Typed{\AT{\Chan[k]_1}{\Role_1}}{T_1'}\\ \vdash \Typed{\tilde{\Args}_1}{\tilde{\Sort}} \quad \vdash \Typed{\tilde{\Args[v]}_1}{\tilde{\Sort}} \end{array}}{\Gamma_1 \vdash \POpt{\Role_1}{\tilde{\Role}}{\PPar{P_1}{P_2}}{\tilde{\Args}_1}{\tilde{\Args[v]}_1}{P_3} \triangleright \Delta_{P1-3}'}(\mathsf{Opt})
			\end{align*}
			where $ \Delta_{P1} \otimes \Delta_{P2} = \Delta_{P1-2}', \Typed{\AT{\Chan[k]_1}{\Role_1}}{T_1''}, \Typed{\Role_1}{\OV{\tilde{\Sort}}} $ and
			\begin{align*}
				\Delta_{P1-3}' = \Delta_{P1-2}' \otimes \Delta_{P3}, \Typed{\AT{\Chan[k]_1}{\Role_1}}{\LTOpt{\Role_2}{T_1''}{\Typed{\tilde{\Args[y]}_1}{\tilde{\Sort}}}{T_1'}}
			\end{align*}
			Note that $ \Delta_{P1-3}' $ is obtained from $ \Delta_{P1-3} $ by removing a capability on $ \AT{\Chan[s]}{\Role_3} $ and changing a capability on $ \AT{\Chan[k]_1}{\Role_1} $.
			With Lemma~\ref{lem:typeEC}~(2), then $ \Gamma' \vdash \AECR{\POpt{\Role_1}{\tilde{\Role}}{\PPar{P_1}{P_2}}{\tilde{\Args}_1}{\tilde{\Args[v]}_1}{P_3}} \triangleright \Delta_{\ECR, P1-3}' $, where $ \Delta_{\ECR, P1-3}' = \Delta_{P1-3}' \otimes \Delta_{\ECR} $.

			Moreover, because $ \GetType[\Gamma_2]{\Chan} = \AT{T1}{\Role_3} $,
			\begin{align*}
				\hspace*{-1em}\dfrac{\dfrac{\begin{array}{l} \dfrac{\Gamma_2 \vdash P_4 \triangleright \Delta_{P4}, \Typed{\AT{\Args}{\Role_3}}{T_1} \quad \GetType[\Gamma_2]{\Chan} = \AT{T_1}{\Role_3}}{\Gamma_2 \vdash \PInp{\Chan}{\Args}{P_4} \triangleright \Delta_{P4}}(\mathsf{I}) \quad \Gamma_2 \vdash P_5 \triangleright \Delta_{P5} \end{array}}{\Gamma_2 \vdash \PPar{\PInp{\Chan}{\Args}{P_4}}{P_5} \triangleright \Delta_{P4} \otimes \Delta_{P5}}(\mathsf{Pa}) \begin{array}{l} \Gamma_2 \vdash P_6 \triangleright \Delta_{P6}, \Typed{\AT{\Chan[k]_2}{\Role_2}}{T_2'}\\ \vdash \Typed{\tilde{\Args}_2}{\tilde{\Sort}'} \quad \vdash \Typed{\tilde{\Args[v]}_2}{\tilde{\Sort}'} \end{array}}{\Gamma_2 \vdash \POpt{\Role_2}{\tilde{\Role}}{\PPar{\PInp{\Chan}{\Args}{P_4}}{P_5}}{\tilde{\Args}_2}{\tilde{\Args[v]}_2}{P_6} \triangleright \Delta_{P4-6}}(\mathsf{Opt})
			\end{align*}
			where $ \Delta_{P4} \otimes \Delta_{P5} = \Delta_{P4-5}, \Typed{\AT{\Chan[k]_2}{\Role_2}}{T_2}, \Typed{\Role_2}{\OV{\tilde{\Sort'}}} $, $ \nexists \Role', \tilde{\Sort[K]} \logdot \Typed{\Role'}{\OV{\tilde{\Sort[K]}}} \in \Delta_{P4-5} $, and we have $ \Delta_{P4-6} = \Delta_{P4-5} \otimes \Delta_{P6}, \Typed{\AT{\Chan[k]_2}{\Role_2}}{\LTOpt{\tilde{\Role}}{T_2}{\Typed{\tilde{\Args[y]}_2}{\tilde{\Sort}'}}{T_2'}} $.
			By $ \Gamma_2 \vdash P_4 \triangleright \Delta_{P4}, \Typed{\AT{\Args}{\Role_3}}{T_1} $ and Lemma~\ref{lem:typeSubstA}, we have $ \Gamma_2 \vdash P_4\!\Set[]{ \Subst{\Chan[s]}{\Args} } \triangleright \Delta_{P4}, \Typed{\AT{\Chan[s]}{\Role_3}}{T_1} $.
			Then
			\begin{align*}
				\hspace*{-1em}\dfrac{\dfrac{\Gamma_2 \vdash P_4\!\Set[]{ \Subst{\Chan[s]}{\Args} } \triangleright \Delta_{P4}, \Typed{\AT{\Chan[s]}{\Role_3}}{T_1} \quad \Gamma_2 \vdash P_5 \triangleright \Delta_{P5}}{\Gamma_2 \vdash \PPar{P_4\!\Set[]{ \Subst{\Chan[s]}{\Args} }}{P_5} \triangleright \left( \Delta_{P4}, \Typed{\AT{\Chan[s]}{\Role_3}}{T_1} \right) \otimes \Delta_{P5}}(\mathsf{Pa}) \begin{array}{l} \Gamma_2 \vdash P_6 \triangleright \Delta_{P6}, \Typed{\AT{\Chan[k]_2}{\Role_2}}{T_2'}\\ \vdash \Typed{\tilde{\Args}_2}{\tilde{\Sort}'} \quad \vdash \Typed{\tilde{\Args[v]}_2}{\tilde{\Sort}'} \end{array}}{\Gamma_2 \vdash \POpt{\Role_2}{\tilde{\Role}}{\PPar{P_4\!\Set[]{ \Subst{\Chan[s]}{\Args} }}{P_5}}{\tilde{\Args}_2}{\tilde{\Args[v]}_2}{P_6} \triangleright \Delta_{P4-6}'}(\mathsf{Opt})
			\end{align*}
			where $ \left( \Delta_{P4}, \Typed{\AT{\Chan[s]}{\Role_3}}{T_1} \right) \otimes \Delta_{P5} = \Delta_{P4-5}'', \Typed{\AT{\Chan[k]_2}{\Role_2}}{T_2}, \Typed{\Role_2}{\OV{\tilde{\Sort'}}} $ and we have $ \Delta_{P4-6}' = \Delta_{P4-5}'' \otimes \Delta_{P6}, \Typed{\AT{\Chan[k]_2}{\Role_2}}{\LTOpt{\tilde{\Role}}{T_2}{\Typed{\tilde{\Args[y]}_2}{\tilde{\Sort}'}}{T_2'}} $.
			Here $ \Delta_{P4-6}' $ is obtained from $ \Delta_{P4-6} $ by a adding a single capability on $ \AT{\Chan[s]}{\Role_3} $.
			With Lemma~\ref{lem:typeEC}~(2), then $ \Gamma' \vdash \AECR[E']{\POpt{\Role_2}{\tilde{\Role}}{\PPar{P_4\!\Set[]{ \Subst{\Chan[s]}{\Args} }}{P_5}}{\tilde{\Args}_2}{\tilde{\Args[v]}_2}{P_6}} \triangleright \Delta_{\ECR', P4-6}' $, where $ \Delta_{\ECR', P4-6}' = \Delta_{P4-6}' \otimes \Delta_{\ECR'} $.

			Since $ \Delta_{\ECR, P1-3} \otimes \Delta_{\ECR', P4-6} $ is defined, so is $ \Delta_{\ECR, P1-3}' \otimes \Delta_{\ECR', P4-6}' $.
			Hence, by Rule~$ (\mathsf{Pa}) $, the judgement $ \Gamma' \vdash \AECR{\POpt{\Role_1}{\tilde{\Role}}{\PPar{P_1}{P_2}}{\tilde{\Args}_1}{\tilde{\Args[v]}_1}{P_3}} \triangleright \Delta_{\ECR, P1-3}' $, and $ \Gamma' \vdash \AECR[E']{\POpt{\Role_2}{\tilde{\Role}}{\PPar{P_4\!\Set[]{ \Subst{\Chan[s]}{\Args} }}{P_5}}{\tilde{\Args}_2}{\tilde{\Args[v]}_2}{P_6}} \triangleright \Delta_{\ECR', P4-6}' $, we have
			\begin{align*}
				\Gamma' \vdash \PPar{\AECR{\POpt{\Role_1}{\tilde{\Role}}{\PPar{P_1}{P_2}}{\tilde{\Args}_1}{\tilde{\Args[v]}_1}{P_3}}}{\AECR[E']{\POpt{\Role_2}{\tilde{\Role}}{\PPar{P_4\!\Set[]{ \Subst{\Chan[s]}{\Args} }}{P_5}}{\tilde{\Args}_2}{\tilde{\Args[v]}_2}{P_6}}} \triangleright \; \Delta_{\ECR, P1-3}' \otimes \Delta_{\ECR', P4-6}'
			\end{align*}
			With Lemma~\ref{lem:typeEC}~(2) we conclude with $ \Gamma \vdash P' \triangleright \left( \Delta_{\ECR, P1-3}' \otimes \Delta_{\ECR', P4-6}' \right) \otimes \Delta_{\EC} $.

			It remains to show that $ \Delta_{\EC} \otimes \Delta_P \mapsto \left( \Delta_{\ECR, P1-3}' \otimes \Delta_{\ECR', P4-6}' \right) \otimes \Delta_{\EC} $.
			Because $ \Delta_P = \left( \Delta_{\ECR} \otimes \Delta_{P1-3} \right) \otimes \left( \Delta_{\ECR'} \otimes \Delta_{P4-6} \right) $ with $ \Delta_{P1-3} = \Delta_{P1-2} \otimes \Delta_{P3}, \Typed{\AT{\Chan[k]_1}{\Role_1}}{\LTOpt{\Role_2}{T_1}{\Typed{\tilde{\Args[y]}_1}{\tilde{\Sort}}}{T_1'}} $ and $ \Delta_{P4-6} = \Delta_{P4-5} \otimes \Delta_{P6}, \Typed{\AT{\Chan[k]_2}{\Role_2}}{\LTOpt{\Role_1}{T_2}{\Typed{\tilde{\Args[y]}_2}{\tilde{\Sort}'}}{T_2'}} $, we obtain
			\begin{align*}
				\dfrac{\dfrac{\dfrac{}{\left( \Delta_{P1-2} \otimes \Delta_{P3}, \Typed{\AT{\Chan[k]_1}{\Role_1}}{T_1} \right) \otimes \left( \Delta_{P4-5} \otimes \Delta_{P6}, \Typed{\AT{\Chan[k]_2}{\Role_2}}{T_2} \right) \mapsto \Delta_{P1-3}' \otimes \Delta_{P4-6}'} (\mathsf{comS}')}{\Delta_{P1-3} \otimes \Delta_{P4-6} \mapsto \Delta_{P1-3}' \otimes \Delta_{P4-6}'} (\mathsf{optCom})}{\Delta_{\EC} \otimes \Delta_P \mapsto \left( \Delta_{\ECR, P1-3}' \otimes \Delta_{\ECR', P4-6}' \right) \otimes \Delta_{\EC}} (\mathsf{par})
			\end{align*}
			where we first reorder the session environments modulo $ \otimes $ and remove with Rule~$ (\mathsf{par}) $ all assignments on the contexts, \ie $ \Delta_{\ECR} $, $ \Delta_{\ECR'} $, and $ \Delta_{\EC} $.
		\item[Case $ (\mathsf{jO}) $:]
			In this case we have
			\begin{align*}
				P & = \AEC{\PPar{\AECR{P_{1-3}}}{\AECR[E']{P_{4-6}}}}\\
				P_{1-3} & = \POpt{\Role_1}{\tilde{\Role}}{\PPar{\PReq{\Chan[s]}{\Role_3}{\Role_4}{\Role_5}{\Chan[k]}{P_1}}{P_2}}{\tilde{\Args}_1}{\tilde{\Args[v]}_1}{P_3}\\
				P_{4-6} & = \POpt{\Role_2}{\tilde{\Role}}{\PPar{\PEnt{\Chan[s]}{\Role_3}{\Role_4}{\Role_5}{\Args}{P_4}}{P_5}}{\tilde{\Args}_2}{\tilde{\Args[v]}_2}{P_6}\\
				P' & = \AEC{\PPar{\AECR{\POpt{\Role_1}{\tilde{\Role}}{\PPar{P_1}{P_2}}{\tilde{\Args}_1}{\tilde{\Args[v]}_1}{P_3}}}{\AECR[E']{\POpt{\Role_2}{\tilde{\Role}}{\PPar{P_4\!\Set[]{ \Subst{\Chan[k]}{\Args} }}{P_5}}{\tilde{\Args}_2}{\tilde{\Args[v]}_2}{P_6}}}}
			\end{align*}
			With $ \Gamma \vdash P \triangleright \Delta $ and Lemma~\ref{lem:typeEC}~(1), there exist $ \Delta_P, \Delta_{\EC}, \Gamma' $ such that $ \Gamma \subseteq \Gamma' $, $ \Delta = \Delta_{\EC} \otimes \Delta_P $, and $ \Gamma' \vdash \PPar{\AECR{P_{1-3}}}{\AECR[E']{P_{4-6}}} \triangleright \Delta_P $.
			By the rules in Figure~\ref{fig:typingRules} the proof of the judgement has to start (modulo Rule~(\textsf{S2})) with Rule~$ (\mathsf{Pa}) $, that splits $ \Delta_P $ such that $ \Delta_P = \Delta_{\ECR, P1-3} \otimes \Delta_{\ECR', P4-6} $. Again by Lemma~\ref{lem:typeEC}~(1), there exist $ \Delta_{P1-3} $, $ \Delta_{P4-6} $, $ \Delta_{\ECR} $, $ \Delta_{\ECR'} $, $ \Gamma_1 $, and $ \Gamma_2 $ such that $ \Gamma' \subseteq \Gamma_1 $, $ \Gamma' \subseteq \Gamma_2 $, $ \Delta_{\ECR, P1-3} = \Delta_{\ECR} \otimes \Delta_{P1-3} $, $ \Delta_{\ECR', P4-6} = \Delta_{\ECR'} \otimes \Delta_{P4-6} $, $ \Gamma_1 \vdash P_{1-3} \triangleright \Delta_{P1-3} $, and $ \Gamma_2 \vdash P_{4-6} \triangleright \Delta_{P4-6} $. Then:
			\begin{align*}
				\dfrac{\dfrac{\begin{array}{l} \dfrac{\begin{array}{l} \Gamma_1 \vdash P_1 \triangleright \Delta_{P1}, \Typed{\AT{\Chan[s]}{\Role_3}}{T_1}\\ \GetType[\Gamma_1]{\Prot} = \TypeOfProt{\tilde{\Role}_6}{\tilde{\Args[y]}_1}{\tilde{\Role}_7}{G} \quad \ProjS{G\!\Set[]{ \Subst{\tilde{\Args[z]}_1}{\tilde{\Args[y]}_1} }}{}{\Role_5} = T_5 \end{array}}{\Gamma_1 \vdash \PReq{\Chan[s]}{\Role_3}{\Role_4}{\Role_5}{\Chan[k]}{P_1} \triangleright \Delta_{P1}'}(\mathsf{P}) \quad \Gamma_1 \vdash P_2 \triangleright \Delta_{P2} \end{array}}{\Gamma_1 \vdash \PPar{\PReq{\Chan[s]}{\Role_3}{\Role_4}{\Role_5}{\Chan[k]}{P_1}}{P_2} \triangleright \Delta_{P1}' \otimes \Delta_{P2}}(\mathsf{Pa}) \begin{array}{l} \Gamma_1 \vdash P_3 \triangleright \Delta_{P3}, \Typed{\AT{\Chan[k]_1}{\Role_1}}{T_1'}\\ \vdash \Typed{\tilde{\Args}_1}{\tilde{\Sort}} \quad \vdash \Typed{\tilde{\Args[v]}_1}{\tilde{\Sort}} \end{array}}{\Gamma_1 \vdash \POpt{\Role_1}{\tilde{\Role}}{\PPar{\PReq{\Chan[s]}{\Role_3}{\Role_4}{\Role_5}{\Chan[k]}{P_1}}{P_2}}{\tilde{\Args}_1}{\tilde{\Args[v]}_1}{P_3} \triangleright \Delta_{P1-3}}(\mathsf{Opt})
			\end{align*}
			where $ \nexists \Role', \tilde{\Sort[K]} \logdot \Typed{\Role'}{\OV{\tilde{\Sort[K]}}} \in \Delta_{P1-2} $ and
			\begin{align*}
				\Delta_{P1}' & = \Delta_{P1}, \Typed{\AT{\Chan[s]}{\Role_3}}{\LTReq{\Prot}{\Role_5}{\tilde{\Args[z]}_1}{\Role_4}{T_1}}, \Typed{\ATI{\Chan[k]}{\Role_5}}{T_5}\\
				\Delta_{P1}' \otimes \Delta_{P2} & = \Delta_{P1-2}, \Typed{\AT{\Chan[k]_1}{\Role_1}}{T_1}, \Typed{\Role_1}{\OV{\tilde{\Sort}}}\\
				\Delta_{P1-3} & = \Delta_{P1-2} \otimes \Delta_{P3}, \Typed{\AT{\Chan[k]_1}{\Role_1}}{\LTOpt{\tilde{\Role}}{T_1}{\Typed{\tilde{\Args[y]}_1}{\tilde{\Sort}}}{T_1'}}
			\end{align*}
			Since $ \Delta_{P1}' \otimes \Delta_{P2} $ is defined and because $ \Chan[s] \neq \Chan[k] $, we obtain
			\begin{align*}
				\dfrac{\dfrac{\begin{array}{l} \Gamma_1 \vdash P_1 \triangleright \Delta_{P1}, \Typed{\AT{\Chan[s]}{\Role_3}}{T_1} \quad \Gamma_1 \vdash P_2 \triangleright \Delta_{P2} \end{array}}{\Gamma_1 \vdash \PPar{P_1}{P_2} \triangleright \left( \Delta_{P1}, \Typed{\AT{\Chan[s]}{\Role_3}}{T_1} \right) \otimes \Delta_{P2}}(\mathsf{Pa}) \quad \begin{array}{l} \Gamma_1 \vdash P_3 \triangleright \Delta_{P3}, \Typed{\AT{\Chan[k]_1}{\Role_1}}{T_1'}\\ \vdash \Typed{\tilde{\Args}_1}{\tilde{\Sort}} \quad \vdash \Typed{\tilde{\Args[v]}_1}{\tilde{\Sort}} \end{array}}{\Gamma_1 \vdash \POpt{\Role_1}{\tilde{\Role}}{\PPar{P_1}{P_2}}{\tilde{\Args}_1}{\tilde{\Args[v]}_1}{P_3} \triangleright \Delta_{P1-3}'}(\mathsf{Opt})
			\end{align*}
			where $ \left( \Delta_{P1}, \Typed{\AT{\Chan[s]}{\Role_3}}{T_1} \right) \otimes \Delta_{P2} = \Delta_{P1-2}', \Typed{\AT{\Chan[k]_1}{\Role^1}}{T_1} $ and
			\begin{align*}
				\Delta_{P1-3}' = \Delta_{P1-2}' \otimes \Delta_{P3}, \Typed{\AT{\Chan[k]_1}{\Role_1}}{\LTOpt{\tilde{\Role}}{T_1}{\Typed{\tilde{\Args[y]}_1}{\tilde{\Sort}}}{T_1'}}, \Typed{\Role_1}{\OV{\tilde{\Sort}}}
			\end{align*}
			Note that $ \Delta_{P1-3}' $ is obtained from $ \Delta_{P1-3} $ by removing a capability on $ \AT{\Chan[k]}{\Role_5} $ and reducing a capability on $ \AT{\Chan[s]}{\Role_3} $.
			With Lemma~\ref{lem:typeEC}~(2), then $ \Gamma' \vdash \AECR{\POpt{\Role_1}{\tilde{\Role}}{\PPar{P_1}{P_2}}{\tilde{\Args}_1}{\tilde{\Args[v]}_1}{P_3}} \triangleright \Delta_{\ECR, P1-3}' $, where $ \Delta_{\ECR, P1-3}' = \Delta_{P1-3}' \otimes \Delta_{\ECR} $.

			Moreover, because of $ \GetType[\Gamma_2]{\Prot} = \TypeOfProt{\tilde{\Role}_6}{\tilde{\Args[y]}_1}{\tilde{\Role}_7}{G} $ and $ \ProjS{G\!\Set[]{ \Subst{\tilde{\Args[z]}_1}{\tilde{\Args[y]}_1} }}{}{\Role_5} = T_5 $,
			\begin{align*}
				\dfrac{\dfrac{\dfrac{\begin{array}{l} \Gamma_2 \vdash P_4 \triangleright \Delta_{P4}, \Typed{\AT{\Chan[s]}{\Role_4}}{T_4}, \Typed{\AT{\Args}{\Role_5}}{T_5}\\ \GetType[\Gamma_2]{\Prot} = \TypeOfProt{\tilde{\Role}_6}{\tilde{\Args[y]}_1}{\tilde{\Role}_7}{G} \quad \ProjS{G\!\Set[]{ \Subst{\tilde{\Args[z]}_1}{\tilde{\Args[y]}_1} }}{}{\Role_5} = T_5 \end{array}}{\Gamma_2 \vdash \PEnt{\Chan[s]}{\Role_3}{\Role_4}{\Role_5}{\Args}{P_4} \triangleright \Delta_{P4}'}(\mathsf{J}) \quad \Gamma_2 \vdash P_5 \triangleright \Delta_{P5}}{\Gamma_2 \vdash \PPar{\PEnt{\Chan[s]}{\Role_3}{\Role_4}{\Role_5}{\Args}{P_4}}{P_5} \triangleright \Delta_{P4}' \otimes \Delta_{P5}}(\mathsf{Pa}) \begin{array}{l} \Gamma_2 \vdash P_6 \triangleright \Delta_{P6}, \Typed{\AT{\Chan[k]_2}{\Role_2}}{T_2'}\\ \vdash \Typed{\tilde{\Args}_2}{\tilde{\Sort}'} \quad \vdash \Typed{\tilde{\Args[v]}_2}{\tilde{\Sort}'} \end{array}}{\Gamma_2 \vdash \POpt{\Role_2}{\tilde{\Role}}{\PPar{\PEnt{\Chan[s]}{\Role_3}{\Role_4}{\Role_5}{\Args}{P_4}}{P_5}}{\tilde{\Args}_2}{\tilde{\Args[v]}_2}{P_6} \triangleright \Delta_{P4-6}}(\mathsf{Opt})
			\end{align*}
			where $ \nexists \Role', \tilde{\Sort[K]} \logdot \Typed{\Role'}{\OV{\tilde{\Sort[K]}}} \in \Delta_{P4-5} $ and
			\begin{align*}
				\Delta_{P4}' & = \Delta_{P4}, \Typed{\AT{\Chan[s]}{\Role_4}}{\LTEnt{\Prot}{\Role_5}{\tilde{\Args[z]}_2}{\Role_3}{T_4}}\\
				\Delta_{P4}' \otimes \Delta_{P5} & = \Delta_{P4-5}, \Typed{\AT{\Chan[k]_2}{\Role_2}}{T_2}, \Typed{\Role_2}{\OV{\tilde{\Sort}'}}\\
				\Delta_{P4-6} & = \Delta_{P4-5} \otimes \Delta_{P6}, \Typed{\AT{\Chan[k]_2}{\Role_2}}{\LTOpt{\tilde{\Role}}{T_2}{\Typed{\tilde{\Args[y]}_2}{\tilde{\Sort}'}}{T_2'}}
			\end{align*}
			By $ \Gamma_2 \vdash P_4 \triangleright \Delta_{P4}, \Typed{\AT{\Chan[s]}{\Role_4}}{T_4}, \Typed{\AT{\Args}{\Role_5}}{T_5} $ and Lemma~\ref{lem:typeSubstA}, we have $ \Gamma_2 \vdash P_4\!\Set[]{ \Subst{\Chan[k]}{\Args} } \triangleright \Delta_{P4}, \Typed{\AT{\Chan[s]}{\Role_4}}{T_4}, \Typed{\AT{\Args[k]}{\Role_5}}{T_3'} $.
			Then
			\begin{align*}
				\dfrac{\dfrac{\begin{array}{l} \Gamma_2 \vdash P_4\!\Set[]{ \Subst{\Chan[k]}{\Args} } \triangleright \Delta_{P4}, \Typed{\AT{\Chan[s]}{\Role_4}}{T_4}, \Typed{\AT{\Args[k]}{\Role_5}}{T_5} \quad \Gamma_2 \vdash P_5 \triangleright \Delta_{P5} \end{array}}{\Gamma_2 \vdash \PPar{P_4\!\Set[]{ \Subst{\Chan[k]}{\Args} }}{P_5} \triangleright \Delta_{P4-5}''}(\mathsf{Pa}) \begin{array}{l} \Gamma_2 \vdash P_6 \triangleright \Delta_{P6}, \Typed{\AT{\Chan[k]_2}{\Role_2}}{T_2'}\\ \vdash \Typed{\tilde{\Args}_2}{\tilde{\Sort}'} \quad \vdash \Typed{\tilde{\Args[v]}_2}{\tilde{\Sort}'} \end{array}}{\Gamma_2 \vdash \POpt{\Role_2}{\tilde{\Role}}{\PPar{P_4\!\Set[]{ \Subst{\Chan[k]}{\Args} }}{P_5}}{\tilde{\Args}_2}{\tilde{\Args[v]}_2}{P_6} \triangleright \Delta_{P4-6}'}(\mathsf{Opt})
			\end{align*}
			where $ \Delta_{P4-5}'' = \left( \Delta_{P4}, \Typed{\AT{\Chan[s]}{\Role_4}}{T_4}, \Typed{\AT{\Args[k]}{\Role_5}}{T_5} \right) \otimes \Delta_5 = \Delta_{P4-5}''', \Typed{\AT{\Chan[k]_2}{\Role_2}}{T_2} $ and $ \Delta_{P4-6}' = \Delta_{P4-5}''' \otimes \Delta_{P6}, \Typed{\AT{\Chan[k]_2}{\Role_2}}{\LTOpt{\tilde{\Role}}{T_2}{\Typed{\tilde{\Args[y]}^2}{\tilde{\Sort}'}}{T_2'}}, \Typed{\Role_2}{\OV{\tilde{\Sort}'}} $.
			Here $ \Delta_{P4-6}' $ is obtained from $ \Delta_{P4-6} $ by reducing a capability on $ \AT{\Chan[s]}{\Role_4} $ and a adding a capability on $ \AT{\Chan[k]}{\Role_5} $.
			With Lemma~\ref{lem:typeEC}~(2), then $ \Gamma' \vdash \AECR[E']{\POpt{\Role_2}{\tilde{\Role}}{\PPar{P_4\!\Set[]{ \Subst{\Chan[k]}{\Args} }}{P_5}}{\tilde{\Args}_2}{\tilde{\Args[v]}_2}{P_6}} \triangleright \Delta_{\ECR', P4-6}' $, where $ \Delta_{\ECR', P4-6}' = \Delta_{P4-6}' \otimes \Delta_{\ECR'} $.

			Since $ \Delta_{\ECR, P1-3} \otimes \Delta_{\ECR', P4-6} $ is defined, so is $ \Delta_{\ECR, P1-3}' \otimes \Delta_{\ECR', P4-6}' $.
			Hence, by Rule~$ (\mathsf{Pa}) $, the judgement $ \Gamma' \vdash \AECR{\POpt{\Role_1}{\tilde{\Role}}{\PPar{P_1}{P_2}}{\tilde{\Args}_1}{\tilde{\Args[v]}_1}{P_3}} \triangleright \Delta_{\ECR, P1-3}' $, and $ \Gamma' \vdash \AECR[E']{\POpt{\Role_2}{\tilde{\Role}}{\PPar{P_4\!\Set[]{ \Subst{\Chan[k]}{\Args} }}{P_5}}{\tilde{\Args}_2}{\tilde{\Args[v]}_2}{P_6}} \triangleright \Delta_{\ECR', P4-6}' $, we have
			\begin{align*}
				\Gamma' \vdash \PPar{\AECR{\POpt{\Role_1}{\tilde{\Role}}{\PPar{P_1}{P_2}}{\tilde{\Args}_1}{\tilde{\Args[v]}_1}{P_3}}}{\AECR[E']{\POpt{\Role_2}{\tilde{\Role}}{\PPar{P_4\!\Set[]{ \Subst{\Chan[k]}{\Args} }}{P_5}}{\tilde{\Args}_2}{\tilde{\Args[v]}_2}{P_6}}} \triangleright \; \Delta_{\ECR, P1-3}' \otimes \Delta_{\ECR', P4-6}'
			\end{align*}
			With Lemma~\ref{lem:typeEC}~(2) we conclude with $ \Gamma \vdash P' \triangleright \left( \Delta_{\ECR, P1-3}' \otimes \Delta_{\ECR', P4-6}' \right) \otimes \Delta_{\EC} $.

			It remains to show that $ \Delta_{\EC} \otimes \Delta_P \mapsto \left( \Delta_{\ECR, P1-3}' \otimes \Delta_{\ECR', P4-6}' \right) \otimes \Delta_{\EC} $.
			Because $ \Delta_P = \left( \Delta_{\ECR} \otimes \Delta_{P1-3} \right) \otimes \left( \Delta_{\ECR'} \otimes \Delta_{P4-6} \right) $ and because $ \Delta_{P1-3}' $ and $ \Delta_{P4-6}' $ are obtained from $ \Delta_{P1-3} $ and $ \Delta_{P4-6} $ by
			\begin{itemize}
				\item changing $ \Typed{\ATI{\Chan[k]}{\Role_5}}{T_5} $ to $ \Typed{\AT{\Chan[k]}{\Role_5}}{T_5} $,
				\item reducing $ \Typed{\AT{\Chan[s]}{\Role_3}}{\LTReq{\Prot}{\Role_5}{\tilde{\Args[z]}_1}{\Role_4}{T_1}} $ to $ \Typed{\AT{\Chan[s]}{\Role_3}}{T_1} $, and
				\item reducing $ \Typed{\AT{\Chan[s]}{\Role_4}}{\LTEnt{\Prot}{\Role_5}{\tilde{\Args[z]}_2}{\Role_3}{T_4}} $ to $ \Typed{\AT{\Chan[s]}{\Role_4}}{T_4} $
			\end{itemize}
			we have
			\begin{align*}
				\dfrac{\dfrac{\dfrac{}{\left( \Delta_{P1-2} \otimes \Delta_{P3}, \Typed{\AT{\Chan[k]_1}{\Role_1}}{T_1} \right) \otimes \left( \Delta_{P4-5} \otimes \Delta_{P6}, \Typed{\AT{\Chan[k]_2}{\Role_2}}{T_2} \right) \mapsto \Delta_{P1-3}' \otimes \Delta_{P4-6}'} (\mathsf{join}')}{\Delta_{P1-3} \otimes \Delta_{P4-6} \mapsto \Delta_{P1-3}' \otimes \Delta_{P4-6}'} (\mathsf{optCom})}{\Delta_{\EC} \otimes \Delta_P \mapsto \left( \Delta_{\ECR, P1-3}' \otimes \Delta_{\ECR', P4-6}' \right) \otimes \Delta_{\EC}} (\mathsf{par})
			\end{align*}
		\item[Case $ (\mathsf{cSO}) $:]
			In this case we have
			\begin{align*}
				P & = \AEC{\PPar{\AECR{P_{1-3}}}{\AECR[E']{P_{4-6}}}}\\
				P_{1-3} & = \POpt{\Role_1}{\tilde{\Role}}{\PPar{\PSend{\Chan[k]}{\Role_3}{\Role_4}{\Labe_j}{\tilde{\Args[v]}}{P_1}}{P_2}}{\tilde{\Args}_1}{\tilde{\Args[v]}_1}{P_3}\\
				P_{4-6} & = \POpt{\Role_2}{\tilde{\Role}}{\PPar{\PGet{\Chan[k]}{\Role_3}{\Role_4}{_{i \in \indexSet} \Set{ \PLab{\Labe_i}{\tilde{\Args}_i}{P_{4, i}} }}}{P_5}}{\tilde{\Args}_2}{\tilde{\Args[v]}_2}{P_6}\\
				P' & = \AEC{\PPar{\AECR{\POpt{\Role_1}{\tilde{\Role}}{\PPar{P_1}{P_2}}{\tilde{\Args}_1}{\tilde{\Args[v]}_1}{P_3}}}{\AECR[E']{\POpt{\Role_2}{\tilde{\Role}}{\PPar{P_{4, j}\!\Set[]{ \Subst{\tilde{\Args[v]}}{\tilde{\Args}_j} }}{P_5}}{\tilde{\Args}_2}{\tilde{\Args[v]}_2}{P_6}}}}
			\end{align*}
			With $ \Gamma \vdash P \triangleright \Delta $ and Lemma~\ref{lem:typeEC}~(1), there exist $ \Delta_P, \Delta_{\EC}, \Gamma' $ such that $ \Gamma \subseteq \Gamma' $, $ \Delta = \Delta_{\EC} \otimes \Delta_P $, and $ \Gamma' \vdash \PPar{\AECR{P_{1-3}}}{\AECR[E']{P_{4-6}}} \triangleright \Delta_P $.
			By the rules in Figure~\ref{fig:typingRules} the proof of the judgement has to start (modulo Rule~(\textsf{S2})) with Rule~$ (\mathsf{Pa}) $, that splits $ \Delta_P $ such that $ \Delta_P = \Delta_{\ECR, P1-3} \otimes \Delta_{\ECR', P4-6} $. Again by Lemma~\ref{lem:typeEC}~(1), there exist $ \Delta_{P1-3} $, $ \Delta_{P4-6} $, $ \Delta_{\ECR} $, $ \Delta_{\ECR'} $, $ \Gamma_1 $, and $ \Gamma_2 $ such that $ \Gamma' \subseteq \Gamma_1 $, $ \Gamma' \subseteq \Gamma_2 $, $ \Delta_{\ECR, P1-3} = \Delta_{\ECR} \otimes \Delta_{P1-3} $, $ \Delta_{\ECR', P4-6} = \Delta_{\ECR'} \otimes \Delta_{P4-6} $, $ \Gamma_1 \vdash P_{1-3} \triangleright \Delta_{P1-3} $, and $ \Gamma_2 \vdash P_{4-6} \triangleright \Delta_{P4-6} $. Then:
			\begin{align*}
				\dfrac{\dfrac{\begin{array}{l} \dfrac{\Gamma_1 \vdash P_1 \triangleright \Delta_{P_1}, \Typed{\AT{\Chan[k]}{\Role^3}}{T_{1, j}} \quad \vdash \Typed{\tilde{\Args[v]}}{\tilde{\Sort}_j'}}{\Gamma_1 \vdash \PSend{\Chan[k]}{\Role_3}{\Role_4}{\Labe_j}{\tilde{\Args[v]}}{P_1} \triangleright \Delta_{P1}'}(\mathsf{S}) \quad \Gamma_1 \vdash P_2 \triangleright \Delta_{P2} \end{array}}{\Gamma_1 \vdash \PPar{\PSend{\Chan[k]}{\Role_3}{\Role_4}{\Labe_j}{\tilde{\Args[v]}}{P_1}}{P_2} \triangleright \Delta_{P1}' \otimes \Delta_{P2}}(\mathsf{Pa}) \begin{array}{l} \Gamma_1 \vdash P_3 \triangleright \Delta_{P3}, \Typed{\AT{\Chan[k]_1}{\Role_1}}{T_1'}\\ \vdash \Typed{\tilde{\Args}_1}{\tilde{\Sort}} \quad \vdash \Typed{\tilde{\Args[v]}_1}{\tilde{\Sort}} \end{array}}{\Gamma_1 \vdash \POpt{\Role_1}{\tilde{\Role}}{\PPar{\PSend{\Chan[k]}{\Role_3}{\Role_4}{\Labe_j}{\tilde{\Args[v]}}{P_1}}{P_2}}{\tilde{\Args}_1}{\tilde{\Args[v]}_1}{P_3} \triangleright \Delta_{P1-3}}(\mathsf{Opt})
			\end{align*}
			where $ \nexists \Role', \tilde{\Sort[K]} \logdot \Typed{\Role'}{\OV{\tilde{\Sort[K]}}} \in \Delta_{P1-2} $ and
			\begin{align*}
				\Delta_{P1}' & = \Delta_{P1}, \Typed{\AT{\Chan[k]}{\Role_3}}{\LTSend{\Role_4}{_{i \in \indexSet} \Set{ \LTLab{\Labe_i}{\Typed{\tilde{\Args[z]}_i}{\tilde{\Sort}_i'}}{T_{1, i}} }}}\\
				\Delta_{P1}' \otimes \Delta_{P2} & = \Delta_{P1-2}, \Typed{\AT{\Chan[k]_1}{\Role_1}}{T_1}, \Typed{\Role_1}{\OV{\tilde{\Sort}}}\\
				\Delta_{P1-3} & = \Delta_{P1-2} \otimes \Delta_{P3}, \Typed{\AT{\Chan[k]_1}{\Role_1}}{\LTOpt{\tilde{\Role}}{T_1}{\Typed{\tilde{\Args[y]}_1}{\tilde{\Sort}}}{T_1'}}
			\end{align*}
			Since $ \Delta_{P1}' \otimes \Delta_{P2} $ is defined and because $ \Chan[s] \neq \Chan[k] $, we obtain
			\begin{align*}
				\dfrac{\dfrac{\begin{array}{l} \Gamma' \vdash P_1 \triangleright \Delta_{P1}, \Typed{\AT{\Chan[k]}{\Role_3}}{T_{1, j}} \quad \Gamma' \vdash P_2 \triangleright \Delta_{P2} \end{array}}{\Gamma' \vdash \PPar{P_1}{P_2} \triangleright \left( \Delta_{P1}, \Typed{\AT{\Chan[k]}{\Role_3}}{T_{1, j}} \right) \otimes \Delta_{P2}}(\mathsf{Pa}) \quad \begin{array}{l} \Gamma' \vdash P_3 \triangleright \Delta_{P3}, \Typed{\AT{\Chan[k]_1}{\Role_1}}{T_1'}\\ \vdash \Typed{\tilde{\Args}_1}{\tilde{\Sort}} \quad \vdash \Typed{\tilde{\Args[v]}_1}{\tilde{\Sort}} \end{array}}{\Gamma' \vdash \POpt{\Role_1}{\tilde{\Role}}{\PPar{P_1}{P_2}}{\tilde{\Args}_1}{\tilde{\Args[v]}_1}{P_3} \triangleright \Delta_{P1-3}'}(\mathsf{Opt})
			\end{align*}
			where $ \left( \Delta_{P1}, \Typed{\AT{\Chan[k]}{\Role_3}}{T_{1, j}} \right) \otimes \Delta_{P2} = \Delta_{P1-2}', \Typed{\AT{\Chan[k]_1}{\Role_1}}{T_1}, \Typed{\Role_1}{\OV{\tilde{\Sort}}} $ and
			\begin{align*}
				\Delta_{P1-3}' = \Delta_{P1-2}' \otimes \Delta_{P3}, \Typed{\AT{\Chan[k]_1}{\Role_1}}{\LTOpt{\tilde{\Role}}{T_1}{\Typed{\tilde{\Args[y]}_1}{\tilde{\Sort}}}{T_1'}}
			\end{align*}
			Note that $ \Delta_{P1-3}' $ is obtained from $ \Delta_{P1-3} $ by reducing a capability on $ \AT{\Chan[k]}{\Role_3} $.
			With Lemma~\ref{lem:typeEC}~(2), then $ \Gamma' \vdash \AECR{\POpt{\Role_1}{\tilde{\Role}}{\PPar{P_1}{P_2}}{\tilde{\Args}_1}{\tilde{\Args[v]}_1}{P_3}} \triangleright \Delta_{\ECR, P1-3}' $, where $ \Delta_{\ECR, P1-3}' = \Delta_{P1-3}' \otimes \Delta_{\ECR} $.

			Moreover
			\begin{align*}
				\dfrac{\begin{array}{l} \dfrac{\begin{array}{l} \dfrac{\left( \Gamma_2 \vdash P_{4, i} \triangleright \Delta_{P4}, \Typed{\AT{\Chan[k]}{\Role_4}}{T_{4, i}} \quad \vdash \Typed{\tilde{\Args}_i}{\tilde{\Sort}_i'} \right)_{i \in \indexSet}}{\Gamma_2 \vdash \PGet{\Chan[k]}{\Role_3}{\Role_4}{_{i \in \indexSet} \Set{ \PLab{\Labe_i}{\tilde{\Args}_i}{P_{4, i}} }} \triangleright \Delta_{P4}'}(\mathsf{C}) \quad \Gamma_2 \vdash P_5 \triangleright \Delta_{P5} \end{array}}{\Gamma_2 \vdash \PPar{\PGet{\Chan[k]}{\Role_3}{\Role_4}{_{i \in \indexSet} \Set{ \PLab{\Labe_i}{\tilde{\Args}_i}{P_{4, i}} }}}{P_5} \triangleright \Delta_{P4}' \otimes \Delta_{P5}}(\mathsf{Pa})\\ \Gamma_2 \vdash P_6 \triangleright \Delta_{P6}, \Typed{\AT{\Chan[k]_2}{\Role_2}}{T_2'} \quad \vdash \Typed{\tilde{\Args}_2}{\tilde{\Sort}''} \quad \vdash \Typed{\tilde{\Args[v]}_2}{\tilde{\Sort}''} \end{array}}{\Gamma_2 \vdash \POpt{\Role_2}{\tilde{\Role}}{\PPar{\PGet{\Chan[k]}{\Role_3}{\Role_4}{_{i \in \indexSet} \Set{ \PLab{\Labe_i}{\tilde{\Args}_i}{P_{4, i}} }}}{P_5}}{\tilde{\Args}_2}{\tilde{\Args[v]}_2}{P_6} \triangleright \Delta_{P4-6}}(\mathsf{Opt})
			\end{align*}
			where $ \nexists \Role', \tilde{\Sort[K]} \logdot \Typed{\Role'}{\OV{\tilde{\Sort[K]}}} \in \Delta_{P4-5} $ and
			\begin{align*}
				\Delta_{P4}' & = \Delta_{P4}, \Typed{\AT{\Chan[k]}{\Role_4}}{\LTGet{\Role_3}{_{i \in \indexSet} \Set{ \LTLab{\Labe_i}{\Typed{\tilde{\Args[z]}'}{\tilde{\Sort}_i'}}{T_{4, i}} }}}\\
				\Delta_{P4}' \otimes \Delta_{P5} & = \Delta_{P4-5}, \Typed{\AT{\Chan[k]_2}{\Role_2}}{T_2}, \Typed{\Role_2}{\OV{\tilde{\Sort}''}}\\
				\Delta_{P4-6} & = \Delta_{P4-5} \otimes \Delta_{P6}, \Typed{\AT{\Chan[k]_2}{\Role_2}}{\LTOpt{\tilde{\Role}}{T_2}{\Typed{\tilde{\Args[y]}_2}{\tilde{\Sort}''}}{T_2'}}
			\end{align*}
			By $ \Gamma_2 \vdash P_{4, j} \triangleright \Delta_{P4}, \Typed{\AT{\Chan[k]}{\Role_4}}{T_{4, j}} $, $ \vdash \Typed{\tilde{\Args[v]}}{\tilde{\Sort}_j'} $, $ \vdash \Typed{\tilde{\Args[x]}_j}{\tilde{\Sort}_j'} $, and Lemma~\ref{lem:typeSubstB}, we have $ \Gamma_2 \vdash P_{4, j}\!\Set[]{ \Subst{\tilde{\Args[v]}}{\tilde{\Args}_j} } \triangleright \Delta_{P4}, \Typed{\AT{\Chan[k]}{\Role_4}}{T_{4, j}} $.
			Then
			\begin{align*}
				\dfrac{\dfrac{\begin{array}{l} \Gamma_2 \vdash P_{4, j}\!\Set[]{ \Subst{\tilde{\Args[v]}}{\tilde{\Args}_j} } \triangleright \Delta_{P4}, \Typed{\AT{\Chan[k]}{\Role_4}}{T_{4, j}}\\ \Gamma_2 \vdash P_5 \triangleright \Delta_{P5} \end{array}}{\Gamma_2 \vdash \PPar{P_4\!\Set[]{ \Subst{\tilde{\Args[v]}}{\tilde{\Args}_j} }}{P_5} \triangleright \Delta_{P4 - 5}''}(\mathsf{Pa}) \begin{array}{l} \Gamma_2 \vdash P_6 \triangleright \Delta_{P6}, \Typed{\AT{\Chan[k]_2}{\Role_2}}{T_2'}\\ \vdash \Typed{\tilde{\Args}_2}{\tilde{\Sort}''} \quad \vdash \Typed{\tilde{\Args[v]}_2}{\tilde{\Sort}''} \end{array}}{\Gamma_2 \vdash \POpt{\Role_2}{\tilde{\Role}}{\PPar{P_4\!\Set[]{ \Subst{\tilde{\Args[v]}}{\tilde{\Args}_j} }}{P_5}}{\tilde{\Args}_2}{\tilde{\Args[v]}_2}{P_6} \triangleright \Delta_{P4-6}'}(\mathsf{Opt})
			\end{align*}
			where $ \nexists \Role', \tilde{\Sort[K]} \logdot \Typed{\Role'}{\OV{\tilde{\Sort[K]}}} \in \Delta_{P4-5}''' $ and
			\begin{align*}
				\Delta_{P4-5}'' & = \left( \Delta_{P4}, \Typed{\AT{\Chan[k]}{\Role_4}}{T_{4, j}} \right) \otimes \Delta_5 = \Delta_{P4-5}''', \Typed{\AT{\Chan[k]_2}{\Role_2}}{T_2}, \Typed{\Role_2}{\OV{\tilde{\Sort}''}}\\
				\Delta_{P4-6}' & = \Delta_{P4-5}''' \otimes \Delta_{P6}, \Typed{\AT{\Chan[k]_2}{\Role_2}}{\LTOpt{\tilde{\Role}}{T_2}{\Typed{\tilde{\Args[y]}_2}{\tilde{\Sort}'}}{T_2'}}
			\end{align*}
			Hence $ \Delta_{P4-6}' $ is obtained from $ \Delta_{P4-6} $ by changing a capabilities for $ \Role_4 $.
			With Lemma~\ref{lem:typeEC}~(2), then $ \Gamma' \vdash \AECR[E']{\POpt{\Role_2}{\tilde{\Role}}{\PPar{P_4\!\Set[]{ \Subst{\tilde{\Args[v]}}{\tilde{\Args}_j} }}{P_5}}{\tilde{\Args}_2}{\tilde{\Args[v]}_2}{P_6}} \triangleright \Delta_{\ECR', P4-6}' $, where $ \Delta_{\ECR', P4-6}' = \Delta_{P4-6}' \otimes \Delta_{\ECR'} $.

			Since $ \Delta_{\ECR, P1-3} \otimes \Delta_{\ECR', P4-5} $ is defined, so is $ \Delta_{\ECR, P1-3}' \otimes \Delta_{\ECR', P4-6}' $.
			Hence, by Rule~$ (\mathsf{Pa}) $, the judgement $ \Gamma' \vdash \AECR{\POpt{\Role_1}{\tilde{\Role}}{\PPar{P_1}{P_2}}{\tilde{\Args}_1}{\tilde{\Args[v]}_1}{P_3}} \triangleright \Delta_{\ECR, P1-3}' $, and $ \Gamma' \vdash \AECR[E']{\POpt{\Role_2}{\tilde{\Role}}{\PPar{P_4\!\Set[]{ \Subst{\tilde{\Args[v]}}{\tilde{\Args}_j} }}{P_5}}{\tilde{\Args}_2}{\tilde{\Args[v]}_2}{P_6}} \triangleright \Delta_{\ECR', P4-6}' $, we have
			\begin{align*}
				\Gamma' \vdash \PPar{\AECR{\POpt{\Role_1}{\tilde{\Role}}{\PPar{P_1}{P_2}}{\tilde{\Args}_1}{\tilde{\Args[v]}_1}{P_3}}}{\AECR[E']{\POpt{\Role_2}{\tilde{\Role}}{\PPar{P_4\!\Set[]{ \Subst{\tilde{\Args[v]}}{\tilde{\Args}_j} }}{P_5}}{\tilde{\Args}_2}{\tilde{\Args[v]}_2}{P_6}}} \triangleright \Delta_{\ECR, P1-3}' \otimes \Delta_{\ECR', P4-6}'
			\end{align*}
			With Lemma~\ref{lem:typeEC}~(2) we conclude with $ \Gamma \vdash P' \triangleright \left( \Delta_{\ECR, P1-3}' \otimes \Delta_{\ECR', P4-6}' \right) \otimes \Delta_{\EC} $.

			It remains to show that $ \Delta_{\EC} \otimes \Delta_P \mapsto \left( \Delta_{\ECR, P1-3}' \otimes \Delta_{\ECR', P4-6}' \right) \otimes \Delta_{\EC} $.
			Because $ \Delta_P = \left( \Delta_{\ECR} \otimes \Delta_{P1-3} \right) \otimes \left( \Delta_{\ECR'} \otimes \Delta_{P4-6} \right) $ and because $ \Delta_{P1-3}' $ and $ \Delta_{P4-6}' $ are obtained from $ \Delta_{P1-3} $ and $ \Delta_{P4-6} $ by
			\begin{itemize}
				\item reducing $ \Typed{\AT{\Chan[k]}{\Role_3}}{\LTSend{\Role_4}{_{i \in \indexSet} \Set{ \LTLab{\Labe_i}{\Typed{\tilde{\Args[z]}_i}{\tilde{\Sort}_i'}}{T_{1, i}} }}} $ to $ \Typed{\AT{\Chan[k]}{\Role_3}}{T_{1, j}} $, and
				\item reducing $ \Typed{\AT{\Chan[k]}{\Role_4}}{\LTGet{\Role_3}{_{i \in \indexSet} \Set{ \LTLab{\Labe_i}{\Typed{\tilde{\Args[z]}'}{\tilde{\Sort}_i'}}{T_{4, i}} }}} $ to $ \Typed{\AT{\Chan[k]}{\Role_4}}{T_{4, j}} $
			\end{itemize}
			we have
			\begin{align*}
				\dfrac{\dfrac{\dfrac{}{\left( \Delta_{P1-2} \otimes \Delta_{P3}, \Typed{\AT{\Chan[k]_1}{\Role_1}}{T_1} \right) \otimes \left( \Delta_{P4-5} \otimes \Delta_{P6}, \Typed{\AT{\Chan[k]_2}{\Role_2}}{T_2} \right) \mapsto \Delta_{P1-3}' \otimes \Delta_{P4-6}'} (\mathsf{comS}')}{\Delta_{P1-3} \otimes \Delta_{P4-6} \mapsto \Delta_{P1-3}' \otimes \Delta_{P4-6}'} (\mathsf{optCom})}{\Delta_{\EC} \otimes \Delta_P \mapsto \left( \Delta_{\ECR, P1-3}' \otimes \Delta_{\ECR', P4-6}' \right) \otimes \Delta_{\EC}} (\mathsf{par})
			\end{align*}
	\end{description}
	To obtain the proof for the smaller type system we simply omit the Cases~$ (\mathsf{subs}) $, $ (\mathsf{join}) $, $ (\mathsf{cSO}) $, $ (\mathsf{cCO}) $, and $ (\mathsf{jO}) $. This is possible, because no other case relies on one of the Rules~$ (\mathsf{P}) $, $ (\mathsf{J}) $, or $ (\mathsf{New}) $.
\end{proof}

\subsection{Progress and Completion}

Apart from subject reduction we are interested in progress and completion.
Following \cite{DemangeonHonda12} we use coherence to prove progress and completion.
A session environment is \emph{coherent} if it is composed of the projections of well-formed global types with global types for all external invitations (also guarded once). In other words if the session environment is coherent we can use the projection rules in the reversed direction to reconstruct complete global types. In particular coherence ensures that in the case of a communication from $ \Role_1 $ to $ \Role_2 $ on a channel $ \Args $ the session environment maps the type of the sender to $ \AT{\Args}{\Role_1} $ and the type of the receiver to $ \AT{\Args}{\Role_2} $ (or vice versa). This also ensures that the type of the transmitted value and the type of the received value have to correspond and that for each sender there is the matching receiver and vice versa.

Most of the reduction rules preserve coherence.
Only the rules to call a sub-session and to handle its internal and external invitations as well as the failing of optional blocks can temporary invalidate this property.
By removing the protocol call and a strict subset of these internal and external invitations, we obtain a process and a corresponding session type that does not directly result from the projection of a global type, since it neither refers to the session initialisation containing all internal and external invitations nor to the global type of the content of this sub-session without open invitations.
A failing optional block is not a problem for the process itself, because the continuation of the process is instantiated with the default value and this process with a corresponding session environment correspond to the projection of the global type of the continuation.
But a failing optional block may cause another part of the network, \ie a parallel process, to lose coherence.
If another, parallel optional block is waiting for a communication with the former, it is doomed to fail.
This situation of a single optional block without its dual communication partner cannot result from the projection of a global type.
Due to the interleaving of steps, an execution starting in a process with a coherent session environment may lead to a state in which there are open internal and external invitations for several different protocols and/or several single optional blocks at the same time. However, coherence ensures that for all such reachable processes there is a finite sequence of steps that restores coherence and thus ensures progress and completion.

The rules of Figure~\ref{fig:sessionTypeReductions} allow to restore coherence.

\begin{lemma}
	\label{lem:coherence}
	For both type systems:\\
	If $ \Delta $ is coherent and $ \Delta \mapsto \Delta' $ then there exists $ \Delta'' $ such that $ \Delta' \mapsto^* \Delta'' $ and $ \Delta'' $ is coherent.
\end{lemma}

\begin{proof}
	Again we consider the larger type system first.
	The proof is by induction on the rules that are necessary to derive $ \Delta \mapsto \Delta' $. Here most of the cases refer to base cases; only the rules~$ (\mathsf{choice}') $, $ (\mathsf{opt}) $, $ (\mathsf{optCom}) $, and $ (\mathsf{par}) $ refer to induction steps.
	\begin{description}
		\item[Case $ (\mathsf{comS}') $:]
			In this case $ \Delta $ contains two type statements for a channel $ \Chan[k] $ on two different roles $ \Role_1 $ and $ \Role_2 $:
			\begin{align*}
				\Typed{\AT{\Chan[k]}{\Role_1}}{\LTSend{\Role_2}{_{i \in \indexSet} \Set{ \LTLab{\Labe_i}{\Typed{\tilde{\Args}_i}{\tilde{\Sort}_i}}{T_i} }}}, \Typed{\AT{\Chan[k]}{\Role_2}}{\LTGet{\Role_1}{_{i \in \indexSet{}} \Set{ \LTLab{\Labe_i}{\Typed{\tilde{\Args}_i'}{\tilde{\Sort}_i}}{T_i'} }}}
			\end{align*}
			Since $ \Delta $ is coherent and cannot contain other type statements for $ \AT{\Chan[k]}{\Role_1} $ or $ \AT{\Chan[k]}{\Role_2} $, these two local types have to be the result of the projection of a single global type describing a communication from $ \Role_1 $ to $ \Role_2 $ on channel $ \Chan[k] $.
			Moreover the possible continuations of this global type are projected into the pairs of local types $ T_i $ and $ T_i' $ such that for all $ i \in \indexSet $ the combination of $ \Typed{\AT{\Chan[k]}{\Role_1}}{T_i} $ and $ \Typed{\AT{\Chan[k]}{\Role_2}}{T_i'} $ is the result of the projection of the respective continuation of the global type.
			Because of this, $ \Delta' $ (the two type statements are replaced by $ \Typed{\AT{\Chan[k]}{\Role_1}}{T_j}, \Typed{\AT{\Chan[k]}{\Role_2}}{T_j'} $) is coherent.
		\item[Case $ (\mathsf{choice}') $:]
			In this case we have $ \Delta_1, \Typed{\AT{\Chan[s]}{\Role}}{T_i} \mapsto \Delta_1', \Typed{\AT{\Chan[s]}{\Role}}{T_i'} $ for some $ i \in \Set[]{1, 2} $ and $ \Delta = \Delta_1, \Typed{\AT{\Chan[s]}{\Role}}{T_1 \oplus T_2} $.
			By the induction hypothesis and $ \Delta_1, \Typed{\AT{\Chan[s]}{\Role}}{T_i} \mapsto \Delta_1', \Typed{\AT{\Chan[s]}{\Role}}{T_i'} $, the resulting $ \Delta_1', \Typed{\AT{\Chan[s]}{\Role}}{T_i'} $ is coherent for both instantiations of $ i $.
		\item[Case $ (\mathsf{comC}') $:]
			In this case $ \Delta = \Delta_1, \Typed{\ATE{\Chan[s]}{\Role}}{T} $ and $ \Delta' = \Delta_1, \Typed{\AT{\Chan[s]}{\Role}}{\Role} $.
			We observe that this rule does not change the types, but only lifts the status of $ \Typed{\ATE{\Chan[s]}{\Role}}{T} $ from 'needs to be invited with type $ T $' to 'is present'.
			However, because of the open external invitation $ \Typed{\ATE{\Chan[s]}{\Role}}{T} $, the session environment $ \Delta $ is not coherent and thus the implication holds trivially.
			We need this rule to restore coherence in the Case~$ (\mathsf{subs}') $.
		\item[Case $ (\mathsf{join}') $:]
			In this case $ \Delta = \Delta_1, \Typed{\AT{\Chan[s]}{\Role_1}}{\LTReq{\Prot}{\Role_3}{\tilde{\Args[v]}}{\Role_2}{T_1}}, \Typed{\ATI{\Chan[k]}{\Role_3}}{T_3}, \Typed{\AT{\Chan[s]}{\Role_2}}{\LTEnt{\Prot}{\Role_3}{\tilde{\Args[v]}'}{\Role_1}{T_2}} $ and $ \Delta' = \Delta_1, \Typed{\AT{\Chan[s]}{\Role_1}}{T_1}, \Typed{\AT{\Chan[s]}{\Role_2}}{T_2}, \Typed{\AT{\Chan[k]}{\Role_3}}{T_3} $.
			An internal invitation is accepted by reducing the corresponding request $ \mathtt{req} $ and its acceptance notification $ \mathtt{ent} $, and by lifting the status of $ \Typed{\ATI{\Chan[k]}{\Role_3}}{T_3} $ from 'needs to be invited with type $ T_3 $' to 'is present' ($ \Typed{\AT{\Chan[k]}{\Role_3}}{T_3} $).
			Again the session environment $ \Delta $ is not coherent, because of $ \Typed{\ATI{\Chan[k]}{\Role_3}}{T_3} $, and thus the implication holds trivially.
			We need this rule to restore coherence in the Case~$ (\mathsf{subs}') $.
		\item[Case $ (\mathsf{subs}') $:] In this case the $ \Typed{\AT{\Chan[s]}{\Role''}}{\LTCall{\Prot}{G}{\tilde{\Args[v]}}{\Typed{\tilde{\Args[y]}}{\tilde{\Sort}}}{\tilde{\Role}'}{T}} $ of $ \Delta $ is reduced to $ \Typed{\AT{\Chan[s]}{\Role''}}{T} $ in $ \Delta' $ and the statements $ \Typed{\ATI{\Chan[k]}{\Role_1}}{T'_1}, \ldots, \Typed{\ATI{\Chan[k]}{\Role_n}}{T'_n}, \Typed{\ATE{\Chan[k]}{\Role'_1}}{T'_{n + 1}}, \ldots, \Typed{\ATE{\Chan[k]}{\Role'_m}}{T'_{n + m}} $ are added to $ \Delta' $.
			Since $ \Delta $ is coherent $ \LTCall{\Prot}{G}{\tilde{\Args[v]}}{\Typed{\tilde{\Args[y]}}{\tilde{\Sort}}}{\tilde{\Role}'}{T} $ results from the projection of the global type for the declaration (with $ \mathtt{let} $) of $ \Prot $ and its call $ \GTCall{\Role_A}{\Prot}{\tilde{\Role}}{\tilde{\Args[y]}}{G} $.
			The the statements $ \Typed{\ATI{\Chan[k]}{\Role_1}}{T'_1}, \ldots, \Typed{\ATI{\Chan[k]}{\Role_n}}{T'_n}, \Typed{\ATE{\Chan[k]}{\Role'_1}}{T'_{n + 1}}, \ldots, \Typed{\ATE{\Chan[k]}{\Role'_m}}{T'_{n + m}} $ refer to the open internal and external invitations.
			Because of these statements, $ \Delta' $ is not coherent but we can restore coherence by accepting all open invitations, \ie by moving to the projection of $ G $ the global type for the continuation of the call $ \GTCall{\Role_A}{\Prot}{\tilde{\Role}}{\tilde{\Args[y]}}{G} $.
			
			The open internal invitations $ \Typed{\ATI{\Chan[k]}{\Role_i}}{T'_i} $ are handled by requests $ \mathtt{req} $ and acceptance notifications $ \mathtt{ent} $ that result from the projection of the call $ \GTCall{\Role_A}{\Prot}{\tilde{\Role}}{\tilde{\Args[y]}}{G} $. These are unguarded by Rule~$ (\mathsf{subs}) $ in the type judgement.
			Since $ \Delta $ is coherent and because all internal invitations as well as their acceptance notifications are generated by the same projection of the call, $ \Delta $ has to contain exactly one pair $ \Typed{\AT{\Chan[s]}{\Role_i'}}{\LTReq{\Prot}{\Role_i}{\tilde{\Args[v]}}{\Role_i''}{T_i''}}, \Typed{\AT{\Chan[s]}{\Role_i''}}{\LTEnt{\Prot}{\Role_i}{\tilde{\Args[v]}'}{\Role_i'}{T_i'''}} $ for each $ \Typed{\ATI{\Chan[k]}{\Role_i}}{T'_i} $.
			Because of that we can reduce the open internal invitations by $ n $ applications of Rule~$ (\mathsf{join}') $.
			As result the requests and acceptance notifications are reduced to their respective continuations, and the $ \Typed{\ATI{\Chan[k]}{\Role_i}}{T'_i} $ are turned into $ \Typed{\AT{\Chan[k]}{\Role_i}}{T'_i} $.
			Accordingly the $ n $ applications of Rule~$ (\mathsf{join}') $ lead to $ \Delta \mapsto^n \Delta_1 $, where $ \Delta_1 $ is obtained from $ \Delta $ by replacing $ \Typed{\AT{\Chan[s]}{\Role''}}{\LTCall{\Prot}{G}{\tilde{\Args[v]}}{\Typed{\tilde{\Args[y]}}{\tilde{\Sort}}}{\tilde{\Role}'}{T}} $ and the corresponding $ n - 1 $ parallel acceptance notifications $ \Typed{\AT{\Chan[s]}{\Role_i''}}{\LTEnt{\Prot}{\Role_i}{\tilde{\Args[v]}'}{\Role_i'}{T_i'''}} $ by $ T' $ and $ T_i''' $, where $ T' $ is obtained from $ T $ by replacing the corresponding requests and the acceptance notification of the caller by their continuations.
			The remaining open external invitations are the only reason that prevents $ \Delta_1 $ from being coherent.

			The open external invitations $ \Typed{\ATE{\Chan[k]}{\Role'_j}}{T'_{n + j}} $ are accepted with $ m $ applications of Rule~$ (\mathsf{comC}') $ (which does not influence other parts of the session environments and also does not require other parts of $ \Delta_1 $ to contain specific local types).
			As result the $ \Typed{\ATE{\Chan[k]}{\Role'_j}}{T'_{n + j}} $ are turned into $ \Typed{\AT{\Chan[k]}{\Role'_j}}{T'_{n + j}} $ that correspond to the projection of $ m $ global types on the respective roles $ \tilde{\Role}' $ for the sub-session $ \Chan[k] $. To obtain the global type that restores coherence, these $ m $ global types of the external communication partners are placed in parallel to the global type of the continuation $ G $.
		\item[Case $ (\mathsf{opt}') $:]
			In this case we have $ \Delta_1, \Typed{\AT{\Chan[s]}{\Role_1}}{T_1} \mapsto \Delta_1', \Typed{\AT{\Chan[s]}{\Role_1}}{T_1'} $,
			\begin{align*}
				\Delta = \Delta_1, \Typed{\AT{\Chan[s]}{\Role_1}}{\LTOpt{\tilde{\Role}}{T_1}{\Typed{\tilde{\Args[y]}}{\tilde{\Sort}}}{T_2}} \quad \text{ and } \quad \Delta' = \Delta_1', \Typed{\AT{\Chan[s]}{\Role_1}}{\LTOpt{\tilde{\Role}}{T_1'}{\Typed{\tilde{\Args[y]}}{\tilde{\Sort}}}{T_2}}
			\end{align*}
			Since $ \Delta_1, \Typed{\AT{\Chan[s]}{\Role_1}}{T_1} $ results from $ \Delta $ by removing an optional block and its continuation while extracting the content of the optional block and since $ \Delta $ is coherent, $ T_1 $ is the result of projecting the global type representing the content of the optional block and thus $ \Delta_1, \Typed{\AT{\Chan[s]}{\Role_1}}{T_1} $ is also coherent.
			Then, by the induction hypothesis, $ \Delta_1, \Typed{\AT{\Chan[s]}{\Role_1}}{T_1} \mapsto \Delta_1', \Typed{\AT{\Chan[s]}{\Role_1}}{T_1'} $ implies that there is some $ \Delta_1'' $ such that $ \Delta_1', \Typed{\AT{\Chan[s]}{\Role_1}}{T_1'} \mapsto^* \Delta_1'' $ and $ \Delta_1'' $ is coherent.
			Note that this sequence may reduce the local type $ T_1' $ assigned to $ \AT{\Chan[s]}{\Role_1} $ to $ T_1'' = \LTEnd $.
			By applying Rule~$ (\mathsf{opt}') $ around each step of $ \Delta_1', \Typed{\AT{\Chan[s]}{\Role_1}}{T_1'} \mapsto^* \Delta_1'' $ we obtain the derivation $ \Delta \mapsto^* \Delta_1''' $, where $ \Delta_1''' $ is obtained from $ \Delta_1'' $ by replacing $ \Typed{\AT{\Chan[s]}{\Role_1}}{T_1''} $ in $ \Delta_1'' $ by $ \Typed{\AT{\Chan[s]}{\Role_1}}{\LTOpt{\tilde{\Role}}{T_1''}{\Typed{\tilde{\Args[y]}}{\tilde{\Sort}}}{T_2}} $.
			Since $ \Delta_1'' $ is coherent, $ T_1'' $ is the result of a projection of a global type and the remaining type statements add to a coherent session environment.
			Since $ \Delta $ is coherent, $ T_2 $ is the result of a projection of a global type and the parts of $ \Delta $ that are not changed in $ \Delta \mapsto^* \Delta_1''' $ contain the dual projection of the optional block.
			Because of this, $ \Delta_1''' $ is coherent.
		\item[Case $ (\mathsf{optCom}) $:]
			In this case we have $ \Delta_1, \Typed{\AT{\Chan[s]}{\Role_1}}{T_1}, \Typed{\AT{\Chan[s]}{\Role_2}}{T_2} \mapsto \Delta_1', \Typed{\AT{\Chan[s]}{\Role_1}}{T_1'}, \Typed{\AT{\Chan[s]}{\Role_2}}{T_2'} $,
			\begin{align*}
				\Delta &= \Delta_1, \Typed{\AT{\Chan[s]}{\Role_1}}{\LTOpt{\tilde{\Role}}{T_1}{\Typed{\tilde{\Args[y]}_1}{\tilde{\Sort}_1}}{T_3}}, \Typed{\AT{\Chan[s]}{\Role_2}}{\LTOpt{\tilde{\Role}}{T_2}{\Typed{\tilde{\Args[y]}_2}{\tilde{\Sort}_2}}{T_4}} \quad \text{ and}\\
				\Delta' &= \Delta_1', \Typed{\AT{\Chan[s]}{\Role_1}}{\LTOpt{\tilde{\Role}}{T_1'}{\Typed{\tilde{\Args[y]}_1}{\tilde{\Sort}_1}}{T_3}}, \Typed{\AT{\Chan[s]}{\Role_2}}{\LTOpt{\tilde{\Role}}{T_2'}{\Typed{\tilde{\Args[y]}_2}{\tilde{\Sort}_2}}{T_4}}
			\end{align*}
			Since $ \Delta_1, \Typed{\AT{\Chan[s]}{\Role_1}}{T_1}, \Typed{\AT{\Chan[s]}{\Role_2}}{T_2} $ results from $ \Delta $ by removing two optional blocks and their continuations while extracting the content of the optional blocks and since $ \Delta $ is coherent, $ T_1 $ and $ T_2 $ are the result of projecting the global types representing the content of the optional blocks and thus $ \Delta_1, \Typed{\AT{\Chan[s]}{\Role_1}}{T_1}, \Typed{\AT{\Chan[s]}{\Role_2}}{T_2} $ is also coherent.
			Then, by the induction hypothesis, $ \Delta_1, \Typed{\AT{\Chan[s]}{\Role_1}}{T_1}, \Typed{\AT{\Chan[s]}{\Role_2}}{T_2} \mapsto \Delta_1', \Typed{\AT{\Chan[s]}{\Role_1}}{T_1'}, \Typed{\AT{\Chan[s]}{\Role_2}}{T_2'} $ implies that there is some $ \Delta_1'' $ such that $ \Delta_1', \Typed{\AT{\Chan[s]}{\Role_1}}{T_1'}, \Typed{\AT{\Chan[s]}{\Role_2}}{T_2'} \mapsto^* \Delta_1'' $ and $ \Delta_1'' $ is coherent.
			Let this sequence reduce $ T_1' $ and $ T_2' $ to $ T_1'' $ and $ T_2'' $.
			By applying Rule~$ (\mathsf{optCom}') $ around each step of $ \Delta_1', \Typed{\AT{\Chan[s]}{\Role_1}}{T_1'}, \Typed{\AT{\Chan[s]}{\Role_2}}{T_2'} \mapsto^* \Delta_1'' $ we obtain the derivation $ \Delta \mapsto^* \Delta_1''' $, where $ \Delta_1''' $ is obtained from $ \Delta_1'' $ by replacing $ \Typed{\AT{\Chan[s]}{\Role_1}}{T_1''}, \Typed{\AT{\Chan[s]}{\Role_2}}{T_2''} $ in $ \Delta_1'' $ by $ \Typed{\AT{\Chan[s]}{\Role_1}}{\LTOpt{\tilde{\Role}}{T_1''}{\Typed{\tilde{\Args[y]}_1}{\tilde{\Sort}_1}}{T_3}}, \Typed{\AT{\Chan[s]}{\Role_2}}{\LTOpt{\tilde{\Role}}{T_2''}{\Typed{\tilde{\Args[y]}_2}{\tilde{\Sort}_2}}{T_4}} $.
			Since $ \Delta_1'' $ is coherent, $ T_1'' $ and $ T_2'' $ are the result of a projection of a global type and the remaining type statements add to a coherent session environment.
			Since $ \Delta $ is coherent, $ T_3 $ and $ T_4 $ are the result of a projection of a global type.
			Because of this, $ \Delta_1''' $ is coherent.
		\item[Case $ (\mathsf{fail}') $:]
			In this case
			\begin{align*}
				\Delta = \Delta_1 \otimes \Delta_2, \Typed{\AT{\Chan[s]}{\Role_1}}{\LTOpt{\tilde{\Role}}{T_1}{\Typed{\tilde{\Args[y]}}{\tilde{\Sort}}}{T_1'}} \quad \text{ and } \quad \Delta' = \Delta_2, \Typed{\AT{\Chan[s]}{\Role_1}}{T_1'}
			\end{align*}
			and $ \Gamma \vdash P \triangleright \Delta_1, \Typed{\AT{\Chan[s]}{\Role_1}}{T}, \Typed{\Role_1}{\OV{\tilde{\Sort}}} $ for some $ \Gamma, P, \tilde{\Sort} $.
			Since $ \Delta $ is coherent, $ \Delta_2 $ contains all optional blocks with participants $ \tilde{\Role} $ that depend on the failed block.
			By one more application of Rule~$ (\mathsf{fail}') $ for each such block, we remove these optional blocks to avoid deadlocked communication attempts with the former failed block, \ie we have $ \Delta \mapsto \Delta' \mapsto^* \Delta'' $ such that $ \Delta'' $ is obtained by reducing statements of the form $ \Typed{\AT{\Chan[s]}{\Role_2}}{\LTOpt{\tilde{\Role}}{T_2}{\Typed{\tilde{\Args[y]}}{\tilde{\Sort}}}{T_2'}} $ with $ \Role_2 \in \tilde{\Role} $ in $ \Delta' $ to $ \Typed{\AT{\Chan[s]}{\Role_2}}{T_2'} $.
			Since $ \Delta $ is coherent, $ T_1' $ and all the $ T_2' $ are projections of global types for the continuations of the respective blocks.
			Because of that, $ \Delta'' $ is coherent.
		\item[Case $ (\mathsf{succ}') $:]
			In this case
			\begin{align*}
				\Delta = \Delta_1, \Typed{\AT{\Chan[s]}{\Role_1}}{\LTOpt{\tilde{\Role}}{\LTEnd}{\Typed{\tilde{\Args[y]}}{\tilde{\Sort}}}{T_1}} \quad \text{ and } \quad \Delta' = \Delta_1, \Typed{\AT{\Chan[s]}{\Role_1}}{T_1}
			\end{align*}
			Since $ \Delta $ is coherent, $ T_1 $ is a projection of a global type (for the continuation of the considered optional block).
			Applying Rule~$ (\mathsf{fail}') $ as in the last case, we reduce all optional blocks on the same participants.
			We obtain $ \Delta \mapsto \Delta' \mapsto^* \Delta'' $, where $ \Delta'' $ is obtained from $ \Delta' $ by reducing statements of the form $ \Typed{\AT{\Chan[s]}{\Role_2}}{\LTOpt{\tilde{\Role}}{T_2}{\Typed{\tilde{\Args[y]}}{\tilde{\Sort}}}{T_2'}} $ with $ \Role_2 \in \tilde{\Role} $ in $ \Delta' $ to $ \Typed{\AT{\Chan[s]}{\Role_2}}{T_2'} $.
			Since $ \Delta $ is coherent, $ T_1' $ and all the $ T_2' $ are projections of global types for the continuations of the respective blocks.
			Because of that, $ \Delta'' $ is coherent.
		\item[Case $ (\mathsf{par}) $:]
			In this case we have $ \Delta_1 \mapsto \Delta_1' $,
			\begin{align*}
				\Delta = \Delta_1 \otimes \Delta_2 \quad \text{ and } \quad \Delta' = \Delta_1' \otimes \Delta_2
			\end{align*}
			Since $ \Delta $ is coherent, either $ \Delta_1 $ is coherent or there are some optional blocks in $ \Delta_2 $ that are missing in $ \Delta_1 $ to turn it into a coherent session environment.
			In the latter case we can move the respective blocks over $ \otimes $ and, by applying Rule~$ (\mathsf{par}) $, obtain a derivation $ \Delta_3 \mapsto \Delta_3' $ such that $ \Delta_3 $ is coherent. Let $ \Delta_2' $ be the remainder of $ \Delta_2 $, \ie $ \Delta_1 \otimes \Delta_2 = \left( \Delta_1 \otimes \Delta_4 \right) \otimes \Delta_2' $ and $ \Delta_3 = \Delta_1 \otimes \Delta_4 $.
			Since we can also move $ \emptyset $ this way, the second case is more general.
			By the induction hypothesis, then there is some $ \Delta_3'' $ such that $ \Delta_3' \mapsto^* \Delta_3'' $ and $ \Delta_3'' $ is coherent.
			Since $ \Delta_3 \otimes \Delta_2' $ is defined, so is $ \Delta_3'' \otimes \Delta_2' $.
			Hence we obtain $ \Delta \mapsto \Delta' \mapsto \Delta_3'' \otimes \Delta_2' $.
			Since $ \Delta $ is coherent and $ \Delta_2' $ does not contain optional blocks with counterparts in $ \Delta_3 $, we conclude that $ \Delta_2' $ is coherent.
			With the coherence of $ \Delta_3'' $, then $ \Delta_3'' \otimes \Delta_2' $ is coherent.
	\end{description}
	For the type system without sub-sessions the Rules~$ (\mathsf{subs}') $ and $ (\mathsf{join}') $ are superfluous. Since $ (\mathsf{subs}') $ is the only case that relies on the presence of these two rules, these two cases can be removed and the statement holds for the smaller type system.
\end{proof}

Let weak coherence describe the session environments that only temporary lost coherence.
More precisely, a session environment $ \Delta $ is \emph{weakly coherent} if there is some $ \Delta' $ such that $ \Delta' $ is coherent and $ \Delta' \mapsto \Delta $.
As it can be shown easily by an induction on the rules of Figure~\ref{fig:sessionTypeReductions} and the definition of coherence, a weakly coherent session environment results from missing optional blocks for pairs of dual communication partners and/or missing $ \mathtt{call} $-type statements together with a strict subset of missing open invitations of the respective protocol.
Note that, due to the open external invitations for the parent session, all presented examples are not coherent but only weakly coherent.
Since weak coherence results from reducing a coherent session environment, we can always perform some more reductions to restore coherence.

\begin{lemma}
	\label{lem:weakCoherence}
	For both type systems:\\
	If $ \Delta $ is weakly coherent then there exists $ \Delta' $ such that $ \Delta \mapsto^* \Delta' $ and $ \Delta' $ is coherent.
\end{lemma}

\begin{proof}
	The proof for both type systems is the same except for the handling of sub-sessions and invitations that can be ignored in the simpler case.
	
	If $ \Delta $ is coherent, then choose $ \Delta' = \Delta $ and we are done.
	Otherwise, because $ \Delta $ is weakly coherent, there is some $ \Delta_0 $ such that $ \Delta_0 \mapsto^* \Delta $ and $ \Delta_0 $ is coherent.
	By recalling the proof of Lemma~\ref{lem:coherence}, then $ \Delta $ is only weakly coherent, because in comparison with $ \Delta_0 $ there are missing $ \mathtt{call} $ (due to Rule~$ (\mathsf{subs}') $) with already reduced invitations (due to the Rules~$ (\mathsf{join}') $ or $ (\mathsf{comC}') $) or missing optional blocks (due to the Rules~$ (\mathsf{fail}') $ or $ (\mathsf{succ}') $), whose counterparts are contained in $ \Delta $.
	
	For the former case, Lemma~\ref{lem:coherence} tells us that it suffices to answer the remaining invitations. Since $ \Delta_0 $ is coherent and $ \Delta_0 \mapsto^* \Delta $, all necessary internal acceptance notifications $ \mathtt{ent} $ are contained in $ \Delta $ and thus the invitations can be removed as described in Lemma~\ref{lem:coherence} in the Case~$ (\mathsf{subs}') $ using Rule~$ (\mathsf{join}') $ followed by the removal of the external invitations using Rule~$ (\mathsf{comC}') $.
	
	In the latter case, Lemma~\ref{lem:coherence} tells us that all problematic optional blocks can be removed by Rule~$ (\mathsf{fail}') $ that can be applied whenever there is an unguarded optional block.
	
	Thus, following Lemma~\ref{lem:coherence}, we can remove all problematic open invitations and optional blocks without counterparts and obtain $ \Delta \mapsto^* \Delta' $ such that $ \Delta' $ is coherent.
\end{proof}

Accordingly, our extension of the type system with optional blocks cannot cause deadlock, because optional blocks can always be aborted using Rule~$ (\mathsf{fail}) $.

Due to initial external invitations $ \PInp{\Chan[a]}{\Chan[s]}{\ldots} $ to the parent session, our examples are not coherent. Since this design decision allows for modularity using sub-sessions, we do not want to restrict our attention to coherent session environments. Instead, to better cover these cases, we relax the definition of coherence for initial session environments.
Let a session environment $ \Delta $ be \emph{initially coherent} if it is obtained from a coherent environment, \ie $ \Delta_0 \mapsto^* \Delta $ for some coherent $ \Delta_0 $, and neither contains internal open invitations nor optional blocks without their counterparts.

\emph{Progress} ensures that well-typed processes cannot get stuck unless their protocol requires them to.
In comparison to standard formulations of progress from literature and in comparison to \cite{DemangeonHonda12}, we add that the respective sequence of steps does not require any optional blocks to be unreliable.
We denote an optional block as \emph{unreliable} \wrt to a sequence of steps if it does fail within this sequence and else as \emph{reliant}.
In other words we ensure progress despite arbitrary (and any number of) failures of optional blocks.

\begin{theorem}[Progress]
	\label{thm:progress}
	For both type systems:\\
	If $ \Gamma \vdash P \triangleright \Delta $ such that $ \Delta $ is initially coherent, then either $ P = \PEnd $ or there exists $ P' $ such that $ P \longmapsto^+ P' $, $ \Gamma \vdash P' \triangleright \Delta' $, where $ \Delta \mapsto^* \Delta' $ and $ \Delta' $ is coherent, and $ P \longmapsto^+ P' $ does not require any optional block to be unreliable.
\end{theorem}

\begin{proof}
	The proof is the same for both type systems.
	Assume $ \Gamma \vdash P \triangleright \Delta $ such that $ \Delta $ is initially coherent and $ P \neq \PEnd $.

	Then, by the Lemmata~\ref{lem:coherence} and \ref{lem:weakCoherence}, we can answer all open external invitations in the sequence $ \Delta \mapsto \Delta_1 $ without Rule~$ (\mathsf{fail}') $ such that $ \Delta_1 $ is coherent.
	Because of $ \Gamma \vdash P \triangleright \Delta $ and the typing rules of Figure~\ref{fig:typingRules}, we can map this sequence to $ P \longmapsto^* P_1 $ and, by Theorem~\ref{thm:subjectReduction}, $ \Gamma \vdash P_1 \triangleright \Delta_1 $.
	Since $ \Delta \mapsto^* \Delta_1 $ does not use Rule~$ (\mathsf{fail}') $, no optional block fails in $ P \longmapsto^* P_1 $.
	
	If $ P_1 = \PEnd $ then, since $ P \neq \PEnd $, there was at least one open external invitation and thus $ P \longmapsto^+ P_1 $ and we are done.
	
	If $ P_1 \neq \PEnd $ then, because of $ \Gamma \vdash P_1 \triangleright \Delta_1 $, the projection rules in Figure~\ref{fig:projectionRules}, and since $ \Delta_1 $ is coherent, $ P_1 $ contains unguarded
	\begin{compactitem}
		\item both parts (sender and receiver) of the projection of a global type for communication,
		\item all counterparts of the projection of a global type of an optional block, or
		\item (in the case of the larger type system) all internal acceptance notifications and the call guarding internal invitations and one acceptance notification that result from the projection of a global type of a sub-session call.
	\end{compactitem}
	In all three cases, coherence and the projection rules ensure that there is at least one step to reduce $ P_1 $ in which no optional block fails, \ie there is some $ P_1' $ such that $ P_1 \longmapsto P_1' $ without Rule~$ (\mathsf{fail}) $.
	By Theorem~\ref{thm:subjectReduction}, then $ \Gamma \vdash P_1' \triangleright \Delta_1' $ for some $ \Delta_1' $ such that $ \Delta_1 \mapsto \Delta_1' $.
	Since $ P $ is initially coherent and $ P \longmapsto^+ P_1' $ does not use Rule~$ (\mathsf{fail}) $, for each optional block in $ P $ there are either all matching counterpart or $ P_1 \longmapsto P_1' $ was using Rule~$ (\mathsf{succ}) $.
	
	In the former case we can use Lemma~\ref{lem:coherence}, to obtain $ \Delta' $ without using Rule~$ (\mathsf{fail}') $ such that $ \Delta \mapsto^* \Delta_1 \mapsto \Delta_1' \mapsto^* \Delta' $ and $ \Delta' $ is coherent.
	With $ \Gamma \vdash P_1' \triangleright \Delta_1' $, the typing rules in Figure~\ref{fig:typingRules}, and Theorem~\ref{thm:subjectReduction}, then $ P \longmapsto^* P_1 \longmapsto P_1' \longmapsto^* P' $ such that $ \Gamma \vdash P' \triangleright \Delta' $ and $ P \longmapsto^+ P' $ does not require any optional block to be unreliable.
	
	In the latter case, coherence ensures that the counterparts of the successfully terminated optional block does not need to communicate with this optional block.
	By repeating the above argument for the content of the counterparts (that are by coherence obtained from a global type), where is some $ P' $ such that $ P_1' \longmapsto^* P' $ without Rule~$ (\mathsf{fail}) $ and $ \Gamma \vdash P' \triangleright \Delta_1' $ that successfully resolves the remaining counterparts such that $ P \longmapsto^* P_1 \longmapsto P_1' \longmapsto^* P' $ does not require any optional block to be unreliable, $ \Delta \mapsto^* \Delta_1 \mapsto^* \Delta_1' $ and $ \Delta_1' $ is coherent.
\end{proof}

\emph{Completion} is a special case of progress for processes without infinite recursions.
It ensures that well-typed processes, without infinite recursion or a loop resulting from calling sub-sessions infinitely often, follow their protocol and then terminate.
Similarly to progress, we prove that completion holds despite arbitrary failures of optional blocks but does not require any optional block to be unreliable.

\begin{theorem}[Completion]
	\label{thm:completion}
	For both type systems:\\
	If $ \Gamma \vdash P \triangleright \Delta $ such that $ \Delta $ is initially coherent and $ P $ does not contain infinite recursions and cannot infinitely often call a sub-session, then $ P \longmapsto^* \PEnd $, $ \Gamma \vdash \PEnd \triangleright \emptyset $, and $ P \longmapsto^* \PEnd $ does not require any optional block to be unreliable.
\end{theorem}

\begin{proof}
	By the typing rules in Figure~\ref{fig:typingRules}, $ \Gamma \vdash P' \triangleright \Delta' $ implies that $ P' = \PEnd $ if and only if $ \Delta' = \emptyset $.

	By Theorem~\ref{thm:progress}, if $ \Gamma \vdash P \triangleright \Delta $ such that $ \Delta $ is initially coherent, then either $ P = \PEnd $ or there exists $ P' $ such that $ P \longmapsto^+ P' $, $ \Gamma \vdash P' \triangleright \Delta' $, where $ \Delta \mapsto^* \Delta' $ and $ \Delta' $ is coherent, and $ P \longmapsto^+ P' $ does not require any optional block to be unreliable.

	In the first case ($ P = \PEnd $) we are done.
	Otherwise, since coherence implies initial coherence, we do perform at least one step and can apply Theorem~\ref{thm:progress} on $ \Gamma \vdash P' \triangleright \Delta' $ again.
	By repeating this argument we either construct an infinite reduction sequence or reach $ \PEnd $ after finitely many steps as required.
	Remember that we equate structural congruent session environments.
	But, since we assume that $ P $ and accordingly $ \Delta $ do not do an infinite sequence of recursions, applying structural congruence cannot increase the session environment infinitely often.

	Along with the reduction sequence for processes we construct a reduction sequence $ \Delta \mapsto^* \Delta' \mapsto^* \Delta'' \mapsto^* \ldots $.
	By inspecting the rules of Figure~\ref{fig:sessionTypeReductions}, it is easy to check that each reduction step strictly reduces the according session environment.
	Rule~$ (\mathsf{subs}') $ introduces new parts to the session environment but therefore has to reduce a $ \mathtt{call} $ in another part of the session environment.
	Since we assume that $ P $ cannot infinitely often call a sub-session, we can easily construct a potential function to prove that the session environment strictly decreases whenever no recursion is unfolded.
	Because of that and since $ \Delta $ is finite, the sequence $ \Delta \mapsto^* \Delta' \mapsto^* \Delta'' \mapsto^* \ldots $ eventually reaches $ \emptyset $, \ie $ \Delta \mapsto^* \emptyset $.
	With that we reach $ \PEnd $.
\end{proof}

\subsection{Summary}

A simple but interesting consequence of the Completion property is, that for each well-typed process there is a sequence of steps that successfully resolves all optional blocks. This is because we type the content of optional blocks and that our type system ensures that these contents reach exactly one success reporting message $ \POptEnd{\Role}{\tilde{\Args[v]}} $ in exactly one of its parallel branches (and in each of its choice branches).

\begin{corollary}[Reliance]
	\label{cor:reliance}
	For both type systems:\\
	If $ \Gamma \vdash P \triangleright \Delta $ such that $ \Delta $ is initially coherent and $ P $ does not contain infinite recursions and cannot infinitely often call a sub-session, then $ P \longmapsto^* \PEnd $ such that all optional blocks are successfully resolved in this sequence.
\end{corollary}

To summarize our type systems have the following properties.

\begin{theorem}[Properties]
	\label{thm:properties}
	For both type systems:
	\begin{description}
		\item[Subject Reduction:] If $ \Gamma \vdash P \triangleright \Delta $ and $ P \longmapsto P' $ then there exists $ \Delta' $ such that $ \Gamma \vdash P' \triangleright \Delta' $ and $ \Delta \mapsto^* \Delta' $.
		\item[Progress:] If $ \Gamma \vdash P \triangleright \Delta $ such that $ \Delta $ is initially coherent, then either $ P = \PEnd $ or there exists $ P' $ such that $ P \longmapsto^+ P' $, $ \Gamma \vdash P' \triangleright \Delta' $, where $ \Delta \mapsto^* \Delta' $ and $ \Delta' $ is coherent, and $ P \longmapsto^+ P' $ does not require any optional block to be unreliable.
		\item[Completion:] If $ \Gamma \vdash P \triangleright \Delta $ such that $ \Delta $ is initially coherent and $ P $ does not contain infinite recursions and cannot infinitely often call a sub-session, then $ P \longmapsto^* \PEnd $, $ \Gamma \vdash \PEnd \triangleright \emptyset $, and $ P \longmapsto^* \PEnd $ does not require any optional block to be unreliable.
		\item[Reliance:] If $ \Gamma \vdash P \triangleright \Delta $ such that $ \Delta $ is initially coherent and $ P $ does not contain infinite recursions and cannot infinitely often call a sub-session, then $ P \longmapsto^* \PEnd $ such that all optional blocks are successfully resolved in this sequence.
	\end{description}
\end{theorem}

$ \PRC{n} $ is well-typed \wrt to the initially coherent session environment $ \Gamma $ that does not contain recursions. Thus, by the completion property of Theorem~\ref{thm:properties}, our implementation $ \PRC{n} $ of the rotating coordinator algorithm terminates despite arbitrary failures of optional blocks. Note that, although establishing the type system and proving Theorem~\ref{thm:properties} was elaborate, to check whether a process is well-typed is straightforward and can be automated easily and efficiently.

Since all communication steps of the algorithm are captured in optional blocks and since failure of optional blocks containing a single communication step represents a link failure/message loss, $ \PRC{n} $ terminates despite arbitrary occurrences of link failures.

Session types usually also ensure \emph{communication safety}, \ie freedom of communication error, and \emph{session fidelity}, \ie a well-typed process exactly follows the specification described by its global type. With optional blocks we lose these properties, because they model failures. As a consequence communications may fail and whole parts of the specified protocol in the global type might be skipped. In order to still provide some guarantees on the behaviour of well-typed processes, we however limited the effect of failures by encapsulation in optional blocks. It is trivial to see, that in the failure-free case, \ie if no optional block fails, we inherit communication safety and session fidelity from the underlying session types in \cite{BettiniAtall08,BocciAtall10} and \cite{DemangeonHonda12}. Even in the case of failing optional blocks, we inherit communication safety and session fidelity for the parts of protocols outside of optional blocks and the inner parts of successful optional blocks, since our extension ensures that all optional blocks that depend on a failure are doomed to fail and the remaining parts work as specified by the global type.

\subsection{System Failures}

If we use optional blocks the cover a single transmission over an unreliable link, each use of Rule~(\textsf{fail}) refers to a single link failure. Whether a specification, \ie a global type, implements link failures can be checked easily, by analysing whether all communication steps on unreliable links are encapsulated by the above described binary optional blocks $ \GUL{\Role[src]}{\Args[v]_{\Role[src]}}{\Role[trg]}{\Args[v]_{\Role[trg]}} $.
Notice that this way we can model systems that contain reliable as well as unreliable links.

The properties encapsulation, isolation, and safety guarantee that the above described unreliable links meet our intuition of the considered class of failure and their effect. Restricting our attention to link failures, where in the case of failure a default value is provided, as well as the restriction on protocols to compute some values might appear as a rather strong limitation. But this limitation actually matches the intuition used for many distributed algorithms. We consider systems that use some method to determine at which point a certain failure has occurred---\eg by a time out or more abstractly a failure detector. But apart from the detection of the failure, the system does usually not provide any informations about it or its source. We match this intuition by restricting the way the modelled system can react on a failure.

\subsection{Crash Failures}

Crash failures can be considered as a special case of link failures: After the first link failure all communications with the respective sender of the first failure have to fail.
Following this intuition, a system with crash failures can be obtained from a system with link failures by excluding all executions that do not meet the above criterion.
Accordingly, all algorithms that terminate despite link failures also terminate despite crash failures.
There are however algorithms that do not guarantee termination despite link failures but only despite crash failures. Consider once more Example~\ref{exa:RCAlgorithm}. This algorithm satisfies termination despite link failures; but it will not be able to ensure agreement in this scenario, \ie cannot ensure that despite link failures all participants decide consistently \cite{Lynch96}. Agreement despite crash failures is ensured.
Similarly an algorithm might satisfy termination only with respect to a maximal amount of failures or under the assumption that a certain process never fails.

The simplest way to express crash failures with optional blocks is to encapsulate the specification of a whole algorithm in an optional block on all participating roles. Projection then results in local types $ T_i $ for each role $ \Role_i $, that are completely encapsulated by an optional block $ \LTOptS{\tilde{\Role}}{T_i}{\cdot} $.
A process crashes iff its optional block fails. Here we need to encapsulate all communications between the participating roles $ \tilde{\Role} $ in optional blocks of the form $ \GUL{\Role[src]}{\Args[v]_{\Role[src]}}{\Role[trg]}{\Args[v]_{\Role[trg]}} $, \ie have to model all links as unreliable, to ensure that the crash of one process does not doom the whole system. With that the specification of systems that contain both, link and crash failures, is easy. We can also model a process crash with recovery this way, using a recursion $ \LTRec{\TermV}{\LTOpt{\tilde{\Role}}{T_i}{\Typed{\tilde{\Args}}{\tilde{\Sort}}}}{t} $ and the default values $ \tilde{\Args} $ to capture the initial values of the process.
The main difficulty are systems with unreliable processes but reliable links. Here we have to ensure that all communication failures result from a crashed process. An easy way to tackle this problem is to let the reduction semantics keep track of the processes that are crashed or are currently considered alive as it was done \eg in \cite{KuhnrichNestmann09,wagnerNestmann14} or for exceptions in \cite{capecchi2016}. With that the semantics can ensure that a communication error causes a process to crash---or is caused by a crashed process---and that the optional block of a crashed process will eventually fail. The interesting question here is how the type system can be used to guarantee termination in systems with unreliable processes, if the algorithm does not terminate in the presence of arbitrary link failures. Even more challenging is the analysis of algorithms that tolerate only a bounded amount of failures. In the presented approach we concentrate---as a first step---on link failures/message loss and algorithms that terminate despite arbitrary link failures.

\section{Conclusions}
\label{sec:conclusions}

We extend standard session types with optional blocks with default values. Thereby, we obtain a type system for progress and completion/termination despite link failures that can be used to reason about fault-tolerant distributed algorithms.
Our approach is limited with respect to two aspects: We only cover algorithms that
\begin{inparaenum}[(1)]
	\item allow us to specify default values for all unreliable communication steps and
	\item terminate despite arbitrary link failures.
\end{inparaenum}
Accordingly, this approach is only a first step towards the analysis of distributed algorithms with session types. It shows however that it is possible to analyse distributed algorithms with session types and how the latter can solve the otherwise often complicated and elaborate task of proving termination.
Note that, optional blocks can contain larger parts of protocols than a single communication step. Thus they may also allow for more complicated failure patterns than simple link failures/message loss.

In \cite{adameitPetersNestmann17} we extend a simple type system with optional blocks. The (for many distributed algorithms interesting) concept of rounds is obtained instead by using the more complicated nested protocols (as defined in \cite{DemangeonHonda12}) with optional blocks. Due to lack of space, the type systems with nested protocols/sub-sessions and optional blocks as well as more interesting examples with and without explicit (and of course overlapping) rounds were postponed to this report. As presented above the inclusion of sub-session is straightforward and does not require to change the concept of optional blocks as presented in \cite{adameitPetersNestmann17}.
In combination with sub-sessions our attempt respects two important aspects of fault-tolerant distributed algorithms:
\begin{inparaenum}[(1)]
	\item The modularity as \eg present in the concept of rounds in many algorithms can be expressed naturally, and
	\item the model respects the asynchronous nature of distributed systems such that messages are not necessarily delivered in the order they are sent and the rounds may overlap.
\end{inparaenum}

Our extension offers new possibilities for the analysis of distributed algorithms and widens the applicability of session types to unreliable network structures.
We hope to inspire further work in particular to cover larger classes of algorithms and system failures.

\newcommand*{\doi}[1]{\href{http://dx.doi.org/#1}{doi: #1}}
\bibliographystyle{plainnat}
\bibliography{SessionTypesForLinkFailures}

\begin{thebibliography}{17}
\providecommand{\natexlab}[1]{#1}
\providecommand{\url}[1]{\texttt{#1}}
\expandafter\ifx\csname urlstyle\endcsname\relax
  \providecommand{\doi}[1]{doi: #1}\else
  \providecommand{\doi}{doi: \begingroup \urlstyle{rm}\Url}\fi

\bibitem[Adameit et~al.(2017)Adameit, Peters, and
  Nestmann]{adameitPetersNestmann17}
Manuel Adameit, Kirstin Peters, and Uwe Nestmann.
\newblock {Session Types for Link Failures}.
\newblock In Ahmed Bouajjani and Alexandra Silva, editors, \emph{Proceedings of
  FORTE}, LNCS. Springer, 2017.
\newblock To appear.

\bibitem[Bettini et~al.(2008)Bettini, Coppo, D{'{}}Antoni, Luca,
  Dezani-Ciancaglini, and Yoshida]{BettiniAtall08}
Lorenzo Bettini, Mario Coppo, Loris D{'{}}Antoni, Marco~De Luca, Mariangiola
  Dezani-Ciancaglini, and Nobuko Yoshida.
\newblock {Global Progress in Dynamically Interleaved Multiparty Sessions}.
\newblock In Franck van Breugel and Marsha Chechik, editors, \emph{Proceedings
  of CONCUR}, volume 5201 of \emph{LNCS}, pages 418--433. Springer, 2008.
\newblock \doi{10.1007/978-3-540-85361-9\_33}.

\bibitem[Bocchi et~al.(2010)Bocchi, Honda, Tuosto, and Yoshida]{BocciAtall10}
Laura Bocchi, Kohei Honda, Emilio Tuosto, and Nobuko Yoshida.
\newblock {A Theory of Design-by-Contract for Distributed Multiparty
  Interactions}.
\newblock In Paul Gastin and Fran\c{c}is Laroussinie, editors,
  \emph{Proceedings of CONCUR}, volume 6269 of \emph{LNCS}, pages 162--176.
  Springer, 2010.
\newblock \doi{10.1007/978-3-642-15375-4\_12}.

\bibitem[Boudol(1992)]{boudol92}
Gérard Boudol.
\newblock {Asynchrony and the $\pi$-calculus}.
\newblock Note RR-1702, INRIA, Mai 1992.
\newblock URL \url{https://hal.inria.fr/inria-00076939/}.

\bibitem[Capecchi et~al.(2016)Capecchi, Giachino, and Yoshida]{capecchi2016}
Sara Capecchi, Elena Giachino, and Nobuko Yoshida.
\newblock {Global escape in multiparty sessions}.
\newblock \emph{Mathematical Structures in Computer Science}, 26\penalty0
  (2):\penalty0 156--205, 2016.
\newblock \doi{10.1017/S0960129514000164}.

\bibitem[Carbone et~al.(2008)Carbone, Honda, and
  Yoshida]{CarboneHondaYoshida08}
Marco Carbone, Kohei Honda, and Nobuko Yoshida.
\newblock {Structured Interactional Exceptions in Session Types}.
\newblock In Franck van Breugel and Marsha Chechik, editors, \emph{Proceedings
  of CONCUR}, volume 5201 of \emph{LNCS}, pages 402--417. Springer, 2008.
\newblock \doi{10.1007/978-3-540-85361-9\_32}.

\bibitem[Demangeon(2015)]{Demangeon15}
Romain Demangeon.
\newblock {Nested Protocols in Session Types}.
\newblock Personal communication about an extended version of
  \cite{DemangeonHonda12} that is currently prepared by R. Demangeon., 2015.

\bibitem[Demangeon and Honda(2012)]{DemangeonHonda12}
Romain Demangeon and Kohei Honda.
\newblock {Nested Protocols in Session Types}.
\newblock In Maciej Koutny and Irek Ulidowski, editors, \emph{Proceedings of
  CONCUR}, volume 7454 of \emph{LNCS}, pages 272--286. Springer, 2012.
\newblock \doi{10.1007/978-3-642-32940-1\_20}.

\bibitem[Honda and Tokoro(1991)]{hondaTokoro91}
Kohei Honda and Mario Tokoro.
\newblock {An Object Calculus for Asynchronous Communication}.
\newblock In Pierre America, editor, \emph{Proceedings of ECOOP}, volume 512 of
  \emph{LNCS}, pages 133--147. Springer, 1991.
\newblock \doi{10.1007/BFb0057019}.

\bibitem[Kouzapas et~al.(2014)Kouzapas, Gutkovas, and
  Gay]{KouzapasGutkovasGay14}
Dimitrios Kouzapas, Ram\={u}nas Gutkovas, and Simon~J. Gay.
\newblock {Session Types for Broadcasting}.
\newblock In Alastair~F. Donaldson and Vasco~T. Vasconcelos, editors,
  \emph{Proceedings of PLACES}, volume 155 of \emph{EPTCS}, pages 25--31, 2014.
\newblock \doi{10.4204/EPTCS.155.4}.

\bibitem[K{\"u}hnrich and Nestmann(2009)]{KuhnrichNestmann09}
Morten K{\"u}hnrich and Uwe Nestmann.
\newblock {On Process-Algebraic Proof Methods for Fault Tolerant Distributed
  Systems}.
\newblock In David Lee, Ant\'{o}nia Lopes, and Arnd Poetzsch-Heffter, editors,
  \emph{Proceedings of FORTE}, volume 5522 of \emph{LNCS}, pages 198--212,
  2009.
\newblock \doi{10.1007/978-3-642-02138-1\_13}.

\bibitem[Lynch(1996)]{Lynch96}
Nancy~A. Lynch.
\newblock \emph{{Distributed Algorithms}}.
\newblock Morgan Kaufmann, 1996.

\bibitem[Milner et~al.(1992)Milner, Parrow, and Walker]{milnerParrowWalker92}
Robin Milner, Joachim Parrow, and David Walker.
\newblock {A Calculus of Mobile Processes, Part I and II}.
\newblock \emph{Information and Computation}, 100\penalty0 (1):\penalty0 1--77,
  1992.
\newblock \doi{10.1016/0890-5401(92)90008-4}.

\bibitem[Palamidessi(2003)]{palamidessi03}
Catuscia Palamidessi.
\newblock {Comparing the Expressive Power of the Synchronous and the
  Asynchronous $\pi$-calculi}.
\newblock \emph{Mathematical Structures of Computer Science}, 13\penalty0
  (5):\penalty0 685--719, 2003.
\newblock \doi{10.1017/S0960129503004043}.

\bibitem[Peters and Nestmann(2012)]{fossacs12_pi}
Kirstin Peters and Uwe Nestmann.
\newblock {Is it a ``Good'' Encoding of Mixed Choice?}
\newblock In Lars Birkedal, editor, \emph{Proceedings of FoSSaCS}, volume 7213
  of \emph{LNCS}, pages 210--224. Springer, 2012.
\newblock \doi{10.1007/978-3-642-28729-9\_14}.

\bibitem[Tel(1994)]{Tel94}
Gerard Tel.
\newblock \emph{{Introduction to Distributed Algorithms}}.
\newblock Cambridge University Press, 1994.

\bibitem[Wagner and Nestmann(2014)]{wagnerNestmann14}
Christoph Wagner and Uwe Nestmann.
\newblock {States in Process Calculi}.
\newblock In Johannes Borgstr\"{o}m and Silvia Crafa, editors,
  \emph{Proceedings of EXPRESS/SOS}, volume 160 of \emph{EPTCS}, pages 48--62,
  2014.
\newblock \doi{10.4204/EPTCS.160.6}.

\end{thebibliography}

\end{document}